\documentclass[12pt,letter]{article}
\usepackage[utf8]{inputenc}
\usepackage[margin=1in]{geometry}
\usepackage{anyfontsize}
\usepackage{color}
\usepackage{mathrsfs}
\usepackage{amsmath}
\usepackage{amsthm}
\usepackage{amssymb}
\usepackage{cancel}
\usepackage{mathdots}
\usepackage{wasysym}
\usepackage{enumerate}
\usepackage{float}
\usepackage{verbatim}
\usepackage[authoryear]{natbib}
\usepackage[unicode=true,
 bookmarks=true,bookmarksnumbered=true,bookmarksopen=true,bookmarksopenlevel=1,
 breaklinks=true,pdfborder={0 0 1},backref=section,colorlinks=true]{hyperref}
\hypersetup{citecolor=blue,linkcolor=blue,urlcolor=blue}

\makeatletter
\usepackage{amsfonts}
\usepackage{graphicx}
\usepackage{color}
\usepackage{rotating}
\usepackage{tikz}
\providecommand{\U}[1]{\protect\rule{.1in}{.1in}}

\newtheorem{theorem}{Theorem}[section]

\newtheorem{assumption}{Assumption}

\newtheorem{lemma}[theorem]{Lemma}

\newtheorem{proposition}[theorem]{Proposition}
\newtheorem{remark}{Remark}[section]

\renewcommand{\baselinestretch}{1.2}

\@addtoreset{lemma}{Appendix B} 

\numberwithin{equation}{section}

\usepackage{booktabs}
\usepackage{tabularx}
\usepackage{dcolumn}
\usepackage{array}
\makeatletter
\newlength\mylena
\newlength\mylenb
\newcommand\mystrut[1][2]{%
    \setlength\mylena{#1\ht\@arstrutbox}%
    \setlength\mylenb{#1\dp\@arstrutbox}%
    \rule[\mylenb]{0pt}{\mylena}}
\makeatother
\newcolumntype{C}[1]{>{\centering\let\newline\\\arraybackslash\hspace{0pt}}p{#1}}
\newcolumntype{x}{>{\centering\let\newline\\\arraybackslash\hspace{0pt}}X}
\newcolumntype{R}[1]{>{\raggedleft\let\newline\\\arraybackslash\hspace{0pt}}p{#1}}
\newcolumntype{y}{>{\raggedleft\let\newline\\\arraybackslash\hspace{0pt}}X}
\newcolumntype{L}[1]{>{\raggedright\let\newline\\\arraybackslash\hspace{0pt}}p{#1}}
\newcolumntype{z}{>{\raggedright\let\newline\\\arraybackslash\hspace{0pt}}X}
\newcolumntype{d}[1]{D..{#1}}  

\usepackage{multirow,lmodern}

\usepackage{threeparttable}

\usepackage{etoolbox}
\BeforeBeginEnvironment{tabularx}{\centering\footnotesize}
\AfterEndEnvironment{tabularx}{}             
\BeforeBeginEnvironment{tabular*}{\centering\footnotesize}
\AfterEndEnvironment{tabular*}{}             
\usepackage{times} 

\makeatother

\begin{document}

\def\spacingset#1{\renewcommand{\baselinestretch}%
{#1}\small\normalsize} 
\spacingset{1}

\title{Panel Stochastic Frontier Models with Latent Group Structures\thanks{We thank the editor, Michal Koles\'ar, an Associate Editor, and two
anonymous referees for their helpful comments. We are grateful for
Bin Peng, Valentin Zelenyuk, and the participants of the Econometric Society
Australasian Meeting 2024 and AE$^2$ conference 2025 for their helpful feedback. The first author acknowledges that this research was supported by the Commonwealth through an Australian Government Research Training Program Scholarship \href{https://doi.org/10.82133/C42F-K220}{[DOI: https://doi.org/10.82133/C42F-K220]}, as well as the Japan-IMF Scholarship.  The second author gratefully acknowledges the financial support of the
National Natural Science Foundation of China (No. 72394392 and No. 72573158).}}
\author{{Kazuki Tomioka}\thanks{Graduate School of Humanities and Social Sciences, Hiroshima University, Hiroshima, Japan. Email: \href{mailto://ktomioka1@hiroshima-u.ac.jp}{\texttt{ktomioka1@hiroshima-u.ac.jp}}.}\\
 Hiroshima University \and {Thomas T. Yang}\thanks{Corresponding author. Research School of Economics, The Australian
National University, Canberra, ACT 2601, Australia. Email: \href{mailto:tao.yang@anu.edu.au}{\texttt{tao.yang@anu.edu.au}}.}\\
 Australian National University \and { Xibin Zhang}\thanks{Department of Econometrics and Business Statistics, Monash University,
Caulfield East, Victoria 3145, Australia. Email: \href{mailto:xibin.zhang@monash.edu}{\texttt{xibin.zhang@monash.edu}}.}\\
 Monash University}

\maketitle
\vspace{0.5cm}

\begin{abstract}

\noindent Stochastic frontier models have attracted considerable attention due to the incorporation of an inefficiency term in addition to the conventional error term. In this paper, we propose a general estimation framework for panel stochastic frontier models that accommodates potential heterogeneity through latent group structures. The framework is tailored to the distinctive features of stochastic frontier models and is paired with a practical hybrid estimation procedure that combines individual-level and joint panel estimation. We illustrate the estimation framework using a panel stochastic frontier model that treats the inefficiency term as a random effect, and show that it can be readily extended to a range of fixed effects specifications common in the literature. Simulation studies indicate strong finite-sample performance, and we further demonstrate the practicality of the approach in an empirical application to the cost efficiency of the U.S. commercial banking sector.

\end{abstract}
\begin{description}
\item [{Keywords:}] Classification, Group Structures, Panel Data, Stochastic
Frontier 
\item [{\emph{JEL classification:}}] C23, C33, C38, C51 
\end{description}

\pagebreak{}

\newpage\spacingset{1.8} 
\section{Introduction}

\citet{Aigneretal1977} and \citet{MeeusenvanDenBroeck1977} introduced
the stochastic frontier (SF) model to study the productive (in)efficiency
of firms, and such SF models have since attracted considerable attention.\footnote{We refer readers to \citet{KumbhakarLovell2000} for early developments and \citet{Kumbhakaretal2022a} and \citet{Tsionas2023} for recent advancements and a comprehensive review of the literature.} One distinct feature of a SF model is the decomposition of the error term, typically expressed as $\varepsilon=v-u$, where $v$ is
a random disturbance and $u\geq0$ is the inefficiency term.
Unlike standard regression models, the identification of $v$ and
$u$ is crucial in SF modeling because of their economic interpretations.

In this paper, we develop a general estimation framework for time-varying panel SF models that can accommodate potential heterogeneity across firms through latent group structures. To illustrate our methodology, we consider a model similar to \citet{YaoZhangKum2019}'s that allows for heterogeneous time-varying coefficients: 
\begin{align}
    \label{EQ:intromodel}
    y_{it} &=\alpha_{i}^{0}+\alpha_{i}(\tau_{t})+x_{it}^{\prime}\beta_{i}(\tau_{t})+\varepsilon_{it},\quad\text{with}\quad\varepsilon_{it}=v_{it}-u_{i}, \nonumber\\
    &= \alpha_{i}^{0}-u_{i}+\alpha_{i}(\tau_{t})+x_{it}^{\prime}\beta_{i}(\tau_{t})+v_{it},
\end{align}
for firm $i=1,2,\dots,N$ and time $t=1,2,\dots,T$, where $\tau_{t}=t/T\in(0,1]$. In this specification of the model, $\alpha_{i}^{0}$ and $\alpha_{i}(\tau_{t})$ denote the constant and time-varying intercepts of the frontier, respectively. The term $\beta_{i}(\tau_{t})$ denotes a non-random, time-varying coefficient vector, $v_{it}$ is a zero-mean random error term, and $u_{i}\geq0$ represents a firm-specific random inefficiency term. In addition to allowing for heterogeneity in the efficient frontier, we permit the variance of $v_{it}$ to vary across firms, and allow for the possibility that $u_{i}$ follows a mixture distribution.

Our study is motivated by the role of heterogeneity in the measurement of the efficient frontier and the inefficiency in panel SF models. Heterogeneity can either shift the efficient frontier or distort the location and scale of inefficiency estimates \citep[see, for example,][]{Galan2014}.  According to \citet{Greene2005a}, the true underlying frontier may include unmeasured firm-specific characteristics that reflect the technology in use. \citet{Greene2005a} was instrumental in expanding SF models to incorporate firm heterogeneity by allowing for either true fixed effects or true random effects. Building on this, heterogeneity was further explored in subsequent work by allowing the inefficiency term to be purely transient or to contain both transient and persistent components \citep[see, for example][]{Colombi2014,Kumbhakar2014,Tsionas2014}.
These approaches help disentangle unobserved heterogeneity from inefficiency. However, a homogeneous frontier implicitly assumes that all firms operate under the same technology, so any systematic deviation is attributed to inefficiency.  As emphasized in \citep{Greene2005a, Greene2005},  failing to distinguish the frontiers and inefficiency leads to a structural misinterpretation: persistent heterogeneity is absorbed into the inefficiency term, resulting in biased efficiency estimates.

The framework we propose addresses this issue by introducing a latent group structure for frontier parameters. Firms are partitioned into a small number of groups, with each group sharing a common frontier that reflects a particular technological regime (e.g., distinct business models or scale-specific production technologies). The finite collection of latent frontiers provides a middle ground between a fully homogeneous frontier and unrestricted firm-specific frontiers.

Formally, we assume that firms can be classified into one of $K^{*}\geq1$ groups. Within each group, firms share a common set of parameters $\left\{ \alpha(\tau_{t}),\beta(\tau_{t}),\text{Var}(v_{t})\right\}$. We show that the classification step is consistent, ensuring that firms are benchmarked against the correct regime with probability approaching one.
Importantly, we do not impose the same group structure for the distribution of the inefficiency term, $u$ or on the constant term, $\alpha^{0}$. We find that defining group membership for $\{\alpha^{0},u\}$ in the
same way as for the other parameters is inappropriate, the reasons of which we discuss in detail
in Section~\ref{SEC:inefficiency}. Instead, we account for potential heterogeneity
in $\alpha^{0}-u$ by modeling it as a mixture of distributions.

The idea of uncovering latent group structures in panel data
models has been extensively studied in recent years. Existing methods can be broadly categorized into two main approaches. The first  approach relies on a two-step procedure in which individual-level estimation is followed by applying a clustering algorithm such as K-means \citep{LinNg2012, BonhommeManresa, AndoBai2016}, or Hierarchical Agglomerative Clustering (HAC) \citep{Chen2019}.
The second approach, introduced by
\citet{SuEtal2016} performs simultaneous estimation and classification by penalizing the joint mean squared errors. Further developments along this line include \citet{SuEtal2019}, \citet{HuangEtal2020} and \citet{WangSu2021}. Conceptually, the latter approach pools all observations
across firms for joint estimation and classification, whereas the
former relies on individual (non-pooled) firm-specific estimates as the basis of classification.

These existing approaches cannot be applied directly in our setting for the following reason. Estimating the parameters that govern the underlying distribution of $u$ requires pooling observations across firms, typically via maximum likelihood based on the likelihood function in (\ref{EQ:likelihood_appro}). However, since the log-likelihood function is complex and highly nonlinear, applying the method of \citet{SuEtal2016} that requires simultaneous estimation and classification for SF models is nontrivial. This leads to a dilemma of whether to pool or not to pool observations for inference. To resolve this, we propose a new hybrid approach that combines pooled and non-pooled estimation methods that is flexible enough to be applied to a broad class of SF models.


Our paper makes several contributions to the literature on SF models
and classification methods for panel data models. First, we estimate
a robust panel SF model that incorporates heterogeneity across firms.
We carefully set out the framework, clarifying what is feasible and what is not, and provide detailed explanations of the underlying rationale.\footnote{Previous work on robust estimation of the SF model has primarily focused on semi-parametric or non-parametric specifications for either the efficient frontier \citep{ParkSimar1994, YaoZhangKum2019} or the error term distribution \citep{Greene2005, LaiKumbhakar2023}. Additionally, there is a growing literature on specification tests for the distribution of inefficiency \citep{Chengetal2024}.} Second, to the best of our knowledge, this is one of the few papers in the literature to model and allow the variance of the error term to exhibit latent group structures. A notable
study in this area is \citet{LoyoBoot}, where they focus on modeling the variance of the error term, primarily for efficiency gains and/or the specific role played by the variance. In contrast, we aim to uncover heterogeneity in both the error and inefficiency terms, for the economic interpretations given in SF models. We emphasize the importance of modeling heterogeneity in both the variance of the error term, $v$, and the distribution of inefficiency, $u$, since inefficiency estimates depend on the variance of $v$ in the widely used \citet{Jondrowetal1982}'s point estimator. Given the importance of the joint modeling of both the variance of $v$ and the distribution of $u$ in SF models, we consider this to be a substantial contribution. Moreover, we propose a hybrid estimation procedure that combines firm-level and joint panel estimation to address the potential heterogeneity present in both the frontier and error components. 
As the description of the SF model in (\ref{EQ:intromodel}) alluded, we present the hybrid approach in the main text using a random effects specification of the SF model, and show in Appendix \ref{APP:FE} how it extends to other variants, including the fixed effects SF models studied by \citet{Greene2005a, Greene2005, Chenetal2014}, and \citet{ZhouEtal2020}.  Finally, the separation between the frontier and inefficiency may be complicated by model misspecification. We rule out this possibility by assuming correct model specification. Developing procedures that are more robust to misspecification is an important direction for future research.


The rest of this paper is organized as follows. In Section~\ref{SEC:model}, we introduce the estimation procedure and propose information criteria to determine the number of groups and whether $u$ follows a unique or a mixture of distributions. Section~\ref{SEC:Theoretical} examines the theoretical properties of our procedure, specifying the conditions required for tuning parameters. In Section~\ref{SEC:simulations}, we evaluate the small-sample performance of our method, providing practical recommendations for tuning parameters that meet the conditions and perform well in simulations. In Section \ref{SEC:application}, we apply our procedure to a dataset of large U.S. commercial banks, using the recommended tuning parameters from the simulations. Our findings confirm the presence of heterogeneity in the frontiers and a mixture distribution for the inefficiency term. The Online Appendix includes several important supplementary results. Appendix \ref{APP:FE1} extends our approach to fixed effects SF models under distributional assumptions. Appendix \ref{APP:FE2} further generalizes this to fixed effects SF models without any distributional restrictions, allowing for more flexible conditions on the inefficiency term. Appendix \ref{APP:innocuous} discusses a relaxation of the normalization condition. Appendix \ref{APP:ZISF} addresses the possibility that the inefficiency term equals zero with positive probability. Appendix \ref{APP:more_mix} considers the mixture distribution with more than two components. Other parts of the Appendix support the results in the main body of the paper. Specifically, Appendices \ref{APP:likeli} and \ref{APP:HAC} describe, respectively, the approximate likelihood function of the model and the clustering method used. Proofs of the theorems and propositions from Section \ref{SEC:Theoretical} are provided in Appendix~\ref{APP:mainproofs}, with technical lemmas presented in Appendix \ref{APP:lemmaProofs}. Additional simulation and application results are included in Appendix~ \ref{APP:Tables}, and Appendix~\ref{App:derivatives} further discusses the properties of the information matrix in support of Appendix~\ref{APP:mainproofs}. 

\section{Model and Estimation}
\label{SEC:model} 
This section presents the model and the procedure. To facilitate exposition,
we assume $\alpha_{i}^{0}=\alpha^{0}$ and $\text{Var}(u_{i}) = \sigma_{ui}^{2} = \sigma_{u}^{2}$ for all $i$ in Sections \ref{sec:The-Model-More} to \ref{SEC:Estimation}. These restrictions are relaxed to the general case subsequently. The technical conditions and main theorems are postponed to the next section.

\subsection{The Model}\label{sec:The-Model-More}
We illustrate our approach using the panel SF model with random effects (RE) and relegate other variants of the model to Appendix \ref{APP:FE}. The particular model we consider modifies \citet{YaoZhangKum2019}'s by incorporating heterogeneous, time-varying coefficients: 
\begin{align}
    y_{it} &=\alpha^{0} + \alpha_{i}\left(\tau_{t}\right)+x_{it}^{\prime}\beta_{i}\left(\tau_{t}\right)+\varepsilon_{it}\nonumber \\
    &=\alpha^{0}-u_{i}+\alpha_{i}\left(\tau_{t}\right)+\sum_{l=1}^{p}x_{itl}\beta_{il}\left(\tau_{t}\right)+v_{it},\label{EQ:model}
\end{align}
noting the temporal restriction $\alpha_{i}^{0}=\alpha^{0}$,
for expositional purposes to be generalized subsequently.
The subscript $l$ denotes the $l$-th element of a vector $x_{it}$, which
is $p$ by $1$. In this specification, $\varepsilon_{it}=v_{it}-u_{i}$,
where $v_{it}$ is a mean-zero random error term, and $u_{i}$ is
a non-negative term capturing the inefficiency of firm $i$.\footnote{We describe the model, estimation method, and simulations in terms
of the production frontier model, but use the cost frontier model
for our application. The only difference is that the inefficiency
term, $u_{i}$, enters the model negatively (production frontier)
or positively (cost frontier). This distinction is minor, and one
can let $\varepsilon_{it}\equiv v_{it}+u_{i}$ for cost frontiers.} We let $\alpha_{i}(\cdot)$ and 
$\beta_{i}(\cdot)$ evolve smoothly in $\tau_{t}$ to capture the gradual technological drift, regulatory cycles, and business-model adjustments that are pervasive in long panels (for example, banking costs). Smoothness avoids implausible jumps while allowing flexible, low-frequency changes that pooled static frontiers cannot capture. 

As in \citet{YaoZhangKum2019}, we assume that 
\begin{equation}
\label{why.RE}
    v_{it}\sim N(0,\sigma_{vi}^{2}), \quad u_{i}\sim|N(0,\sigma_{u}^{2})|, \quad\text{and} \quad v_{it}\perp u_{i}\perp x_{it}.
\end{equation}
We restrict $\sigma_{u}^{2}$ to be identical for all $i$ (no group structure) momentarily, again to facilitate exposition. We note that the two parameters where homogeneity is imposed, $\alpha^{0}$ and $\sigma_{u}^{2}$, exhibit distinctive features compared to other parameters. Given the importance of these parameters in SF models, we devote a separate section to discuss these differences in detail in Section \ref{SEC:inefficiency}.

The assumption given by (\ref{why.RE}) implies that inefficiency arises from managerial, organizational or behavioral factors that are unrelated to observed input or output variables in the frontier. The resulting panel SF model is RE in the sense of \citet{Greene2005a}. In panel SF models, RE is attractive because fixed effects (FE) type likelihoods face an incidental-parameters problem in finite $T$. RE avoids this by treating unit effects probabilistically, which \citet{Tsionas2014} emphasize when arguing for a fully likelihood-based/Bayesian route for panel SF model analysis with multiple error components. We adopt this setup primarily to illustrate our proposed approach, although the methodology can be extended to alternative model specifications, such as the four-component panel stochastic frontier model introduced by \citet{Tsionas2014} and \citet{LaiKumbhakar2023}. FE-type models are discussed separately in Appendices \ref{APP:FE1} and \ref{APP:FE2}, as their treatment is relatively straightforward given the procedure developed in the main body of the paper.

We assume that there are $K^{\ast}\geq1$ groups of parameters, and
each firm's parameters belong to one of these groups. Mathematically,
\begin{equation}
    \left\{ \alpha_{i}\left(\tau_{t}\right),\beta_{i}\left(\tau_{t}\right),\sigma_{vi}\right\} =\sum_{k=1}^{K^{\ast}}\left\{ \alpha_{(k)}^{\ast}\left(\tau_{t}\right),\beta_{(k)}^{\ast}\left(\tau_{t}\right),\sigma_{v(k)}^{\ast}\right\} \boldsymbol{1}(i\in G_{k}),\label{EQ:group_para}
\end{equation}
where $\boldsymbol{1}(\cdot)$ is the indicator function, equaling
1 if $(\cdot)$ is true and 0 otherwise. Additionally, parameters from
different groups are distinct, meaning 
$
\left\{ \alpha_{(k)}^{\ast}\left(\tau_{t}\right), \beta_{(k)}^{\ast}\left(\tau_{t}\right),\sigma_{v(k)}^{\ast}\} \neq\{ \alpha_{(j)}^{\ast}\left(\tau_{t}\right),\beta_{(j)}^{\ast}\left(\tau_{t}\right),\sigma_{v(j)}^{\ast}\right\} ,
$ for $j\neq k$. The group membership sets satisfy 
$G_{j}\cap G_{k}=\emptyset\text{ and }\bigcup_{k=1}^{K^{\ast}}G_{k}=\left\{ 1,2,\dots,N\right\}$.

It is worth noting that we impose the following assumption on each $\alpha_{(k)}^{*}(s)$:
\begin{equation}
    \int_{0}^{1}\alpha_{(1)}^{\ast}(s)\,\textrm{d}s = \int_{0}^{1}\alpha_{(2)}^{\ast}(s)\,\textrm{d}s =\ldots=\int_{0}^{1}\alpha_{(K^{\ast})}^{\ast}(s)\, \textrm{d}s,\label{eq:restriction_level}
\end{equation}
although we generalize it to allow them to differ in Appendix \ref{APP:innocuous}.\footnote{We explain in detail why we impose this condition, how we adjust the procedure without it, and when we recommend relaxing the condition in Appendix \ref{APP:innocuous}.} The normalization adopted here is $\int_{0}^{1}\alpha_{(k)}^{\ast}(s)\,\textrm{d}s=0$.
Clearly, when $K^{\ast}=1$ (the homogeneous case), this normalization
is innocuous; however, it is not in the general case due to the restriction in (\ref{eq:restriction_level}). When the intercept term does not vary over time, $\alpha_{i}(s)=0$. This normalization ensures that $\alpha_{i}(s)$ captures the time-varying component of the intercept term.

\subsection{Approximation of \texorpdfstring{$\alpha(\cdot)$}{a} and \texorpdfstring{$\beta(\cdot)$}{b} }\label{SEC:serial_approximation}

The approximations we adopt are standard in the literature. Let $L^{2}\left[0,1\right]=\{f\left(s\right):\int_{0}^{1}f^{2}\left(s\right)\text{d}s<\infty\}$
represent the space of square-integrable functions. The inner product equipped on this space is defined as $\left\langle f_{1},f_{2}\right\rangle \equiv\int_{0}^{1}f_{1}\left(s\right)f_{2}\left(s\right)\text{d}s$,
and the induced norm is $\left\Vert f\right\Vert =\left\langle f,f\right\rangle ^{1/2}$. Following \citet{DongLinton2018} and \citet{Ataketal}, we use cosine functions as basis functions. In particular, $B_{0}\left(s\right)=1$ and $B_{j}\left(s\right)=\sqrt{2}\cos(j\pi s)$
for $j\geq1$. The set $\{B_{j}\left(s\right)\}_{j=0}^{\infty}$ then forms
an orthonormal basis for the Hilbert space $L^{2}\left[0,1\right]$,
such that $\left\langle B_{i},B_{j}\right\rangle =\delta_{ij}$, where
$\delta_{ij}$ is the Kronecker delta.

Suppose $f\in L^{2}\left[0,1\right]$ is $\kappa$-th order continuously differentiable. Then, we have 
\begin{align*}
    f\left(s\right) &= \sum_{j=0}^{\infty}B_{j}\left(s\right)v_{j}^{0}=\sum_{j=0}^{m-1}B_{j}\left(s\right)v_{j}^{0}+\sum_{j=m}^{\infty}B_{j}\left(s\right)v_{j}^{0}\\
    &= \sum_{j=0}^{m-1}B_{j}\left(s\right)v_{j}^{0}+O\left(m^{-\kappa}\right)\equiv\mathbb{B}^{m}\left(s\right)^{\prime}v^{0}+O\left(m^{-\kappa}\right),
\end{align*}
where $\mathbb{B}^{m}\left(s\right)\equiv\left(B_{0}\left(s\right),B_{1}\left(s\right),\dots,B_{m-1}\left(s\right)\right)^{\prime}$,
$v_{j}^{0}=\left\langle f,B_{j}\right\rangle $, and $v^{0}=\left(v_{0}^{0},v_{1}^{0},\dots,v_{m-1}^{0}\right)^{\prime}$.
Here, $\sum_{j=m}^{\infty}B_{j}\left(s\right)v_{j}^{0}$ is the bias
term from using only the first $m-1$ terms contained in $\mathbb{B}^{m}\left(s\right)$ to approximate $f\left(s\right)$. If $f\left(s\right)$ is $\kappa$-th order differentiable,
the bias term is $\sum_{j=m}^{\infty}B_{j}\left(s\right)v_{j}^{0} = O\left(m^{-\kappa}\right)$. When $\int_{0}^{1}f\left(s\right)\text{d}s=0$ is imposed, we approximate $f$ using $\mathbb{B}_{-0}^{m}\left(s\right)\equiv\left(B_{1}\left(s\right),\dots,B_{m-1}\left(s\right)\right)^{\prime}$,
since $B_{0}(s)=1$  and $\int_{0}^{1}B_{j}\left(s\right)\text{d}s=0$ for $j\geq1$. Thanks to this property,  it is more convenient than using alternative bases such as B-splines. Similarly,
\[
f\left(s\right)=\mathbb{B}_{-0}^{m}\left(s\right)v_{-0}^{0}+O\left(m^{-\kappa}\right),
\]
for some $v_{-0}^{0}=\left(v_{1}^{0},v_{2}^{0},\dots,v_{m-1}^{0}\right)^{\prime}$. 

We apply this approximation to our case. For each firm $i$, we have
\begin{align}
y_{it} & =\alpha^{0}-u_{i}+\alpha_{i}\left(\tau_{t}\right)+\sum_{l=1}^{p}x_{itl}\beta_{il}\left(\tau_{t}\right)+v_{it}\nonumber \\
 & \approx\alpha^{0}-u_{i}+\mathbb{B}_{-0}^{m}\left(\tau_{t}\right)^{\prime}\pi_{i0}^{0}+\sum_{l=1}^{p}x_{itl}\mathbb{B}^{m}\left(\tau_{t}\right)^{\prime}\pi_{il}^{0}+v_{it}\nonumber \\
 & \equiv\alpha^{0}-u_{i}+\left[\mathbb{B}_{-0}^{m}\left(\tau_{t}\right)^{\prime},\left(x_{it}\otimes\mathbb{B}^{m}\left(\tau_{t}\right)\right)^{\prime}\right]\pi_{i}^{0}+v_{it}\nonumber \\
 & \equiv\alpha^{0}-u_{i}+z_{it}^{\prime}\pi_{i}^{0}+v_{it}\equiv\tilde{z}_{it}^{\prime}\tilde{\pi}_{i}^{0}+v_{it},\label{EQ:approx}
\end{align}
where the terms $\mathbb{B}_{-0}^{m}\left(\tau_{t}\right)^{\prime}\pi_{i0}^{0}$
and $\mathbb{B}^{m}\left(\tau_{t}\right)^{\prime}\pi_{il}^{0}$ represent
approximations of $\alpha_{i}\left(\tau_{t}\right)$ and $\beta_{il}\left(\tau_{t}\right)$,
respectively, and 
\begin{align*}
x_{it}\otimes\mathbb{B}^{m}\left(\tau_{t}\right) & \equiv\left(x_{it1}B_{0}\left(\tau_{t}\right),\dots,x_{it1}B_{m-1}\left(\tau_{t}\right),\dots,x_{itp}B_{0}\left(\tau_{t}\right),\dots,x_{itp}B_{m-1}\left(\tau_{t}\right)\right)^{\prime},\\
\pi_{i}^{0} & \equiv\left(\pi_{i0}^{0\prime},\pi_{i1}^{0\prime},\dots,\pi_{ip}^{0\prime}\right)^{\prime},\quad\tilde{\pi}_{i}^{0}\equiv\left(\alpha^{0}-u_{i},\pi_{i0}^{0\prime},\pi_{i1}^{0\prime},\dots,\pi_{ip}^{0\prime}\right)^{\prime},\\
z_{it} & \equiv\left[\mathbb{B}_{-0}^{m}\left(\tau_{t}\right)^{\prime},\left(x_{it}\otimes\mathbb{B}^{m}\left(\tau_{t}\right)\right)^{\prime}\right]',\quad\text{and}\quad\tilde{z}_{it}\equiv\left(1,z_{it}^{\prime}\right)'.
\end{align*}
The last two lines of equation (\ref{EQ:approx}) represent three
equivalent ways of expressing the approximation.

\subsection{The Estimation}\label{SEC:Estimation}

The variance of inefficiency, $\sigma_{u}^{2}$ is identified
through the skewness in the distribution of $\alpha^{0}-u_{i}$ across
$i$. As a result, $\sigma_{u}^{2}$ cannot be identified or estimated
without pooling observations across different $i$. However, pooling observations for estimation introduces challenges for numerical optimization, since the log-likelihood functions in (\ref{EQ:likelihood_appro}) and (\ref{EQ:log_approx_p}) are complex and highly nonlinear. This issue exacerbates as the number of unknown parameters increases and becomes particularly pronounced in pooled
estimation with classification methods. This creates a dilemma regarding
whether to pool observations for estimation.

To address these challenges, we propose a hybrid procedure that combines
estimations with and without pooling. We present the detailed steps of the procedure below. 

\subsubsection*{Step 1: Individual Estimation}

Using the approximation from (\ref{EQ:approx}), we regress $y_{it}$
against $\mathbb{B}^{m}\left(\tau_{t}\right)$ and $x_{it}\otimes\mathbb{B}^{m}\left(\tau_{t}\right)$
for $t=1,2,\ldots,T$ to obtain $\widehat{\tilde{\pi}}_{i}$. From
this we obtain $\hat{\sigma}_{vi}^{2}$ as the sample variance of
the regression residuals.

Specifically, consider the following expression: $Z_{im}\equiv (\tilde{z}_{i1},...,\tilde{z}_{iT})'$, a $T\times m(p+1)$ vector. Then the OLS estimator for firm $i$ is given by 
\begin{equation}
\widehat{\tilde{\pi}}_{i}=\left(Z_{im}^{\prime}Z_{im}\right)^{-1}Z_{im}^{\prime}y_{i},\label{eq:pihat}
\end{equation}
with $y_{i}=\left(y_{i1},\ldots,y_{iT}\right)^{\prime}$. We then
obtain an estimate of $\sigma_{vi}^{2}$ as $\hat{\sigma}_{vi}^{2}=\frac{1}{T-1}\sum_{t=1}^{T}\left(y_{it}-\tilde{z}_{it}'\widehat{\tilde{\pi}}_{i}\right)^{2}.$ Excluding the first element in $\widehat{\tilde{\pi}}_{i}$, we let
$\hat{\pi}_{i}$ denote the estimated coefficients associated with
$z_{it}$. The estimates $\hat{\pi}_{i}$ and $\hat{\sigma}_{vi}$
are collected to form an estimate of $\vartheta_{i}$: 
$\hat{\vartheta}_{i}=\left(\hat{\pi}_{i}^{\prime},\hat{\sigma}_{vi}\right)^{\prime},$ based on which we form groups.

\subsubsection*{Step 2: Classification}

Having obtained $\hat{\vartheta}_{1},\hat{\vartheta}_{2},\ldots,\hat{\vartheta}_{N}$
from Step 1, we use the $L_{2}$ norm to measure the distance between
$\hat{\vartheta}_{i}$ and $\hat{\vartheta}_{j}$. Based on this distance
measure, we then apply the classical HAC algorithm to the estimates of each firm's functional coefficient to determine group memberships. The HAC is a widely used algorithm for clustering, and several variants of it are employed in heterogeneous panel data models (see, for e.g., \citet{Chen2019}). Details of the HAC method are provided in Appendix \ref{APP:HAC} and refer the readers to \citet{Everittetal} for a comprehensive treatment. Given a value
for $K$, we apply the HAC to obtain an estimate of the group membership, denoted as $\left(\hat{G}_{1|K},\hat{G}_{2|K},\ldots,\hat{G}_{K|K}\right),$ which forms a partition of the set $\left\{ 1,2,\ldots,N\right\} $.

\subsubsection*{Step 3: Post-Classification Estimation and Determination of $K^{\ast}$}

Within each group, we now have significantly more observations available
for pooling. Recognizing this, we set the number of sieve terms to
$\underline{m}$, which is substantially larger than $m$. Within each estimated group, $\hat{G}_{k|K}$ for $1\leq k\leq K$, we conduct post-classification estimation using standard within-panel data estimation methods. Let $\underline{z}_{it}=\left[\mathbb{B}_{-0}^{\underline{m}}\left(\tau_{t}\right)^{\prime},\left(x_{it}\otimes\mathbb{B}^{\underline{m}}\left(\tau_{t}\right)\right)^{\prime}\right]^{\prime}$ denote the new regressors. At this stage, we do not consider the inefficiency term $\alpha^{0}-u_{i}$. The group specific coefficient is given by
\[
\hat{\pi}_{(k|K)}=\arg\min_{\pi}\sum_{i\in\hat{G}_{k|K}}\sum_{t=1}^{T}\left(\ddot{y}_{it}-\underline{\ddot{z}}_{it}^{\prime}\pi\right)^{2},
\]
where $\ddot{y}_{it}=y_{it}-\frac{1}{T}\sum_{t=1}^{T}y_{it}$, and $\underline{\ddot{z}}_{it}=\underline{z}_{it}-\frac{1}{T}\sum_{t=1}^{T}\underline{z}_{it}.$ The estimate of the variance of $v_{it}$ for group $\hat{G}_{k|K}$ is  
\[
\hat{\sigma}_{v(k|K)}^{2}=\frac{1}{N_{k}(T-1)}\sum_{i\in\hat{G}_{k|K}}\sum_{t=1}^{T}\left(\ddot{y}_{it}-\underline{\ddot{z}}_{it}^{\prime}\hat{\pi}_{(k|K)}\right)^{2},
\]
where $N_{k}=\sharp\{\hat{G}_{k|K}\}$ is the number of elements in
$\hat{G}_{k|K}$. For simplicity, we do not explicitly distinguish
between $\hat{N}_{k}=\sharp\{\hat{G}_{k|K^{*}}\}$ and $N_{k}=\sharp\{G_{k|K^{*}}\}$.
Similarly,   $\hat{\vartheta}_{(k|K)}=\left(\hat{\pi}_{(k|K)}^{\prime},\hat{\sigma}_{v(k|K)}\right)^{\prime}.$

Inspired by the pseudo log-likelihood, we construct an information
criterion to determine the optimal number of groups as follows: 
\begin{align}
\text{IC}(K,\lambda_{NT}) & =\sum_{k=1}^{K}\left\{ N_{k}T\log(\hat{\sigma}_{v(k|K)})+\sum_{i\in\hat{G}_{k|K}}\sum_{t=1}^{T}\frac{\left(\ddot{y}_{it}-\underline{\ddot{z}}_{it}^{\prime}\hat{\pi}_{(k|K)}\right)^{2}}{\hat{\sigma}_{v(k|K)}^{2}}\right\} +\lambda_{NT}K\nonumber \\
 & =\sum_{k=1}^{K}\left\{ N_{k}T\log(\hat{\sigma}_{v(k|K)})+N_{k}(T-1)\right\} +\lambda_{NT}K,\label{eq:IC_group}
\end{align}
where $\lambda_{NT}$ is a suitable penalty term. The optimal number of groups is the minimizer of (\ref{eq:IC_group}) 
\[
\hat{K}(\lambda_{NT})=\arg\min_{K=1,2,\ldots,\bar{K}}\text{IC}(K,\lambda_{NT}),
\]
given a suitable $\bar{K}$. For brevity, we henceforth refer to this as
$\hat{K}$. The final group estimates are then given by 
\[
\hat{\vartheta}_{(k|\hat{K})}=\left(\hat{\pi}_{(k|\hat{K})},\hat{\sigma}_{v(k|\hat{K})}\right),\quad k=1,2,\ldots,\hat{K}.
\]

\subsubsection*{Step 4: Estimation of $\alpha^{0}$ and $\sigma_{u}^{2}$}

We estimate $\alpha^{0}$ and $\sigma_{u}^{2}$ pooling all observations
via maximum likelihood estimation (MLE). Specifically, the estimate is given by 
\[
\left(\hat{\alpha}^{0},\hat{\sigma}_{u}^{2}\right)=\arg\max_{(s,\delta_{u}^{2})}\sum_{k=1}^{\hat{K}}\sum_{i\in\hat{G}_{k|\hat{K}}}\log f\left(y_{i}\mid x_{i};s,\delta_{u}^{2},\hat{\vartheta}_{(k|\hat{K})}\right),
\]
where $y_{i}=\left(y_{i1},\ldots,y_{iT}\right)^{\prime}$, $x_{i}=\left(x_{i1},\ldots,x_{iT}\right)^{\prime}$,
and $f\left(y_{i}\mid x_{i};s,\delta_{u}^{2},\hat{\vartheta}_{(k|\hat{K})}\right)$
is defined in (\ref{EQ:likelihood_appro}), noting that we use the post-classification
estimates of $\vartheta$. Since only two parameters, $\alpha^{0}$
and $\sigma_{u}^{2}$, are being estimated at this stage, the numerical
optimization is straightforward.

\subsection{Inefficiency Term}\label{SEC:inefficiency}

We now consider the general case where $\alpha^{0}$ can be heterogeneous
across $i$, and we will use $\alpha_{i}^{0}$ from this point onward.
Modeling the underlying structure of $\alpha_{i}^{0}-u_{i}$ differs
from that of $\left\{\alpha_{i}\left(\tau_{t}\right),\beta_{i}\left(\tau_{t}\right),\sigma_{vi}^{2}\right\}$
because $u_{i}$ is assumed to be random effects, and $\alpha_{i}^{0}-u_{i}$
naturally varies across $i$, even when $\alpha_{i}^{0}$ is identical.
For this reason, we focus on identifying the distribution of $\alpha_{i}^{0}-u_{i}$
rather than the actual values. While we can uncover the underlying
distribution, consistently estimating group membership remains challenging.

Consider the following example to illustrate this point. Suppose we
have two random variables $\varepsilon_{1}$ and $\varepsilon_{2}$,
with $\varepsilon_{1}\sim0-\left|N(0,2)\right|$ and $\varepsilon_{2}\sim1-\left|N(0,1)\right|$.
If we mix i.i.d. realizations of $\varepsilon_{1}$ and $\varepsilon_{2}$,
such as $\left\{ \varepsilon_{11},\varepsilon_{12},\ldots,\varepsilon_{1n},\varepsilon_{21},\varepsilon_{22},\ldots,\varepsilon_{2n}\right\} $,
it is likely that many $\varepsilon_{1i}$ and $\varepsilon_{2j}$ values lie very close to one another. For example, a small-scale Monte Carlo experiment
with $n=100$ suggests that about 28\% of $\varepsilon_{1i}$ 
have at least one $\varepsilon_{2j}$ within a radius of 0.01. In
such cases, swapping their memberships would likely have a minimal impact
on the likelihood function, complicating their distinct identification
from the data.

Misclassification of group memberships can have a serious impact on
the inefficiency term, unlike parameters at the frontiers, where only
similar frontiers can be misclassified together due to low power or
minor estimation errors. Continuing the previous example, suppose
$\varepsilon_{1i}=0$ (highly efficient with $u_{1i}=0$) is misclassified
as $\varepsilon_{2}$, then the inefficiency term for $\varepsilon_{1i}$
would be calculated as 1 (indicating inefficiency). Conversely, if
$\varepsilon_{2j}=0$ (originally not efficient with $u_{2j}=1$)
is misclassified as $\varepsilon_{1}$, then the inefficiency term
for $\varepsilon_{2j}$ would be calculated as 0 (indicating high
efficiency).

Given these challenges and the serious implications of misclassification,
we adopt a mixture distribution approach as follows. Suppose there exist an integer $\mathcal{K}^{*}\geq1$, such
that with probability $\tau_{j}^{0}$, it is distributed as $\alpha_{(j)}^{0}-\left|N(0,\sigma_{u(j)}^{2})\right|$
for $j=1,2,...,\mathcal{K}^{*}-1$, and with probability $\tau_{\mathcal{K}^{*}}^{0}=1-\tau_{1}^{0}-...-\tau_{\mathcal{K}^{*}-1}^{0}$,
as $\alpha_{(\mathcal{K}^{*})}^{0}-\left|N(0,\sigma_{u(\mathcal{K}^{*})}^{2})\right|$,
where $\left(\alpha_{(j)}^{0},\sigma_{u(j)}^{2}\right)$, $j=1,2,...,\mathcal{K}^{*}$,
are distinct vectors, $0<\tau_{j}^{0}<1,$ $j=1,2,...,\mathcal{K}^{*}-1,$
and $1-\tau_{1}^{0}-...-\tau_{\mathcal{K}^{*}-1}^{0}>0$. When $\mathcal{K}^{*}=1$,
the error distribution is reduced to that of a
unique distribution. The mixture distribution on the inefficiency
term is similar in spirit to the latent class model in \citet{Greene2005}.
However, the latent class model is only a small part of \citet{Greene2005} and so the treatment is very brief. We examine this issue in depth by proposing an information criterion to determine the
number of components, rigorously establish its theoretical properties,
and assess its small-sample performance through simulation studies.

An alternative approach is to model the composite term $\alpha_i^{0}-u_i$ by assuming that they follow a mixture distribution as a whole. However, without additional identifying restrictions, $\alpha_i^{0}$ and $u_i$ cannot be separately identified. As emphasized in the Introduction, separating these components is central to stochastic frontier analysis. For this reason, we keep our current specification.

The potential presence of a mixture distribution significantly alters
the interpretation of the results. With uniquely distributed inefficiency term, we
can remove the subscript $i$ from $\alpha_{i}^{0}$ because $\alpha_{i}^{0}-u_{i}\overset{d}{\sim}\alpha^{0}-\left|N\left(0,\sigma_{u}^{2}\right)\right|$.
A point estimate of $\alpha_{}^{0}-u_{i}$ is
\[
\widehat{\alpha^{0}-u_{i}}=\frac{1}{T}\sum_{t=1}^{T}\left(y_{it}-z_{it}'\hat{\pi}_{i}\right),
\]
where $\hat{\pi}_{i}$ is a sub-vector of $\hat{\tilde{\pi}}_{i}$
defined in (\ref{eq:pihat}) in Step 1. Thus, $\alpha^{0}-u_{i}$ can be estimated consistently as $T\to\infty$, allowing us to rank firms according to inefficiency
because $\alpha^{0}$ is identical across $i$. However, in the case
of a mixture distribution, although we can still consistently estimate
$\widehat{\alpha_{i}^{0}-u_{i}}$ (similar to the above), it is not
possible to rank firms as in the former scenario. This limitation
arises because memberships, or equivalently, the values of $\alpha_{i}^{0}$, cannot
be identified. This observation aligns with the findings for the cross-sectional
case discussed in \citet{Greene2005}.

The presence of a mixture distribution in the distribution of $\alpha_{i}^{0}-u_{i}$
does not impact the estimation of $\vartheta_{i}$ given independence
among $u_{i}$, $v_{it}$, and $x_{it}$. Consequently, Steps 1, 2,
and 3 remain unchanged. Details of the revised Step 4, now referred
to as Step 4', are provided below.

\subsubsection*{Step 4': Estimation of $\alpha_{(j)}^{0},\sigma_{u(j)}^{2},\text{ and }\tau_{j}^{0}$}

We adopt the mixture distribution for $\alpha_{i}^{0}-u_{i}$ as previously
described. Assuming that the inefficiency terms come from $\mathcal{K\geq}1$
distributions, we obtain an estimate of $\left(\alpha_{(1)}^{0},\sigma_{u(1)}^{2},...,\alpha_{(\mathcal{K})}^{0},\sigma_{u(\mathcal{K})}^{2},\tau_{1}^{0},...,\tau_{\mathcal{K}-1}^{0}\right)$
using MLE as follows: 
\begin{align*}
 & \left(\hat{\alpha}_{(1)}^{0},\hat{\sigma}_{u(1)}^{2},...,\hat{\alpha}_{(\mathcal{K})}^{0},\hat{\sigma}_{u(\mathcal{K})}^{2},\hat{\tau}_{1},...,\hat{\tau}_{\mathcal{K}-1}\right)\\
= & \arg\max_{(s,\delta_{u}^{2},\tau)}\sum_{k=1}^{\hat{K}}\sum_{i\in\hat{G}_{k|\hat{K}}}\log\tilde{f}\left(y_{i}\left\vert x_{i};s_{(1)},\delta_{u\left(1\right)}^{2},...,s_{(\mathcal{K})},\delta_{u\left(\mathcal{K}\right)}^{2},\tau_{1},...,\tau_{\mathcal{K}-1},\hat{\vartheta}_{(k|\hat{K})}\right.\right)
\end{align*}
where $\tilde{f}$ is the likelihood function defined in (\ref{EQ:log_approx_p}),
and we incorporate estimates from Step 3, as detailed in Section \ref{SEC:Estimation}.


\subsubsection*{Step 5: Determination of the Distributional Structures of the Inefficiency
Term}

To determine the optimal number of mixtures, we introduce a new information criterion for this task: 
\begin{equation}
\widetilde{\mathrm{IC}}(\mathcal{K},\tilde{\lambda}_{NT})=-\sum_{k=1}^{\hat{K}}\sum_{i\in\hat{G}_{k|\hat{K}}}\log\tilde{f}\left(y_{i}\left\vert x_{i};\hat{\alpha}_{(1)}^{0},\hat{\sigma}_{u(1)}^{2},...,\hat{\alpha}_{(\mathcal{K})}^{0},\hat{\sigma}_{u(\mathcal{K})}^{2},\hat{\tau}_{1},...,\hat{\tau}_{\mathcal{K}-1},\hat{\vartheta}_{(k|\hat{K})}\right.\right)+\mathcal{K}\tilde{\lambda}_{NT},\label{eq:IC_tide2}
\end{equation}
where $\tilde{\lambda}_{NT}$ is a suitable penalty term, and the
estimates are as obtained from Steps 3 and 4'.

The optimal number of mixtures is the minimizer of (\ref{eq:IC_tide2})  
\[
\hat{\mathcal{K}}(\tilde{\lambda}_{NT})=\arg\min_{\mathcal{K}=1,2,\ldots,\bar{\mathcal{K}}}\widetilde{\mathrm{IC}}(\mathcal{K},\tilde{\lambda}_{NT}),
\]
and we write $\hat{\mathcal{K}}$ for short. Finally, the estimated parameters are 
\[
\left(\hat{\alpha}_{(1)}^{0},\hat{\sigma}_{u(1)}^{2},...,\hat{\alpha}_{(\mathcal{\hat{\mathcal{K}}})}^{0},\hat{\sigma}_{u(\mathcal{\hat{\mathcal{K}}})}^{2},\hat{\tau}_{1},...,\hat{\tau}_{\mathcal{\hat{\mathcal{K}}}-1}\right).
\]

\subsection{A Summary of the Estimation Procedure}

The outline of the estimation procedure is as follows: 
\begin{enumerate}
\item Conduct estimations of the frontiers for each firm as described in
Step 1. 
\item Apply the HAC algorithm using the individual estimations
from Step 1 for $K=1,2,\ldots,\bar{K}$. 
\item Use the information criterion in (\ref{eq:IC_group}) from Step 3
to determine the optimal number of groups and group memberships.
Obtain an estimate of the frontier with the determined groups. 
\item Based on the group assignments from Step 3, perform joint estimation
for the inefficiency term as in Step 4'. 
\item Use the information criterion in (\ref{eq:IC_tide2}) to determine
the distribution of $\alpha_{i}^{0}-u_{i}$, as in Step~5. 
\item Collect results. The estimates of the frontiers and the distribution
of the inefficiency term are derived from Steps~3 and 5, respectively. 
\end{enumerate}
\section{Asymptotic Properties} \label{SEC:Theoretical} 

We examine classification consistency in Section \ref{SEC:Classification}.
Subsequently, we discuss the large sample properties of the post-classification
estimators in Section \ref{SEC:post}.

\subsection{Classification}

\label{SEC:Classification}


\begin{assumption}\label{A:mixing} The process $\{(x_{it}',v_{it}),t=1,\ldots,T\}$
is strong mixing with a mixing coefficient $\alpha(j)$ that satisfies
$\alpha(j)\leq C_{\alpha}\rho^{j}$ for some positive $C_{\alpha}$
and $0<\rho<1$, and this holds for $i=1,\ldots,N$. \end{assumption}

\begin{assumption}\label{A:moment} $\max_{1\leq i\leq N,1\leq t\leq T}\text{\emph{E}}\|x_{it}\|^{q}\leq\bar{C}_{x}<\infty,$
and $\max_{1\leq i\leq N,1\leq t\leq T}\text{\emph{E}}|v_{it}|^{q}\leq\bar{C}_{v}<\infty,$
for some $q>4.$ \end{assumption}

\begin{assumption}\label{A:rank} Denote $\tilde{x}_{it}\equiv(1,x_{it}^{\prime})^{\prime}$. Let
$\mu_{\min}$ and $\mu_{\max}$ denote the minimum and maximum eigenvalues
of a matrix, respectively. There exist $\underline{C}_{xx}$ and $\bar{C}_{xx}$ with $0<\underline{C}_{xx} \leq \bar{C}_{xx} <\infty$ such that
\[
0<\underline{C}_{xx}\leq\min_{1\leq i\leq N,1\leq t\leq T}\mu_{\min}[\text{\emph{E}}(\tilde{x}_{it}\tilde{x}_{it}')]\leq\max_{1\leq i\leq N,1\leq t\leq T}\mu_{\max}[\text{\emph{E}}(\tilde{x}_{it}\tilde{x}_{it}^{\prime})]\leq\bar{C}_{xx}<\infty.
\]
\end{assumption}

\begin{assumption}\label{A:coef} For $k=1,2,\ldots,K^{*}$, $\alpha_{(k)}^{*}(s)$,
$\beta_{(k)1}^{*}(s)$, ..., and $\beta_{(k)p}^{*}(s)$ belong to
$L^{2}\left[0,1\right]$ and are $\kappa$ times continuously differentiable.
\end{assumption}

\begin{assumption}\label{A:group_diff} There exists a positive $\underline{C}^{*}$,
such that 
\[
\min_{1\leq j\neq k\leq K^{*}}\left\{ \|\alpha_{(j)}^{*}-\alpha_{(k)}^{*}\|+\sum_{l=1}^{p}\|\beta_{(j)l}^{*}-\beta_{(k)l}^{*}\|+|\sigma_{v(j)}^{*}-\sigma_{v(k)}^{*}|\right\} \geq\underline{C}^{*}>0,
\]
and $\min_{1\leq k\leq K^{*}}\sigma_{v(k)}^{*2}>0.$ \end{assumption}

\begin{assumption}\label{A:tuningPara} (i) $N$ can either be a finite,
or divergent. If $N$ is divergent, $N=O(T^{C})$ for some positive $C$ as $T\rightarrow\infty$. (ii) $m\rightarrow\infty$ as $T\rightarrow\infty.$
$Nm^{q/2+2}(\log N)^{2q}/T^{q/2-1}\rightarrow0,$ with $q$ in
Assumption \ref{A:moment}. \end{assumption}


Assumption \ref{A:mixing} imposes a condition of weak dependence
across $t$, noting that independence across $i$ is not required
for classifications. Assumption \ref{A:moment} requires that $x$ and
$v$ have finite $q$-th moment. Assumption \ref{A:rank} is the classic
full rank condition. Assumption \ref{A:coef} stipulates that the
coefficients are $\kappa$-th order differentiable, a standard condition
for nonparametric or semiparametric estimation. Assumption \ref{A:group_diff}
requires that at least one of the coefficients, including the variance
of $v$, must differ across groups.

Assumption \ref{A:tuningPara} specifies that the moment conditions
must be sufficiently large or that $T$ grows fast enough. $N$ can
be fixed. If $N$ diverges, it cannot be too fast, e.g., at the rate
of $\exp(T)$. In the empirical application,  $(N,T)=(466,80)$
and $466\approx80^{1.4}$, thus any $C\geq1.4$ works in (i). The
most stringent requirements arise from the estimation of the ``design''
matrix $\frac{1}{T}Z_{im}^{\prime}Z_{im}$ with diverging dimensions,
used in $\widehat{\tilde{\pi}}_{i}$ (see (\ref{eq:pihat})); a similar
condition was imposed in \citet{Chen2019}. We take a logarithm of
all covariates before estimation, and $q$ can be reasonably considered
large, e.g., $q\geq8$. If we set $m=T^{1/5}$, Assumption \ref{A:tuningPara}
is satisfied. We do not have the usual bias and variance tradeoff
here, as explained below. The results of Theorem \ref{TH:classify}
are underpinned by the uniform convergence of $\widehat{\tilde{\pi}}_{i}$
without any rate requirement. As a result, it is not necessary to
consider the trade-off between the bias and variance of the estimates
for this aspect, when deciding $m$. Of course, we do need $m\rightarrow\infty$ to ensure the uniform convergence. However, to achieve the optimal convergence rate for the post-classification estimates, this consideration becomes crucial, as reflected in Assumption \ref{A:tuning2} (ii) in the subsequent section. With the aforementioned technical conditions, we demonstrate the consistency
of the classification.

\begin{theorem}\label{TH:classify} Suppose Assumptions \ref{A:mixing}
through \ref{A:tuningPara} hold. Then:

\noindent (i) For any small positive $\epsilon$, 
\[
\Pr\left(\max_{i=1,2,\ldots,N}\left\Vert \hat{\vartheta}_{i}-\vartheta_{i}\right\Vert >\epsilon\right)=o(1);
\]
(ii) Assuming $K=K^{*}$, denote the event 
\[
\mathcal{M}\equiv\left\{ \left(\hat{G}_{1|K^{*}},\hat{G}_{2|K^{*}},\ldots,\hat{G}_{K^{*}|K^{*}}\right)=\left(G_{1|K^{*}},G_{2|K^{*}},\ldots,G_{K^{*}|K^{*}}\right)\right\} ,
\]
then $\Pr(\mathcal{M})\rightarrow1$. \end{theorem}

Theorem \ref{TH:classify} (i) establishes the uniform convergence
of $\hat{\vartheta}_{i}$ for $i=1,2,\ldots,N$ provided $m\to\infty$. Building on this,
Theorem \ref{TH:classify} (ii) demonstrates that the probability
of correct classification approaches 1, provided that $K=K^{*}$.
In the next section, we will argue that the $K$ we choose converges
to $K^{*}$ with probability approaching 1, and we discuss the asymptotic
properties of the post-classification estimates.

We note that significantly fewer assumptions are required for consistency in classification compared to those needed for post-classification and determining the number of groups.  For instance, we do not need independence or weak
dependence across $i$, nor do we require specific distributional assumptions on $v_{it}$ and $u_{i}$.

\subsection{The Number of Groups and Post-Classification Estimator}
\label{SEC:post}

In this section, we address the question regarding the choice of $K$ and the post-classification estimation. We begin by presenting additional assumptions necessary for this analysis.

\begin{assumption}\label{A:additional-1} $(x_{i},\varepsilon_{i})$,
for $i=1,2,\ldots,N$ are independent across $i$. 
\end{assumption}

\begin{assumption}\label{A:group} $N_{k}\propto N$ for each $k=1,2,\ldots,K^{*}$.
\end{assumption}

\begin{assumption}\label{A:stochasticF} $v_{it}\overset{iid}{\sim}N\left(0,\sigma_{vi}^{2}\right)$ across $t.$ There exist an integer $\mathcal{K}^{*}\geq1$, such that, 
\[
\alpha_{i}^{0}-u_{i}\overset{d}{\sim}\alpha_{(j)}^{0}-\left\vert N\left(0,\sigma_{u(j)}^{2}\right)\right\vert \textrm{ with probability }\tau_{j}^{0},j=1,2,...,\mathcal{K}^{*},
\]
where $0<\tau_{j}^{0}<1$ and $\sigma_{u(j)}^{2}>C>0$ for $j=1,2,...,\mathcal{K}^{*},$
$\tau_{1}^{0}+\tau_{2}^{0}+...+\tau_{\mathcal{K}^{*}}^{0}=1,$ $\left(\alpha_{(j)}^{0},\sigma_{u(j)}^{2}\right)$
differ across $j=1,2,...,\mathcal{K}^{*}$. Lastly, the sequences
$\{v_{it}\}_{t=1}^{T}$, $u_{i}$, and $\{x_{it}\}_{t=1}^{T}$ are
mutually independent. \end{assumption}

\begin{assumption}\label{A:tuning2} (i) $T\propto N^{C_{*}}$ for
some positive $C_{*}$ as $N\to\infty$; (ii) $\underline{m}\to\infty$
as $T,N\to\infty$. Additionally, $\underline{m}/T\to0$, $\underline{m}^{q/2+2}(\log N)^{2q}/(NT)^{q/2-1}\to0$,
and $NT/\underline{m}^{1+2\kappa}\to0$; (iii) $C_{*}>1/(2\kappa)$
and $(q-2)\kappa>3$. \end{assumption}

Assumption \ref{A:additional-1} further imposes independence across
$i.$ Assumption \ref{A:group} says the number of members in each
group is proportional to $N$. This condition is not necessary, but
it facilitates expositions. Assumption \ref{A:stochasticF} specifies the
distributional conditions on the error terms. As explained in Section \ref{SEC:inefficiency}, these conditions are essential for the identification of $\sigma_{u}^{2}$, and common in the literature, (see, for e.g., \citet{YaoZhangKum2019}). 
Assumption \ref{A:tuning2} places
restrictions on the rate of growth of $T$ relative to $N$, and the rate of growth of the tuning parameter $\underline{m}$. The condition in (iii) is set to ensure that the set of $\underline{m}$ that satisfies (ii) is
not empty. We need $\underline{m}/T\rightarrow0$ so that the ``design''
matrix $\text{E}\left(\tilde{z}_{it}\tilde{z}_{it}'\right)$ is still
well-behaved, as required in Lemma \ref{LE:splines}. However, condition \textit{$\underline{m}/T\rightarrow0$}
can be restrictive when $N$ is much larger than $T$. $\left.\underline{m}^{q/2+2}\left(\log N\right)^{2q}\right/\left(NT\right)^{q/2-1}\rightarrow0$
is assumed to ensure the sample version of the design matrix, namely $Q_{(k),zz}$
defined in (\ref{eq:Qkzz}), is of full rank with very high probability.
Note that the dimension of $Q_{(k),zz}$ is diverging, so the consistency
of this matrix requires uniform convergence of all elements and hence
this restriction. $NT\left/\underline{m}^{1+2\kappa}\right.\rightarrow0$
ensures the bias term is asymptotically negligible. In the special
case where $\kappa\geq2,$ we need $C_{*}>1/4,$ and $\left(q-2\right)\kappa>3$
is satisfied due to $q>4$ in Assumption \ref{A:moment}. One can
then set, for example, $\underline{m}=\left(NT\right)^{1/4.8}$, which
satisfies condition (ii).

We show the asymptotic properties of our estimators for the case where
$\mathcal{K}^{*}\geq2$. The case in which $\alpha_{i}^{0}-u_{i}$ comes
from a unique distribution is straightforward given this result.

Recall that $\mathbb{B}_{-0}^{m}\left(\tau_{t}\right)'\pi_{i0}^{0}$
and $\mathbb{B}^{m}\left(\tau_{t}\right)^{\prime}\pi_{il}^{0}$ represent
the approximations of $\alpha_{i}\left(\tau_{t}\right)$ and $\beta_{il}\left(\tau_{t}\right).$
For the coefficients on the frontiers, let $\theta\left(s\right)\equiv\left(\alpha\left(s\right),\beta\left(s\right)'\right)',$
and correspondingly $\hat{\theta}\left(s\right)=\left(\mathbb{B}_{-0}^{\underline{m}}\left(\tau_{t}\right)'\hat{\pi}_{0},\mathbb{B}^{\underline{m}}\left(\tau_{t}\right)^{\prime}\hat{\pi}_{1},...,\mathbb{B}^{\underline{m}}\left(\tau_{t}\right)^{\prime}\hat{\pi}_{p}\right)'$.
For the parameters in the distribution of $\alpha_{i}^{0}-u_{i}$,
we denote 
\[
\varrho^{0}\equiv\left(\alpha_{(1)}^{0},\sigma_{u(1)}^{2},...,\alpha_{(\mathcal{K^{*}})}^{0},\sigma_{u(\mathcal{K}^{*})}^{2},\tau_{1}^{0},...,\tau_{\mathcal{\mathcal{K}^{*}}-1}^{0}\right),
\]
and correspondingly 
\[
\hat{\varrho}=\left(\hat{\alpha}_{(1)}^{0},\hat{\sigma}_{u(1)}^{2},...,\hat{\alpha}_{(\mathcal{K}^{*})}^{0},\hat{\sigma}_{u(\mathcal{\mathcal{K}^{*}})}^{2},\hat{\tau}_{1},...,\hat{\tau}_{\mathcal{\mathcal{K}^{*}}-1}\right).
\]
The following notations are used to characterize the asymptotic distribution. Denote 
\[
\mathbb{M_{B}}\left(s\right)\equiv\left(\begin{array}{cccc}
\mathbb{B}_{-0}^{\underline{m}}\left(s\right)' & 0 & \cdots & 0\\
0 & \mathbb{B}^{\underline{m}}\left(s\right)' & \cdots & 0\\
\vdots & \vdots & \ddots & \vdots\\
0 & 0 & \cdots & \mathbb{B}^{\underline{m}}\left(s\right)'
\end{array}\right)_{\left(p+1\right)\times\left(\underline{m}-1+\underline{m}p\right)},
\]
\begin{equation}
Q_{(k),zz}=\frac{1}{N_{k}T}\sum_{i\in G_{k|K^{*}}}\sum_{t=1}^{T}\underline{\ddot{z}}_{it}\underline{\ddot{z}}_{it}',\label{eq:Qkzz}
\end{equation}
and 
\[
\mathbb{S}_{(k)}\left(s\right)=\frac{\sigma_{v\left(k\right)}^{*2}}{\underline{m}}\mathbb{M_{B}}\left(s\right)Q_{(k),zz}^{-1}\mathbb{M_{B}}\left(s\right)',
\]
noting that $\mathbb{S}_{(k)}\left(s\right)$ is positive definite with very high probability; which we show it in equation (\ref{eq:var(Ak2)}) of Appendix \ref{APP:lemmaProofs}. For notation convenience, write 
\[
\tilde{f}_{i\left(k\right)}\left(\varrho\right)\equiv\tilde{f}\left(y_{i}\left\vert x_{i};\varrho,\vartheta_{\left(k|K^{*}\right)}\right.\right),
\]
and 
\[
\mathbb{I}\equiv\left.-\text{E}\left[\frac{1}{N}\sum_{k=1}^{K^{*}}\sum_{i\in G_{k}}\frac{\partial^{2}}{\partial\varrho\partial\varrho'}\log\tilde{f}_{i\left(k\right)}\left(\varrho\right)\right]\right|_{\varrho=\varrho^{0}},
\]
with $\mathbb{I}^{1/2}$ denoting the matrix such that $\mathbb{I}^{1/2}\mathbb{I}^{1/2\prime}=\mathbb{I}$.
$\mathbb{I}$ is a positive definite matrix with finite eigenvalues,
as shown in Appendix \ref{App:derivatives}. As before, we
show the asymptotic property of $\hat{\theta},\hat{\sigma}_{v}^{2}$ and
$\hat{\varrho}$ pretending that we know $K^{*}$ and $\mathcal{K}^{*}.$
We then show that $\hat{K}$ and $\hat{\mathcal{K}}$ converge to
$K^{*}$ and $\mathcal{K}^{*}$, respectively, with probability approaching
1.

\begin{theorem}\label{TH:post-estimation} Suppose Assumptions \ref{A:mixing}
through \ref{A:tuning2} hold, $\hat{K}(\lambda_{NT})=K^{*}$ and
$\hat{\mathcal{K}}\left(\tilde{\lambda}_{NT}\right)=\mathcal{K}^{*}$.
Let $I_{l}$ denote the $l\times l$ identity matrix. Then, for each
$k=1,2,\ldots,K^{*}$,

\noindent (i) 
\[
\sqrt{\frac{N_{k}T}{\underline{m}}}\mathbb{S}_{(k)}^{-1/2}(s)\left(\hat{\theta}_{(k|K^{*})}(s)-\theta_{(k)}^{*}(s)\right)\overset{d}{\rightarrow}N(0,I_{p+1});
\]
(ii) 
\[
\sqrt{N_{k}T}\left(\hat{\sigma}_{v(k|K^{*})}^{2}-\sigma_{v(k)}^{*2}\right)\overset{d}{\rightarrow}N\left(0,\text{\emph{Var}}(v_{it}^{2}|i\in G_{k|K^{*}})\right);
\]
(iii) 
\[
\sqrt{N}\mathbb{I}^{-1/2}(\hat{\varrho}-\varrho^{0})\overset{d}{\rightarrow}N(0,I_{3\mathcal{K}^{*}-1}).
\]

\end{theorem}

This theorem establishes the asymptotic properties of the post-classification estimators. As expected, $\hat{\theta}$ converges at a nonparametric
rate, while $\hat{\sigma}_{v}^{2}$ converges at a parametric rate.
The convergence rate of $\hat{\varrho}$ is $\sqrt{N}$ and does not
depend on $T$. This result may appear odd, but there is a simple explanation. Note that
$\varrho$ collects only the parameters that govern the distribution of $u_{i}$. The best scenario of estimating $\varrho$ is that we observe $u_{1},u_{2},...,u_{N}$ directly, in which case the rate of convergence of $\hat{\varrho}$ is $\sqrt{N}$. In theory, the value of $T$ does not impact the convergence rate of $\hat{\varrho}$. However, in finite samples, large $T$ can potentially ensure a more precise estimation of $u_{i}$, and thus can possibly improve the finite-sample performance of $\hat{\varrho}$.

In addition, the validity of the proposed information criteria relies
on these properties, as they depend on the accuracy and consistency
of the post-classification estimators, as demonstrated above. For example,
$\lambda_{NT}$ depends on both $N$ and $T$ (due to the rates of
$\hat{\theta}$ and $\hat{\sigma}_{v}^{2}$), while $\tilde{\lambda}_{NT}$
depends only on $N$ (due to the rate of $\hat{\varrho}$).

\begin{proposition}\label{Prop:classify} Suppose Assumptions \ref{A:mixing} through \ref{A:tuning2} hold.

\noindent (i) Select a value of $\lambda_{NT}$ such that $(NT)^{-1/2}\lambda_{NT}\rightarrow\infty$
and $(NT)^{-1}\lambda_{NT}\rightarrow0$. Then, 
\[
\Pr\left(\hat{K}\left(\lambda_{NT}\right)=K^{*}\right)\rightarrow1.
\]

\noindent (ii) Select a value of $\tilde{\lambda}_{NT}$ such that $\tilde{\lambda}_{NT}\rightarrow\infty$
and $N^{-1}\tilde{\lambda}_{NT}\rightarrow0$. Then, 
\[
\Pr\left(\hat{\mathcal{K}}\left(\tilde{\lambda}_{NT}\right)=\mathcal{K}^{*}\right)\rightarrow1.
\]

\end{proposition}

Proposition \ref{Prop:classify} presents conditions under which the
information criteria are valid, focusing on the tuning parameters
$\lambda_{NT}$ and $\tilde{\lambda}_{NT}$. As highlighted earlier,
selecting the correct range for $\lambda_{NT}$ and $\tilde{\lambda}_{NT}$
is crucial. In the subsequent section, we will evaluate specific values
for $\lambda_{NT}$ and $\tilde{\lambda}_{NT}$, identify those that
perform well in simulations, and recommend practical choices.

\section{Monte Carlo Simulations}
\label{SEC:simulations} 

\subsection{Simulation Designs}

Heterogeneity in panel SF models arises from various
sources. Thus, we design three different Monte Carlo experiments that
allow us to examine the finite-sample performance of the proposed
method and its ability to identify sources of heterogeneity. In
the first design, we study the classification for the case with heterogeneity from frontiers yet with constant variances of $v_{it}$. In the second design, we study the case with heterogeneous variances of $v_{it}$, yet homogeneous frontiers. In the third design, we check the performance of our methods in a more general and much more complicated scenario. In all three designs, we consider two sub-cases where
the term $\alpha^{0}-u$ comes from either unique or from mixture
distribution, which previous methods did not consider. Due to the similarity and the space constraint, we defer the description and simulation results of Designs 1 and 2 to Appendix \ref{APP:Tables} and only present Design 3 here.

\textbf{Design 3:} In our third design (consisting of DGP3U and DGP3M),
we study the performance of our method in a setting similar to those
in \citet{YaoZhangKum2019}, where there are three groups for both
the frontiers and variances, with two regressors. The DGP is 
\[
y_{it}=\alpha_{i}^{0}-u_{i}+\alpha_{i}(\tau_{t})+x_{it1}\beta_{i1}(\tau_{t})+x_{it2}\beta_{i2}(\tau_{t})+v_{it},
\]
where $x_{itl}\sim N(1,0.5^{2})$ for both regressors $l=1,2$.
Group 1 frontiers and error term are specified as $\alpha_{(1)}(s)=-\frac{1}{1+3s}-\varpi_{1}$,
$\beta_{(1)1}(s)=2s^{3}$, $\beta_{(1)2}(s)=\ln(5s)$, $v_{it}\overset{iid}{\sim}N(0,\sigma_{v(1)}^{2})$
with $\sigma_{v(1)}=0.75$ and $\varpi_{1}$ is a mean of $-\frac{1}{1+3s}$.
Group 2 frontiers and error term are specified as $\alpha_{(2)}(s)=-\cos(4s)-\varpi_{2}$,
$\beta_{(2)1}(s)=\sin(4s)$, $\beta_{(2)2}(s)=\ln(\frac{s}{1-s})$,
$v_{it}\overset{iid}{\sim}N(0,\sigma_{v(2)}^{2})$ with $\sigma_{v(2)}=1.25$
and $\varpi_{2}$ is a mean of $-\cos(4s)$. Group 3 frontiers and
error term are specified as $\alpha_{(3)}(s)=5s^{2}-s+1-\varpi_{3}$,
$\beta_{(3)1}(s)=\exp{(-s)}+\sin(5s)$, $\beta_{(3)2}(s)=-5\sin(s)\cos(5s)+1$,
$v_{it}\overset{iid}{\sim}N(0,\sigma_{v(3)}^{2})$, with $\sigma_{v(3)}=1.25$
and $\varpi_{3}$ is a mean of $5s^{2}-s+1$. We consider two
sub-cases of $\alpha^{0}-u$,  which we denote
them as DGP3U and DGP3M.  In DGP3U,  $\alpha^{0}-u$
comes from a unique distribution, with $\alpha^{0}=0.5$ and $u_{i}\overset{iid}{\sim}|N(0,\sigma_{u}^{2})|$,
where $\sigma_{u}=1$. In DGP3M, we let $\alpha^{0}-u$ to come from $\alpha_{(1)}^{0}-|N(0,\sigma_{u(1)}^{2})|$
with probability $\tau^{0}$ and $\alpha_{(2)}^{0}-|N(0,\sigma_{u(2)}^{2})|$
with probability $1-\tau^{0}$, where $\alpha_{(1)}^{0}=1$, $\alpha_{(2)}^{0}=-1$,
$\sigma_{u(1)}=0.75$, $\sigma_{u(2)}=1.25$ and $\tau^{0}=0.5$. It is important to note that the mixture structure of $\alpha^{0}-u$ is
independent of the grouping structure.

We evaluate the performance of each model and the case for any combination
of $N=100,250,$ or $500$ and $T=50,75,$ or $100$. Thus,
there are $3\times2\times9=54$ different cases. We assess the finite
sample properties of our method with 500 MC replications.

Note $(N,T)=(466,80)$ in the empirical application of the paper,
so our simulations, including the recommended tuning parameters in
the next section, offer meaningful and practical guidance.

\subsection{Choices of Tuning Parameters} \label{SEC:tuning}

We set $m=\left\lfloor T^{1/5}\right\rfloor $, where $\left\lfloor \cdot\right\rfloor$ denotes the integer part and $\underline{m}=\left\lfloor \left(N_{k}T\right)^{1/4.8}\right\rfloor $
for each group $k$. The value of the two tuning parameters align with standard choices in the literature. 

Theoretically, the valid ranges for $\lambda_{NT}$ and $\tilde{\lambda}_{NT}$
are quite broad. Based on simulation evidence, we recommend setting
$\lambda_{NT}=\left(c_{\lambda}\sqrt{NT}\log(NT)\right)/2$ and $\tilde{\lambda}_{NT}=\left(\tilde{c}_{\lambda}\sqrt{N}\log N\right)/8$,
where $c_{\lambda}$ and $\tilde{c}_{\lambda}$ are constants. These
values of $\lambda_{NT}$ and $\tilde{\lambda}_{NT}$ meet the conditions
specified in Proposition \ref{Prop:classify}. The constants $c_{\lambda}$
and $\tilde{c}_{\lambda}$ serve as sensitivity parameters, over which
we conduct sensitivity analyses for the choice of $\lambda_{NT}$
and $\tilde{\lambda}_{NT}$. Specifically, we test values of $c_{\lambda}$
and $\tilde{c}_{\lambda}$ in $\{3/2,1,3/4\}$, with $c_{\lambda}=\tilde{c}_{\lambda}=1$
as the benchmark setting.

In finding the number of components in the mixture distribution, we restrict our attention to one or two components, i.e., $\mathcal{K}=1,2$. We find that the small-sample properties of mixtures with more than two components perform poorly in simulations. We conjecture that this is due to the complicated log-likelihood functions
(see, e.g., equations (\ref{EQ:likelihood_appro}) and (\ref{EQ:log_approx_p})), which cannot effectively
handle mixtures with more than two components. In Appendix \ref{APP:more_mix}, we propose an alternative method to
identify the mixture structure with potentially more than two components.
We also conduct simulations to evaluate its finite-sample performance.
The simulation results for this alternative method for allowing mixtures with more than two components suggest that the method performs well in determining the correct number of components, but the parameter estimates can be off. Although not perfect, these findings provide a foundation for future research in this direction.

\subsection{Simulation Results}

We report results for DGP3M, the most complex model, featuring three
groups and a mixture distribution structure in $\alpha_{i}^{0}-u_{i}$. Results for the remaining DGPs are reported in Appendix \ref{APP:Tables}.

We first report the performance of the IC in Step 3 for coefficient
groups and Step 5 for the mixture distribution structure in Table \ref{tab:IC_DGP3M} for the benchmark specification with $c_{\lambda}=\tilde{c}_{\lambda}=1$. Additionally, Table \ref{tab:IC_DGP3M} includes the classification errors, denoted as $\bar{\text{Pr}}(\bar{F})$ . It is defined as the average percentage of observations misclassified to other groups across 500 replications. 
The performance of the IC in Step 3 for selecting the correct number
of groups, $K^{*}=3$, is reasonable. For $N=500$, the classification
error in Step 3 is less than 1 percent for each $T=50,75,100$. The
performance of the Step 5 IC is also strong for DGP3M, choosing the
correct specification (mixture distribution) with a probability close
to 1. For DPG3U, where $\alpha_{i}^{0}-u_{i}$ comes from a unique
distribution, Table \ref{tab:IC_DGP3U} in Appendix \ref{APP:Tables}
shows that the probability of Step 5 IC selecting the correct distribution (unique distribution) quickly approaches 1 as $N$ increases. Sensitivity analyses for both ICs in Steps 3 and 5, shown in Tables \ref{tab:Sensitivity_Step3_IC_DGP1U}
- \ref{tab:Sensitivity_Step5_IC_DGP3M}, demonstrate that the results
are robust to the selected range of tuning parameters.

We assess the accuracy of the estimates of $\{\sigma_{v}\text{s},\alpha_{1}^{0},\sigma_{u(1)},\alpha_{2}^{0},\sigma_{u(2)},\tau^{0}\}$
using two measures: (i) bias (BIAS) and (ii) root mean squared errors
(RMSE). The reported values in Table \ref{tab:DGP3M_parameters},
obtained by averaging over 500 MC iterations, are reasonable.

Illustrated in Figures \ref{fig:Grouped_frontiers_DGP6_N500_T50},
\ref{fig:Grouped_frontiers_DGP6_N500_T75} and \ref{fig:Grouped_frontiers_DGP6_N500_T100}
in Appendix \ref{APP:Tables} are the estimates of time-varying frontiers
for $(N,T)=(500,50),(500,75),\text{ and }(500,100)$. Black solid lines
depict the true time-varying frontier, dotted lines show the mean
of the estimated grouped frontiers averaged over 500 MC iterations,
and the gray shaded region depicts the 90th percentile of the estimates.
It is clear from Figure \ref{fig:Grouped_frontiers_DGP6_N500_T50}
that, while the mean over MC iterations is reasonably close to the
true frontiers, the 90th percentile bands are wide, suggesting possible
classification errors between neighboring groups. Figures \ref{fig:Grouped_frontiers_DGP6_N500_T75}
and \ref{fig:Grouped_frontiers_DGP6_N500_T100} show that the accuracy
of frontier grouping improves as $T$ increases. This is consistent
with the theory developed, since the Step 2 classification using HAC
is based on $\hat{\vartheta}_{i}=\left(\hat{\pi}_{i}^{\prime},\hat{\sigma}_{vi}\right)^{\prime}$
obtained using $T$ observations.


\begin{table}[!htb]
\centering
\caption{Performance of ICs for DGP3M}
\vspace{0.2cm}
\begin{tabularx}{\textwidth}{X XXXX X XX}
\toprule
$(N,T)$ & $K=1$ & $K=2$ & $K=3$ & $K=4$ & $\bar{\text{Pr}}(\bar{F})$ & $\alpha^0-u$ uni & $\alpha^0-u$ mix  \\
\midrule

(100,50) & 0.000 & 0.352 & 0.648 & 0.000 & 0.106 & 0.008 & 0.992 \\ 
(100,75) & 0.000 & 0.078 & 0.922 & 0.000 & 0.024 & 0.000 & 1.000 \\ 
(100,100) & 0.000 & 0.000 & 1.000 & 0.000 & 0.024 & 0.000 & 1.000 \\ 
(250,50) & 0.000 & 0.086 & 0.914 & 0.000 & 0.026 & 0.002 & 0.998 \\ 
(250,75) & 0.000 & 0.000 & 1.000 & 0.000 & 0.026 & 0.000 & 1.000 \\ 
(250,100) & 0.000 & 0.000 & 1.000 & 0.000 & 0.026 & 0.000 & 1.000 \\ 
(500,50) & 0.000 & 0.006 & 0.994 & 0.000 & 0.002 & 0.004 & 0.996 \\ 
(500,75) & 0.000 & 0.000 & 1.000 & 0.000 & 0.002 & 0.000 & 1.000 \\ 
(500,100) & 0.000 & 0.000 & 1.000 & 0.000 & 0.002 & 0.000 & 1.000 \\ 
\bottomrule

\end{tabularx}
\begin{tablenotes}
      \footnotesize
      \item 
      \emph{Note}: Results for the baseline case $c_{\lambda}=\tilde{c}_{\lambda}=1$. Reported numbers are probabilities across replications.
\end{tablenotes}
\label{tab:IC_DGP3M}

\end{table}

    
\begin{table}[!htb]
\centering
\caption{BIAS and RMSE over 500 MC iterations for DGP3M}
\vspace{0.2cm}

\scalebox{0.71}{
\begin{tabularx}{1.4\textwidth}{X c XX c XX c XX c XX c XX c XX c XX c XX}
\toprule
&& \multicolumn{2}{c}{$\hat{\sigma}_{v(1)}$} && \multicolumn{2}{c}{$\hat{\sigma}_{v(2)}$} && \multicolumn{2}{c}{$\hat{\sigma}_{v(3)}$} && \multicolumn{2}{c}{$\hat{\tau}$} && \multicolumn{2}{c}{$\hat\alpha^0_{(1)}$} && \multicolumn{2}{c}{$\hat{\sigma}_{u(1)}$} && \multicolumn{2}{c}{$\hat\alpha^0_{(2)}$} && \multicolumn{2}{c}{$\hat{\sigma}_{u(2)}$} \\
\cmidrule{3-4} \cmidrule{6-7} \cmidrule{9-10} \cmidrule{12-13} \cmidrule{15-16} \cmidrule{18-19} \cmidrule{21-22} \cmidrule{24-25}
$(N,T)$ && BIAS & RMSE && BIAS & RMSE && BIAS & RMSE && BIAS & RMSE && BIAS & RMSE && BIAS & RMSE && BIAS & RMSE && BIAS & RMSE\\
\midrule
(100,50) && 0.268 & 0.463 && 0.127 & 0.215 && 0.367 & 0.594 && 0.021 & 0.029 && 0.057 & 0.074 && 0.116 & 0.152 && 0.120 & 0.157 && 0.135 & 0.171 \\ 
(100,75) && 0.082 & 0.214 && 0.074 & 0.174 && 0.154 & 0.372 && 0.015 & 0.027 && 0.042 & 0.052 && 0.103 & 0.145 && 0.100 & 0.131 && 0.126 & 0.162 \\ 
(100,100) && 0.024 & 0.099 && 0.027 & 0.093 && 0.053 & 0.196 && 0.013 & 0.018 && 0.040 & 0.052 && 0.090 & 0.116 && 0.091 & 0.118 && 0.108 & 0.137 \\ 
(250,50) && 0.204 & 0.373 && 0.138 & 0.243 && 0.327 & 0.562 && 0.013 & 0.018 && 0.036 & 0.049 && 0.077 & 0.107 && 0.080 & 0.102 && 0.089 & 0.112 \\ 
(250,75) && 0.034 & 0.129 && 0.035 & 0.116 && 0.069 & 0.241 && 0.009 & 0.012 && 0.026 & 0.033 && 0.062 & 0.079 && 0.062 & 0.076 && 0.070 & 0.089 \\ 
(250,100) && 0.005 & 0.033 && 0.008 & 0.032 && 0.014 & 0.064 && 0.008 & 0.011 && 0.024 & 0.030 && 0.057 & 0.075 && 0.059 & 0.074 && 0.076 & 0.095 \\ 
(500,50) && 0.145 & 0.307 && 0.110 & 0.218 && 0.242 & 0.484 && 0.011 & 0.013 && 0.026 & 0.035 && 0.056 & 0.074 && 0.053 & 0.067 && 0.069 & 0.087 \\ 
(500,75) && 0.009 & 0.064 && 0.009 & 0.044 && 0.018 & 0.100 && 0.007 & 0.009 && 0.018 & 0.022 && 0.042 & 0.055 && 0.042 & 0.053 && 0.053 & 0.068 \\ 
(500,100) && 0.002 & 0.003 && 0.004 & 0.005 && 0.007 & 0.009 && 0.006 & 0.008 && 0.017 & 0.021 && 0.041 & 0.054 && 0.040 & 0.053 && 0.056 & 0.070 \\
\bottomrule
\end{tabularx}}

\label{tab:DGP3M_parameters}
\begin{tablenotes}
      \tiny
      \item 
\end{tablenotes}
\end{table}

\section{Application to the U.S. Commercial Banking Sector}
\label{SEC:application} 

In this section, we apply the developed method for stochastic cost
frontier model to analyze the cost efficiency of the U.S. large commercial banks in presence of a series of gradual deregulation that allowed banks to increase their capacity of operation. We use the same dataset used by \citet{Fengetal2017}, and focus our analysis on a sample
of banks that operate continuously over the period 1986 to 2005 (thereby
mitigating the impact of entry and exit) with assets of at least \$1
billion in 1986 dollars. Data supporting the findings of this
study are available upon request. As briefly explained in \citet{Fengetal2017}, the banking sector over this period saw a number of gradual deregulation that allowed banks to increase the capacity of operation. In particular, the exact timing of the deregulation varied at the state level, and it was not until June 1997 that banks were allowed to operate across states as a result of the Riegle-Neal Interstate Banking and Branching Efficiency Act of 1994.\footnote{See \citet{Fengetal2017} and \citet{JayaratneStrahan1997} for more
detailed discussion of the history of deregulation in the banking
sector.} Given this context, our method that allows us to group banks based on the time-varying frontiers is well suited to capture the effect
of gradual deregulation, as well as to analyze the inefficiency
of banks in presence of such deregulation.

To set the stage, let $i=1,2,\ldots,N$ denote the banks, $t=1,2,\ldots,T$
denote the time periods. The data is recorded in quarterly frequency,
over 1986 to 2005, with $T=80$ and consists of $N=466$ banks. We
assume that banks use three inputs to generate three outputs. Specifically, the inputs used are: (i) price of labor, $W_{it1}$, (ii) price of purchased funds, $W_{it2}$, and (iii) price of core deposits, $W_{it3}$. Generated outputs are: (i) consumer loans, $Y_{it1}$, (ii) non-consumer loans, $Y_{it2}$, consisting of industrial, commercial, and real estate loans, and (iii) securities, $Y_{it3}$, which includes all non-loan financial assets. Summary statistics of these variables are
reported in Table \ref{tab:Summary_Statistics} in Appendix \ref{APP:Tables}.

We estimate a cost frontier $C(Y_{it}, W_{it})$. Accordingly, $Y_{it}$ are output quantities (loan categories, securities) and $W_{itj}$ are input prices. In particular, $W_{it2}$ is the \emph{price of purchased funds}, not an output; loans are treated as outputs under the intermediation view of banking services. This mapping is consistent with cost duality and with our specification, in which the frontier uses \emph{grouped, smoothly time-varying coefficients} to capture heterogeneous technological regimes under staggered state-level deregulation.

The particular variant of the panel SF model we study is a panel stochastic cost frontier model adapted from \citet{Greene2005}: 
\begin{align}
\log c_{it}^{*}=\alpha_{i}^{0} & +u_{i} +\alpha_{i}(\tau_{t})+\beta_{i1}(\tau_{t})\log w_{it1}+\beta_{i2}(\tau_{t})\log w_{it2}\label{eqn:application_model}\nonumber \\
 & +\beta_{i3}(\tau_{t})\log y_{it1}+\beta_{i4}(\tau_{t})\log y_{it2}+\beta_{i5}(\tau_{t})\log y_{it3}+v_{it},
\end{align}
where linear homogeneity is imposed in input prices by the normalizations: $c_{it}^{*}=C_{it}/W_{it3}$, $w_{itl}=W_{itl}/W_{it3}$ for $l=1,2$ and $y_{itl}=Y_{itl}/W_{it3}$ for $l=1,2,3$. The inefficiency term $u_{i}\geq0$ enters the model positively as cost frontier models are derived from the dual cost minimization problem of the firm where the cost function is assumed to be Cobb-Douglas. 

In a cost frontier, interest-rate conditions primarily operate through $W_{it}$, while local demand affects the output  $Y_{it}$.
Our specification already conditions on $(Y_{it},W_{it})$ which are allowed to vary smoothly over time and across latent regimes.
Adding rate or local controls would double-count channels captured by $(Y_{it},W_{it})$.


We estimate the model in \eqref{eqn:application_model} using the
method described in the previous section, setting the tuning parameters
as in Section \ref{SEC:tuning}. We also check the sensitivity of
parameter $c_{\lambda}$ in Step 3 and $\tilde{c}_{\lambda}$ in Step
5 of the proposed method as in Section \ref{SEC:tuning}. Different
values of $c_{\lambda}$ and $\tilde{c}_{\lambda}$ deliver the same
classification results.
\begin{figure}[!htb]
\caption{Scatter plot of elements in $\hat{\vartheta}_{i} = (\hat{\pi}_{i}',\hat{\sigma}_{vi})'$}
\centering
\includegraphics[width=0.8\textwidth]{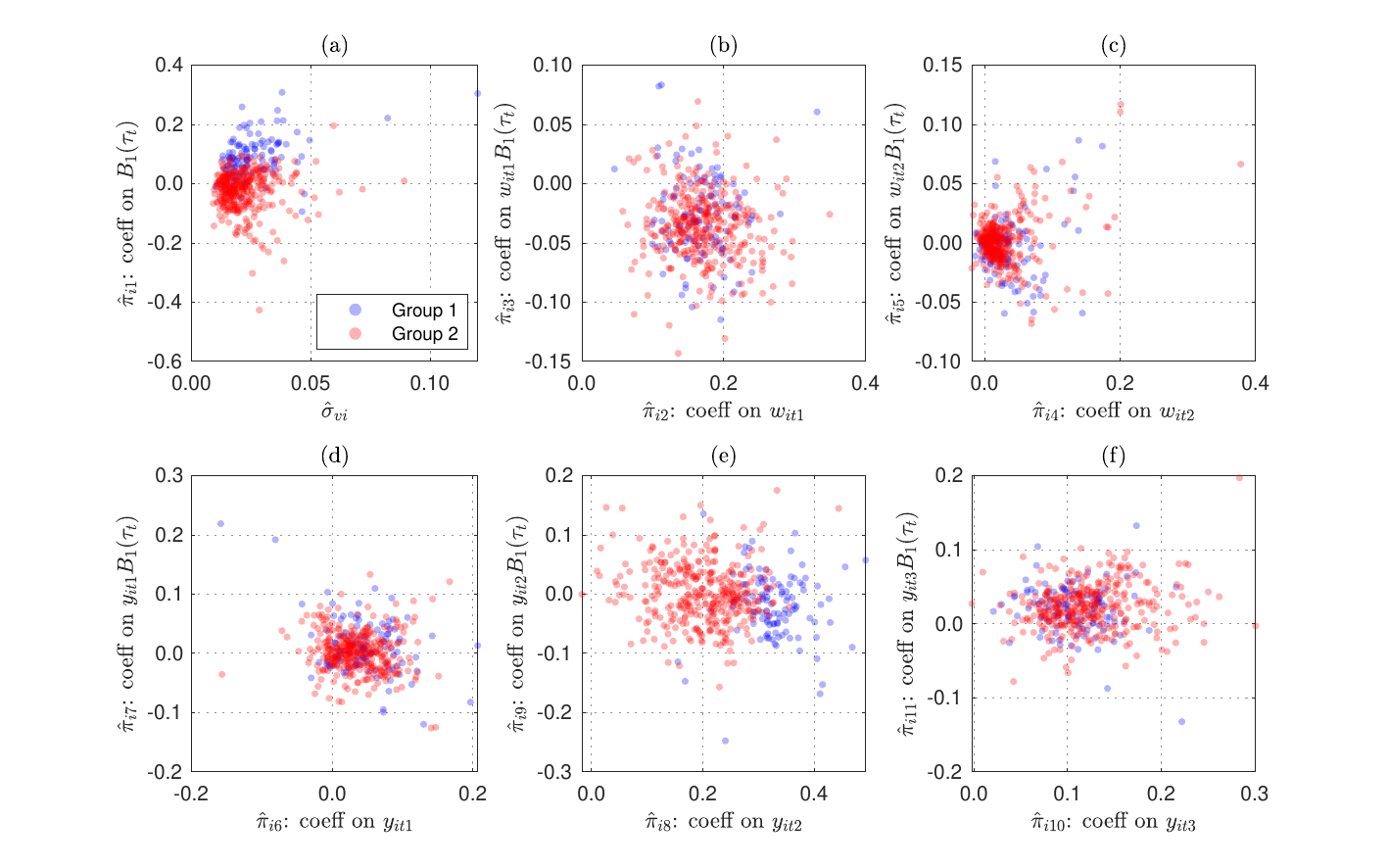}
\label{fig:Grouped_Frontiers_pihat}
\begin{tablenotes}
      \vspace{-24pt}
      \footnotesize
      \item 
      \emph{Note}: Estimates classified as group 1 are plotted as blue dots, while that for group 2 are plotted as red dots. 
\end{tablenotes}
\end{figure}

As in the simulations, we set $\bar{K}=4$. The information criteria
in step 3 selects the optimal number of group for the banks to be
two, splitting $N=466$ banks into $(N_{1},N_{2})=(113,353)$. Figure
\ref{fig:Grouped_Frontiers_pihat} depict the scatter plots of the
elements in $\hat{\vartheta}_{i}=(\hat{\pi}_{i}',\hat{\sigma}_{vi})'$,
that collects the parameters obtained from individual level estimation
in step 1 for classification in step 2. Note $m=\left\lfloor T^{1/5}\right\rfloor =2$,
so each $\hat{\vartheta}_{i}$ is a 12 by 1 vector. Individual estimates classified as groups 1 and 2 are depicted as blue and red dots, respectively. Panel (a) depicts the scatter plot of the estimates $\hat{\pi}_{i1}$ against $\hat{\vartheta}_{i12}$, while panels (b)--(f) depict the respective coefficients on the inputs/outputs $(w_{itl},y_{itl})$ against $(w_{itl}B_{1}(\tau_{t}),y_{itl}B_{1}(\tau_{t}))$. We discuss what drives the classification in Appendix \ref{app:read-groups-lean}.

\begin{figure}[!htb]
\caption{Grouped Frontiers of the U.S. Large Commercial Banks}
\centering
\includegraphics[width=0.8\textwidth]{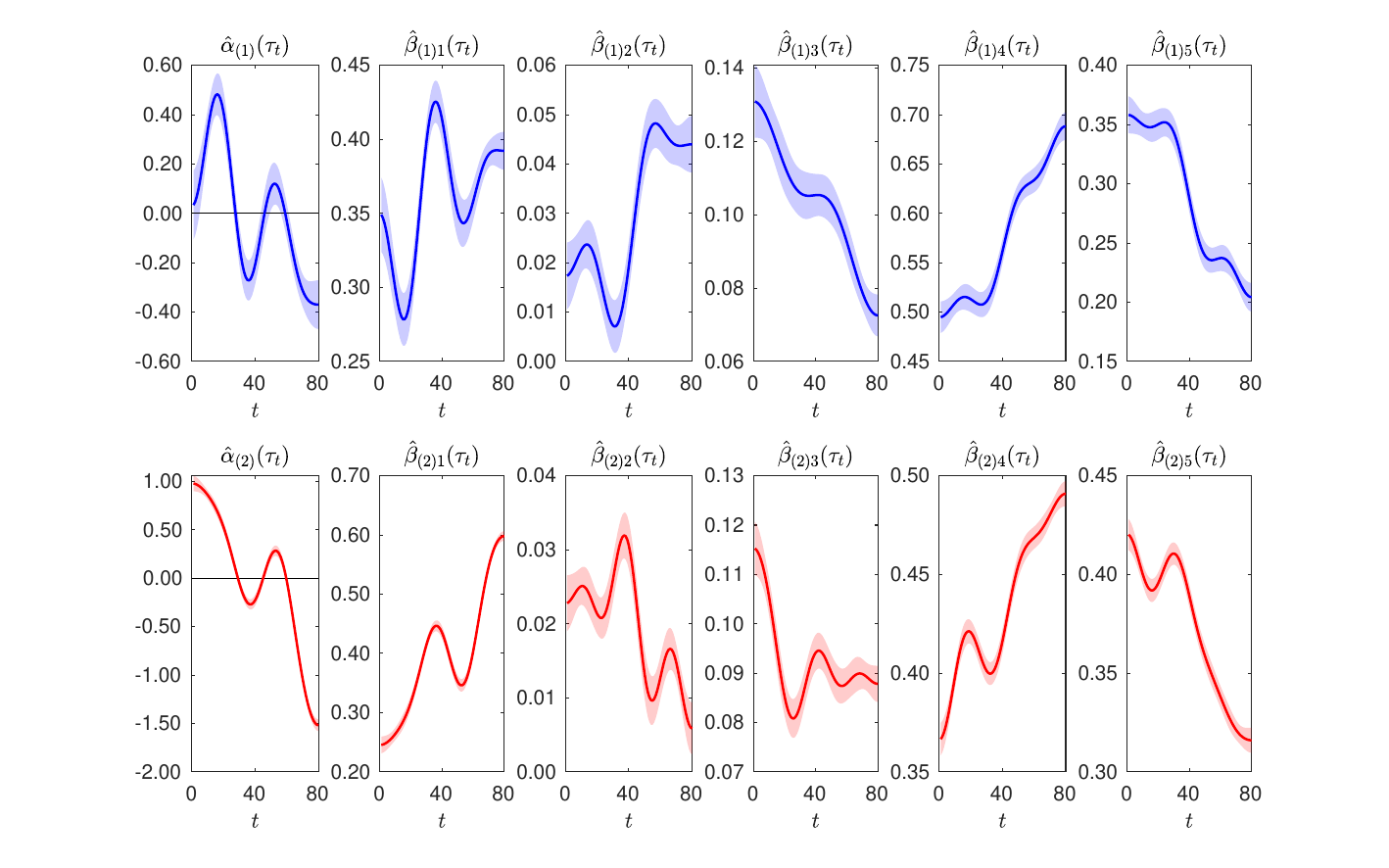}
\label{fig:Grouped_Frontiers_Application_SE}
\begin{tablenotes}
      \vspace{-24pt}
      \footnotesize
      \item 
      \emph{Note}: Top row depict the group 1 time-varying cost frontiers while the bottom row depict that of group 2. Solid lines are the point estimates and the shaded regions are 95\% confidence interval.
\end{tablenotes}
\end{figure}

Figure \ref{fig:Grouped_Frontiers_Application_SE} depicts the frontiers.
The top row depicts the time-varying frontiers of group 1, while the bottom row depicts that of group 2. Solid lines in blue and red are the point estimates for group 1 and group 2 respectively, and the shaded regions depict the 95\% confidence interval. It is evident that there
are substantial time-variations in the estimates, which may be a result
of increasing the capacity of operation as a result of deregulation
in the banking sector. Figure \ref{fig:ES_Application} shows the
estimated economies of scale experienced by two groups of banks, $k=1,2$
defined by the inverse of the sum of elasticities of output, $1/(\hat{\beta}_{(k)3}(\tau_{t})+\hat{\beta}_{(k)4}(\tau_{t})+\hat{\beta}_{(k)5}(\tau_{t}))-1$.
The estimates on economies of scale are comparable to the ones found
in \citet{Greene2005} and suggest some considerable time-variations
for both groups, with group 2 banks enjoying larger economies of scale.

\begin{figure}[ht]
\caption{Estimates of Economy of Scale }
\centering
\includegraphics[width=0.6\textwidth]{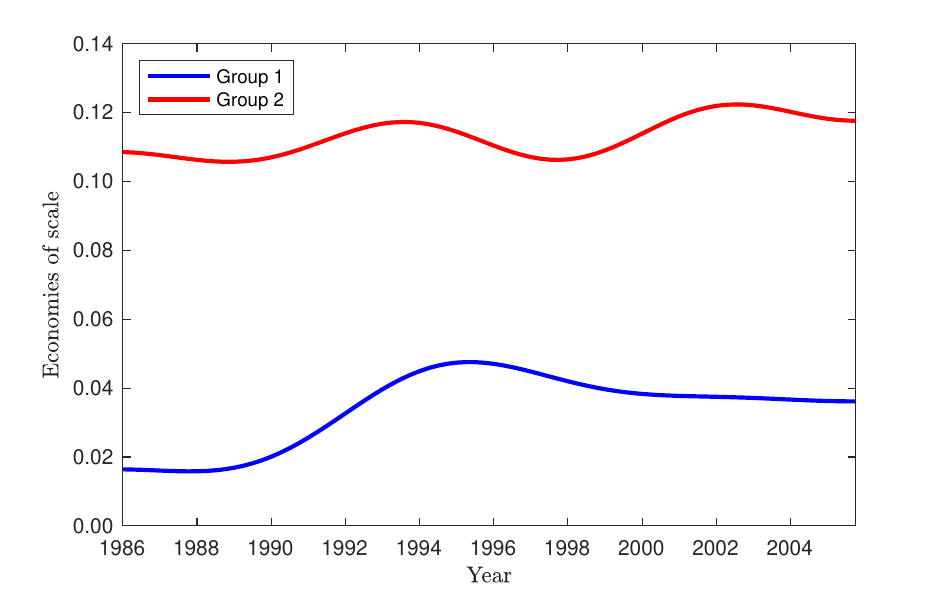}
\label{fig:ES_Application}
\begin{tablenotes}
      \footnotesize
      \item 
      \emph{Note}: Estimates of the economies of scale for each groups are calculated using the point estimates $\hat{\beta}_{(k)l}(\tau_{t})$ for $l=3,4,5$ and $k=1,2$.
\end{tablenotes}
\end{figure}

The results from Step 5 of the proposed method suggest that intercept and idiosyncratic random
effects inefficiency terms, $\alpha^{0}+u$, possess a mixture distribution structure. This result indicates that not only do frontiers form two distinct groups, but so do the level terms that represent the inefficiency of individual banks. The estimated values of the parameters along with the standard errors are presented in Table \ref{tab:Application}. The results
suggest that there are no substantial differences in the standard deviation of random noise, $\hat{\sigma}_{v}$s, although there are significant differences in the standard deviation of the inefficiency terms.

As we briefly mentioned in Section \ref{SEC:inefficiency}, since $\alpha_{i}^{0}$ differs across $i$, we cannot make a valid ranking of the inefficiencies. Luckily, $\hat{\alpha}_{\left(1\right)}^{0}$
and $\hat{\alpha}_{\left(2\right)}^{0}$ are not statistically different  (by the likelihood-ratio test) and as such we can view them the same and construct a ranking of the inefficiencies. Recall that we were estimating cost frontiers. From the estimation results, we can view 
\begin{equation}
\alpha_{i}^{0}+u_{i}\overset{d}{\sim}\begin{cases}
\begin{array}{c}
\hat{\alpha}^{*}+\left|N\left(0,\hat{\sigma}_{u\left(1\right)}^{2}\right)\right|\\
\hat{\alpha}^{*}+\left|N\left(0,\hat{\sigma}_{u\left(2\right)}^{2}\right)\right|
\end{array} & \begin{array}{c}
\textrm{with probability }\hat{\tau}\\
\textrm{with probability }1-\hat{\tau}
\end{array},\end{cases}\label{eq:mixture_distribution}
\end{equation}
where $\hat{\alpha}^{*}=\hat{\tau}\hat{\alpha}_{(1)}^{0}+\left(1-\hat{\tau}\right)\hat{\alpha}_{(2)}^{0}.$ We  compute $\widehat{\textrm{E}\left(\alpha_{i}^{0}+u_{i}|\varepsilon_{i1},...,\varepsilon_{iT}\right)}$ using (\ref{eq:inefficiency_post}).
We compare the ranking in the homogeneous case where the frontiers and
the variances of $v_{it}$ are assumed the same across firms and the
inefficiency term comes from one distribution.  The result of top 60 is reported in Figure \ref{fig:Inefficiency_Ranking_Ranking} in Appendix \ref{APP:application_fig}. We can see that the two rankings differ greatly after the top 3. This highlights the importance of classification to ensure valid inference of inefficiency term. 

\begin{table}[H]
\centering
\caption{Estimates of $\hat{\sigma}_{v}$s and $\hat{\varrho}$}
\vspace{0.2cm}

\begin{tabularx}{\textwidth}{X X X X  X  XX }
\toprule

$\hat{\sigma}_{v(1)}$ & $\hat{\sigma}_{v(2)}$ & $\hat{\tau}$ & $\hat\alpha^0_{(1)}$ & $\hat{\sigma}_{u(1)}$ & $\hat\alpha^0_{(2)}$ & $\hat{\sigma}_{u(2)}$ \\
\midrule
0.0862 & 0.0855 & 0.8748 & 0.0157 & 0.4426 & 0.6161 &   0.7756\\ 
(0.0041) & (0.0008) & (0.1017) & (0.3960) & (0.0362) & (0.1708) & (0.0235) \\
\bottomrule
\end{tabularx}
\label{tab:Application}
\begin{tablenotes}
      \footnotesize
      \item 
      \emph{Note}: Reported in parentheses are the standard errors. 
\end{tablenotes}
\end{table}

\section{Conclusion}

In this paper, we develop a general framework for panel SF models with latent
group structures.  A natural concern is whether allowing for multiple frontiers weakens the interpretation of inefficiency. Our results suggest the opposite. By accounting for latent technological regimes, we prevent unobserved heterogeneity from being mistakenly absorbed into the inefficiency term. Inefficiency in our framework is always measured relative to the appropriate group frontier. This distinction is crucial in empirical applications, such as our U.S. banking study, where ignoring heterogeneity would miscalculate inefficiency.

Two extensions are worth mentioning. First, our framework cannot be
directly generalized to endogenous cases where covariates, $x$ are correlated with the error term, $v$. Extending the framework to accommodate endogeneity is an important direction for future work. Second, it would be valuable to explore a one-step HAC algorithm that avoids the use of information criteria, as proposed by \citet{Mugnier2025}.

\newpage{}

\appendix
\renewcommand\thefigure{\thesection.\arabic{figure}}  
\renewcommand\thetable{\thesection.\arabic{table}}  
\setcounter{table}{0} \setcounter{figure}{0} \setcounter{section}{0}
\numberwithin{equation}{section}

\begin{center}
{\large\textbf{Online Appendix to}}$\vspace{0.08in}$ 
\par\end{center}

\begin{center}
{\large\textbf{``Panel Stochastic Frontier Models with Latent Group
Structures''}}{\Large\par}
\par\end{center}

\begin{center}
 
\par\end{center}

\begin{center}
(NOT for Publication) 
\par\end{center}

\begin{center}
{\small\ \ \ \ \ \ \ \ \ } 
\par\end{center}

\noindent\textit{Additional Notation.} For the deterministic series
$\{a_{n},b_{n}\}_{n=1}^{\infty}$, we denote $a_{n}\lesssim b_{n}$
if $\limsup_{n\rightarrow\infty}\left\vert a_{n}/b_{n}\right\vert \leq C$
for some constant $C$ that does not depend on $n$, $a_{n}\gtrsim b_{n}$
if $b_{n}\lesssim a_{n}$, $a_{n}\ll b_{n}$ if $a_{n}=o(b_{n})$,
and $a_{n}\gg b_{n}$ if $b_{n}\ll a_{n}$. $\propto_{P}$ denotes
proportional in probability, e.g., $x_{n}\propto_{P}y_{n}$ indicates
that both $x_{n}=O_{P}(y_{n})$ and $y_{n}=O_{P}(x_{n})$ hold. $A^{c}$
denotes the complement of $A$. $C$ and $M$ denote some positive
constants that may vary from line to line.

\section{Other Stochastic Frontier Models }

\label{APP:FE}

\subsection{Fixed Effects Model with Distributional Conditions }\label{APP:FE1}

We consider the fixed effects (FE) SF model when error terms are assumed
to follow certain parametric distributions. The classification of
the FE-SF model follows naturally from the methodology developed in
the main body of the paper. Incorporating time-varying heterogeneous
coefficients, the model proposed by \citet{Greene2005a,Greene2005}
and \citet{Chenetal2014} takes the form: 
\begin{align*}
y_{it} & =\alpha_{i}^{0}+\alpha_{i}(\tau_{t})+x_{it}'\beta_{i}(\tau_{t})+\varepsilon_{it}\\
 & =\alpha_{i}^{0}+\alpha_{i}(\tau_{t})+x_{it}'\beta_{i}(\tau_{t})+v_{it}-u_{it},
\end{align*}
where $\varepsilon_{it}=v_{it}-u_{it}$, and the normalization condition
$\int_{0}^{1}\alpha_{i}(\tau_{t})\,d\tau_{t}=0$ is imposed. We note
that this normalization is innocuous as in \citet{Ataketal}, because
the the constant (even varying acorss $i$) can be absorbed in $\alpha_{i}^{0}$.
The standard assumptions in these papers include: 
\[
v\perp u\perp(\alpha^{0},x),\quad v\sim N(0,\sigma_{v}^{2}),\quad u\sim|N(0,\sigma_{u}^{2})|,
\]
with $\{v_{it},u_{it}\}$ independent across time $t$. To address
the incidental parameters problem, \citet{Chenetal2014} proposed
the within MLE. Allowing for heterogeneous distributions across firms,
we have: 
\[
v_{it}\sim N(0,\sigma_{vi}^{2}),\quad u_{it}\sim|N(0,\sigma_{ui}^{2})|.
\]

We next define the group structure. Specifically, we assume that there
are $K^{*}$ distinct groups of parameter sets, and each firm's parameters
belong to one of these groups. Formally: 
\begin{equation}
\left\{ \alpha_{i}(\tau_{t}),\beta_{i}(\tau_{t}),\sigma_{vi},\sigma_{ui}\right\} =\sum_{k=1}^{K^{*}}\left\{ \alpha_{(k)}^{*}(\tau_{t}),\beta_{(k)}^{*}(\tau_{t}),\sigma_{v(k)}^{*},\sigma_{u(k)}^{*}\right\} \cdot\boldsymbol{1}(i\in G_{k}).\label{EQ:group_para-1}
\end{equation}
Parameters across different groups are distinct, and each firm is
uniquely assigned to one of the $K^{*}$ groups.

The key distinction here is that we also classify firms based on $\sigma_{ui}$.
The intuition is as follows: unlike the framework in the main body
of the paper, we can consistently estimate $\sigma_{ui}$ for each
firm $i$ individually, because we allow $u_{it}$ to vary over time.
This variation enables us to estimate $\sigma_{ui}$ without pooling
observations across firms. Thanks to this insight, we only need to
apply Steps 1--3 in Section~\ref{SEC:Estimation} for the FE model.
Before detailing the procedure, we approximate the model using sieve
expansions as in equation~(\ref{EQ:approx}), leading to: 
\begin{align}
y_{it} & =\alpha_{i}^{0}+\alpha_{i}(\tau_{t})+\sum_{l=1}^{p}x_{itl}\beta_{il}(\tau_{t})+\varepsilon_{it}\nonumber \\
 & \approx\alpha_{i}^{0}+\mathbb{B}_{-0}^{m}(\tau_{t})'\pi_{i0}^{0}+\sum_{l=1}^{p}x_{itl}\mathbb{B}^{m}(\tau_{t})'\pi_{il}^{0}+\varepsilon_{it}\nonumber \\
 & =\alpha_{i}^{0}+z_{it}'\pi_{i}^{0}+\varepsilon_{it}=\alpha_{i}^{0}+z_{it}'\pi_{i}^{0}+v_{it}-u_{it},\label{EQ:approx-1}
\end{align}
where $z_{it}$ collects all the basis function terms and their interactions
with covariates. We now outline the estimation procedure.

\subsubsection*{Step $\mathbf{1}^{*}$: Individual Estimation}

We begin by applying the within transformation, and define the following
notation: 
\[
\ddot{y}_{it}=y_{it}-\frac{1}{T}\sum_{t=1}^{T}y_{it},\quad\ddot{z}_{it}=z_{it}-\frac{1}{T}\sum_{t=1}^{T}z_{it},
\]
with $\ddot{v}_{it}$, $\ddot{u}_{it}$, and $\ddot{\varepsilon}_{it}$
defined analogously. The transformed model becomes: 
\[
\ddot{y}_{it}\approx\ddot{z}_{it}'\pi_{i}^{0}+\ddot{\varepsilon}_{it}=\ddot{z}_{it}'\pi_{i}^{0}+\ddot{v}_{it}-\ddot{u}_{it}.
\]

Let $\ddot{{\varepsilon}}_{i}=(\ddot{\varepsilon}_{i1},\ddot{\varepsilon}_{i2},\ldots,\ddot{\varepsilon}_{i,T-1})'$
denote the vector of the first $T-1$ transformed residuals. Based
on the results in \citet{Chenetal2014}, $\ddot{{\varepsilon}}_{i}$
follows a closed skew-normal distribution: 
\begin{align}
\mathrm{CSN}_{T-1,T}\Bigg( & {0}_{T-1},(\sigma_{ui}^{2}+\sigma_{vi}^{2})\left(I_{T-1}-\frac{1}{T-1}\iota_{T-1}\iota_{T-1}'\right),\nonumber \\
 & -\frac{\sigma_{ui}/\sigma_{vi}}{\sqrt{\sigma_{ui}^{2}+\sigma_{vi}^{2}}}\begin{pmatrix}I_{T-1}\\
-\iota_{T-1}'
\end{pmatrix},{0}_{T},I_{T}+\frac{\sigma_{ui}^{2}}{T\sigma_{vi}^{2}}\iota_{T}\iota_{T}'\Bigg),\label{eq:likelihoodFE}
\end{align}
where ${0}_{T-1}$ is a $(T-1)\times1$ vector of zeros, and $\iota_{T}$
is a $T\times1$ vector of ones. The notation $\mathrm{CSN}_{p,q}$
denotes the closed skew-normal distribution with the following density
for a $p$-dimensional random variable ${S}$: 
\[
f_{\text{CNS}}({s})=C\,\phi_{p}({s};{\mu},\varSigma)\,\Phi_{q}(D({s}-{\mu});{\upsilon},\varDelta),
\]
where: 
\begin{itemize}
\item $\phi_{p}(\cdot;{\mu},\varSigma)$ is the density of a $p$-dimensional
normal distribution, 
\item $\Phi_{q}(\cdot;{\upsilon},\varDelta)$ is the CDF of a $q$-dimensional
normal distribution, 
\item ${\mu}\in\mathbb{R}^{p}$, $\varSigma\in\mathbb{R}^{p\times p}$, 
\item $D\in\mathbb{R}^{q\times p}$, ${\upsilon}\in\mathbb{R}^{q}$, and
$\varDelta\in\mathbb{R}^{q\times q}$. 
\end{itemize}
Since $\ddot{\varepsilon}_{it}\approx\ddot{y}_{it}-\ddot{z}_{it}'\pi_{i}^{0}$,
we can estimate the parameters $(\pi_{i}^{0},\sigma_{ui},\sigma_{vi})$
by substituting the expression for $\ddot{\varepsilon}_{it}$ into
the density function in \eqref{eq:likelihoodFE} and applying MLE.
This yields consistent estimators $\left(\hat{\pi}_{i}^{0},\hat{\sigma}_{ui},\hat{\sigma}_{vi}\right)$.

\subsubsection*{Step $\mathbf{2}^{*}$: Classification}

From Step $1^{*}$, we obtain individual estimates for each firm:
\[
\left(\hat{\pi}_{1}^{0},\hat{\sigma}_{u1},\hat{\sigma}_{v1}\right),\left(\hat{\pi}_{2}^{0},\hat{\sigma}_{u2},\hat{\sigma}_{v2}\right),\ldots,\left(\hat{\pi}_{N}^{0},\hat{\sigma}_{uN},\hat{\sigma}_{vN}\right).
\]
As in the original Step 2, given a prespecified number of groups $K$,
we estimate the group membership structure. The estimated partition
of the $N$ firms is denoted by: 
\[
\left(\hat{G}_{1|K},\hat{G}_{2|K},\ldots,\hat{G}_{K|K}\right),
\]
which forms a disjoint partition of the index set $\left\{ 1,2,\ldots,N\right\} $.

\subsubsection*{Step $\mathbf{3}^{*}$: Post-Classification Estimation and Determination
of $K^{\ast}$}

As in Step 3, we set the number of sieve terms to $\underline{m}$,
which is substantially larger than $m$. We define the new set of
regressors as $\underline{z}_{it}$: 
\[
\underline{z}_{it}=\left[\mathbb{B}_{-0}^{\underline{m}}(\tau_{t})^{\prime},\left(x_{it}\otimes\mathbb{B}^{\underline{m}}(\tau_{t})\right)^{\prime}\right]^{\prime},
\]
and approximate $\alpha(\cdot)$ and $\beta(\cdot)$ accordingly.

Within each group, say $\hat{G}_{k|K}$ for $1\leq k\leq K$, we perform
post-classification estimation using MLE by maximizing the CSN density
across all observations within the group. Specifically, 
\[
\left(\hat{\pi}_{(k|K)},\hat{\sigma}_{u(k|K)}^{2},\hat{\sigma}_{v(k|K)}^{2}\right)=\arg\max_{(\pi,\delta_{u}^{2},\delta_{v}^{2})}\sum_{i\in\hat{G}_{k|K}}\log f_{\text{CNS}}\left(\ddot{y}_{i}-\underline{\ddot{z}}_{i}\pi;\delta_{u}^{2},\delta_{v}^{2}\right),
\]
where $\ddot{y}_{i}$ and $\underline{\ddot{z}}_{i}$ collect the
first $T-1$ elements of $\ddot{y}_{it}$ and $\underline{\ddot{z}}_{it}$,
respectively, and $f_{\text{CNS}}(\cdot;\delta_{u}^{2},\delta_{v}^{2})$
denotes the CSN density evaluated at the specified variance parameters.

To determine the number of groups, we use the following information
criterion: 
\begin{align*}
\text{IC}_{\text{FE}}(K,\lambda_{NT}^{\text{FE}})=-\sum_{k=1}^{K}\left\{ \sum_{i\in\hat{G}_{k|K}}\log f_{\text{CNS}}\left(\ddot{y}_{i}-\underline{\ddot{z}}_{i}\hat{\pi}_{(k|K)};\hat{\sigma}_{u(k|K)}^{2},\hat{\sigma}_{v(k|K)}^{2}\right)\right\} +\lambda_{NT}^{\text{FE}}K,
\end{align*}
where $\lambda_{NT}^{\text{FE}}$ is an appropriate penalty term.

The optimal number of groups is chosen as 
\[
\hat{K}(\lambda_{NT})=\arg\min_{K=0,1,\ldots,\bar{K}}\text{IC}_{\text{FE}}(K,\lambda_{NT}),
\]
for a suitably chosen upper bound $\bar{K}$. For simplicity, we denote
this as $\hat{K}$. The final parameter estimates are then given by
\[
\hat{\vartheta}_{(k|\hat{K})}=\left(\hat{\pi}_{(k|\hat{K})},\hat{\sigma}_{u(k|\hat{K})},\hat{\sigma}_{v(k|\hat{K})}\right),\quad k=1,2,\ldots,\hat{K}.
\]

An investigation of the small sample properties of this procedure
is beyond the scope of this paper and is left for future research.


\subsection{Nonparametric Fixed Effects Model }\label{APP:FE2}

One strand of the panel FE SF literature attempts to avoid imposing
distributional assumptions on the error terms; see, e.g., \citet{ZhouEtal2020}.
We take the framework in that paper as an example to demonstrate that
our approach can be readily extended to such settings. The model remains
the same as in the previous section: 
\[
\begin{aligned}y_{it} & =\alpha_{i}^{0}+\alpha_{i}(\tau_{t})+x_{it}'\beta_{i}(\tau_{t})+\varepsilon_{it}\\
 & =\alpha_{i}^{0}+\alpha_{i}(\tau_{t})+x_{it}'\beta_{i}(\tau_{t})+v_{it}-u_{it},
\end{aligned}
\]
with the constraint $\int_{0}^{1}\alpha_{i}(\tau_{t})\,d\tau_{t}=0$.
As in the last section, this normalization is innocuous due to the
fixed effects $\alpha_{i}^{0}$.

As noted in \citet{ZhouEtal2020}, the cost of this relaxation is
that the covariates influencing the inefficiency term differ from
those appearing in the frontier function. To accommodate this, we
denote the covariate that solely affects the inefficiency term by
$h_{it}$ and, without loss of generality, assume that $h_{it}$ is
univariate to simplify notation and analysis. Formally, we impose
the restriction: 
\[
\text{E}(u_{it}\mid h_{it},x_{i})=\text{E}(u_{it}\mid h_{it})\equiv\varpi_{i}(h_{it}).
\]
For the noise term $v_{it}$, we assume the usual conditional moment
condition: 
\[
\text{E}(v_{it}\mid h_{it},x_{i})=0.
\]
We normalize $h_{it}$ so that its support is $[0,1]$, which enables
the use of orthogonal basis functions, as introduced in Section~\ref{SEC:serial_approximation},
to approximate $\varpi_{i}(h_{it})$. The key idea in \citet{ZhouEtal2020}
is to isolate $\varpi_{i}(h_{it})$ from the inefficiency term, yielding
a new error term with zero conditional mean. Specifically, 
\[
\begin{aligned}y_{it} & =\alpha_{i}^{0}+\alpha_{i}(\tau_{t})+x_{it}'\beta_{i}(\tau_{t})-\varpi_{i}(h_{it})+v_{it}-\left[u_{it}-\varpi_{i}(h_{it})\right]\\
 & \equiv\alpha_{i}^{0}+\alpha_{i}(\tau_{t})+x_{it}'\beta_{i}(\tau_{t})-\varpi_{i}(h_{it})+\varepsilon_{it}^{*}.
\end{aligned}
\]
To identify $\varpi_{i}$, we impose the additional normalization:
\[
\int_{0}^{1}\varpi_{i}(h)\,\text{d}h=0.
\]
This normalization implies that inefficiency levels cannot be directly
compared across firms, which is a limitation of the approach in \citet{ZhouEtal2020}.
However, one can still analyze how inefficiency varies with changes
in $h$. We refer readers to the original paper for further details.
Under the above assumptions, the new error term $\varepsilon_{it}^{*}$
satisfies: 
\[
\text{E}(\varepsilon_{it}^{*}\mid h_{it},x_{i})=0,
\]
which behaves like a standard error term in fixed effects panel data
models.

In the present setting, our goal is to estimate and classify both
the frontier function and the conditional mean of the inefficiency
term. The variances of the error components ($v_{it}$ and $\varepsilon_{it}^{*}$)
are of secondary importance. This contrasts with earlier settings,
where the variance of $v_{it}$ was essential for the likelihood function
in MLE for estimating inefficiency distribution. In our case, however,
we directly estimate the conditional mean of the inefficiency $\text{E}(u_{it}\mid h_{it})$
as $\varpi_{i}(h_{it})$. With this in mind, we define the group structure
based solely on the frontier and the conditional mean of the inefficiency
term. Specifically, we assume that there are $K^{*}$ distinct groups
of parameter sets, and each firm belongs to exactly one of these groups.
Formally, we write: 
\[
\left\{ \alpha_{i}(\tau_{t}),\beta_{i}(\tau_{t}),\varpi_{i}(\cdot)\right\} =\sum_{k=1}^{K^{*}}\left\{ \alpha_{(k)}^{*}(\tau_{t}),\beta_{(k)}^{*}(\tau_{t}),\varpi_{(k)}^{*}(\cdot)\right\} \cdot\boldsymbol{1}(i\in G_{k}).
\]
Each parameter set across the $K^{*}$ groups is distinct, and each
firm is uniquely assigned to one of these groups.

The approximation using orthogonal basis functions can be conducted
in a similar manner. For each firm $i$, we have: 
\[
\begin{aligned}y_{it} & =\alpha_{i}^{0}+\alpha_{i}(\tau_{t})+x_{it}'\beta_{i}(\tau_{t})-\varpi_{i}(h_{it})+\varepsilon_{it}^{*}\\
 & \approx\alpha_{i}^{0}+\mathbb{B}_{-0}^{m}(\tau_{t})'\pi_{i0}^{0}+\sum_{l=1}^{p}x_{itl}\mathbb{B}^{m}(\tau_{t})'\pi_{il}^{0}+\mathbb{B}_{-0}^{m}(h_{it})'\pi_{ip+1}^{0}+\varepsilon_{it}^{*}\\
 & \equiv\alpha_{i}^{0}+\left[\mathbb{B}_{-0}^{m}(\tau_{t})',(x_{it}\otimes\mathbb{B}^{m}(\tau_{t}))',\mathbb{B}_{-0}^{m}(h_{it})'\right]\pi_{i}^{0}+\varepsilon_{it}^{*}\\
 & \equiv\alpha_{i}^{0}+z_{it}'\pi_{i}^{0}+\varepsilon_{it}^{*},
\end{aligned}
\]
where we slightly abuse notation by reusing $z_{it}$ and $\pi_{i}^{0}$,
although they represent different quantities in this section. Specifically,
\[
z_{it}\equiv\left[\mathbb{B}_{-0}^{m}(\tau_{t})',(x_{it}\otimes\mathbb{B}^{m}(\tau_{t}))',\mathbb{B}_{-0}^{m}(h_{it})'\right]'\quad\text{and}\quad\pi_{i}^{0}\equiv\left(\pi_{i0}^{0\prime},\pi_{i1}^{0\prime},\dots,\pi_{ip}^{0\prime},\pi_{ip+1}^{0\prime}\right)'.
\]

With this setup, the classification and post-classification estimation
become straightforward and follow the essentially same procedure as
Steps 1--3 in Section~\ref{SEC:Estimation}.

\subsubsection*{Step $\mathbf{1}^{\bigstar}$: Individual Estimation}

Like Step $1^{*},$ we begin by applying the within transformation
for each firm $i.$ We reuse following notation: 
\[
\ddot{y}_{it}=y_{it}-\frac{1}{T}\sum_{t=1}^{T}y_{it},\quad\ddot{z}_{it}=z_{it}-\frac{1}{T}\sum_{t=1}^{T}z_{it},
\]
with $\ddot{\varepsilon}_{it}^{*}$ defined analogously. The transformed
model becomes: 
\[
\ddot{y}_{it}\approx\ddot{z}_{it}'\pi_{i}^{0}+\ddot{\varepsilon}_{it}^{*},t=1,2,...,T.
\]
We regress $\ddot{y}_{it}$ on $\ddot{z}_{it}$ for each fixed $i$
using observations of $t=1,2,...,T$.

Specifically, 
\[
\widehat{\pi}_{i}=\left(\sum_{t=1}^{T}\ddot{z}_{it}\ddot{z}_{it}'\right)^{-1}\left(\sum_{t=1}^{T}\ddot{z}_{it}\ddot{y}_{it}\right),
\]
based on which we form groups.

\subsubsection*{Step $\mathbf{2}^{\bigstar}$: Classification}

From Step $1^{\bigstar}$, we obtain individual estimates for each
firm: 
\[
\widehat{\pi}_{1},\widehat{\pi}_{2},...,\widehat{\pi}_{N}.
\]
Given a pre-specified number of groups $K$, we estimate the group
membership structure using the HAC algorithm. The estimated partition
of the $N$ firms is denoted by: 
\[
\left(\hat{G}_{1|K},\hat{G}_{2|K},\ldots,\hat{G}_{K|K}\right),
\]
which forms a disjoint partition of the index set $\left\{ 1,2,\ldots,N\right\} $.

\subsubsection*{Step $\mathbf{3}^{\bigstar}$: Post-Classification Estimation and
Determination of $K^{\ast}$}

As in Step 3, we set the number of sieve terms to $\underline{m}$,
which is substantially larger than $m$. We define the new set of
regressors as $\underline{z}_{it}$: 
\[
\underline{z}_{it}=\left[\mathbb{B}_{-0}^{\underline{m}}(\tau_{t})^{\prime},\left(x_{it}\otimes\mathbb{B}^{\underline{m}}(\tau_{t})\right)^{\prime},\mathbb{B}_{-0}^{\underline{m}}\left(h_{it}\right)^{\prime}\right]^{\prime},
\]
and approximate $\alpha(\cdot)$, $\beta(\cdot)$ and $\varpi\left(\cdot\right)$
accordingly.

Within each group, say $\hat{G}_{k|K}$ for $1\leq k\leq K$, we perform
the post-classification within estimation. Specifically, 
\[
\hat{\pi}_{(k|K)}=\left(\sum_{i\in\hat{G}_{k|K}}\sum_{t=1}^{T}\ddot{\underline{z}}_{it}\ddot{\underline{z}}_{it}'\right)^{-1}\left(\sum_{i\in\hat{G}_{k|K}}\sum_{t=1}^{T}\ddot{\underline{z}}_{it}\ddot{y}_{it}\right).
\]

To determine the number of groups, we use the following information
criterion: 
\begin{align*}
 & \text{IC}_{\text{FENP}}(K,\lambda_{NT}^{\text{FENP}})\\
= & -\sum_{k=1}^{K}\left\{ \frac{N_{k}T}{2}\log\left[\frac{1}{\left(T-1\right)N_{k}}\sum_{i\in\hat{G}_{k|K}}\sum_{t=1}^{T}\left(\ddot{y}_{it}-\underline{\ddot{z}}_{it}^{\prime}\hat{\pi}_{(k|K)}\right)^{2}\right]\right\} +\lambda_{NT}^{\text{FENP}}K,
\end{align*}
where $\lambda_{NT}^{\text{FENP}}$ is an appropriate penalty term.

The optimal number of groups is chosen as 
\[
\hat{K}(\lambda_{NT}^{\text{FENP}})=\arg\min_{K=0,1,\ldots,\bar{K}}\text{IC}_{\text{FENP}}(K,\lambda_{NT}^{\text{FENP}}),
\]
for a suitably chosen upper bound $\bar{K}$. For simplicity, we denote
this as $\hat{K}$. The final parameter estimates are then given by
\[
\hat{\pi}_{(k|\hat{K})},\quad k=1,2,\ldots,\hat{K}.
\]

The theoretical properties of the above procedure are relatively straightforward.
However, complications may arise when $h_{it}$ is multidimensional.
In such cases, one may need to employ tensor product bases to approximate
$\varpi(\cdot)$, which could result in a slower convergence rate.
As an alternative, assuming an additive structure for $\varpi(\cdot)$
allows the convergence rate to remain unchanged. A rigorous investigation
of the theoretical properties under these settings, as well as an
analysis of the small-sample performance, is beyond the scope of this
paper and is left for future research.

\subsection{A Group Innocuous Normalization}\label{APP:innocuous}

As discussed in Section \ref{sec:The-Model-More}, the normalization
$\int_{0}^{1}\alpha_{i}\left(s\right)\textrm{d}s=0$ is not innocuous
when $\alpha_{i}\left(s\right)$ varies across firms and possesses
a group structure, because we need to assume that the levels of $\alpha_{\left(k\right)}^{*}\left(\tau_{t}\right)$
(that is, $\int_{0}^{1}\alpha_{\left(k\right)}^{*}\left(s\right)\textrm{d}s$)
are the same across groups. In this section, we generalize the result
in the main paper to the case in which we allow the levels to differ
across groups but they remain the same within each group.

As before, we assume that there are $K^{\ast}$ groups of specific
parameters, and each firm's parameters belong to one of these groups:
\[
\left\{ \alpha_{i}\left(\tau_{t}\right),\beta_{i}\left(\tau_{t}\right),\sigma_{vi}\right\} =\sum_{k=1}^{K^{\ast}}\left\{ \alpha_{(k)}^{\ast}\left(\tau_{t}\right),\beta_{(k)}^{\ast}\left(\tau_{t}\right),\sigma_{v(k)}^{\ast}\right\} \boldsymbol{1}(i\in G_{k}).
\]
Note that our procedure is silent on the classification of the levels
of $\alpha_{i}\left(\tau_{t}\right)$, so the classification is based
solely on the time-varying part of $\alpha_{i}\left(\tau_{t}\right).$
As such, we titled the section ``Group'' innocuous, since it is
not entirely innocuous. In other words, for all $i\in G_{k},$ we
need to assume that $\int_{0}^{1}\alpha_{i}\left(s\right)\textrm{d}s=c_{k}$
(before normalization), which is identical within the group.

We recommend this procedure when, on average, each group contains
at least a few hundred observations, due to the difficulty of uncovering
the mixture structure with very small samples; see, e.g., \citet{OlsonEtal},
\citet{SimarWilson}, and \citet{ChristianEtal2018}.

A consequence of weakening the restriction on the levels is that the
distribution of $\alpha_{i}^{0}-u_{i}$ naturally differs across groups,
because for group $k$, $\int_{0}^{1}\alpha_{i}(s)\,\textrm{d}s=c_{k}$
is absorbed into $\alpha_{i}^{0}$ after normalization. As such, we
assume that for observation $i$ in group $k$, there exists a $\mathcal{K}_{(k)}^{*}\geq1$
such that 
\[
\alpha_{i}^{0}-u_{i}\overset{d}{\sim}\alpha_{(k)(j)}^{0}-\left|N\left(0,\sigma_{u(k)(j)}^{2}\right)\right|\quad\text{with probability }\tau_{(k)j}^{0},\quad j=1,2,\ldots,\mathcal{K}_{(k)}^{*},
\]
where $0<\tau_{(k)j}^{0}<1$ and $\sigma_{u(k)(j)}^{2}>C>0$ for $j=1,2,\ldots,\mathcal{K}_{(k)}^{*}$,
with $\tau_{(k)1}^{0}+\tau_{(k)2}^{0}+\ldots+\tau_{(k)\mathcal{K}_{(k)}^{*}}^{0}=1$,
and $\left(\alpha_{(k)(j)}^{0},\sigma_{u(k)(j)}^{2}\right)$ differ
across $j=1,2,\ldots,\mathcal{K}_{(k)}^{*}$. While the notation is
rather tedious, the estimation remains straightforward. We simply
estimate the distribution of $\alpha_{i}^{0}-u_{i}$ for each group
separately using the previous strategy.

We present the details of the procedure below. The generalization
that allows the levels to differ across groups has no impact on the
first three steps, as the intercept term is not involved.

\subsubsection*{Step $\mathbf{1}^{\blacklozenge}$ : Individual Estimation}

Same as Step 1.

\subsubsection*{Step $\mathbf{2}^{\blacklozenge}$: Classification}

Same as Step 2.

\subsubsection*{Step $\mathbf{3}^{\blacklozenge}$: Post-Classification Estimation
and Determination of $K^{\ast}$}

Same as Step 3.

\subsubsection*{Step $\mathbf{4}^{\blacklozenge}$: Estimation of $\alpha_{\left(k\right)(j)}^{0},\sigma_{u\left(k\right)(j)}^{2},\text{ and }\tau_{\left(k\right)j}^{0}$}

For group $k,$ assuming that the error term comes from $\mathcal{K}{}_{\left(k\right)}$
distributions, we obtain an estimate of $\left(\alpha_{\left(k\right)(1)}^{0},\sigma_{u\left(k\right)(1)}^{2},...,\alpha_{\left(k\right)(\mathcal{K}_{\left(k\right)})}^{0},\sigma_{u\left(k\right)(\mathcal{K}_{\left(k\right)})}^{2},\tau_{\left(k\right)1}^{0},...,\tau_{\left(k\right)\mathcal{K}_{\left(k\right)}-1}^{0}\right)$
using MLE as follows: 
\begin{align*}
 & \left(\hat{\alpha}_{\left(k\right)(1)}^{0},\hat{\sigma}_{u\left(k\right)(1)}^{2},...,\hat{\alpha}_{\left(k\right)(\mathcal{K}_{\left(k\right)})}^{0},\hat{\sigma}_{u\left(k\right)(\mathcal{K}_{\left(k\right)})}^{2},\hat{\tau}_{\left(k\right)1},...,\hat{\tau}_{\left(k\right)\mathcal{K}_{\left(k\right)}-1}\right)\\
= & \arg\max_{(s,\delta_{u}^{2},\tau)}\sum_{i\in\hat{G}_{k|\hat{K}}}\log\tilde{f}\left(y_{i}\left\vert x_{i};s_{1},\delta_{u\left(1\right)}^{2},...,s_{\mathcal{K}_{\left(k\right)}},\delta_{u\left(\mathcal{K}_{\left(k\right)}\right)}^{2},\tau_{1},...,\tau_{\mathcal{K}_{\left(k\right)}-1},\hat{\vartheta}_{(k|\hat{K})}\right.\right)
\end{align*}
where $\tilde{f}$ is the likelihood function, defined in (\ref{EQ:log_approx_p}),
and we incorporate estimates from Step 3$^{\blacklozenge}$.

Collecting results from Steps 4$^{\blacklozenge}$, we proceed to
Step 5$^{\blacklozenge}$ to determine the specification of $\alpha_{i}^{0}-u_{i}$
for each group.

\subsubsection*{Step $\mathbf{5}^{\blacklozenge}$: Determination of the Distributional
Structures of the Inefficiency Term}

The information criterion for group $k$ is constructed as 
\[
\widetilde{\mathrm{IC}}_{k}(\mathcal{K}_{\left(k\right)},\tilde{\lambda}_{NT})=-\sum_{i\in\hat{G}_{k|\hat{K}}}\log\tilde{f}\left(y_{i}\left\vert x_{i};s_{1},\delta_{u\left(1\right)}^{2},...,s_{\mathcal{K}_{\left(k\right)}},\delta_{u\left(\mathcal{K}_{\left(k\right)}\right)}^{2},\tau_{1},...,\tau_{\mathcal{K}_{\left(k\right)}-1},\hat{\vartheta}_{(k|\hat{K})}\right.\right)+\mathcal{K}_{\left(k\right)}\tilde{\lambda}_{NT},
\]
where $\tilde{\lambda}_{NT}$ is a suitable penalty term and $k=1,2,...,\hat{K}.$
We take 
\[
\hat{\mathcal{K}}_{\left(k\right)}(\tilde{\lambda}_{NT})=\arg\min_{\mathcal{K}_{\left(k\right)}=0,1,\ldots,\bar{\mathcal{K}}}\widetilde{\mathrm{IC}}_{k}(\mathcal{K}_{\left(k\right)},\tilde{\lambda}_{NT}),
\]
and we write $\hat{\mathcal{K}}_{\left(k\right)}$ for short.

Collecting results across $\hat{K}$ groups, the estimated parameters
are 
\[
\left\{ \left(\hat{\alpha}_{\left(k\right)(1)}^{0},\hat{\sigma}_{u\left(k\right)(1)}^{2},...,\hat{\alpha}_{\left(k\right)(\hat{\mathcal{K}}_{\left(k\right)})}^{0},\hat{\sigma}_{u\left(k\right)(\hat{\mathcal{K}}_{\left(k\right)})}^{2},\hat{\tau}_{\left(k\right)1},...,\hat{\tau}_{\left(k\right)\hat{\mathcal{K}}_{\left(k\right)}-1}\right)\right\} _{k=1}^{\hat{K}}.
\]

\bigskip{}

The theoretical properties of the above procedure remain the same
as those in the main body of the paper, provided that $N_{k}\propto N$
for $k=1,2,\ldots,K^{*}$. Before we conclude this section, we emphasize
that we recommend this procedure only when practitioners have ample
observations available, e.g., at least a few hundred for each group.

\subsection{Latent Structure with Zero Inefficiency}\label{APP:ZISF}

\citet{KumbhakarEtal2013} and \citet{RhoSchmidt} proposed the zero
inefficiency stochastic frontier (ZISF) model, in which the inefficiency
term is exactly zero with a positive probability. Testing the validity
of the ZISF model amounts to examining whether the variance of the
inefficiency term is zero in one component of the mixture distribution.
This leads to a nonstandard testing problem, as the null hypothesis
places the parameter on the boundary of the parameter space. Consequently,
standard inference methods such as the $t$-test may not be appropriate,
and alternative testing procedures, as suggested by \citet{KumbhakarEtal2013},
may be required. Compounding the challenge, \citet{RhoSchmidt} highlighted
several identification issues inherent in this framework. Given these
complications, it is worthwhile to explore this question from a different
perspective.

We propose to address this problem using information criteria, following
a similar approach to that used in the main body of the paper. To
this end, we incorporate the ZISF assumption into our model. Naturally,
this introduces an additional layer of complexity and difficulty.
For clarity of exposition, we focus on a mixture distribution with
at most two components. While extending to more components is conceptually
straightforward, it may distract from the central focus of this section.
We adopt the model from the main body of the paper, maintaining the
same group structure on the frontier. The only difference lies in
the assumption regarding the distribution of $\alpha_{i}^{0}-u_{i}$. 

Specifically, we consider the following three specifications and propose
an IC to choose one. 
\begin{align*}
\textrm{Spec1: } & \alpha_{i}^{0}-u_{i}\overset{d}{\sim}\alpha^{0}-N\left(0,\sigma_{u}^{2}\right),\\
\textrm{Spec2: } & \alpha_{i}^{0}-u_{i}\overset{d}{\sim}\begin{cases}
\begin{array}{c}
\alpha_{\left(1\right)}^{0}\\
\alpha_{\left(2\right)}^{0}-N\left(0,\sigma_{u(2)}^{2}\right)
\end{array} & \begin{array}{c}
\textrm{ with probability }\tau^{0}\\
\textrm{ with probability }1-\tau^{0}
\end{array}\end{cases},\textrm{ and}\\
\textrm{Spec3: } & \alpha_{i}^{0}-u_{i}\overset{d}{\sim}\begin{cases}
\begin{array}{c}
\alpha_{\left(1\right)}^{0}-N\left(0,\sigma_{u(1)}^{2}\right)\\
\alpha_{\left(2\right)}^{0}-N\left(0,\sigma_{u(2)}^{2}\right)
\end{array} & \begin{array}{c}
\textrm{ with probability }\tau^{0}\\
\textrm{ with probability }1-\tau^{0}
\end{array},\end{cases}
\end{align*}
 where $0<\tau^{0}<1,$ $\sigma_{u}^{2},\sigma_{u(1)}^{2},\sigma_{u(2)}^{2}>0,$
and $\left(\alpha_{\left(1\right)}^{0},\sigma_{u(1)}^{2}\right)\neq\left(\alpha_{\left(2\right)}^{0},\sigma_{u(2)}^{2}\right).$
Note that S2 is the assumption in the ZISF model. We do not consider
the case when $\alpha_{i}^{0}-u_{i}=\alpha^{0}$ because the fully
efficient case is rare in practice. 

We slightly abuse notation by using the same symbols for parameters
in both Spec2 and Spec3. However, the context should make it clear
which set of parameters is being referenced.

Recall that the only change in this setting is the assumption regarding
the distribution of $\alpha_{i}^{0}-u_{i}$. Consequently, the estimation
of the frontier parameters proceeds exactly as before --- specifically,
following Steps 1 to 3 in Section \ref{SEC:Estimation}. For clarity,
we relabel these steps as Step $1^{\blacktriangledown},$ Step $2^{\blacktriangledown},$
and Step $3^{\blacktriangledown}$.

\subsubsection*{Steps $\mathbf{1}^{\blacktriangledown},$ $\mathbf{2}^{\blacktriangledown},$
and $\mathbf{3}^{\blacktriangledown}$}

Same as Steps 1 to 3.

\subsubsection*{Step $\mathbf{4}^{\blacktriangledown}$: Estimation of $\alpha^{0},\sigma_{u}^{2},\text{ and }\tau^{0}$}

For Spec1, $\left(\alpha^{0},\sigma_{u}^{2}\right)$ can be estimated
the same as the way in Step 4. That is, 
\[
\left(\hat{\alpha}^{0},\hat{\sigma}_{u}^{2}\right)=\arg\max_{(s,\delta_{u}^{2})}\sum_{k=1}^{\hat{K}}\sum_{i\in\hat{G}_{k|\hat{K}}}\log f\left(y_{i}\mid x_{i};s,\delta_{u}^{2},\hat{\vartheta}_{(k|\hat{K})}\right),
\]
where $f\left(y_{i}\mid x_{i};s,\delta_{u}^{2},\hat{\vartheta}_{(k|\hat{K})}\right)$
is defined in (\ref{EQ:likelihood_appro}), and we use the post-classification
estimates $\hat{\vartheta}_{(k|\hat{K})}$ from Step $3^{\blacktriangledown}$. 

For Spec2, $\left(\alpha_{\left(1\right)}^{0},\alpha_{\left(2\right)}^{0},\sigma_{u(2)}^{2},\tau^{0}\right)$
can be obtained as the restricted estimation in Step 4' (by forcing
$\sigma_{u(1)}^{2}=0$) when $\mathcal{K}=2$. That is
\begin{align*}
\left(\hat{\alpha}_{(1)}^{0},\hat{\alpha}_{(2)}^{0},\hat{\sigma}_{u(2)}^{2},\hat{\tau}\right) & =\arg\max_{(s,\delta_{u}^{2},\tau)}\sum_{k=1}^{\hat{K}}\sum_{i\in\hat{G}_{k|\hat{K}}}\log\tilde{f}\left(y_{i}\left\vert x_{i};s_{1},0,s_{2},\delta_{u\left(2\right)}^{2},\tau,\hat{\vartheta}_{(k|\hat{K})}\right.\right),
\end{align*}
where $\tilde{f}$ is the likelihood function, defined in (\ref{EQ:log_approx_p}). 

For Spec3, $\left(\alpha_{\left(1\right)}^{0},\sigma_{u(1)}^{2},\alpha_{\left(2\right)}^{0},\sigma_{u(2)}^{2},\tau^{0}\right)$
is the estimation in Step 4' when $\mathcal{K}=2$. Thus,

\[
\left(\hat{\alpha}_{(1)}^{0},\hat{\sigma}_{u(1)}^{2},\hat{\alpha}_{(2)}^{0},\hat{\sigma}_{u(2)}^{2},\hat{\tau}\right)=\arg\max_{(s,\delta_{u}^{2},\tau)}\sum_{k=1}^{\hat{K}}\sum_{i\in\hat{G}_{k|\hat{K}}}\log\tilde{f}\left(y_{i}\left\vert x_{i};s_{1},\delta_{u\left(1\right)}^{2},s_{2},\delta_{u\left(2\right)}^{2},\tau,\hat{\vartheta}_{(k|\hat{K})}\right.\right).
\]

Collecting results from Step $4^{\blacktriangledown}$, we proceed
to Step $5^{\blacktriangledown}$ to determine the specification of
$\alpha_{i}^{0}-u_{i}$.

\subsubsection*{Step \textmd{$\mathbf{5}^{\blacktriangledown}$}: Determination of
the Distributional Structures of the Inefficiency Term}

We calculate three IC values to determine the distribution of $\alpha_{i}^{0}-u_{i}$. 

We consider Spec1 and Spec3 first, because they are special cases
of Step 4' when $\mathcal{K}=1$ and $\mathcal{K}=2$, respectively.
Using this insight, we define the information criteria as follows:
\\
For Spec1:
\[
\widetilde{\textrm{IC}}_{1}\left(\tilde{\lambda}_{NT}\right)=-\sum_{k=1}^{\hat{K}}\sum_{i\in\hat{G}_{k|\hat{K}}}\log f\left(y_{i}\left\vert x_{i};\hat{\alpha}^{0},\hat{\sigma}_{u}^{2}\hat{\vartheta}_{(k|\hat{K})}\right.\right)+\tilde{\lambda}_{NT}.
\]
For Spec3:
\[
\widetilde{\textrm{IC}}_{3}\left(\tilde{\lambda}_{NT}\right)=-\sum_{k=1}^{\hat{K}}\sum_{i\in\hat{G}_{k|\hat{K}}}\log\tilde{f}\left(y_{i}\left\vert x_{i};\hat{\alpha}_{(1)}^{0},\hat{\sigma}_{u(1)}^{2},\hat{\alpha}_{(2)}^{0},\hat{\sigma}_{u(2)}^{2},\hat{\tau},\hat{\vartheta}_{(k|\hat{K})}\right.\right)+2\tilde{\lambda}_{NT}.
\]
Both $\textrm{IC}_{1}\left(\tilde{\lambda}_{NT}\right)$ and $\textrm{IC}_{3}\left(\tilde{\lambda}_{NT}\right)$
are special cases of $\widetilde{\mathrm{IC}}(\mathcal{K},\tilde{\lambda}_{NT})$
(defined in (\ref{eq:IC_tide2})) when $\mathcal{K}=1$ and 2, respectively.

Note that Spec1, Spec2 and Spec3 contain 2, 4, and 5 parameters, respectively.
Thus, the number of parameters increases by 3 when moving from Spec1
to Spec3, and by 2 when moving from Spec1 to Spec2. Based on this
observation, we define:
\[
\widetilde{\textrm{IC}}_{2}\left(\tilde{\lambda}_{NT}\right)=-\sum_{k=1}^{\hat{K}}\sum_{i\in\hat{G}_{k|\hat{K}}}\log\tilde{f}\left(y_{i}\left\vert x_{i};\hat{\alpha}_{(1)}^{0},0,\hat{\alpha}_{(2)}^{0},\hat{\sigma}_{u(2)}^{2},\hat{\tau},\hat{\vartheta}_{(k|\hat{K})}\right.\right)+\frac{5}{3}\tilde{\lambda}_{NT}.
\]
\\
We then select the model that minimizes the information criterion:
\[
\hat{l}=\arg\min_{l=1,2,3}\widetilde{\textrm{IC}}_{l}\left(\tilde{\lambda}_{NT}\right),
\]
and adopt Spec$\hat{l}$ accordingly.

A further investigation of this procedure is left for future.

\section{Allowing More Than Two Components in the Mixture}\label{APP:more_mix}

As discussed in the main body of the paper, estimating mixtures with
more than two components, as in Step~4$'$, is numerically challenging.
The difficulty arises because the original log-likelihood function
is highly nonlinear and cannot be effectively optimized when the
number of mixture components exceeds two.

The motivation for the alternative method stems from the observation
that the original log-likelihood function is overly complicated, and
we seek a numerically less demanding approach.

The procedure is as follows. For each $i=1,2,\ldots,N$, we first
compute
\begin{align}
\widehat{\alpha_{i}^{0}-u_{i}}
&=\frac{1}{T}\sum_{t=1}^{T}\left(y_{it}-\underline{z}_{it}'\hat{\pi}_{(k|K)}\right),
\quad\text{for } i\in\hat{G}_{k|K}\nonumber \\[6pt]
&=\alpha_{i}^{0}-u_{i}+\frac{1}{T}\sum_{t=1}^{T}v_{it}
   +O_{P}\!\left(\sqrt{\tfrac{\underline{m}}{N_{k}T}}\right)\nonumber \\[6pt]
&=\alpha_{i}^{0}-u_{i}
   +O_{P}\!\left(\tfrac{1}{\sqrt{T}}+\sqrt{\tfrac{\underline{m}}{N_{k}T}}\right)\nonumber \\[6pt]
&\approx \alpha_{i}^{0}-u_{i}.
\label{eq:hat=a-u}
\end{align}

Recall that for $\alpha_{i}^{0}-u_{i}$ there exists an integer $\mathcal{K}^{*}\geq 1$
such that, with probability $\tau_{j}^{0}$, it follows the distribution
\[
\alpha_{(j)}^{0}-\lvert N(0,\sigma_{u(j)}^{2})\rvert,\quad j=1,2,\ldots,\mathcal{K}^{*},
\]
where $\bigl(\alpha_{(j)}^{0},\sigma_{u(j)}^{2}\bigr)$, $j=1,2,\ldots,\mathcal{K}^{*}$,
are distinct parameter pairs, with $0<\tau_{j}^{0}<1$ for
$j=1,2,\ldots,\mathcal{K}^{*}-1$, and
$1-\tau_{1}^{0}-\cdots-\tau_{\mathcal{K}^{*}-1}^{0}>0$.

The likelihood function of $\alpha_{i}^{0}-u_{i}$ then takes the
form
\begin{align*}
&f_{\textrm{mix}}\!\left(s \,\middle|\,
\alpha_{(1)}^{0},\sigma_{u(1)}^{2},\ldots,\alpha_{(\mathcal{K}^{*})}^{0},
\sigma_{u(\mathcal{K}^{*})}^{2},\tau_{1}^{0},\ldots,\tau_{\mathcal{K}^{*}-1}^{0}\right) \\
&\quad= \sum_{j=1}^{\mathcal{K}^{*}}\tau_{j}^{0}\,
\frac{2}{\sqrt{2\pi}}
\exp\!\left\{ -\frac{\bigl(s-\alpha_{(j)}^{0}\bigr)^{2}}{2\sigma_{u(j)}^{2}} \right\}
\mathbf{1}\!\left(s\leq\alpha_{(j)}^{0}\right),
\end{align*}
with $\tau_{\mathcal{K}^{*}}^{0}
=1-\tau_{1}^{0}-\cdots-\tau_{\mathcal{K}^{*}-1}^{0}$.
This likelihood is considerably simpler and more tractable than the
one introduced in Step~4$'$.

We propose to estimate the mixture structure of $\alpha_{i}^{0}-u_{i}$
by using $\widehat{\alpha_{i}^{0}-u_{i}}$ as a surrogate and maximizing
the likelihood function defined above. The information criterion can
then be constructed analogously to before:
\[
\mathrm{IC}_{\mathrm{mix}}(\mathcal{K},\tilde{\lambda}_{NT})
=-\sum_{i=1}^{N}\log f_{\mathrm{mix}}\!\left(
\widehat{\alpha_{i}^{0}-u_{i}} \,\middle|\,
\hat{\alpha}_{(1)}^{0},\hat{\sigma}_{u(1)}^{2},\ldots,
\hat{\alpha}_{(\mathcal{K})}^{0},\hat{\sigma}_{u(\mathcal{K})}^{2},
\hat{\tau}_{1}^{0},\ldots,\hat{\tau}_{\mathcal{K}-1}^{0}\right)
+\mathcal{K}\tilde{\lambda}_{NT},
\]
for a given number of mixture components $\mathcal{K}$.

Finally, we select the number of components by
\[
\hat{\mathcal{K}}(\tilde{\lambda}_{NT})
=\arg\min_{\mathcal{K}=0,1,\ldots,\bar{\mathcal{K}}}
\mathrm{IC}_{\mathrm{mix}}(\mathcal{K},\tilde{\lambda}_{NT}),
\]
and write $\hat{\mathcal{K}}$ for short. The corresponding parameter
estimates are
\[
\bigl(\hat{\alpha}_{(1)}^{0},\hat{\sigma}_{u(1)}^{2},\ldots,
\hat{\alpha}_{(\hat{\mathcal{K}})}^{0},\hat{\sigma}_{u(\hat{\mathcal{K}})}^{2},
\hat{\tau}_{1},\ldots,\hat{\tau}_{\hat{\mathcal{K}}-1}\bigr).
\]
Since $T\propto N$, it follows directly that $\tilde{\lambda}_{NT}$
in Proposition~\ref{Prop:classify} remains valid in this setting.
Therefore, in practice we adopt the recommended $\tilde{\lambda}_{NT}$
from Section~\ref{SEC:tuning}.

We conduct simulations to evaluate the small sample properties of
the above procedure. We study DGP 1M, 2M, and 3M, where the number of components in the mixture distribution is two. In addition, we consider DGP1T, DGP2T and DGP3T, where we allow the inefficiency terms to arise from the mixture of three distributions. DGPs 1T, 2T, and 3T share the same frontier and error distribution of $v$ as DGPs 1M, 2M, and 3M, respectively.  Specifically, for DGPs 1T, 2T and 3T, we let $\alpha^{0}-u$ to come from $\alpha_{(1)}^{0} - |N(0,\sigma_{u(1)}^{2})|$ with probability $\tau_{1}$, $\alpha_{(2)}^{0} - |N(0,\sigma_{u(2)}^{2})|$ with probability $\tau_{2}$, and $\alpha_{(3)}^{0} - |N(0,\sigma_{u(3)}^{2})|$ with probability $\tau_{3} $, where $\{\alpha_{(1)}^{0}, \alpha_{(2)}^{0}, \alpha_{(3)}^{0}\} = \{0.5, -1, 2\}$, $\{\sigma_{u(1)}, \sigma_{u(2)}, \sigma_{u(3)}\} = \{1,1,1\}$,  and $\{\tau_{1},\tau_{2},\tau_{3}\} = \{0.3,0.4,0.3\}$. 

Simulation results are reported in Appendix \ref{APP:tables_more}. We find the following: 
\begin{enumerate}
\item Using the previously recommended tuning parameters, the method performs
reasonably well in identifying the correct number of components in
the mixture for the original three DGPs (1M, 2M, and 3M) with two
components, as well as three additional DGPs (1T, 2T, and 3T) with
three components.
\item The parameter estimates converge at the rate of $\sqrt{T}$ when $N\geq T$.
This is because we use $\widehat{\alpha_{i}^{0}-u_{i}}$ as observations,
and $\widehat{\alpha_{i}^{0}-u_{i}}$ converges to $\alpha_{i}^{0}-u_{i}$
at the rate $\sqrt{T}$. In the simulations, we further set $T$ smaller
than $N$. 
\item For fixed $T$, the estimates deteriorate as $N$ increases for the
following reasons:
\begin{enumerate}
\item The theoretical convergence rate is $\sqrt{T}$, so increasing $N$
does not yield improvements in theory.
\item To illustrate the intuition, suppose the mixture has only one component
with $\alpha_{i}^{0}-u_{i}\overset{d}{\sim}\alpha^{0}-|N(0,\sigma_{u}^{2})|$.
Continuing from (\ref{eq:hat=a-u}), we obtain
\begin{align*}
\widehat{\alpha_{i}^{0}-u_{i}}
&=\alpha^{0}-u_{i}+\frac{1}{T}\sum_{t=1}^{T}v_{it}
   +\left(\frac{1}{T}\sum_{t=1}^{T}\underline{z}_{it}'\right)\!
     \left(\pi_{(k|K)}^{0}-\hat{\pi}_{(k|K)}\right)
   +\frac{1}{T}\sum_{t=1}^{T}b_{i0}(\tau_{t}) \\[6pt]
&\equiv \alpha^{0}-u_{i}+\hat{\epsilon}_{i},
\end{align*}
where $b_{i0}(\cdot)$ is the approximation bias term, and
\[
\hat{\epsilon}_{i}\equiv \frac{1}{T}\sum_{t=1}^{T}v_{it}
+\left(\frac{1}{T}\sum_{t=1}^{T}\underline{z}_{it}'\right)\!
   \left(\pi_{(k|K)}^{0}-\hat{\pi}_{(k|K)}\right)
+\frac{1}{T}\sum_{t=1}^{T}b_{i0}(\tau_{t}).
\]
If we treat $\widehat{\alpha_{i}^{0}-u_{i}}$ as observations, then,
analogous to the uniform distribution case, we have a closed-form
solution for $\hat{\alpha}^{0}$:
\begin{align*}
\hat{\alpha}^{0}
&=\max_{i}\left\{\widehat{\alpha^{0}-u_{i}}\right\} \\
&=\alpha^{0}-\min_{i}\left\{u_{i}-\hat{\epsilon}_{i}\right\}.
\end{align*}
Since $\hat{\epsilon}_{i}$ can take negative values, the probability
of observing extreme negative realizations of $\hat{\epsilon}_{i}$
increases with $N$, thereby reducing the accuracy of the estimate.
\end{enumerate}
\end{enumerate}

For these reasons, we recommend that practitioners use this approach
only when the primary goal is to identify the number of mixture components
or to model three or more components, and only when $T$ is sufficiently
large to ensure reliable estimation.

\section{The Approximation of the Likelihood Function }

\label{APP:likeli}

\subsection{The Inefficiency Term with a Unique Distribution}
We derive the likelihood function incorporating approximations of
$\alpha$ and $\beta$. Using the last line of (\ref{EQ:approx}),
we have 
\begin{align}
y_{it} & \approx\alpha_{i}^{0}-u_{i}+z_{it}'\pi_{i}^{0}+v_{it}\nonumber \\
 & =\alpha_{i}^{0}+z_{it}'\pi_{i}^{0}+\varepsilon_{it}.\label{eq:new_appro}
\end{align}
The primary distinction between the above approximation and (\ref{EQ:approx})
is that we do not separate $u_{i}$ from $\varepsilon_{it}$.

We first derive the likelihood function when the distribution of $\alpha_{i}^{0}-u_{i}$
is unique. We adopt the notation used in \citet{YaoZhangKum2019}
for the following presentation. Let 
\[
\sigma_{i}^{2}=\sigma_{vi}^{2}+T\sigma_{u}^{2},
\]
\[
\rho_{i}=\sigma_{u}/\sigma_{vi},
\]
\[
\mu_{i\ast}=-\frac{\sigma_{u}^{2}}{\sigma_{i}^{2}}\sum_{t=1}^{T}\varepsilon_{it},\quad\text{and}
\]
\[
\sigma_{i\ast}^{2}=\frac{\sigma_{u}^{2}\sigma_{vi}^{2}}{\sigma_{i}^{2}},
\]
then 
\[
\sigma_{vi}^{2}=\frac{\sigma_{i}^{2}}{1+T\rho_{i}^{2}}.
\]
We denote the density and cumulative distribution functions of a standard
normal distribution as $\phi(\cdot)$ and $\Phi(\cdot)$, respectively.
Following calculations similar to those in \citet{YaoZhangKum2019},
the density of $\varepsilon_{i}=\left(\varepsilon_{i1},\ldots,\varepsilon_{iT}\right)^{\prime}$
is given by 
\begin{align*}
f\left(\varepsilon_{i};\sigma_{u}^{2},\sigma_{vi}^{2}\right) & =\frac{2}{\sigma_{vi}^{T-1}\sigma_{i}}\left[1-\Phi\left(-\frac{\mu_{i\ast}}{\sigma_{i\ast}}\right)\right]\left[\prod_{t=1}^{T}\phi\left(\frac{\varepsilon_{it}}{\sigma_{vi}}\right)\right]\exp\left(\frac{1}{2}\left(-\frac{\mu_{i\ast}}{\sigma_{i\ast}}\right)^{2}\right),
\end{align*}
which implies 
\begin{align*}
\log f\left(\varepsilon_{i};\sigma_{u}^{2},\sigma_{vi}^{2}\right) & =C-\frac{(T-1)}{2}\log\sigma_{vi}^{2}-\frac{1}{2}\log\left(\sigma_{vi}^{2}+T\sigma_{u}^{2}\right)\\
 & \quad+\log\left[1-\Phi\left(-\frac{\mu_{i\ast}}{\sigma_{i\ast}}\right)\right]+\frac{1}{2}\left(\frac{\mu_{i\ast}}{\sigma_{i\ast}}\right)^{2}-\frac{\sum_{t=1}^{T}\varepsilon_{it}^{2}}{2\sigma_{vi}^{2}},
\end{align*}
for some constant $C$ that does not depend on the parameters to be
estimated. Recall that $\vartheta_{i}=\left(\pi_{i}^{0\prime},\sigma_{vi}^{2}\right)'$.
Using the approximation in (\ref{eq:new_appro}), the log-likelihood
density function for $y_{i}=\left(y_{i1},\ldots,y_{iT}\right)^{\prime}$
is given by 
\begin{align}
\log f\left(y_{i}\mid x_{i};\alpha_{i}^{0},\sigma_{u}^{2},\vartheta_{i}\right) & \approx C-\frac{(T-1)}{2}\log\sigma_{vi}^{2}-\frac{1}{2}\log\left(\sigma_{vi}^{2}+T\sigma_{u}^{2}\right)\label{EQ:likelihood_appro}\\
 & \quad+\log\left[1-\Phi\left(-\frac{\tilde{\mu}_{\ast i}}{\sigma_{\ast i}}\right)\right]+\frac{1}{2}\left(\frac{\tilde{\mu}_{\ast i}}{\sigma_{\ast i}}\right)^{2}-\frac{\sum_{t=1}^{T}\tilde{\varepsilon}_{it}^{2}}{2\sigma_{vi}^{2}},\nonumber 
\end{align}
with 
\begin{align*}
\tilde{\varepsilon}_{it} & =y_{it}-\alpha_{i}^{0}-z_{it}'\pi_{i}^{0},\text{ and}\\
\tilde{\mu}_{i\ast} & =-\frac{\sigma_{u}^{2}}{\sigma_{i}^{2}}\sum_{t=1}^{T}\tilde{\varepsilon}_{it}.
\end{align*}

\subsection{The Inefficiency Term with a Mixture Distribution}

We now consider the case when $\alpha_{i}^{0}-u_{i}$ is distributed
as $\alpha_{(j)}^{0}-\left\vert N\left(0,\sigma_{u(j)}^{2}\right)\right\vert $
with probability $\tau_{j}^{0}$, $j=1,2,...,\mathcal{K}^{*}$. This
log-likelihood function is denoted as $\log\tilde{f}$ and can be
derived as: 
\begin{align}
 & \log\tilde{f}\left(y_{i}\left\vert x_{i};\alpha_{(1)}^{0},\sigma_{u(1)}^{2},...,\alpha_{(\mathcal{K}^{*})}^{0},\sigma_{u(\mathcal{K}^{*})}^{2},\tau_{1}^{0},...,\tau_{\mathcal{K}^{*}-1}^{0},\vartheta_{i}\right.\right)\label{EQ:log_approx_p}\\
= & \log\left[\tau_{1}^{0}f\left(y_{i}\left\vert x_{i};\alpha_{(1)}^{0},\sigma_{u(1)}^{2},\vartheta_{i}\right.\right)+\tau_{2}^{0}f\left(y_{i}\left\vert x_{i};\alpha_{(2)}^{0},\sigma_{u(2)}^{2},\vartheta_{i}\right.\right)+...\right.\nonumber \\
 & \left.+\left(1-\tau_{1}^{0}-...-\tau_{\mathcal{K}^{*}-1}^{0}\right)f\left(y_{i}\left\vert x_{i};\alpha_{(\mathcal{K}^{*})}^{0},\sigma_{u(\mathcal{K}^{*})}^{2},\vartheta_{i}\right.\right)\right],\nonumber 
\end{align}
where $f\left(y_{i}\mid x_{i};\alpha^{0},\sigma_{u}^{2},\vartheta_{i}\right)$
is defined in (\ref{EQ:likelihood_appro}). We note that $\tau_{\mathcal{K}^{*}}^{0}=1-\tau_{1}^{0}-...-\tau_{\mathcal{K}^{*}-1}^{0}$,
so the last term in the above is equivalent to $\tau_{\mathcal{K}^{*}}^{0}f\left(\cdot|\cdot\right)$.

\subsection{Computation of the Inefficiency Term Post Estimation}

In the case when $\alpha_i$ does not vary across $i$, we are able to compute the expectation of the inefficiency term. We take the case in the empirical application as an example. Other cases with more than two components in the mixture can be studied similarly.

Using (\ref{eq:mixture_distribution}) and $T^{-1}\sum_{t=1}^{T}v_{it}\overset{d}{\sim}N\left(0,\hat{\sigma}_{v\left(k\right)}^{2}\left/T\right.\right)$
for $i\in\hat{G}_{k|K},$  \citet{Greene2005} implies that 
\begin{align}
 & \widehat{\textrm{E}\left(\alpha_{i}^{0}+u_{i}|\varepsilon_{i1},...,\varepsilon_{iT}\right)} \notag\\
= & \hat{\alpha}^{*}+\hat{\tau}\left[\mu_{i\left(1\right)}^{*}+\sigma_{\left(1\right)}^{*}\frac{\phi\left(\mu_{i\left(1\right)}^{*}/\sigma_{\left(1\right)}^{*}\right)}{\Phi\left(\mu_{i\left(1\right)}^{*}/\sigma_{\left(1\right)}^{*}\right)}\right]+\left(1-\hat{\tau}\right)\left[\mu_{i\left(2\right)}^{*}+\sigma_{\left(2\right)}^{*}\frac{\phi\left(\mu_{i\left(2\right)}^{*}/\sigma_{\left(2\right)}^{*}\right)}{\Phi\left(\mu_{i\left(2\right)}^{*}/\sigma_{\left(2\right)}^{*}\right)}\right],\label{eq:inefficiency_post}
\end{align}
where 
\begin{align*}
\mu_{i\left(j\right)}^{*} & =\rho_{(j)}^{2}\left[\sum_{t=1}^{T}\left(y_{it}-\underline{z}_{it}'\hat{\pi}_{(k|K)}\right)-\hat{\alpha}^{*}\right],\sigma_{\left(j\right)}^{*2}=\rho_{(j)}^{2}\hat{\sigma}_{v\left(k\right)}^{2},\\
\rho_{(j)}^{2} & =\left.\lambda_{\left(j\right)}^{2}\right/\left(1+T\lambda_{\left(j\right)}^{2}\right),\textrm{ and }\lambda_{\left(j\right)}=\left.\hat{\sigma}_{u\left(j\right)}\right/\hat{\sigma}_{v\left(k\right)},
\end{align*}
for $j=1,2$, and $i\in\hat{G}_{k|K}.$

\section{HAC Method}

\label{APP:HAC}

In this section of the appendix, we describe the Hierarchical Agglomerative
Clustering (HAC) method used as part of the proposed method. We largely
adopt the description from Chapter~4 of \citet{Everittetal}.

HAC is a bottom-up clustering approach that starts by treating each
data point as an individual cluster. At each step, the two ``closest''
clusters are merged to form a new cluster, and this process is repeated
iteratively until a stopping rule is satisfied. The stopping rule may
be a pre-specified number of clusters $K$, a cut-off distance threshold,
or the completion of the full hierarchy where all observations are
eventually merged into a single cluster. The result can be represented
by a \emph{dendrogram}, which is a binary tree structure showing how
clusters are combined at each stage.

A crucial component of HAC is the definition of the \emph{linkage criterion},
which determines how distances between clusters are computed during
the iterative merging process. Common linkage criteria include:
\begin{itemize}
\item \textbf{Single linkage:} distance between two clusters is defined
as the minimum distance between any pair of points across clusters.
\item \textbf{Complete linkage:} distance is the maximum distance between
any pair of points across clusters.
\item \textbf{Average linkage:} distance is the average pairwise distance
between all points across clusters.
\item \textbf{Ward's method:} distance is defined as the increase in within-cluster
variance resulting from a merge, which tends to produce compact and
spherical clusters.
\end{itemize}

In this paper, we focus on Ward's method \citep{Ward1963}, which
is widely regarded as producing more balanced clusters compared to
single or complete linkage. At each iteration, Ward's method merges
the two clusters whose union leads to the smallest possible increase
in the total within-cluster variance (often called the ``error sum
of squares''). This variance-based criterion makes the procedure especially
well-suited for applications where compactness and homogeneity within
clusters are desired.

Formally, the distance metric in Ward's method for clusters $A$ and
$B$ is given by
\[
d_{AB}=\frac{|A||B|}{|A|+|B|}\|\bar{x}_{A}-\bar{x}_{B}\|^{2},
\]
where $|A|$ and $|B|$ denote the cluster sizes, and $\bar{x}_{A},\bar{x}_{B}$
are the corresponding centroids. This expression is derived from the
increment in within-cluster sum of squares that would occur if clusters
$A$ and $B$ were merged. The algorithm therefore prioritizes merging
clusters with centroids that are close to each other, scaled by their
sizes.

In practice, HAC with Ward's method proceeds as follows. First, each
data point is initialized as its own cluster. Second, the pairwise
distances between all clusters are computed using Ward's criterion.
Third, the two clusters with the smallest distance are merged. Fourth,
the cluster distances are updated to reflect the new merged cluster.
These steps are repeated until the stopping rule is satisfied. 

The output of HAC provides a nested sequence of partitions, which
can be cut at different levels depending on the desired granularity.
This flexibility is useful in empirical applications where the ``true''
number of groups is not known ex ante, and different levels of aggregation
can be explored.


\section{Main Proofs}

\label{APP:mainproofs}

The main proofs in this section are built on technical lemmas in Appendix
\ref{APP:lemmaProofs}. We first present some well known results that
are useful for the proofs in this appendix.

We let $b_{il}$, $l=0,1,...,p$, denote the bias term from approximations.
Specifically, 
\begin{equation}
b_{i0}\left(s\right)=\alpha_{i}\left(s\right)-\mathbb{B}_{-0}^{m}\left(s\right)'\pi_{i0}^{0}\text{ and \ensuremath{b_{il}\left(s\right)}\ensuremath{=\ensuremath{\beta_{il}\left(s\right)}-\ensuremath{\mathbb{B}^{m}\left(s\right)^{\prime}\pi_{il}^{0}}}},\label{eq:bil}
\end{equation}
for $s\in\left[0,1\right]$, and $\xi_{it}$ collects the bias term
in $y_{it}$: 
\begin{equation}
\xi_{it}\equiv{\displaystyle b_{i0}\left(\tau_{t}\right)+\sum_{l=1}^{p}}x_{itl}b_{il}\left(\tau_{t}\right).\label{eq:xi_define}
\end{equation}
With this notation and (\ref{EQ:approx}), 
\begin{equation}
y_{it}=\tilde{z}_{it}'\tilde{\pi}_{i}^{0}+\xi_{it}+v_{it}.\label{eq:yit_withbias}
\end{equation}

We know from \citet{Chen2007} that 
\[
\sup_{s\in\left[0,1\right]}\left|b_{il}\left(s\right)\right|=O\left(m^{-\kappa}\right)
\]
holds by Assumption \ref{A:coef}. Since $p$ is finite, clearly,
\begin{equation}
\max_{l=0,...,p}\sup_{s\in\left[0,1\right]}\left|b_{il}\left(s\right)\right|=O\left(m^{-\kappa}\right),\label{eq:bias_rate}
\end{equation}
for $i=1,...,N.$ Using similar logic on $\alpha_{k}^{*}\left(\tau_{t}\right)$
and the fact $K^{*}$is fixed, we can obtain 
\begin{equation}
\max_{k=1,...,K^{*}}\max_{l=0,...,p}\sup_{s\in\left[0,1\right]}\left|b_{kl}^{*}\left(s\right)\right|=O\left(m^{-\kappa}\right),\label{eq:bias_rate_group}
\end{equation}
where 
\[
b_{k0}^{*}\left(s\right)=\alpha_{k}^{*}\left(s\right)-\mathbb{B}_{-0}^{m}\left(s\right)'\pi_{i0}^{*0}\text{ and \ensuremath{b_{kl}^{*}\left(s\right)=\beta_{il}\left(s\right)-\mathbb{B}^{m}\left(s\right)^{\prime}\pi_{il}^{*0}}}\text{ for }l=1,...,p.
\]

\begin{proof}[Proof of Theorem \ref{TH:classify}]
\textbf{(i)} is a direct result of Lemma \ref{LE:uni_converge}.
To see that, 
\begin{align*}
\Pr\left(\max_{i=1,2,...,N}\left\Vert \hat{\vartheta}_{i}-\vartheta_{i}\right\Vert >\epsilon\right) & \leq\Pr\left(\max_{i=1,2,...,N}\left\Vert \hat{\pi}_{i}-\pi_{i}^{0}\right\Vert >\frac{\epsilon}{2}\right)+\Pr\left(\max_{i=1,2,...,N}\left\Vert \hat{\sigma}_{vi}^{2}-\sigma_{vi}^{2}\right\Vert >\frac{\epsilon}{2}\right)\\
 & \leq\Pr\left(\max_{i=1,2,...,N}\left\Vert \widehat{\tilde{\pi}}_{i}-\tilde{\pi}_{i}^{0}\right\Vert >\frac{\epsilon}{2}\right)+\Pr\left(\max_{i=1,2,...,N}\left\Vert \hat{\sigma}_{vi}^{2}-\sigma_{vi}^{2}\right\Vert >\frac{\epsilon}{2}\right)\\
 & =o\left(1\right),
\end{align*}
where the second line holds by the fact that $\hat{\pi}_{i}$ is a
sub-vector of $\widehat{\tilde{\pi}}_{i}$, and the last line applies
the results in Lemma \ref{LE:uni_converge}.

\textbf{(ii) }Denote 
\begin{align*}
L_{ii'} & \equiv\sum_{l=0}^{p}\left\Vert \pi_{il}^{0}-\pi_{i'l}^{0}\right\Vert +\left|\sigma_{vi}-\sigma_{vi'}\right|,\text{ and}\\
\hat{L}_{ii'} & \equiv\sum_{l=0}^{p}\left\Vert \hat{\pi}_{il}-\hat{\pi}_{i'l}\right\Vert +\left|\hat{\sigma}_{vi}-\hat{\sigma}_{vi'}\right|,
\end{align*}
and 
\[
L_{jk}^{*}\equiv\sum_{l=0}^{p}\left\Vert \pi_{jl}^{*0}-\pi_{kl}^{*0}\right\Vert +\left|\sigma_{v(j)}^{\ast}-\sigma_{v(k)}^{\ast}\right|.
\]
for $i,i'=1,...,n$ and $j,k=1,...,K^{*}$. We first claim that 
\begin{equation}
\min_{1\leq j\neq k\leq K^{*}}L_{jk}^{*}\geq\frac{1}{2}\underline{C}^{*}\label{eq:claim}
\end{equation}
holds after some large $T$ (and hence large $m$). We will show this
claim at the end.

To show the result in (ii), it is equivalent to show that 
\[
\Pr\left(\max_{1\leq k\leq K^{*}}\max_{i,i'\in G_{k|K^{*}}}\hat{L}_{ii'}<\min_{1\leq j\neq k\leq K^{*}}\min_{i\in G_{j|K^{*}},i'\in G_{k|K^{*}}}\hat{L}_{ii'}\right)=1-o\left(1\right).
\]
When $i,i'\in G_{k|K^{*}},$ $L_{ii'}=0$. Thus, the uniform convergence
in \textbf{(i)} implies that, for $\epsilon=\underline{C}^{*}/6,$
\begin{equation}
\Pr\left(\max_{1\leq k\leq K^{*}}\max_{i,i'\in G_{k|K^{*}}}\hat{L}_{ii'}<\underline{C}^{*}/6\right)=1-o\left(1\right).\label{eq:Lhatii'}
\end{equation}
Denote this event as 
\[
A=\left\{ \max_{1\leq k\leq K^{*}}\max_{i,i'\in G_{k|K^{*}}}\hat{L}_{ii'}<\underline{C}^{*}/6\right\} .
\]
Conditional on this event $A$, the claim in (\ref{eq:claim}) after
some large $T$ (and hence large $m$), the result in (\ref{eq:Lhatii'}),
and the triangular inequality imply that 
\begin{align*}
 & \min_{1\leq j\neq k\leq K^{*}}\min_{i\in G_{j|K^{*}},i'\in G_{k|K^{*}}}\hat{L}_{ii'}\\
\geq & \min_{1\leq j\neq k\leq K^{*}}L_{jk}^{*}-2\max_{i=1,...,n}\left(\sum_{l=0}^{p}\left\Vert \hat{\pi}_{il}-\pi_{il}^{0}\right\Vert +\left|\hat{\sigma}_{vi}-\sigma_{vi}\right|\right)\\
\geq & \frac{\underline{C}^{*}}{2}-2\cdot\frac{\underline{C}^{*}}{6}=\underline{C}^{*}/6\\
> & \max_{1\leq k\leq K^{*}}\max_{i,i'\in G_{k|K^{*}}}\hat{L}_{ii'}.
\end{align*}

Therefore, after some large $T$ (and hence large $m$)$,$ 
\begin{align*}
 & \Pr\left(\max_{1\leq k\leq K^{*}}\max_{i,i'\in G_{k|K^{*}}}\hat{L}_{ii'}<\min_{1\leq j\neq k\leq K^{*}}\min_{i\in G_{j|K^{*}},i'\in G_{k|K^{*}}}\hat{L}_{ii'}\right)\\
\geq & \Pr\left(\left.\max_{1\leq k\leq K^{*}}\max_{i,i'\in G_{k|K^{*}}}\hat{L}_{ii'}<\min_{1\leq j\neq k\leq K^{*}}\min_{i\in G_{j|K^{*}},i'\in G_{k|K^{*}}}\hat{L}_{ii'}\right|A\right)\Pr\left(A\right)\\
= & \Pr\left(A\right)=1-o\left(1\right),
\end{align*}
as desired.

We now show the claim in (\ref{eq:claim}). Notice that for any $\pi_{1},\pi_{2}\in\mathbb{R}^{m},$
\begin{align}
\left\Vert \mathbb{B}^{m}\left(s\right)^{\prime}\pi_{1}-\mathbb{B}^{m}\left(s\right)^{\prime}\pi_{2}\right\Vert  & =\left[\int_{0}^{1}\left(\sum_{j=0}^{m-1}B_{j}\left(s\right)\pi_{1j}-\sum_{j=0}^{m-1}B_{j}\left(s\right)\pi_{2j}\right)^{2}ds\right]^{1/2}\nonumber \\
 & =\left[\int_{0}^{1}\sum_{j=0}^{m-1}B_{j}\left(s\right)^{2}\left(\pi_{1j}-\pi_{2j}\right)^{2}ds\right]^{1/2}=\left[\sum_{j=0}^{m-1}\left(\pi_{1j}-\pi_{2j}\right)^{2}\right]^{1/2}\nonumber \\
 & =\left\Vert \pi_{1}-\pi_{2}\right\Vert ,\label{eq:pi_diff}
\end{align}
where the second line holds by $\int_{0}^{1}B_{j}\left(s\right)B_{j'}\left(s\right)ds=0$
for $j\neq j',$and the third line holds by $\int_{0}^{1}B_{j}\left(s\right)^{2}ds=1.$

Using (\ref{eq:pi_diff}), 
\begin{align*}
L_{jk}^{*}= & \left\Vert \pi_{j0}^{*0}-\pi_{k0}^{*0}\right\Vert +\sum_{l=1}^{p}\left\Vert \pi_{jl}^{*0}-\pi_{kl}^{*0}\right\Vert +\left|\sigma_{v(j)}^{\ast}-\sigma_{v(k)}^{\ast}\right|\\
= & \left\Vert \mathbb{B}_{-0}^{m}\left(s\right)^{\prime}\pi_{j0}^{*0}-\mathbb{B}_{-0}^{m}\left(s\right)^{\prime}\pi_{k0}^{*0}\right\Vert +\sum_{l=1}^{p}\left\Vert \mathbb{B}^{m}\left(s\right)^{\prime}\pi_{jl}^{*0}-\mathbb{B}^{m}\left(s\right)^{\prime}\pi_{kl}^{*0}\right\Vert +\left|\sigma_{v(j)}^{\ast}-\sigma_{v(k)}^{\ast}\right|\\
= & \left\Vert \alpha_{j}^{*}-b_{j0}^{*}\left(s\right)-\alpha_{k}^{*}+b_{k0}^{*}\left(s\right)\right\Vert +\sum_{l=1}^{p}\left\Vert \beta_{jl}^{\ast}\left(s\right)-b_{jl}^{*}\left(s\right)-\beta_{kl}^{\ast}\left(s\right)+b_{jk}^{*}\left(s\right)\right\Vert +\left|\sigma_{v(j)}^{\ast}-\sigma_{v(k)}^{\ast}\right|\\
\geq & \left\Vert \alpha_{j}^{*}-\alpha_{k}^{*}\right\Vert -\left\Vert b_{j0}^{*}\left(s\right)-b_{k0}^{*}\left(s\right)\right\Vert +\sum_{l=1}^{p}\left[\left\Vert \beta_{jl}^{\ast}-\beta_{kl}^{\ast}\right\Vert -\left\Vert b_{jl}^{*}\left(s\right)-b_{jk}^{*}\left(s\right)\right\Vert \right]+\left|\sigma_{v(j)}^{\ast}-\sigma_{v(k)}^{\ast}\right|\\
\geq & \underline{C}^{*}-O\left(m^{-\kappa}\right),
\end{align*}
where the fourth line holds by triangular inequality, and the last
line holds by Assumption \ref{A:group_diff}, the result in (\ref{eq:bias_rate_group}),
and the fact that $p$ is fixed. Using the result in (\ref{eq:bias_rate_group})
and the fact that $K^{*}$ is fixed, then after some large $T$ (and
hence large $m),$ 
\[
\min_{1\leq j\neq k\leq K^{*}}L_{jk}^{*}=\min_{1\leq j\neq k\leq K^{*}}\sum_{l=0}^{p}\left\Vert \pi_{jl}^{*0}-\pi_{kl}^{*0}\right\Vert +\left|\sigma_{v(j)}^{\ast}-\sigma_{v(k)}^{\ast}\right|\geq\frac{1}{2}\underline{C}^{*},
\]
as desired. 
\end{proof}
\begin{proof}[Proof of Theorem \ref{TH:post-estimation}]
\textbf{(i)} Denote the event of correct classification as 
\[
\mathcal{M}=\left\{ \left(\hat{G}_{1|K^{*}},\hat{G}_{2|K^{*}},\ldots,\hat{G}_{K^{*}|K^{*}}\right)=\left(G_{1|K^{*}},G_{2|K^{*}},\ldots,G_{K^{*}|K^{*}}\right)\right\} .
\]

We first show the results conditional on the event $\mathscr{\mathcal{M}}.$

For each $i\in G_{k|K^{*}}$

\[
y_{it}=\alpha^{0}-u_{i}+\underline{z}_{it}'\pi_{\left(k\right)}^{*0}+\xi_{it}+v_{it},
\]
where similar to how $\pi_{i}^{0}$ is defined, $\pi_{\left(k\right)}^{*0}$
collects coefficients for the approximation of $\alpha_{\left(k\right)}^{*}\left(s\right)$
and $\beta_{\left(k\right)}^{*}\left(s\right)$. Thus 
\begin{equation}
\ddot{y}_{it}=\ddot{\underline{z}}_{it}'\pi_{\left(k\right)}^{*0}+\ddot{\xi}_{it}+\ddot{v}_{it}.\label{eq:y_dots}
\end{equation}
Using (\ref{eq:y_dots}), 
\begin{align}
\hat{\pi}_{\left(k|K^{*}\right)}^{*0}-\pi_{\left(k\right)}^{*0} & =\left(\sum_{i\in G_{k|K^{*}}}\sum_{t=1}^{T}\underline{\ddot{z}}_{it}\underline{\ddot{z}}_{it}'\right)^{-1}\left(\sum_{i\in G_{k|K^{*}}}\sum_{t=1}^{T}\underline{\ddot{z}}_{it}\ddot{\xi}_{it}\right)\nonumber \\
 & +\left(\sum_{i\in G_{k|K^{*}}}\sum_{t=1}^{T}\underline{\ddot{z}}_{it}\underline{\ddot{z}}_{it}'\right)^{-1}\left(\sum_{i\in G_{k|K^{*}}}\sum_{t=1}^{T}\underline{\ddot{z}}_{it}\ddot{v}_{it}\right)\nonumber \\
 & \equiv A_{k1}+A_{k2},\label{eq:pi_diff_decom}
\end{align}
where $\ddot{\xi}_{it}$ is the de-meaned $\xi_{it}$ (defined in
(\ref{eq:xi_define})) over $t$ for observation $i$, and $\ddot{v}_{it}$
is similarly defined. With the decomposition in (\ref{eq:pi_diff_decom}),

\begin{align}
 & \sqrt{N_{k}T/\underline{m}}\mathbb{S}_{(k)}^{-1/2}\left[\hat{\theta}_{\left(k|K^{*}\right)}\left(s\right)-\theta_{\left(k\right)}^{*}\left(s\right)\right]\nonumber \\
 & =\sqrt{N_{k}T/\underline{m}}\mathbb{S}_{(k)}^{-1/2}\mathbb{M_{B}}\left(s\right)\left[\hat{\pi}_{\left(k|K^{*}\right)}^{*0}-\pi_{\left(k\right)}^{*0}\right]\nonumber \\
 & =\sqrt{N_{k}T/\underline{m}}\mathbb{S}_{(k)}^{-1/2}\mathbb{M_{B}}\left(s\right)A_{k1}+\sqrt{N_{k}T/\underline{m}}\mathbb{S}_{(k)}^{-1/2}\mathbb{M_{B}}\left(s\right)A_{k2}\nonumber \\
 & =o_{P}\left(1\right)+\sqrt{N_{k}T/\underline{m}}\mathbb{S}_{(k)}^{-1/2}\mathbb{M_{B}}\left(s\right)A_{k2}\nonumber \\
 & \stackrel{d}{\rightarrow}N\left(0,I_{p+1}\right),\label{eq:theta_converge_d}
\end{align}
where the fourth line uses the result in Lemma \ref{LE:Ak1} and $N_{k}T/\underline{m}^{1+2\kappa}\rightarrow0$
imposed in Assumption \ref{A:tuning2}, and the last line holds by
evoking the Crámer-Wold device on the result in Lemma \ref{LE:Ak2}.

To facilitate exposition, let $\omega_{NT}\equiv\sqrt{N_{k}T/\underline{m}}\left(\mathbb{S}_{(k)}^{-1/2}\right)'\left[\hat{\theta}_{\left(k|K^{*}\right)}\left(s\right)-\theta_{\left(k\right)}^{*}\left(s\right)\right].$
According to the definition of convergence in distribution, (\ref{eq:theta_converge_d})
implies that for any $\epsilon>0$ and any $x\in R$, there exists
a large $N_{1}$ such that for all $N>N_{1}$ 
\[
\left|P\left(\left.\omega_{NT}\leq x\right|\mathcal{M}\right)-\Phi\left(x\right)\right|<\frac{\epsilon}{3},
\]
where $\Phi\left(\cdot\right)$ denotes the cumulative distribution
function of the standard normal.

We now show the general results without conditional on $\mathscr{\mathcal{M}}$.
Theorem \ref{TH:classify} implies that there exists a large $N_{2}$
such that for all $N>N_{2}$ 
\[
P\left(\mathcal{M}\right)>1-\frac{\epsilon}{3}.
\]
Thus, for all $N>\max\left(N_{1},N_{2}\right)$ and any $x\in R$,
\begin{align*}
\left|P\left(\omega_{NT}\leq x\right)-\Phi\left(x\right)\right| & =\left|P\left(\left.\omega_{NT}\leq x\right|\mathcal{M}\right)P\left(\mathcal{M}\right)+P\left(\left.\omega_{NT}\leq x\right|\mathcal{M}^{c}\right)P\left(\mathcal{M}^{c}\right)-\Phi\left(x\right)\right|\\
 & \leq\left|\left[P\left(\left.\omega_{NT}\leq x\right|\mathcal{M}\right)-\Phi\left(x\right)\right]P\left(\mathcal{M}\right)\right|+\left[1-P\left(\mathcal{M}\right)\right]\Phi\left(x\right)\\
 & +P\left(\left.\omega_{NT}\leq x\right|\mathcal{M}^{c}\right)P\left(\mathcal{M}^{c}\right)\\
 & \leq\frac{\epsilon}{3}+\frac{\epsilon}{3}+\frac{\epsilon}{3}=\epsilon,
\end{align*}
which is the desired result by the definition of convergence in distribution.

\textbf{(ii)} We first show the results conditional on the event $\mathscr{\mathcal{M}}.$
Using the representation in (\ref{eq:y_dots}),

\begin{align*}
 & \hat{\sigma}_{v\left(k|K^{*}\right)}^{2}-\sigma_{v\left(k\right)}^{*2}\\
 & =\frac{1}{N_{k}\left(T-1\right)}\sum_{i\in G_{k|K^{*}}}\sum_{t=1}^{T}\left(\ddot{y}_{it}-\underline{\ddot{z}}_{it}'\hat{\pi}_{\left(k\right)}\right)^{2}-\sigma_{v\left(k\right)}^{*2}\\
 & =\frac{1}{N_{k}\left(T-1\right)}\sum_{i\in G_{k|K^{*}}}\sum_{t=1}^{T}\left[\ddot{\underline{z}}_{it}'\left(\pi_{\left(k\right)}^{*0}-\hat{\pi}_{\left(k\right)}\right)+\ddot{\xi}_{it}+\ddot{v}_{it}\right]^{2}-\sigma_{v\left(k\right)}^{*2}\\
 & =\frac{1}{N_{k}\left(T-1\right)}\sum_{i\in G_{k|K^{*}}}\sum_{t=1}^{T}\ddot{v}_{it}^{2}-\sigma_{v\left(k\right)}^{*2}+\frac{1}{N_{k}\left(T-1\right)}\sum_{i\in G_{k|K^{*}}}\sum_{t=1}^{T}\left[\ddot{\underline{z}}_{it}'\left(\pi_{\left(k\right)}^{*0}-\hat{\pi}_{\left(k\right)}\right)\right]^{2}\\
 & +\frac{1}{N_{k}\left(T-1\right)}\sum_{i\in G_{k|K^{*}}}\sum_{t=1}^{T}\ddot{\xi}_{it}^{2}+\frac{2}{N_{k}\left(T-1\right)}\sum_{i\in G_{k|K^{*}}}\sum_{t=1}^{T}\ddot{\underline{z}}_{it}'\left(\pi_{\left(k\right)}^{*0}-\hat{\pi}_{\left(k\right)}\right)\left(\ddot{\xi}_{it}+\ddot{v}_{it}\right)\\
 & +\frac{2}{N_{k}\left(T-1\right)}\sum_{i\in G_{k|K^{*}}}\sum_{t=1}^{T}\ddot{\xi}_{it}\ddot{v}_{it}\\
 & \equiv A_{k1}+A_{k2}+A_{k3}+A_{k4}+A_{k5}.
\end{align*}

We show that $A_{k2},A_{k3},A_{k4},$ and $A_{k5}$ are asymptotically
negligible by demonstrating the rates for $A_{k2},A_{k3},A_{k4},$
and $A_{k5}$. We then move to the asymptotic distribution of $A_{k1}$.

For $A_{k2}$, 
\begin{align}
A_{k2} & =\frac{1}{N_{k}\left(T-1\right)}\sum_{i\in G_{k|K^{*}}}\sum_{t=1}^{T}\left[\ddot{\underline{z}}_{it}'\left(\pi_{\left(k\right)}^{*0}-\hat{\pi}_{\left(k\right)}\right)\right]^{2}\nonumber \\
 & =\left(\pi_{\left(k\right)}^{*0}-\hat{\pi}_{\left(k\right)}\right)'\left(\frac{1}{N_{k}\left(T-1\right)}\sum_{i\in G_{k|K^{*}}}\sum_{t=1}^{T}\ddot{\underline{z}}_{it}\ddot{\underline{z}}_{it}'\right)\left(\pi_{\left(k\right)}^{*0}-\hat{\pi}_{\left(k\right)}\right)\nonumber \\
 & =O_{P}\left(\left\Vert \pi_{\left(k\right)}^{*0}-\hat{\pi}_{\left(k\right)}\right\Vert ^{2}\right)=O_{P}\left(\frac{\underline{m}}{N_{k}T}\right),\label{eq:ak2_}
\end{align}
where the last line holds by full rank condition implied by Lemmas
\ref{LE:z_k'z_k} and \ref{LE:splines-sample-2}, and the rate we
show in \textbf{(i)}.

For $A_{k3},$ by the definition of $\xi_{it}$in (\ref{eq:xi_define}),
the rate in (\ref{eq:bias_rate_group}), and Assumption \ref{A:tuning2}
(ii), we can obtain 
\begin{equation}
A_{k3}=\frac{1}{N_{k}\left(T-1\right)}\sum_{i\in G_{k|K^{*}}}\sum_{t=1}^{T}\ddot{\xi}_{it}^{2}=O_{P}\left(\underline{m}^{-2\kappa}\right)=o_{P}\left(\frac{\underline{m}}{N_{k}T}\right).\label{eq:ak3_}
\end{equation}

For $A_{k4},$ 
\begin{align}
A_{k4} & =\frac{2}{N_{k}\left(T-1\right)}\sum_{i\in G_{k|K^{*}}}\sum_{t=1}^{T}\ddot{\underline{z}}_{it}'\left(\pi_{\left(k\right)}^{*0}-\hat{\pi}_{\left(k\right)}\right)\left(\ddot{\xi}_{it}+\ddot{v}_{it}\right)\nonumber \\
 & =\left(\pi_{\left(k\right)}^{*0}-\hat{\pi}_{\left(k\right)}\right)'\left[\frac{2}{N_{k}\left(T-1\right)}\sum_{i\in G_{k|K^{*}}}\sum_{t=1}^{T}\ddot{\underline{z}}_{it}\ddot{\xi}_{it}\right]+\left(\pi_{\left(k\right)}^{*0}-\hat{\pi}_{\left(k\right)}\right)'\frac{2}{N_{k}\left(T-1\right)}\sum_{i\in G_{k|K^{*}}}\sum_{t=1}^{T}\ddot{\underline{z}}_{it}\ddot{v}_{it}\nonumber \\
 & =O_{P}\left(\sqrt{\frac{\underline{m}}{N_{k}T}}\underline{m}^{-\kappa}\right)+O_{P}\left(\sqrt{\frac{\underline{m}}{N_{k}T}}\cdot\sqrt{\frac{1}{N_{k}T}}\right)=o_{P}\left(\frac{\underline{m}}{N_{k}T}\right),\label{eq:ak4_}
\end{align}
where for the third line we apply the result in Lemma \ref{LE:Ak1},
the mixing condition across $t$ in Assumption \ref{A:mixing}, and
the independence across $i$ in Assumption \ref{A:additional-1}

For $A_{k5},$ again by the mixing condition across $t$, independence
across $i$, and the rate in (\ref{eq:bias_rate_group}), we have
\[
\text{Var}\left(\frac{2}{N_{k}\left(T-1\right)}\sum_{i\in G_{k|K^{*}}}\sum_{t=1}^{T}\ddot{\xi}_{it}\ddot{v}_{it}\right)\propto\frac{\underline{m}^{-2\kappa}}{NT}.
\]
Using the Markov inequality, the above implies that 
\begin{equation}
A_{k5}=O_{P}\left(\frac{\underline{m}^{-\kappa}}{\sqrt{NT}}\right)=o_{P}\left(\frac{\underline{m}}{N_{k}T}\right).\label{eq:ak5_}
\end{equation}

We turn to the leading term $A_{k1}:$ 
\begin{align*}
A_{k1} & =\frac{1}{N_{k}\left(T-1\right)}\sum_{i\in G_{k|K^{*}}}\sum_{t=1}^{T}\ddot{v}_{it}^{2}-\sigma_{v\left(k\right)}^{2}\\
 & =\frac{1}{N_{k}\left(T-1\right)}\sum_{i\in G_{k|K^{*}}}\sum_{t=1}^{T}\left(\ddot{v}_{it}^{2}-\sigma_{v\left(k\right)}^{2}\right)+\frac{1}{T-1}\sigma_{v\left(k\right)}^{2}\\
 & =\frac{T}{T-1}\left[\frac{1}{N_{k}T}\sum_{i\in G_{k|K^{*}}}\sum_{t=1}^{T}\left(\ddot{v}_{it}^{2}-\sigma_{v\left(k\right)}^{2}\right)\right]+\frac{1}{T-1}\sigma_{v\left(k\right)}^{2}\\
 & =\frac{T}{T-1}\left[\frac{1}{N_{k}T}\sum_{i\in G_{k|K^{*}}}\sum_{t=1}^{T}\left(v_{it}^{2}-\sigma_{v\left(k\right)}^{2}\right)\right]-\frac{T}{T-1}\frac{1}{N_{k}}\sum_{i\in G_{k|K^{*}}}\bar{v}_{i}^{2}+\frac{1}{T-1}\sigma_{v\left(k\right)}^{2}\\
 & =\frac{T}{T-1}\left[\frac{1}{N_{k}T}\sum_{i\in G_{k|K^{*}}}\sum_{t=1}^{T}\left(v_{it}^{2}-\sigma_{v\left(k\right)}^{2}\right)\right]-\frac{1}{T-1}\frac{1}{N_{k}}\sum_{i\in G_{k|K^{*}}}\left(T\bar{v}_{i}^{2}-\sigma_{v\left(k\right)}^{2}\right)\\
 & \equiv A_{k11}+A_{k12}.
\end{align*}
By the i.i.d. assumption on $v_{it},$ 
\[
\sqrt{N_{k}T}\cdot A_{k11}\stackrel{d}{\rightarrow}N\left(0,\text{Var}\left(\left.v_{it}^{2}\right|i\in G_{k|K^{*}}\right)\right).
\]
Note $\text{E}\left(T\bar{v}_{i}^{2}\right)=\sigma_{v\left(k\right)}^{2}$
and $\text{E}\left(T^{2}\bar{v}_{i}^{4}\right)=O\left(1\right).$
By the independence across $i$ and Markov's inequality, 
\[
A_{k12}=O_{P}\left(\frac{1}{T\sqrt{N_{k}}}\right),
\]
which implies that 
\[
\sqrt{N_{k}T}\cdot A_{k12}=o_{P}\left(1\right).
\]
Put the asymptotic distribution of $A_{k11}$ and the rate of $A_{k12}$
together, we obtain 
\begin{equation}
\sqrt{N_{k}T}\cdot A_{k1}\stackrel{d}{\rightarrow}N\left(0,\text{Var}\left(\left.v_{it}^{2}\right|i\in G_{k|K^{*}}\right)\right).\label{eq:ak1_}
\end{equation}

Equations (\ref{eq:ak2_}), (\ref{eq:ak3_}), (\ref{eq:ak4_}), (\ref{eq:ak5_}),
and (\ref{eq:ak1_}) imply that conditional on $\mathcal{M},$ 
\[
\sqrt{N_{k}T}\left(\hat{\sigma}_{v\left(k|K^{*}\right)}^{2}-\sigma_{v\left(k\right)}^{*2}\right)\stackrel{d}{\rightarrow}N\left(0,\text{Var}\left(\left.v_{it}^{2}\right|i\in G_{k|K^{*}}\right)\right).
\]
The above result holds unconditionally, using a similar argument as
in the proof of \textbf{(i)}.

\textbf{(iii) }We only need to show the result conditional on the
event $\mathcal{M}$. After that, we can show the unconditional result
using the same logic as in the proofs of \textbf{(i)} and \textbf{(ii)}.
In the following, we assume that the event $\mathcal{M}$ happens.

In Appendix \ref{app:IOmega}, we show that the information matrix
$\mathbb{I}$ and 
\[
\text{E}\left[\left.\sum_{k=1}^{K^{*}}\frac{N_{k}}{N}\cdot\left(\frac{\partial}{\partial\varrho}\log\tilde{f}_{i\left(k\right)}\left(\varrho\right)\right)\left(\frac{\partial}{\partial\varrho}\log\tilde{f}_{i\left(k\right)}\left(\varrho\right)\right)'\right|_{\varrho=\varrho^{0}}\right]
\]
are well behaved; that is, all elements in these two matrices are
finite constants, not degenerated to 0, and they are positive definite.

Further, in \textbf{(i)} and \textbf{(ii)}, we have shown that $\hat{\theta}_{\left(k|K^{*}\right)}\left(s\right)$
and $\hat{\sigma}_{v\left(k|K^{*}\right)}^{2}$ converge to the trues
at the rates of $\sqrt{\frac{NT}{\underline{m}}}$ and $\sqrt{NT}$,
respectively. Thus, the estimation errors of $\hat{\theta}_{\left(k|K^{*}\right)}\left(s\right)$
and $\hat{\sigma}_{v\left(k|\hat{K}\right)}^{2}$ have no impact on
the convergence rate of $\hat{\varrho}$ and its asymptotic distribution,
because $\hat{\varrho}$ supposedly converges to $\varrho^{0}$ at
the rate of $\sqrt{N}$, which is slower than those of $\hat{\theta}_{\left(k|K^{*}\right)}\left(s\right)$
and $\hat{\sigma}_{v\left(k|K^{*}\right)}^{2}$.

By the independence across $i$, some standard analysis for the asymptotics
of the MLE, see, e.g., Section 5.4.3 in \citet{BickelDoksum}, and
the Slutsky's Theorem, the asymptotic variance of $\hat{\varrho}$
is $N^{-1}$$\mathbb{I}$, and, by the Lindeberg Central Limit Theorem,
\[
\sqrt{N}\mathbb{I}^{-1/2}\left(\hat{\varrho}-\varrho\right)\overset{d}{\rightarrow}N\left(0,I_{3\mathcal{K}^{*}-1}\right).
\]
\end{proof}
\begin{proof}[Proof of Proposition \ref{Prop:classify}]
(i) This is a result direct from Lemmas \ref{LE:IC1} and \ref{LE:IC2}.

(ii) This is a result from Lemma \ref{LE:IC3}. 
\end{proof}
\clearpage
\section{DGP 1 and 2, Figures, Tables and Additional Discussions}

\label{APP:Tables}
\subsection{DGPs 1 and 2}

\textbf{Design 1:} In this design (consisting of DGP1U and DGP1M),
we study the classification arising as a result of heterogeneity from
frontiers. To this end, we specify the DGP as 
\[
y_{it}=\alpha_{i}^{0}-u_{i}+\alpha_{i}(\tau_{t})+x_{it}\beta_{i}(\tau_{t})+v_{it},
\]
and suppose there are two groups for frontiers with $x_{it}\sim N(1,1^{2})$.
The error term $v_{it}$ is generated from the same distribution $v_{it}\overset{iid}{\sim}N(0,\sigma_{v}^{2})$,
$\sigma_{v}=1$ for both groups. Group 1 frontiers are specified as
$\alpha_{(1)}(s)=3F(s;0.5,0.1)-\varpi_{1}$, and $\beta_{(1)}(s)=3[2s-4s^{2}+2s^{3}+F(s;0.6,0.1)$,
where $F(s;\cdot,\cdot)$ denotes the logistic CDF and $\varpi_{1}$
is the mean of $3F(s;0.5,0.1),$ that is, $\int_{0}^{1}3F(s;0.5,0.1)\mathrm{d}s.$
Group 2 parameters are specified as $\alpha_{(2)}(s)=3[2s-6s^{2}+4s^{3}+F(s;0.7,0.05)]-\varpi_{2}$,
and $\beta_{(2)}(s)=3[s-3s^{2}+2s^{3}+F(s;0.7,0.04)]$, $\varpi_{2}$
plays the same role as $\varpi_{1}$, and $\varpi_{2}=\int_{0}^{1}3[2s-6s^{2}+4s^{3}+F(s;0.7,0.05)]\mathrm{d}s$.
In what we call this DGP1U, we consider a case where $\alpha^{0}-u$
comes from a unique distribution, with $\alpha^{0}=0.5$ and $u_{i}\overset{iid}{\sim}|N(0,\sigma_{u}^{2})|$,
where $\sigma_{u}=1$. We consider another DGP, which we call DGP1M,
distinguished by letting $\alpha^{0}-u$ to come from a mixture distribution.
Specifically, we let $\alpha^{0}-u$ to come from $\alpha_{(1)}^{0}-|N(0,\sigma_{u(1)}^{2})|$
with probability $\tau$ and $\alpha_{(2)}^{0}-|N(0,\sigma_{u(2)}^{2})|$
with probability $1-\tau$, where $\alpha_{(1)}^{0}=1$, $\alpha_{(2)}^{0}=-1$,
$\sigma_{u(1)}=0.75$, $\sigma_{u(2)}=1.25$ and $\tau=0.5$. It is important to note that the mixture structure of $\alpha^{0}-u$ is
independent of the grouping structure.

\textbf{Design 2:} In the second design (consisting of DGP2U and DGP2M),
we study the case where there are two groups, but the classification
is due to heterogeneity in variances. To this end, we adopt a similar
setup as Design 1. Specifically, the DGP is 
\[
y_{it}=\alpha_{i}^{0}-u_{i}+\alpha(\tau_{t})+x_{it}\beta(\tau_{t})+v_{it},
\]
where $x_{it}\sim N(2,0.75^{2})$. The frontiers for both groups are
specified as $\alpha(s)=\log{s}\sin(6s)-\varpi$, and $\beta(s)=7\sin(5s)\exp(-5s)$,
where $\varpi$ denotes the mean of $\log{s}\sin(6s)$. The error
terms are generated from two distinct groups. Group 1 errors are generated
from $v_{it}\overset{iid}{\sim}N(0,\sigma_{v(1)}^{2})$, while group
2 errors are generated from $v_{it}\overset{iid}{\sim}N(0,\sigma_{v(2)}^{2})$.
Standard deviations are specified as $\sigma_{v(1)}=0.5$ and $\sigma_{v(2)}=1.5$.
Similarly to Design 1, we consider two sub-cases where $\alpha^{0}-u$
either comes from a unique (DGP2U) or mixture distribution (DGP2M),
the same as in Design 1.

\subsection{Additional Tables and Figures for Section \ref{SEC:simulations} and Section \ref{SEC:application}}

\begin{figure}[!htb]
\caption{Estimates of the grouped frontiers for DGP3M with $N=500$ and $T=50$}
\centering
\includegraphics[scale=0.75] {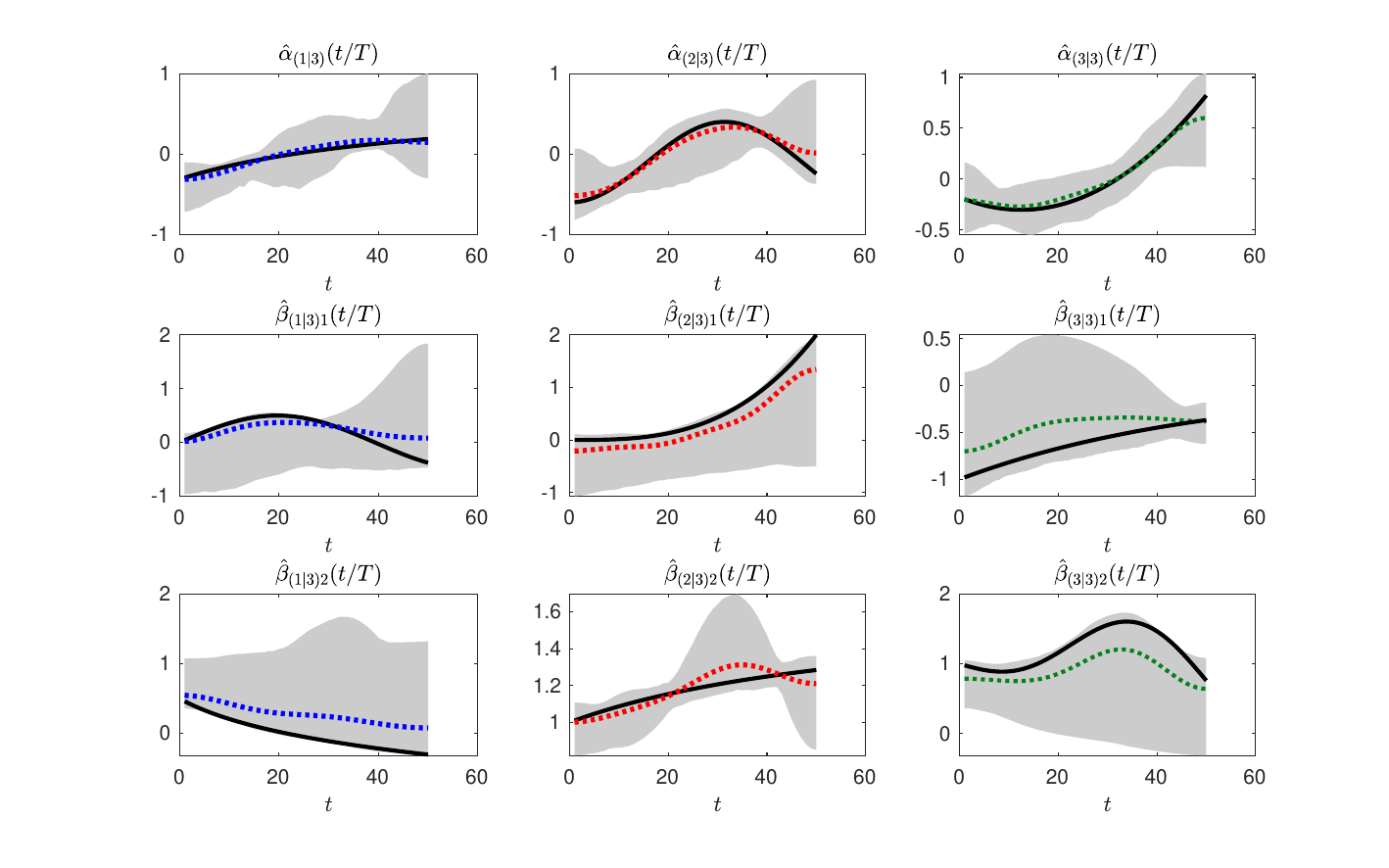}
\label{fig:Grouped_frontiers_DGP6_N500_T50}
\begin{tablenotes}
      \footnotesize
      \item 
      \emph{Note}: Black solid line depict the true time-varying frontier, dotted lines depict the mean of the estimated grouped frontiers averaged over 500 MC iterations, and the grey shaded region depict the 90 percentile of the estimates. The gray area is wide, due to mis-classification errors.
\end{tablenotes}
\end{figure}

\begin{figure}[!htb]
\caption{Estimates of the grouped frontiers for DGP3M with $N=500$ and $T=75$}
\centering
\includegraphics[scale=0.75] {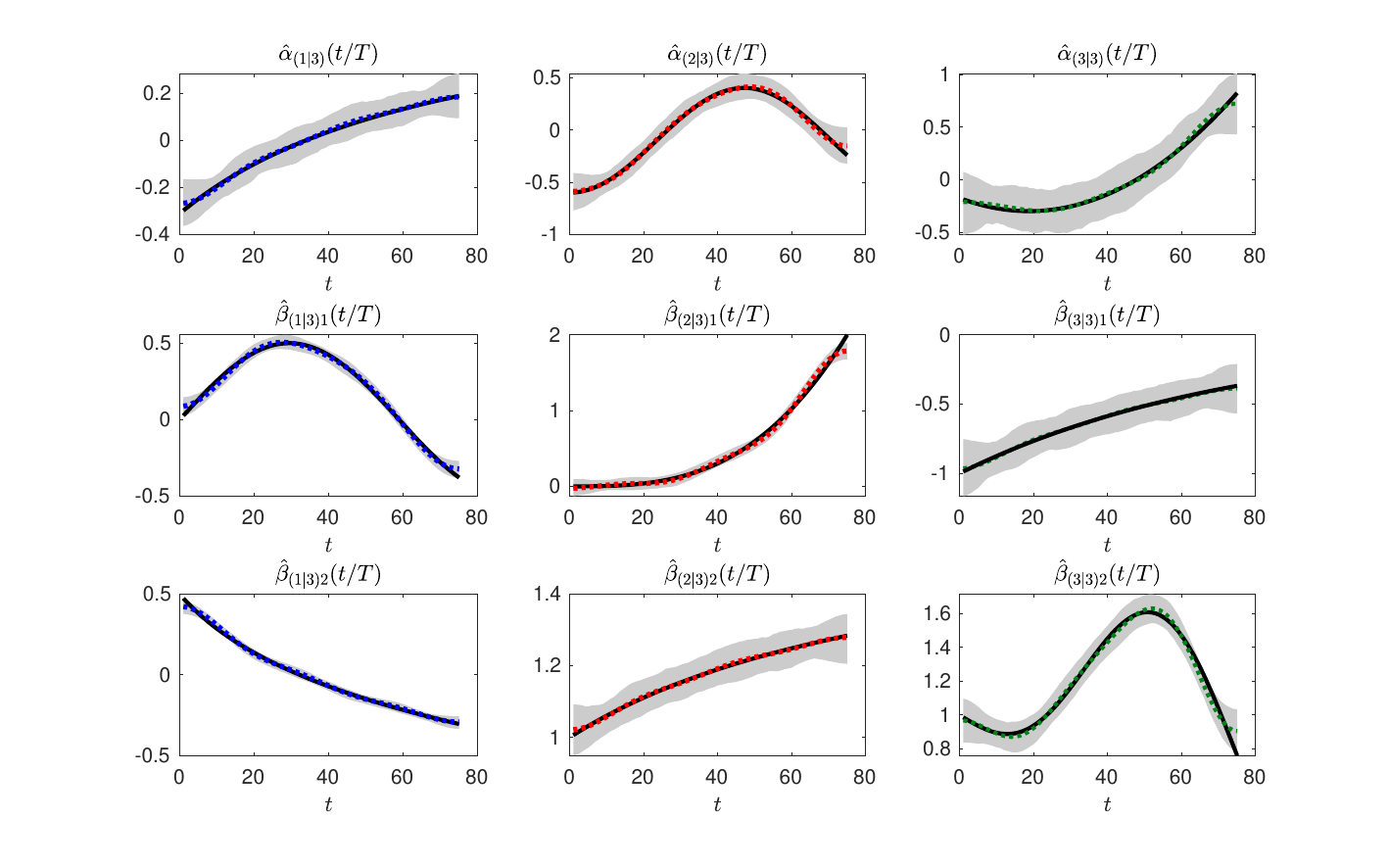}
\label{fig:Grouped_frontiers_DGP6_N500_T75}
\begin{tablenotes}
      \footnotesize
      \item 
      \emph{Note}: Black solid line depict the true time-varying frontier, dotted lines depict the mean of the estimated grouped frontiers averaged over 500 MC iterations, and the gray shaded region depict the 90 percentile of the estimates. 
\end{tablenotes}
\end{figure}

\begin{figure}[!htb]
\caption{Estimates of the grouped frontiers for DGP3M with $N=500$ and $T=100$}
\centering
\includegraphics[scale=0.75] {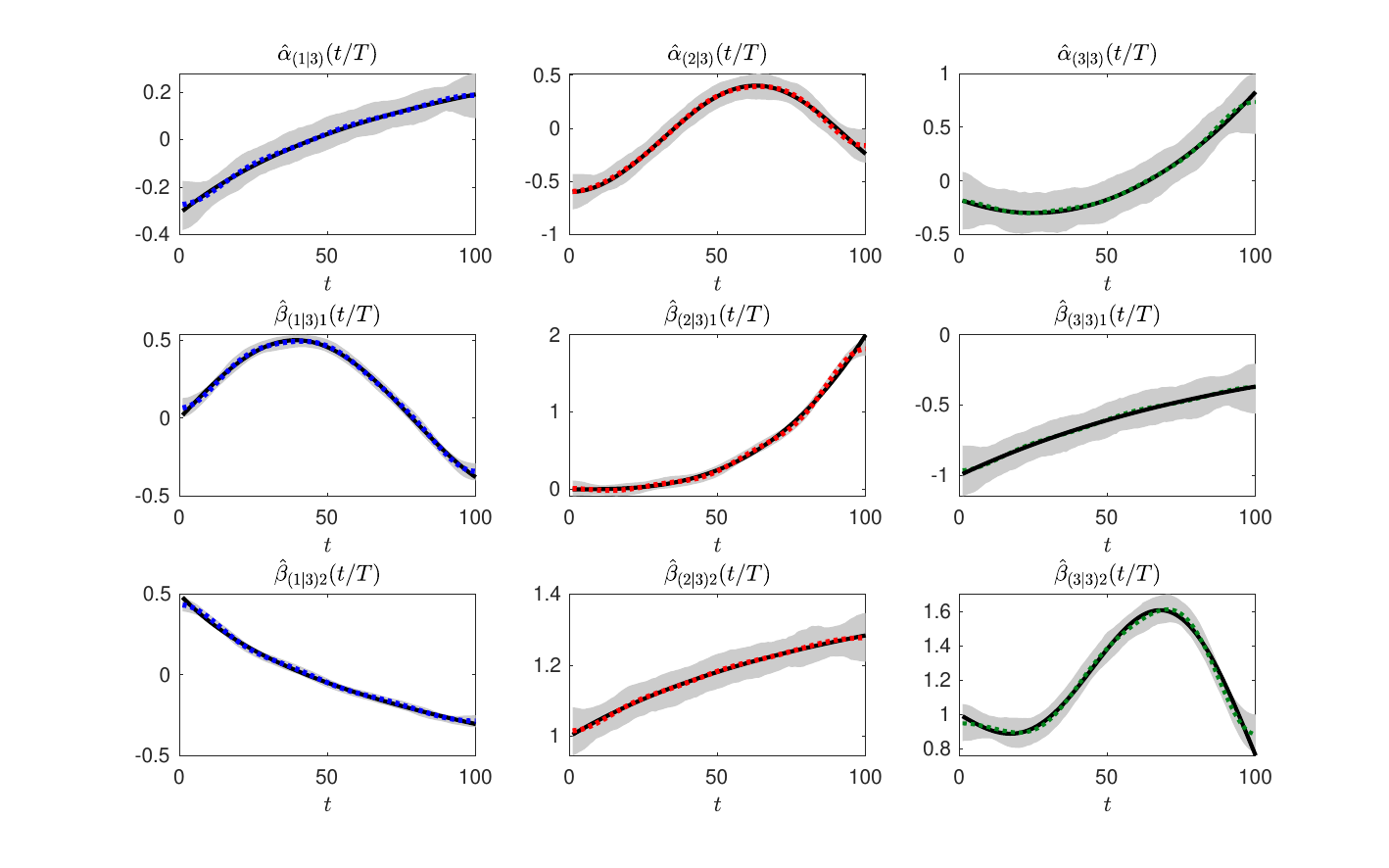}
\label{fig:Grouped_frontiers_DGP6_N500_T100}
\begin{tablenotes}
      \footnotesize
      \item 
      \emph{Note}: Black solid line depict the true time-varying frontier, dotted lines depict the mean of the estimated grouped frontiers averaged over 500 MC iterations, and the gray shaded region depict the 90 percentile of the estimates. 
\end{tablenotes}
\end{figure}

\begin{table}[!htb]
\centering
\caption{BIAS and RMSE over 500 MC iterations for DGP1U}

\begin{tabularx}{\textwidth}{X c XX c XX c XX c XX}
\toprule
&& \multicolumn{2}{c}{$\hat{\sigma}_{v(1)}$} && \multicolumn{2}{c}{$\hat{\sigma}_{v(2)}$} && \multicolumn{2}{c}{$\hat{\alpha}^0$} && \multicolumn{2}{c}{$\hat{\sigma}_{u}$} \\
\cmidrule{3-4} \cmidrule{6-7} \cmidrule{9-10} \cmidrule{12-13}  
$(N,T)$ && BIAS & RMSE && BIAS & RMSE && BIAS & RMSE && BIAS& RMSE \\
\midrule
(100,50) && 0.015 & 0.020 && 0.034 & 0.038 && 0.046 & 0.058 && 0.066 & 0.082 \\ 
(100,75) && 0.014 & 0.019 && 0.033 & 0.036 && 0.039 & 0.049 && 0.062 & 0.077 \\ 
(100,100) && 0.012 & 0.017 && 0.034 & 0.037 && 0.036 & 0.046 && 0.063 & 0.081 \\ 
(250,50) && 0.008 & 0.011 && 0.017 & 0.019 && 0.035 & 0.043 && 0.042 & 0.053 \\ 
(250,75) && 0.007 & 0.009 && 0.016 & 0.018 && 0.027 & 0.034 && 0.042 & 0.052 \\ 
(250,100) && 0.005 & 0.007 && 0.011 & 0.013 && 0.023 & 0.030 && 0.042 & 0.052 \\ 
(500,50) && 0.005 & 0.007 && 0.012 & 0.013 && 0.030 & 0.035 && 0.029 & 0.036 \\ 
(500,75) && 0.005 & 0.006 && 0.011 & 0.012 && 0.023 & 0.028 && 0.028 & 0.035 \\ 
(500,100) && 0.004 & 0.005 && 0.007 & 0.008 && 0.019 & 0.024 && 0.029 & 0.035 \\ 
\bottomrule
\end{tabularx}

\label{tab:DGP1U_parameters}
\end{table}


\begin{table}[!htb]

\centering
\caption{BIAS and RMSE over 500 MC iterations for DGP1M}
\scalebox{0.71}{
\begin{tabularx}{1.4\textwidth}{X c XX c XX c XX c XX c XX c XX c XX}

\toprule
&& \multicolumn{2}{c}{$\hat{\sigma}_{v(1)}$} && \multicolumn{2}{c}{$\hat{\sigma}_{v(2)}$} && \multicolumn{2}{c}{$\hat{\tau}$} && \multicolumn{2}{c}{$\hat\alpha^0_{(1)}$} && \multicolumn{2}{c}{$\hat{\sigma}_{u(1)}$} && \multicolumn{2}{c}{$\hat\alpha^0_{(2)}$} && \multicolumn{2}{c}{$\hat{\sigma}_{u(2)}$}  \\
\cmidrule{3-4} \cmidrule{6-7} \cmidrule{9-10} \cmidrule{12-13} \cmidrule{15-16} \cmidrule{18-19} \cmidrule{21-22}
$(N,T)$ && BIAS & RMSE && BIAS & RMSE && BIAS & RMSE && BIAS & RMSE && BIAS & RMSE && BIAS & RMSE && BIAS & RMSE \\

\midrule
(100,50) && 0.015 & 0.020 && 0.034 & 0.038 && 0.012 & 0.018 && 0.059 & 0.074 && 0.103 & 0.131 && 0.089 & 0.115 && 0.122 & 0.153 \\ 
(100,75) && 0.014 & 0.019 && 0.033 & 0.036 && 0.010 & 0.014 && 0.050 & 0.064 && 0.091 & 0.113 && 0.070 & 0.090 && 0.116 & 0.144 \\ 
(100,100) && 0.012 & 0.017 && 0.034 & 0.037 && 0.011 & 0.016 && 0.046 & 0.058 && 0.097 & 0.122 && 0.072 & 0.095 && 0.114 & 0.144 \\ 
(250,50) && 0.008 & 0.011 && 0.017 & 0.019 && 0.008 & 0.010 && 0.041 & 0.050 && 0.062 & 0.080 && 0.055 & 0.069 && 0.076 & 0.098 \\ 
(250,75) && 0.007 & 0.009 && 0.016 & 0.018 && 0.007 & 0.009 && 0.033 & 0.042 && 0.062 & 0.078 && 0.046 & 0.058 && 0.071 & 0.087 \\ 
(250,100) && 0.005 & 0.007 && 0.011 & 0.013 && 0.007 & 0.009 && 0.027 & 0.035 && 0.058 & 0.073 && 0.044 & 0.056 && 0.072 & 0.089 \\ 
(500,50) && 0.005 & 0.007 && 0.012 & 0.013 && 0.005 & 0.007 && 0.034 & 0.041 && 0.048 & 0.060 && 0.042 & 0.052 && 0.050 & 0.062 \\ 
(500,75) && 0.005 & 0.006 && 0.011 & 0.012 && 0.005 & 0.006 && 0.027 & 0.033 && 0.040 & 0.051 && 0.035 & 0.044 && 0.048 & 0.061 \\ 
(500,100) && 0.004 & 0.005 && 0.007 & 0.008 && 0.004 & 0.005 && 0.022 & 0.027 && 0.041 & 0.051 && 0.032 & 0.040 && 0.051 & 0.063 \\

\bottomrule
\end{tabularx}}

\label{tab:DGP1M_parameters}

\end{table}


    
\begin{table}[!htb]
\centering
\caption{BIAS and RMSE over 500 MC iterations for DGP2U}

\begin{tabularx}{\textwidth}{X c XX c XX c XX c XX}
\toprule
&& \multicolumn{2}{c}{$\hat{\sigma}_{v(1)}$} && \multicolumn{2}{c}{$\hat{\sigma}_{v(2)}$} && \multicolumn{2}{c}{$\hat{\alpha}^0$} && \multicolumn{2}{c}{$\hat{\sigma}_{u}$} \\
\cmidrule{3-4} \cmidrule{6-7} \cmidrule{9-10} \cmidrule{12-13}  
$(N,T)$ && BIAS & RMSE && BIAS & RMSE && BIAS & RMSE && BIAS& RMSE \\
\midrule
(100,50) && 0.006 & 0.008 && 0.017 & 0.022 && 0.039 & 0.050 && 0.068 & 0.084 \\ 
(100,75) && 0.005 & 0.006 && 0.014 & 0.018 && 0.035 & 0.045 && 0.061 & 0.076 \\ 
(100,100) && 0.004 & 0.005 && 0.012 & 0.015 && 0.032 & 0.041 && 0.062 & 0.079 \\ 
(250,50) && 0.004 & 0.005 && 0.011 & 0.014 && 0.024 & 0.031 && 0.041 & 0.052 \\ 
(250,75) && 0.003 & 0.004 && 0.009 & 0.011 && 0.021 & 0.027 && 0.041 & 0.051 \\ 
(250,100) && 0.002 & 0.003 && 0.008 & 0.009 && 0.018 & 0.023 && 0.041 & 0.051 \\ 
(500,50) && 0.002 & 0.003 && 0.008 & 0.009 && 0.017 & 0.021 && 0.030 & 0.037 \\ 
(500,75) && 0.002 & 0.003 && 0.006 & 0.008 && 0.015 & 0.019 && 0.028 & 0.035 \\ 
(500,100) && 0.002 & 0.002 && 0.005 & 0.007 && 0.014 & 0.017 && 0.029 & 0.035 \\

\bottomrule
\end{tabularx}

\label{tab:DGP2U_parameters}
\end{table}

    
\begin{table}[!htb]
\centering
\caption{BIAS and RMSE over 500 MC iterations for DGP2M}
\scalebox{0.71}{
\begin{tabularx}{1.4\textwidth}{X c XX c XX c XX c XX c XX c XX c XX}
\toprule
&& \multicolumn{2}{c}{$\hat{\sigma}_{v(1)}$} && \multicolumn{2}{c}{$\hat{\sigma}_{v(2)}$} && \multicolumn{2}{c}{$\hat{\tau}$} && \multicolumn{2}{c}{$\hat\alpha^0_{(1)}$} && \multicolumn{2}{c}{$\hat{\sigma}_{u(1)}$} && \multicolumn{2}{c}{$\hat\alpha^0_{(2)}$} && \multicolumn{2}{c}{$\hat{\sigma}_{u(2)}$} \\
\cmidrule{3-4} \cmidrule{6-7} \cmidrule{9-10} \cmidrule{12-13} \cmidrule{15-16} \cmidrule{18-19} \cmidrule{21-22}
$(N,T)$ && BIAS & RMSE && BIAS & RMSE && BIAS & RMSE && BIAS & RMSE && BIAS & RMSE && BIAS & RMSE && BIAS & RMSE \\

\midrule
(100,50) && 0.006 & 0.008 && 0.017 & 0.022 && 0.019 & 0.027 && 0.040 & 0.050 && 0.105 & 0.139 && 0.125 & 0.160 && 0.133 & 0.166 \\ 
(100,75) && 0.005 & 0.006 && 0.014 & 0.018 && 0.014 & 0.019 && 0.034 & 0.044 && 0.089 & 0.115 && 0.105 & 0.132 && 0.121 & 0.150 \\ 
(100,100) && 0.004 & 0.005 && 0.012 & 0.015 && 0.015 & 0.021 && 0.032 & 0.041 && 0.098 & 0.126 && 0.100 & 0.129 && 0.121 & 0.153 \\ 
(250,50) && 0.004 & 0.005 && 0.011 & 0.014 && 0.014 & 0.018 && 0.025 & 0.031 && 0.071 & 0.092 && 0.084 & 0.106 && 0.082 & 0.106 \\ 
(250,75) && 0.003 & 0.004 && 0.009 & 0.011 && 0.012 & 0.015 && 0.021 & 0.026 && 0.067 & 0.087 && 0.072 & 0.090 && 0.079 & 0.098 \\ 
(250,100) && 0.002 & 0.003 && 0.008 & 0.009 && 0.010 & 0.013 && 0.017 & 0.022 && 0.063 & 0.080 && 0.070 & 0.086 && 0.079 & 0.097 \\ 
(500,50) && 0.002 & 0.003 && 0.008 & 0.009 && 0.014 & 0.016 && 0.018 & 0.022 && 0.061 & 0.077 && 0.063 & 0.079 && 0.058 & 0.073 \\ 
(500,75) && 0.002 & 0.003 && 0.006 & 0.008 && 0.009 & 0.011 && 0.015 & 0.019 && 0.048 & 0.061 && 0.052 & 0.065 && 0.054 & 0.068 \\ 
(500,100) && 0.002 & 0.002 && 0.005 & 0.007 && 0.008 & 0.010 && 0.013 & 0.017 && 0.048 & 0.062 && 0.051 & 0.063 && 0.055 & 0.070 \\ 

\bottomrule
\end{tabularx}}

\label{tab:DGP2M_parameters}
\end{table}

    
\begin{table}[!htb]
\centering
\caption{BIAS and RMSE over 500 MC iterations for DGP3U}

\begin{tabularx}{\textwidth}{X X X  X  c X  X  c X  X  c X   X c X  X }
\toprule
&& \multicolumn{2}{c}{$\hat{\sigma}_{v(1)}$} && \multicolumn{2}{c}{$\hat{\sigma}_{v(2)}$} && \multicolumn{2}{c}{$\hat{\sigma}_{v(3)}$} && \multicolumn{2}{c}{$\hat{\alpha}^0$} && \multicolumn{2}{c}{$\hat{\sigma}_{u}$} \\
\cmidrule{3-4} \cmidrule{6-7} \cmidrule{9-10} \cmidrule{12-13} \cmidrule{15-16} 
$(N,T)$ && BIAS & RMSE && BIAS & RMSE && BIAS & RMSE && BIAS & RMSE && BIAS & RMSE \\
\midrule
(100,50) && 0.268 & 0.463 && 0.127 & 0.215 && 0.367 & 0.594 && 0.059 & 0.079 && 0.074 & 0.095 \\ 
(100,75) && 0.082 & 0.214 && 0.074 & 0.174 && 0.154 & 0.372 && 0.043 & 0.053 && 0.066 & 0.081 \\ 
(100,100) && 0.024 & 0.099 && 0.027 & 0.093 && 0.053 & 0.196 && 0.040 & 0.052 && 0.064 & 0.081 \\ 
(250,50) && 0.204 & 0.373 && 0.138 & 0.243 && 0.327 & 0.562 && 0.037 & 0.048 && 0.053 & 0.067 \\ 
(250,75) && 0.034 & 0.129 && 0.035 & 0.116 && 0.069 & 0.241 && 0.024 & 0.031 && 0.040 & 0.049 \\ 
(250,100) && 0.005 & 0.033 && 0.008 & 0.032 && 0.014 & 0.064 && 0.023 & 0.029 && 0.041 & 0.050 \\ 
(500,50) && 0.145 & 0.307 && 0.110 & 0.218 && 0.242 & 0.484 && 0.027 & 0.036 && 0.042 & 0.053 \\ 
(500,75) && 0.009 & 0.064 && 0.009 & 0.044 && 0.018 & 0.100 && 0.018 & 0.023 && 0.030 & 0.037 \\ 
(500,100) && 0.002 & 0.003 && 0.004 & 0.005 && 0.007 & 0.009 && 0.016 & 0.020 && 0.029 & 0.036 \\ 
\bottomrule
\end{tabularx}

\label{tab:DGP3U_parameters}
\end{table}

    
\begin{table}[!htb]
\centering
\caption{Performance of ICs for DGP1U}

\begin{tabularx}{\textwidth}{X X X X  X  X X }
\toprule

$(N,T)$ & $K=1$ & $K=2$ & $K=3$ & $K=4$ & $\alpha^0-u$ unique & $\alpha^0-u$ mix  \\
\midrule
(100,50) & 0.000 & 1.000 & 0.000 & 0.000 & 0.926 & 0.074 \\ 
(100,75) & 0.000 & 1.000 & 0.000 & 0.000 & 0.940 & 0.060 \\ 
(100,100) & 0.000 & 1.000 & 0.000 & 0.000 & 0.938 & 0.062 \\ 
(250,50) & 0.000 & 1.000 & 0.000 & 0.000 & 0.996 & 0.004 \\ 
(250,75) & 0.000 & 1.000 & 0.000 & 0.000 & 0.992 & 0.008 \\ 
(250,100) & 0.000 & 1.000 & 0.000 & 0.000 & 0.988 & 0.012 \\ 
(500,50) & 0.000 & 1.000 & 0.000 & 0.000 & 0.996 & 0.004 \\ 
(500,75) & 0.000 & 1.000 & 0.000 & 0.000 & 1.000 & 0.000 \\ 
(500,100) & 0.000 & 1.000 & 0.000 & 0.000 & 1.000 & 0.000 \\

\bottomrule
\multicolumn{7}{l}{\emph{Note}: Results for the baseline case $c_{\lambda}=\tilde{c}_{\lambda}=1$.  Reported numbers are probabilities across replications.}
\end{tabularx}

\label{tab:IC_DGP1U}

\end{table}

    
\begin{table}[!htb]
\centering
\caption{Performance of ICs for DGP1M}

\begin{tabularx}{\textwidth}{X X X X  X  X X }
\toprule

$(N,T)$ & $K=1$ & $K=2$ & $K=3$ & $K=4$ & $\alpha^0-u$ unique & $\alpha^0-u$ mix  \\
\midrule
(100,50) & 0.000 & 1.000 & 0.000 & 0.000 & 0.000 & 1.000 \\ 
(100,75) & 0.000 & 1.000 & 0.000 & 0.000 & 0.000 & 1.000 \\ 
(100,100) & 0.000 & 1.000 & 0.000 & 0.000 & 0.000 & 1.000 \\ 
(250,50) & 0.000 & 1.000 & 0.000 & 0.000 & 0.000 & 1.000 \\ 
(250,75) & 0.000 & 1.000 & 0.000 & 0.000 & 0.000 & 1.000 \\ 
(250,100) & 0.000 & 1.000 & 0.000 & 0.000 & 0.000 & 1.000 \\ 
(500,50) & 0.000 & 1.000 & 0.000 & 0.000 & 0.000 & 1.000 \\ 
(500,75) & 0.000 & 1.000 & 0.000 & 0.000 & 0.000 & 1.000 \\ 
(500,100) & 0.000 & 1.000 & 0.000 & 0.000 & 0.000 & 1.000 \\

\bottomrule
\multicolumn{7}{l}{\emph{Note}:Results for the baseline case $c_{\lambda}=\tilde{c}_{\lambda}=1$.  Reported numbers are probabilities across replications.}
\end{tabularx}

\label{tab:IC_DGP1M}

\end{table}

    
\begin{table}[!htb]
\centering
\caption{Performance of ICs for DGP2U}

\begin{tabularx}{\textwidth}{X X X X  X  X X }
\toprule

$(N,T)$ & $K=1$ & $K=2$ & $K=3$ & $K=4$ & $\alpha^0-u$ unique & $\alpha^0-u$ mix  \\
\midrule
(100,50) & 0.000 & 1.000 & 0.000 & 0.000 & 0.868 & 0.132 \\ 
(100,75) & 0.000 & 1.000 & 0.000 & 0.000 & 0.838 & 0.162 \\ 
(100,100) & 0.000 & 1.000 & 0.000 & 0.000 & 0.878 & 0.122 \\ 
(250,50) & 0.000 & 1.000 & 0.000 & 0.000 & 0.978 & 0.022 \\ 
(250,75) & 0.000 & 1.000 & 0.000 & 0.000 & 0.984 & 0.016 \\ 
(250,100) & 0.000 & 1.000 & 0.000 & 0.000 & 0.970 & 0.030 \\ 
(500,50) & 0.000 & 1.000 & 0.000 & 0.000 & 0.996 & 0.004 \\ 
(500,75) & 0.000 & 1.000 & 0.000 & 0.000 & 0.998 & 0.002 \\ 
(500,100) & 0.000 & 1.000 & 0.000 & 0.000 & 1.000 & 0.000 \\ 

\bottomrule
\multicolumn{7}{l}{\emph{Note}: Results for the baseline case $c_{\lambda}=\tilde{c}_{\lambda}=1$.  Reported numbers are probabilities across replications. }
\end{tabularx}

\label{tab:IC_DGP2U}

\end{table}

    
\begin{table}[!htb]
\centering
\caption{Performance of ICs for DGP2M}

\begin{tabularx}{\textwidth}{X X X X  X  X X }
\toprule

$(N,T)$ & $K=1$ & $K=2$ & $K=3$ & $K=4$ & $\alpha^0-u$ unique & $\alpha^0-u$ mix  \\
\midrule
(100,50) & 0.000 & 1.000 & 0.000 & 0.000 & 0.000 & 1.000 \\ 
(100,75) & 0.000 & 1.000 & 0.000 & 0.000 & 0.000 & 1.000 \\ 
(100,100) & 0.000 & 1.000 & 0.000 & 0.000 & 0.000 & 1.000 \\ 
(250,50) & 0.000 & 1.000 & 0.000 & 0.000 & 0.000 & 1.000 \\ 
(250,75) & 0.000 & 1.000 & 0.000 & 0.000 & 0.000 & 1.000 \\ 
(250,100) & 0.000 & 1.000 & 0.000 & 0.000 & 0.000 & 1.000 \\ 
(500,50) & 0.000 & 1.000 & 0.000 & 0.000 & 0.000 & 1.000 \\ 
(500,75) & 0.000 & 1.000 & 0.000 & 0.000 & 0.000 & 1.000 \\ 
(500,100) & 0.000 & 1.000 & 0.000 & 0.000 & 0.000 & 1.000 \\

\bottomrule
\multicolumn{7}{l}{\emph{Note}: Results for the baseline case $c_{\lambda}=\tilde{c}_{\lambda}=1$.  Reported numbers are probabilities across replications. }
\end{tabularx}

\label{tab:IC_DGP2M}

\end{table}

    
\begin{table}[!htb]
\centering
\caption{Performance of ICs for DGP3U}

\begin{tabularx}{\textwidth}{X X X X  X  X X }
\toprule

$(N,T)$ & $K=1$ & $K=2$ & $K=3$ & $K=4$ & $\alpha^0-u$ unique & $\alpha^0-u$ mix  \\
\midrule
(100,50) & 0.000 & 0.352 & 0.648 & 0.000 & 0.784 & 0.216 \\ 
(100,75) & 0.000 & 0.078 & 0.922 & 0.000 & 0.840 & 0.160 \\ 
(100,100) & 0.000 & 0.000 & 1.000 & 0.000 & 0.874 & 0.126 \\ 
(250,50) & 0.000 & 0.086 & 0.914 & 0.000 & 0.942 & 0.058 \\ 
(250,75) & 0.000 & 0.000 & 1.000 & 0.000 & 0.988 & 0.012 \\ 
(250,100) & 0.000 & 0.000 & 1.000 & 0.000 & 0.992 & 0.008 \\ 
(500,50) & 0.000 & 0.006 & 0.994 & 0.000 & 0.996 & 0.004 \\ 
(500,75) & 0.000 & 0.000 & 1.000 & 0.000 & 1.000 & 0.000 \\ 
(500,100) & 0.000 & 0.000 & 1.000 & 0.000 & 1.000 & 0.000 \\  

\bottomrule
\multicolumn{7}{l}{\emph{Note}: Results for the baseline case $c_{\lambda}=\tilde{c}_{\lambda}=1$.  Reported numbers are probabilities across replications. }
\end{tabularx}

\label{tab:IC_DGP3U}

\end{table}

    
\begin{table}[!htb]
\centering
\caption{Sensitivity analysis for classification error in DGP1U and DGP1M}

\begin{tabularx}{\textwidth}{X c X  X  c X  X  c X  X  c X   X c X  X }
\toprule
&& \multicolumn{2}{c}{$c_\lambda=3/2$} && \multicolumn{2}{c}{$c_\lambda=1$ (bench.)} && \multicolumn{2}{c}{$c_\lambda=3/4$} \\
\cmidrule{3-4} \cmidrule{6-7} \cmidrule{9-10} 
$(N,T)$ &&& $\bar{\text{Pr}}(F)$ &&& $\bar{\text{Pr}}(F)$ &&& $\bar{\text{Pr}}(F)$ \\ 
\midrule

(100,50) &&& 0.144 &&& 0.144 &&& 0.144 \\ 
(100,75) &&& 0.152 &&& 0.152 &&& 0.152 \\ 
(100,100) &&& 0.140 &&& 0.140 &&& 0.140 \\ 
(250,50) &&& 0.140 &&& 0.140 &&& 0.140 \\ 
(250,75) &&& 0.102 &&& 0.102 &&& 0.102 \\ 
(250,100) &&& 0.068 &&& 0.068 &&& 0.068 \\ 
(500,50) &&& 0.070 &&& 0.070 &&& 0.070 \\ 
(500,75) &&& 0.068 &&& 0.068 &&& 0.068 \\ 
(500,100) &&& 0.034 &&& 0.034 &&& 0.034 \\

\bottomrule
\end{tabularx}

\label{tab:Sensitivity_Step3_IC_DGP1U}

\end{table}

\begin{table}[!htb]
\centering
\caption{Sensitivity analysis for classification error in DGP2U and DGP2M}

\begin{tabularx}{\textwidth}{X c X  X  c X  X  c X  X  c X   X c X  X }
\toprule
&& \multicolumn{2}{c}{$c_\lambda=3/2$} && \multicolumn{2}{c}{$c_\lambda=1$ (bench.)} && \multicolumn{2}{c}{$c_\lambda=3/4$} \\
\cmidrule{3-4} \cmidrule{6-7} \cmidrule{9-10} 
$(N,T)$ &&& $\bar{\text{Pr}}(F)$ &&& $\bar{\text{Pr}}(F)$ &&& $\bar{\text{Pr}}(F)$ \\ 
\midrule
(100,50) &&& 0.000 &&& 0.000 &&& 0.000 \\ 
(100,75) &&& 0.000 &&& 0.000 &&& 0.000 \\ 
(100,100) &&& 0.000 &&& 0.000 &&& 0.000 \\ 
(250,50) &&& 0.000 &&& 0.000 &&& 0.000 \\ 
(250,75) &&& 0.000 &&& 0.000 &&& 0.000 \\ 
(250,100) &&& 0.000 &&& 0.000 &&& 0.000 \\ 
(500,50) &&& 0.000 &&& 0.000 &&& 0.000 \\ 
(500,75) &&& 0.000 &&& 0.000 &&& 0.000 \\ 
(500,100) &&& 0.000 &&& 0.000 &&& 0.000 \\ 

\bottomrule
\end{tabularx}

\label{tab:Sensitivity_Step3_IC_DGP2U}

\end{table}

\begin{table}[!htb]
\centering
\caption{Sensitivity analysis for classification error in DGP3U and DGP3M}

\begin{tabularx}{\textwidth}{X c X  X  c X  X  c X  X  c X   X c X  X }
\toprule
&& \multicolumn{2}{c}{$c_\lambda=3/2$} && \multicolumn{2}{c}{$c_\lambda=1$ (bench.)} && \multicolumn{2}{c}{$c_\lambda=3/4$} \\
\cmidrule{3-4} \cmidrule{6-7} \cmidrule{9-10} 
$(N,T)$ &&& $\bar{\text{Pr}}(F)$ && & $\bar{\text{Pr}}(F)$ &&& $\bar{\text{Pr}}(F)$ \\ 
\midrule
(100,50) &&& 0.122 &&& 0.106 &&& 0.198 \\ 
(100,75) &&& 0.046 &&& 0.024 &&& 0.137 \\ 
(100,100) &&& 0.012 &&& 0.024 &&& 0.039 \\ 
(250,50) &&& 0.108 &&& 0.026 &&& 0.294 \\ 
(250,75) &&& 0.037 &&& 0.026 &&& 0.060 \\ 
(250,100) &&& 0.005 &&& 0.026 &&& 0.005 \\ 
(500,50) &&& 0.178 &&& 0.002 &&& 0.228 \\ 
(500,75) &&& 0.012 &&& 0.002 &&& 0.012 \\ 
(500,100) &&& 0.001 &&& 0.002 &&& 0.001 \\

\bottomrule
\end{tabularx}

\label{tab:Sensitivity_Step3_IC_DGP3U}

\end{table}

\begin{table}[!htb]
\centering
\caption{Sensitivity analysis for $\alpha^0-u$ mixture structure for DGP1U}

\begin{tabularx}{\textwidth}{X c X  X  c X  X  c X  X   }
\toprule
&& \multicolumn{2}{c}{$\tilde{c}_\lambda=3/2$} && \multicolumn{2}{c}{$\tilde{c}_\lambda=1$ (bench.)} && \multicolumn{2}{c}{$\tilde{c}_\lambda=3/4$} \\
\cmidrule{3-4} \cmidrule{6-7} \cmidrule{9-10} 
$(N,T)$ && unique & mix && unique & mix && unique & mix  \\
\midrule
(100,50) && 1.000 & 0.000 && 0.994 & 0.006 && 0.912 & 0.088 \\ 
(100,75) && 1.000 & 0.000 && 0.992 & 0.008 && 0.908 & 0.092 \\ 
(100,100) && 1.000 & 0.000 && 0.996 & 0.004 && 0.934 & 0.066 \\ 
(250,50) && 1.000 & 0.000 && 1.000 & 0.000 && 0.990 & 0.010 \\ 
(250,75) && 1.000 & 0.000 && 1.000 & 0.000 && 0.986 & 0.014 \\ 
(250,100) && 1.000 & 0.000 && 1.000 & 0.000 && 0.996 & 0.004 \\ 
(500,50) && 1.000 & 0.000 && 1.000 & 0.000 && 1.000 & 0.000 \\ 
(500,75) && 1.000 & 0.000 && 1.000 & 0.000 && 0.998 & 0.002 \\ 
(500,100) && 1.000 & 0.000 && 1.000 & 0.000 && 1.000 & 0.000 \\

\bottomrule
\end{tabularx}

\label{tab:Sensitivity_Step5_IC_DGP1U}
\end{table}

    
\begin{table}[!htb]
\centering
\caption{Sensitivity analysis for $\alpha^0-u$ mixture structure for DGP1M}

\begin{tabularx}{\textwidth}{X c X  X  c X  X  c X  X   }
\toprule
&& \multicolumn{2}{c}{$\tilde{c}_\lambda=3/2$} && \multicolumn{2}{c}{$\tilde{c}_\lambda=1$ (bench.)} && \multicolumn{2}{c}{$\tilde{c}_\lambda=3/4$} \\
\cmidrule{3-4} \cmidrule{6-7} \cmidrule{9-10} 
$(N,T)$ && unique & mix && unique & mix && unique & mix  \\
\midrule
(100,50) && 0.244 & 0.756 && 0.006 & 0.994 && 0.000 & 1.000 \\ 
(100,75) && 0.210 & 0.790 && 0.000 & 1.000 && 0.000 & 1.000 \\ 
(100,100) && 0.142 & 0.858 && 0.002 & 0.998 && 0.000 & 1.000 \\ 
(250,50) && 0.018 & 0.982 && 0.000 & 1.000 && 0.000 & 1.000 \\ 
(250,75) && 0.010 & 0.990 && 0.000 & 1.000 && 0.000 & 1.000 \\ 
(250,100) && 0.002 & 0.998 && 0.000 & 1.000 && 0.000 & 1.000 \\ 
(500,50) && 0.000 & 1.000 && 0.000 & 1.000 && 0.000 & 1.000 \\ 
(500,75) && 0.000 & 1.000 && 0.000 & 1.000 && 0.000 & 1.000 \\ 
(500,100) && 0.000 & 1.000 && 0.000 & 1.000 && 0.000 & 1.000 \\

\bottomrule
\end{tabularx}

\label{tab:Sensitivity_Step5_IC_DGP1M}
\end{table}

    
\begin{table}[!htb]
\centering
\caption{Sensitivity analysis for $ \alpha^0-u$ mixture structure for DGP2U}

\begin{tabularx}{\textwidth}{X c X  X  c X  X  c X  X   }
\toprule
&& \multicolumn{2}{c}{$\tilde{c}_\lambda=3/2$} && \multicolumn{2}{c}{$\tilde{c}_\lambda=1$ (bench.)} && \multicolumn{2}{c}{$\tilde{c}_\lambda=3/4$} \\
\cmidrule{3-4} \cmidrule{6-7} \cmidrule{9-10} 
$(N,T)$ && unique & mix && unique & mix && unique & mix  \\
\midrule
(100,50) && 1.000 & 0.000 && 0.986 & 0.014 && 0.844 & 0.156 \\ 
(100,75) && 1.000 & 0.000 && 0.986 & 0.014 && 0.844 & 0.156 \\ 
(100,100) && 1.000 & 0.000 && 0.994 & 0.006 && 0.854 & 0.146 \\ 
(250,50) && 1.000 & 0.000 && 1.000 & 0.000 && 0.978 & 0.022 \\ 
(250,75) && 1.000 & 0.000 && 1.000 & 0.000 && 0.978 & 0.022 \\ 
(250,100) && 1.000 & 0.000 && 1.000 & 0.000 && 0.992 & 0.008 \\ 
(500,50) && 1.000 & 0.000 && 1.000 & 0.000 && 1.000 & 0.000 \\ 
(500,75) && 1.000 & 0.000 && 1.000 & 0.000 && 1.000 & 0.000 \\ 
(500,100) && 1.000 & 0.000 && 1.000 & 0.000 && 1.000 & 0.000 \\

\bottomrule
\end{tabularx}

\label{tab:Sensitivity_Step5_IC_DGP2U}
\end{table}

    
\begin{table}[!htb]
\centering
\caption{Sensitivity analysis for $\alpha^0-u$ mixture structure for DGP2M}

\begin{tabularx}{\textwidth}{X c X  X  c X  X  c X  X   }
\toprule
&& \multicolumn{2}{c}{$\tilde{c}_\lambda=3/2$} && \multicolumn{2}{c}{$\tilde{c}_\lambda=1$ (bench.)} && \multicolumn{2}{c}{$\tilde{c}_\lambda=3/4$} \\
\cmidrule{3-4} \cmidrule{6-7} \cmidrule{9-10} 
$(N,T)$ && unique & mix && unique & mix && unique & mix  \\
\midrule
(100,50) && 0.360 & 0.640 && 0.028 & 0.972 && 0.000 & 1.000 \\ 
(100,75) && 0.304 & 0.696 && 0.010 & 0.990 && 0.000 & 1.000 \\ 
(100,100) && 0.210 & 0.790 && 0.008 & 0.992 && 0.000 & 1.000 \\ 
(250,50) && 0.094 & 0.906 && 0.000 & 1.000 && 0.000 & 1.000 \\ 
(250,75) && 0.052 & 0.948 && 0.000 & 1.000 && 0.000 & 1.000 \\ 
(250,100) && 0.022 & 0.978 && 0.000 & 1.000 && 0.000 & 1.000 \\ 
(500,50) && 0.006 & 0.994 && 0.000 & 1.000 && 0.000 & 1.000 \\ 
(500,75) && 0.000 & 1.000 && 0.000 & 1.000 && 0.000 & 1.000 \\ 
(500,100) && 0.000 & 1.000 && 0.000 & 1.000 && 0.000 & 1.000 \\ 

\bottomrule
\end{tabularx}

\label{tab:Sensitivity_Step5_IC_DGP2M}
\end{table}

    
\begin{table}[!htb]
\centering
\caption{Sensitivity analysis for $\alpha^0-u$ mixture structure for DGP3U}

\begin{tabularx}{\textwidth}{X c X  X  c X  X  c X  X   }
\toprule
&& \multicolumn{2}{c}{$\tilde{c}_\lambda=3/2$} && \multicolumn{2}{c}{$\tilde{c}_\lambda=1$ (bench.)} && \multicolumn{2}{c}{$\tilde{c}_\lambda=3/4$} \\
\cmidrule{3-4} \cmidrule{6-7} \cmidrule{9-10} 
$(N,T)$ && unique & mix && unique & mix && unique & mix  \\
\midrule
(100,50) && 1.000 & 0.000 && 0.968 & 0.032 && 0.778 & 0.222 \\ 
(100,75) && 1.000 & 0.000 && 0.970 & 0.030 && 0.810 & 0.190 \\ 
(100,100) && 1.000 & 0.000 && 0.992 & 0.008 && 0.852 & 0.148 \\ 
(250,50) && 1.000 & 0.000 && 0.992 & 0.008 && 0.948 & 0.052 \\ 
(250,75) && 1.000 & 0.000 && 1.000 & 0.000 && 0.976 & 0.024 \\ 
(250,100) && 1.000 & 0.000 && 1.000 & 0.000 && 0.976 & 0.024 \\ 
(500,50) && 1.000 & 0.000 && 1.000 & 0.000 && 0.994 & 0.006 \\ 
(500,75) && 1.000 & 0.000 && 1.000 & 0.000 && 1.000 & 0.000 \\ 
(500,100) && 1.000 & 0.000 && 1.000 & 0.000 && 1.000 & 0.000 \\ 

\bottomrule
\end{tabularx}

\label{tab:Sensitivity_Step5_IC_DGP3U}
\end{table}

    
\begin{table}[!htb]
\centering
\caption{Sensitivity analysis for $\alpha^0-u$ mixture structure for DGP3M}

\begin{tabularx}{\textwidth}{X c X  X  c X  X  c X  X   }
\toprule
&& \multicolumn{2}{c}{$\tilde{c}_\lambda=3/2$} && \multicolumn{2}{c}{$\tilde{c}_\lambda=1$ (bench.)} && \multicolumn{2}{c}{$\tilde{c}_\lambda=3/4$} \\
\cmidrule{3-4} \cmidrule{6-7} \cmidrule{9-10} 
$(N,T)$ && unique & mix && unique & mix && unique & mix  \\
\midrule
(100,50) && 0.496 & 0.504 && 0.064 & 0.936 && 0.006 & 0.994 \\ 
(100,75) && 0.328 & 0.672 && 0.018 & 0.982 && 0.002 & 0.998 \\ 
(100,100) && 0.258 & 0.742 && 0.016 & 0.984 && 0.000 & 1.000 \\ 
(250,50) && 0.148 & 0.852 && 0.010 & 0.990 && 0.006 & 0.994 \\ 
(250,75) && 0.038 & 0.962 && 0.000 & 1.000 && 0.000 & 1.000 \\ 
(250,100) && 0.028 & 0.972 && 0.000 & 1.000 && 0.000 & 1.000 \\ 
(500,50) && 0.006 & 0.994 && 0.002 & 0.998 && 0.002 & 0.998 \\ 
(500,75) && 0.000 & 1.000 && 0.000 & 1.000 && 0.000 & 1.000 \\ 
(500,100) && 0.000 & 1.000 && 0.000 & 1.000 && 0.000 & 1.000 \\  

\bottomrule
\end{tabularx}

\label{tab:Sensitivity_Step5_IC_DGP3M}
\end{table}


\begin{table}[!htb]
\centering
\caption{Summary statistics of data used in the application}

\begin{tabularx}{\textwidth}{X X X X X X }
\toprule
& MEAN & STD & MIN & MAX & $N\times T$ \\
\midrule
&\multicolumn{5}{c}{Panel A: raw variables }  \\\cline{2-6}
$C_{it}$ &  $9.7974\times10^{4}$ & $7.0182\times10^{5}$ & $91.3791$ & $2.2152\times10^{7}$ & $37280$  \\
$W_{it1}$ & $17.4185$ & $9.7993$ & $0.0913$ & $252.9964$ & $37280$ \\
$W_{it2}$ & $0.0073$ & $0.0094$ & $7.9946\times10^{-7}$ & $0.4393$ & $37280$ \\
$W_{it3}$ & $0.02718$ & $0.0200$ & $6.1287\times10^{-4}$ & $0.4423$ & $37280$ \\
$Y_{it1}$ & $2.1086\times10^{5}$ & $1.4001\times10^{6}$ & $4.2900$ & $4.9625\times10^{7}$ & $37280$ \\
$Y_{it2}$ & $1.6894\times10^{6}$ & $1.0978\times10^{7}$ & $3020$ & $3.1730\times10^{8}$ & $37280$ \\
$Y_{it3}$ & $1.3421\times10^{6}$ & $1.1142\times10^{7}$ & $2.4091\times10^{3}$ & $3.9273\times10^{8}$ & $37280$ \\

&\multicolumn{5}{c}{Panel B: variables used in regressions}  \\\cline{2-6}
$\log c_{it}^*$  & 12.8442 & 1.6370 & 9.1203  & 20.3854 & 37280\\
$\log w_{it1}$ & 6.5320  & 0.6060 & 2.8625  & 10.1837 & 37280\\
$\log w_{it2}$ & -1.7023 & 1.0291 & -8.4596 & 2.0101  & 37280\\
$\log y_{it1}$ & 9.8573  & 1.8283 & 1.4563  & 17.7200 & 37280\\
$\log y_{it2}$ & 11.8815 & 1.7312 & 8.0130  & 19.5754 & 37280\\
$\log y_{it3}$ & 11.6393 & 1.6402 & 7.7870  & 19.7886 & 37280 \\
\bottomrule
\end{tabularx}
\begin{tablenotes}
      \footnotesize
      \item 
      \emph{Note}: Raw variables in Panel A are denominated in Millions of 1986 USD. Variables in Panel B have been divided by $W_{it3}$ before taken logs.
\end{tablenotes}
\label{tab:Summary_Statistics}
\end{table}

\clearpage
\subsection{Tables for Simulations in Appendix \ref{APP:more_mix}} \label{APP:tables_more}

\vfill
\begin{table}[!htb]

\centering
\caption{BIAS and RMSE over 500 MC Iterations for DGP1M}
\scalebox{0.71}{
\begin{tabularx}{1.4\textwidth}{X c XX c XX c XX c XX c XX}

\toprule
&& \multicolumn{2}{c}{$\hat\alpha^{0}_{(1)}$} && \multicolumn{2}{c}{$\hat{\sigma}_{u(1)}$} && \multicolumn{2}{c}{$\hat\alpha^{0}_{(2)}$} && \multicolumn{2}{c}{$\hat{\sigma}_{u(2)}$} && \multicolumn{2}{c}{$\hat{\tau}_{1}$} \\
\cmidrule{3-4} \cmidrule{6-7} \cmidrule{9-10} \cmidrule{12-13} \cmidrule{15-16}
$(N,T)$ && BIAS & RMSE && BIAS & RMSE && BIAS & RMSE && BIAS & RMSE && BIAS & RMSE \\

\midrule

(100,50) && 0.116 & 0.140 && 0.242 & 0.300 && 0.106 & 0.136 && 0.119 & 0.149 && 0.031 & 0.040 \\ 
(100,75) && 0.094 & 0.115 && 0.184 & 0.230 && 0.086 & 0.112 && 0.116 & 0.145 && 0.029 & 0.037 \\ 
(100,100) && 0.079 & 0.096 && 0.184 & 0.238 && 0.095 & 0.127 && 0.111 & 0.141 && 0.031 & 0.041 \\ 
(250,50) && 0.173 & 0.189 && 0.328 & 0.365 && 0.087 & 0.112 && 0.087 & 0.110 && 0.031 & 0.041 \\ 
(250,75) && 0.135 & 0.149 && 0.255 & 0.289 && 0.074 & 0.102 && 0.081 & 0.102 && 0.026 & 0.034 \\ 
(250,100) && 0.111 & 0.123 && 0.211 & 0.241 && 0.070 & 0.095 && 0.078 & 0.097 && 0.026 & 0.033 \\ 
(500,50) && 0.206 & 0.215 && 0.361 & 0.386 && 0.071 & 0.092 && 0.069 & 0.088 && 0.029 & 0.037 \\ 
(500,75) && 0.165 & 0.174 && 0.273 & 0.293 && 0.055 & 0.070 && 0.061 & 0.075 && 0.022 & 0.027 \\ 
(500,100) && 0.143 & 0.152 && 0.232 & 0.253 && 0.051 & 0.068 && 0.066 & 0.085 && 0.022 & 0.028 \\ 

\bottomrule
\end{tabularx}}

\label{tab:DGP1M_parameters_AP}

\end{table}
\vfill


\begin{table}[!htb]

\centering
\caption{BIAS and RMSE over 500 MC iterations for DGP1T}

\scalebox{0.71}{
\begin{tabularx}{1.4\textwidth}{X c XX c XX c XX c XX c XX c XX c XXc XX }
\toprule

&& \multicolumn{2}{c}{$\hat\alpha_{(1)}^{0}$} && \multicolumn{2}{c}{$\hat{\sigma}_{u(1)}$} && \multicolumn{2}{c}{$\hat\alpha_{(2)}^{0}$} && \multicolumn{2}{c}{$\hat{\sigma}_{u(2)}$} && \multicolumn{2}{c}{$\hat{\alpha}_{(3)}^{0}$} && \multicolumn{2}{c}{$\hat{\sigma}_{u(3)}$} && \multicolumn{2}{c}{$\hat{\tau}_{1}$} && \multicolumn{2}{c}{$\hat{\tau}_{2}$} \\
\cmidrule{3-4} \cmidrule{6-7} \cmidrule{9-10} \cmidrule{12-13} \cmidrule{15-16} \cmidrule{18-19} \cmidrule{21-22} \cmidrule{24-25} 
$(N,T)$ && BIAS & RMSE && BIAS & RMSE && BIAS & RMSE && BIAS & RMSE && BIAS & RMSE && BIAS & RMSE && BIAS & RMSE && BIAS & RMSE \\
\midrule

(100,50) && 0.095 & 0.123 && 0.240 & 0.312 && 0.080 & 0.103 && 0.104 & 0.133 && 0.093 & 0.117 && 0.117 & 0.145 && 0.038 & 0.047 && 0.038 & 0.048 \\ 
(100,75) && 0.084 & 0.109 && 0.222 & 0.282 && 0.071 & 0.095 && 0.106 & 0.137 && 0.078 & 0.099 && 0.108 & 0.133 && 0.037 & 0.048 && 0.037 & 0.047 \\ 
(100,100) && 0.083 & 0.109 && 0.223 & 0.282 && 0.069 & 0.091 && 0.107 & 0.135 && 0.071 & 0.092 && 0.102 & 0.128 && 0.037 & 0.048 && 0.035 & 0.046 \\ 
(250,50) && 0.073 & 0.094 && 0.235 & 0.293 && 0.066 & 0.088 && 0.074 & 0.096 && 0.125 & 0.148 && 0.147 & 0.173 && 0.033 & 0.042 && 0.038 & 0.048 \\ 
(250,75) && 0.066 & 0.085 && 0.229 & 0.287 && 0.056 & 0.075 && 0.075 & 0.096 && 0.098 & 0.119 && 0.118 & 0.142 && 0.033 & 0.042 && 0.035 & 0.045 \\ 
(250,100) && 0.057 & 0.074 && 0.199 & 0.257 && 0.056 & 0.073 && 0.072 & 0.093 && 0.082 & 0.100 && 0.103 & 0.127 && 0.030 & 0.037 && 0.031 & 0.040 \\ 
(500,50) && 0.060 & 0.075 && 0.263 & 0.317 && 0.057 & 0.073 && 0.059 & 0.077 && 0.163 & 0.179 && 0.183 & 0.200 && 0.030 & 0.039 && 0.043 & 0.052 \\ 
(500,75) && 0.050 & 0.066 && 0.237 & 0.290 && 0.048 & 0.060 && 0.058 & 0.073 && 0.128 & 0.143 && 0.148 & 0.168 && 0.028 & 0.036 && 0.037 & 0.045 \\ 
(500,100) && 0.050 & 0.065 && 0.217 & 0.270 && 0.042 & 0.053 && 0.052 & 0.067 && 0.110 & 0.125 && 0.128 & 0.148 && 0.027 & 0.034 && 0.033 & 0.042 \\

\bottomrule
\end{tabularx}}

\label{tab:DGP1T_parameters_AP}
\end{table}

\vfill
\begin{table}[!htb]
\centering
\caption{BIAS and RMSE over 500 MC Iterations for DGP2M}
\scalebox{0.71}{
\begin{tabularx}{1.4\textwidth}{X c XX c XX c XX c XX c XX }
\toprule
&& \multicolumn{2}{c}{$\hat\alpha^{0}_{(1)}$} && \multicolumn{2}{c}{$\hat{\sigma}_{u(1)}$} && \multicolumn{2}{c}{$\hat\alpha^{0}_{(2)}$} && \multicolumn{2}{c}{$\hat{\sigma}_{u(2)}$} && \multicolumn{2}{c}{$\hat{\tau}_{1}$} \\
\cmidrule{3-4} \cmidrule{6-7} \cmidrule{9-10} \cmidrule{12-13} \cmidrule{15-16} 
$(N,T)$ && BIAS & RMSE && BIAS & RMSE && BIAS & RMSE && BIAS & RMSE && BIAS & RMSE \\

\midrule

(100,50) && 0.059 & 0.075 && 0.168 & 0.221 && 0.125 & 0.159 && 0.123 & 0.155 && 0.036 & 0.046 \\ 
(100,75) && 0.047 & 0.059 && 0.137 & 0.181 && 0.106 & 0.134 && 0.126 & 0.154 && 0.030 & 0.039 \\ 
(100,100) && 0.040 & 0.050 && 0.132 & 0.180 && 0.096 & 0.125 && 0.113 & 0.144 && 0.032 & 0.042 \\ 
(250,50) && 0.082 & 0.094 && 0.196 & 0.231 && 0.102 & 0.128 && 0.104 & 0.127 && 0.033 & 0.040 \\ 
(250,75) && 0.061 & 0.071 && 0.158 & 0.191 && 0.088 & 0.112 && 0.096 & 0.120 && 0.029 & 0.036 \\ 
(250,100) && 0.048 & 0.056 && 0.144 & 0.176 && 0.075 & 0.095 && 0.090 & 0.113 && 0.029 & 0.036 \\ 
(500,50) && 0.104 & 0.111 && 0.221 & 0.245 && 0.088 & 0.110 && 0.087 & 0.108 && 0.027 & 0.034 \\ 
(500,75) && 0.078 & 0.086 && 0.175 & 0.197 && 0.069 & 0.088 && 0.087 & 0.108 && 0.026 & 0.032 \\ 
(500,100) && 0.066 & 0.072 && 0.154 & 0.173 && 0.061 & 0.076 && 0.086 & 0.106 && 0.025 & 0.030 \\ 

\bottomrule
\end{tabularx}}

\label{tab:DGP2M_parameters_AP}
\end{table}
\vfill

\begin{table}[!htb]
\centering
\caption{BIAS and RMSE over 500 MC iterations for DGP2T}
\scalebox{0.71}{
\begin{tabularx}{1.4\textwidth}{X c XX c XX c XX c XX c XX c XX c XXc XX }
\toprule
&& \multicolumn{2}{c}{$\hat\alpha_{(1)}^{0}$} && \multicolumn{2}{c}{$\hat{\sigma}_{u(1)}$} && \multicolumn{2}{c}{$\hat\alpha_{(2)}^{0}$} && \multicolumn{2}{c}{$\hat{\sigma}_{u(2)}$} && \multicolumn{2}{c}{$\hat{\alpha}_{(3)}^{0}$} && \multicolumn{2}{c}{$\hat{\sigma}_{u(3)}$} && \multicolumn{2}{c}{$\hat{\tau}_{1}$} && \multicolumn{2}{c}{$\hat{\tau}_{2}$} \\
\cmidrule{3-4} \cmidrule{6-7} \cmidrule{9-10} \cmidrule{12-13} \cmidrule{15-16} \cmidrule{18-19} \cmidrule{21-22} \cmidrule{24-25}  
$(N,T)$ && BIAS & RMSE && BIAS & RMSE && BIAS & RMSE && BIAS & RMSE && BIAS & RMSE && BIAS & RMSE && BIAS & RMSE && BIAS & RMSE \\
\midrule

(100,50) && 0.085 & 0.114 && 0.221 & 0.287 && 0.078 & 0.104 && 0.108 & 0.140 && 0.172 & 0.209 && 0.147 & 0.177 && 0.037 & 0.048 && 0.039 & 0.051 \\ 
(100,75) && 0.074 & 0.099 && 0.197 & 0.259 && 0.078 & 0.104 && 0.112 & 0.143 && 0.139 & 0.168 && 0.125 & 0.155 && 0.037 & 0.048 && 0.036 & 0.045 \\ 
(100,100) && 0.073 & 0.099 && 0.210 & 0.273 && 0.076 & 0.102 && 0.105 & 0.137 && 0.114 & 0.142 && 0.115 & 0.145 && 0.037 & 0.048 && 0.037 & 0.048 \\ 
(250,50) && 0.061 & 0.085 && 0.201 & 0.266 && 0.064 & 0.085 && 0.071 & 0.091 && 0.238 & 0.266 && 0.205 & 0.230 && 0.030 & 0.039 && 0.041 & 0.052 \\ 
(250,75) && 0.054 & 0.073 && 0.195 & 0.248 && 0.060 & 0.082 && 0.073 & 0.094 && 0.191 & 0.218 && 0.165 & 0.191 && 0.030 & 0.039 && 0.038 & 0.049 \\ 
(250,100) && 0.046 & 0.063 && 0.190 & 0.243 && 0.053 & 0.071 && 0.074 & 0.097 && 0.165 & 0.188 && 0.143 & 0.173 && 0.031 & 0.040 && 0.036 & 0.046 \\ 
(500,50) && 0.045 & 0.059 && 0.206 & 0.259 && 0.053 & 0.069 && 0.063 & 0.080 && 0.301 & 0.318 && 0.257 & 0.273 && 0.026 & 0.034 && 0.042 & 0.052 \\ 
(500,75) && 0.040 & 0.053 && 0.194 & 0.244 && 0.043 & 0.055 && 0.057 & 0.072 && 0.236 & 0.251 && 0.213 & 0.231 && 0.026 & 0.034 && 0.039 & 0.047 \\ 
(500,100) && 0.039 & 0.051 && 0.188 & 0.238 && 0.041 & 0.053 && 0.056 & 0.071 && 0.211 & 0.227 && 0.183 & 0.201 && 0.026 & 0.033 && 0.037 & 0.046 \\

\bottomrule
\end{tabularx}}

\label{tab:DGP2T_parameters_AP}
\end{table}

\begin{table}[!htb]
\centering

\caption{BIAS and RMSE over 500 MC iterations for DGP3M}
\scalebox{0.71}{
\begin{tabularx}{1.4\textwidth}{X c XX c XX c XX c XX c XX  }
\toprule
&& \multicolumn{2}{c}{$\hat\alpha^{0}_{(1)}$} && \multicolumn{2}{c}{$\hat{\sigma}_{u(1)}$} && \multicolumn{2}{c}{$\hat\alpha^{0}_{(2)}$} && \multicolumn{2}{c}{$\hat{\sigma}_{u(2)}$} && \multicolumn{2}{c}{$\hat{\tau}_{1}$} \\
\cmidrule{3-4} \cmidrule{6-7} \cmidrule{9-10} \cmidrule{12-13} \cmidrule{15-16} 
$(N,T)$ && BIAS & RMSE && BIAS & RMSE && BIAS & RMSE && BIAS & RMSE && BIAS & RMSE \\
\midrule

(100,50) && 0.127 & 0.173 && 0.247 & 0.335 && 0.135 & 0.172 && 0.139 & 0.176 && 0.037 & 0.050 \\ 
(100,75) && 0.083 & 0.109 && 0.179 & 0.231 && 0.106 & 0.136 && 0.123 & 0.154 && 0.030 & 0.040 \\ 
(100,100) && 0.069 & 0.088 && 0.148 & 0.194 && 0.102 & 0.131 && 0.109 & 0.139 && 0.030 & 0.039 \\ 
(250,50) && 0.173 & 0.203 && 0.316 & 0.371 && 0.106 & 0.136 && 0.108 & 0.135 && 0.034 & 0.045 \\ 
(250,75) && 0.113 & 0.133 && 0.222 & 0.253 && 0.081 & 0.104 && 0.083 & 0.104 && 0.029 & 0.036 \\ 
(250,100) && 0.089 & 0.107 && 0.175 & 0.203 && 0.075 & 0.095 && 0.089 & 0.112 && 0.026 & 0.032 \\ 
(500,50) && 0.206 & 0.225 && 0.335 & 0.370 && 0.083 & 0.108 && 0.101 & 0.125 && 0.031 & 0.041 \\ 
(500,75) && 0.139 & 0.151 && 0.237 & 0.260 && 0.063 & 0.079 && 0.073 & 0.094 && 0.024 & 0.030 \\ 
(500,100) && 0.116 & 0.127 && 0.207 & 0.226 && 0.059 & 0.076 && 0.070 & 0.089 && 0.023 & 0.028 \\

\bottomrule
\end{tabularx}}

\label{tab:DGP3M_parameters_AP}
\end{table}


\begin{table}[!htb]
\centering
\caption{BIAS and RMSE over 500 MC iterations for DGP3T}
\scalebox{0.71}{
\begin{tabularx}{1.4\textwidth}{X c XX c XX c XX c XX c XX c XX c XX c XX}
\toprule
&& \multicolumn{2}{c}{$\hat\alpha_{(1)}^{0}$} && \multicolumn{2}{c}{$\hat{\sigma}_{u(1)}$} && \multicolumn{2}{c}{$\hat\alpha_{(2)}^{0}$} && \multicolumn{2}{c}{$\hat{\sigma}_{u(2)}$} && \multicolumn{2}{c}{$\hat{\alpha}_{(3)}^{0}$} && \multicolumn{2}{c}{$\hat{\sigma}_{u(3)}$} && \multicolumn{2}{c}{$\hat{\tau}_{1}$} && \multicolumn{2}{c}{$\hat{\tau}_{2}$} \\
\cmidrule{3-4} \cmidrule{6-7} \cmidrule{9-10} \cmidrule{12-13} \cmidrule{15-16} \cmidrule{18-19} \cmidrule{21-22} \cmidrule{24-25} 
$(N,T)$ && BIAS & RMSE && BIAS & RMSE && BIAS & RMSE && BIAS & RMSE && BIAS & RMSE && BIAS & RMSE && BIAS & RMSE && BIAS & RMSE \\
\midrule

(100,50) && 0.087 & 0.116 && 0.211 & 0.275 && 0.102 & 0.132 && 0.114 & 0.150 && 0.203 & 0.249 && 0.366 & 0.442 && 0.039 & 0.050 && 0.041 & 0.054 \\ 
(100,75) && 0.076 & 0.103 && 0.226 & 0.281 && 0.086 & 0.113 && 0.103 & 0.131 && 0.149 & 0.184 && 0.281 & 0.357 && 0.036 & 0.046 && 0.036 & 0.047 \\ 
(100,100) && 0.074 & 0.099 && 0.217 & 0.277 && 0.083 & 0.107 && 0.109 & 0.138 && 0.129 & 0.161 && 0.251 & 0.323 && 0.036 & 0.049 && 0.037 & 0.048 \\ 
(250,50) && 0.063 & 0.085 && 0.189 & 0.250 && 0.075 & 0.096 && 0.079 & 0.099 && 0.232 & 0.265 && 0.520 & 0.582 && 0.032 & 0.041 && 0.044 & 0.056 \\ 
(250,75) && 0.054 & 0.072 && 0.177 & 0.227 && 0.062 & 0.079 && 0.072 & 0.091 && 0.186 & 0.214 && 0.359 & 0.418 && 0.027 & 0.035 && 0.034 & 0.043 \\ 
(250,100) && 0.051 & 0.070 && 0.178 & 0.232 && 0.055 & 0.073 && 0.070 & 0.090 && 0.150 & 0.176 && 0.299 & 0.359 && 0.029 & 0.036 && 0.032 & 0.041 \\ 
(500,50) && 0.051 & 0.065 && 0.187 & 0.243 && 0.064 & 0.081 && 0.064 & 0.079 && 0.296 & 0.317 && 0.674 & 0.721 && 0.028 & 0.037 && 0.048 & 0.058 \\ 
(500,75) && 0.042 & 0.054 && 0.182 & 0.230 && 0.054 & 0.067 && 0.057 & 0.074 && 0.236 & 0.253 && 0.457 & 0.506 && 0.027 & 0.034 && 0.037 & 0.046 \\ 
(500,100) && 0.038 & 0.051 && 0.180 & 0.232 && 0.047 & 0.060 && 0.055 & 0.070 && 0.200 & 0.217 && 0.391 & 0.437 && 0.025 & 0.033 && 0.035 & 0.044 \\

\bottomrule
\end{tabularx}}

\label{tab:DGP3T_parameters_AP}
\end{table}

    
\begin{table}[!htb]
\centering
\caption{Sensitivity analysis for $\alpha^0-u$ mixture structure for DGP1M}

\hspace*{-1.3cm}
\begin{tabularx}{1.15\textwidth}{X c XXX  cXXX  cXXX   }
\toprule
&& \multicolumn{3}{c}{$\tilde{c}_\lambda=3/2$} && \multicolumn{3}{c}{$\tilde{c}_\lambda=1$ (bench.)} && \multicolumn{3}{c}{$\tilde{c}_\lambda=3/4$} \\
\cmidrule{3-5} \cmidrule{7-9} \cmidrule{11-13} 
$(N,T)$ && $\mathcal{K}=1$ & $\mathcal{K}=2$ & $\mathcal{K}=3$ && $\mathcal{K}=1$ & $\mathcal{K}=2$ & $\mathcal{K}=3$ && $\mathcal{K}=1$ & $\mathcal{K}=2$ & $\mathcal{K}=3$   \\
\midrule

(100,50) && 0.114 & 0.876 & 0.010 && 0.024 & 0.886 & 0.090 && 0.008 & 0.808 & 0.184\\ 
(100,75) && 0.050 & 0.938 & 0.012 && 0.006 & 0.946 & 0.048 && 0.000 & 0.882 & 0.118\\ 
(100,100) && 0.048 & 0.946 & 0.006 && 0.004 & 0.948 & 0.048 && 0.000 & 0.900 & 0.100\\ 
(250,50) && 0.012 & 0.888 & 0.100 && 0.002 & 0.738 & 0.260 && 0.002 & 0.614 & 0.384\\ 
(250,75) && 0.012 & 0.938 & 0.050 && 0.000 & 0.800 & 0.200 && 0.000 & 0.698 & 0.302\\ 
(250,100) && 0.006 & 0.960 & 0.034 && 0.000 & 0.870 & 0.130 && 0.000 & 0.774 & 0.226\\ 
(500,50) && 0.000 & 0.766 & 0.234 && 0.000 & 0.492 & 0.508 && 0.000 & 0.348 & 0.652\\ 
(500,75) && 0.000 & 0.842 & 0.158 && 0.000 & 0.646 & 0.354 && 0.000 & 0.498 & 0.502\\ 
(500,100) && 0.000 & 0.906 & 0.094 && 0.000 & 0.714 & 0.286 && 0.000 & 0.576 & 0.424\\ 

\bottomrule
\end{tabularx}

\label{tab:Sensitivity_Step5_IC_DGP1M_AP}
\end{table}

    
\begin{table}[!htb]
\centering
\caption{Sensitivity analysis for $\alpha^{0}-u$ mixture structure for DGP1T}

\hspace*{-1.3cm}
\begin{tabularx}{1.15\textwidth}{X c XXXX  cXXXX  cXXXX   }
\toprule
&& \multicolumn{4}{c}{$\tilde{c}_\lambda=3/2$} && \multicolumn{4}{c}{$\tilde{c}_\lambda=1$ (bench.)} && \multicolumn{4}{c}{$\tilde{c}_\lambda=3/4$} \\
\cmidrule{3-6} \cmidrule{8-11} \cmidrule{13-16} 
$(N,T)$ && $\mathcal{K}=1$ & $\mathcal{K}=2$ & $\mathcal{K}=3$ & $\mathcal{K}=4$ && $\mathcal{K}=1$ & $\mathcal{K}=2$ & $\mathcal{K}=3$ & $\mathcal{K}=4$ && $\mathcal{K}=1$ & $\mathcal{K}=2$ & $\mathcal{K}=3$ & $\mathcal{K}=4$  \\
\midrule

(100,50) && 0.000 & 0.000 & 0.990 & 0.010 && 0.000 & 0.000 & 0.980 & 0.020 && 0.000 & 0.000 & 0.940 & 0.060\\ 
(100,75) && 0.000 & 0.000 & 0.998 & 0.002 && 0.000 & 0.000 & 0.986 & 0.014 && 0.000 & 0.000 & 0.958 & 0.042\\ 
(100,100) && 0.000 & 0.000 & 1.000 & 0.000 && 0.000 & 0.000 & 0.986 & 0.014 && 0.000 & 0.000 & 0.960 & 0.040\\ 
(250,50) && 0.000 & 0.000 & 0.992 & 0.008 && 0.000 & 0.000 & 0.976 & 0.024 && 0.000 & 0.000 & 0.920 & 0.080\\ 
(250,75) && 0.000 & 0.000 & 0.994 & 0.006 && 0.000 & 0.000 & 0.982 & 0.018 && 0.000 & 0.000 & 0.956 & 0.044\\ 
(250,100) && 0.000 & 0.000 & 1.000 & 0.000 && 0.000 & 0.000 & 1.000 & 0.000 && 0.000 & 0.000 & 0.986 & 0.014\\ 
(500,50) && 0.000 & 0.000 & 0.978 & 0.022 && 0.000 & 0.000 & 0.926 & 0.074 && 0.000 & 0.000 & 0.832 & 0.168\\ 
(500,75) && 0.000 & 0.000 & 0.992 & 0.008 && 0.000 & 0.000 & 0.974 & 0.026 && 0.000 & 0.000 & 0.930 & 0.070\\ 
(500,100) && 0.000 & 0.000 & 0.998 & 0.002 && 0.000 & 0.000 & 0.990 & 0.010 && 0.000 & 0.000 & 0.948 & 0.052\\

\bottomrule
\end{tabularx}

\label{tab:Sensitivity_Step5_IC_DGP1T_AP}
\end{table}

    
\begin{table}[!htb]
\centering
\caption{Sensitivity analysis for $\alpha^{0}-u$ mixture structure for DGP2M}

\hspace*{-1.3cm}
\begin{tabularx}{1.15\textwidth}{X c XXX  cXXX  cXXX   }
\toprule
&& \multicolumn{3}{c}{$\tilde{c}_\lambda=3/2$} && \multicolumn{3}{c}{$\tilde{c}_\lambda=1$ (bench.)} && \multicolumn{3}{c}{$\tilde{c}_\lambda=3/4$} \\
\cmidrule{3-5} \cmidrule{7-9} \cmidrule{11-13} 
$(N,T)$ && $\mathcal{K}=1$ & $\mathcal{K}=2$ & $\mathcal{K}=3$ && $\mathcal{K}=1$ & $\mathcal{K}=2$ & $\mathcal{K}=3$ && $\mathcal{K}=1$ & $\mathcal{K}=2$ & $\mathcal{K}=3$   \\
\midrule

(100,50) && 0.124 & 0.874 & 0.002 && 0.014 & 0.962 & 0.024 && 0.002 & 0.924 & 0.074\\ 
(100,75) && 0.068 & 0.930 & 0.002 && 0.004 & 0.978 & 0.018 && 0.000 & 0.946 & 0.054\\ 
(100,100) && 0.046 & 0.950 & 0.004 && 0.004 & 0.986 & 0.010 && 0.002 & 0.966 & 0.032\\ 
(250,50) && 0.020 & 0.978 & 0.002 && 0.002 & 0.954 & 0.044 && 0.000 & 0.878 & 0.122\\ 
(250,75) && 0.012 & 0.982 & 0.006 && 0.000 & 0.966 & 0.034 && 0.000 & 0.932 & 0.068\\ 
(250,100) && 0.002 & 0.998 & 0.000 && 0.000 & 0.992 & 0.008 && 0.000 & 0.958 & 0.042\\ 
(500,50) && 0.002 & 0.978 & 0.020 && 0.000 & 0.900 & 0.100 && 0.000 & 0.778 & 0.222\\ 
(500,75) && 0.000 & 0.994 & 0.006 && 0.000 & 0.934 & 0.066 && 0.000 & 0.862 & 0.138\\ 
(500,100) && 0.000 & 0.994 & 0.006 && 0.000 & 0.966 & 0.034 && 0.000 & 0.890 & 0.110\\

\bottomrule
\end{tabularx}

\label{tab:Sensitivity_Step5_IC_DGP2M_AP}
\end{table}

    
\begin{table}[!htb]
\centering
\caption{Sensitivity analysis for $\alpha^0-u$ mixture structure for DGP2T}

\hspace*{-1.3cm}
\begin{tabularx}{1.15\textwidth}{X c XXXX  cXXXX  cXXXX   }
\toprule
&& \multicolumn{4}{c}{$\tilde{c}_\lambda=3/2$} && \multicolumn{4}{c}{$\tilde{c}_\lambda=1$ (bench.)} && \multicolumn{4}{c}{$\tilde{c}_\lambda=3/4$} \\
\cmidrule{3-6} \cmidrule{8-11} \cmidrule{13-16} 
$(N,T)$ && $\mathcal{K}=1$ & $\mathcal{K}=2$ & $\mathcal{K}=3$ & $\mathcal{K}=4$ && $\mathcal{K}=1$ & $\mathcal{K}=2$ & $\mathcal{K}=3$ & $\mathcal{K}=4$ && $\mathcal{K}=1$ & $\mathcal{K}=2$ & $\mathcal{K}=3$ & $\mathcal{K}=4$  \\
\midrule

(100,50) && 0.000 & 0.000 & 0.998 & 0.002 && 0.000 & 0.000 & 0.992 & 0.008 && 0.000 & 0.000 & 0.974 & 0.026\\ 
(100,75) && 0.000 & 0.000 & 0.996 & 0.004 && 0.000 & 0.000 & 0.990 & 0.010 && 0.000 & 0.000 & 0.982 & 0.018\\ 
(100,100) && 0.000 & 0.000 & 0.994 & 0.006 && 0.000 & 0.000 & 0.986 & 0.014 && 0.000 & 0.000 & 0.976 & 0.024\\ 
(250,50) && 0.000 & 0.000 & 0.986 & 0.014 && 0.000 & 0.000 & 0.980 & 0.020 && 0.000 & 0.000 & 0.940 & 0.060\\ 
(250,75) && 0.000 & 0.000 & 0.984 & 0.016 && 0.000 & 0.000 & 0.972 & 0.028 && 0.000 & 0.000 & 0.952 & 0.048\\ 
(250,100) && 0.000 & 0.000 & 0.998 & 0.002 && 0.000 & 0.000 & 0.992 & 0.008 && 0.000 & 0.000 & 0.974 & 0.026\\ 
(500,50) && 0.000 & 0.000 & 0.984 & 0.016 && 0.000 & 0.000 & 0.960 & 0.040 && 0.000 & 0.000 & 0.890 & 0.110\\ 
(500,75) && 0.000 & 0.000 & 0.992 & 0.008 && 0.000 & 0.000 & 0.972 & 0.028 && 0.000 & 0.000 & 0.940 & 0.060\\ 
(500,100) && 0.000 & 0.000 & 0.994 & 0.006 && 0.000 & 0.000 & 0.962 & 0.038 && 0.000 & 0.000 & 0.920 & 0.080\\

\bottomrule
\end{tabularx}

\label{tab:Sensitivity_Step5_IC_DGP2T_AP}
\end{table}

    
\begin{table}[!htb]
\centering
\caption{Sensitivity analysis for $\alpha^{0}-u$ mixture structure for DGP3M}

\hspace*{-1.3cm}
\begin{tabularx}{1.15\textwidth}{X c XXX  cXXX  cXXX   }
\toprule
&& \multicolumn{3}{c}{$\tilde{c}_\lambda=3/2$} && \multicolumn{3}{c}{$\tilde{c}_\lambda=1$ (bench.)} && \multicolumn{3}{c}{$\tilde{c}_\lambda=3/4$} \\
\cmidrule{3-5} \cmidrule{7-9} \cmidrule{11-13} 
$(N,T)$ && $\mathcal{K}=1$ & $\mathcal{K}=2$ & $\mathcal{K}=3$ && $\mathcal{K}=1$ & $\mathcal{K}=2$ & $\mathcal{K}=3$ && $\mathcal{K}=1$ & $\mathcal{K}=2$ & $\mathcal{K}=3$   \\
\midrule

(100,50) && 0.346 & 0.640 & 0.014 && 0.166 & 0.764 & 0.070 && 0.100 & 0.782 & 0.118\\ 
(100,75) && 0.148 & 0.842 & 0.010 && 0.054 & 0.892 & 0.054 && 0.024 & 0.888 & 0.088\\ 
(100,100) && 0.070 & 0.922 & 0.008 && 0.010 & 0.964 & 0.026 && 0.002 & 0.924 & 0.074\\ 
(250,50) && 0.152 & 0.772 & 0.076 && 0.052 & 0.708 & 0.240 && 0.034 & 0.584 & 0.382\\ 
(250,75) && 0.008 & 0.954 & 0.038 && 0.000 & 0.880 & 0.120 && 0.000 & 0.782 & 0.218\\ 
(250,100) && 0.010 & 0.970 & 0.020 && 0.000 & 0.910 & 0.090 && 0.000 & 0.834 & 0.166\\ 
(500,50) && 0.050 & 0.706 & 0.244 && 0.030 & 0.536 & 0.434 && 0.024 & 0.358 & 0.618\\ 
(500,75) && 0.002 & 0.898 & 0.100 && 0.000 & 0.710 & 0.290 && 0.000 & 0.576 & 0.424\\ 
(500,100) && 0.000 & 0.942 & 0.058 && 0.000 & 0.812 & 0.188 && 0.000 & 0.660 & 0.340\\ 

\bottomrule
\end{tabularx}

\label{tab:Sensitivity_Step5_IC_DGP3M_AP}
\end{table}

    
\begin{table}[!htb]
\centering
\caption{Sensitivity analysis for $\alpha^{0}-u$ mixture structure for DGP3T}

\hspace*{-1.3cm}
\begin{tabularx}{1.15\textwidth}{X c XXXX  cXXXX  cXXXX   }
\toprule
&& \multicolumn{4}{c}{$\tilde{c}_\lambda=3/2$} && \multicolumn{4}{c}{$\tilde{c}_\lambda=1$ (bench.)} && \multicolumn{4}{c}{$\tilde{c}_\lambda=3/4$} \\
\cmidrule{3-6} \cmidrule{8-11} \cmidrule{13-16} 
$(N,T)$ && $\mathcal{K}=1$ & $\mathcal{K}=2$ & $\mathcal{K}=3$ & $\mathcal{K}=4$ && $\mathcal{K}=1$ & $\mathcal{K}=2$ & $\mathcal{K}=3$ & $\mathcal{K}=4$ && $\mathcal{K}=1$ & $\mathcal{K}=2$ & $\mathcal{K}=3$ & $\mathcal{K}=4$  \\
\midrule

(100,50) && 0.000 & 0.000 & 0.998 & 0.002 && 0.000 & 0.000 & 0.992 & 0.008 && 0.000 & 0.000 & 0.974 & 0.026\\ 
(100,75) && 0.000 & 0.000 & 0.996 & 0.004 && 0.000 & 0.000 & 0.990 & 0.010 && 0.000 & 0.000 & 0.982 & 0.018\\ 
(100,100) && 0.000 & 0.000 & 0.994 & 0.006 && 0.000 & 0.000 & 0.986 & 0.014 && 0.000 & 0.000 & 0.976 & 0.024\\ 
(250,50) && 0.000 & 0.000 & 0.986 & 0.014 && 0.000 & 0.000 & 0.980 & 0.020 && 0.000 & 0.000 & 0.940 & 0.060\\ 
(250,75) && 0.000 & 0.000 & 0.984 & 0.016 && 0.000 & 0.000 & 0.972 & 0.028 && 0.000 & 0.000 & 0.952 & 0.048\\ 
(250,100) && 0.000 & 0.000 & 0.998 & 0.002 && 0.000 & 0.000 & 0.992 & 0.008 && 0.000 & 0.000 & 0.974 & 0.026\\ 
(500,50) && 0.000 & 0.000 & 0.984 & 0.016 && 0.000 & 0.000 & 0.960 & 0.040 && 0.000 & 0.000 & 0.890 & 0.110\\ 
(500,75) && 0.000 & 0.000 & 0.992 & 0.008 && 0.000 & 0.000 & 0.972 & 0.028 && 0.000 & 0.000 & 0.940 & 0.060\\ 
(500,100) && 0.000 & 0.000 & 0.994 & 0.006 && 0.000 & 0.000 & 0.962 & 0.038 && 0.000 & 0.000 & 0.920 & 0.080\\

\bottomrule
\end{tabularx}

\label{tab:Sensitivity_Step5_IC_DGP3T_AP}
\end{table}

\clearpage

\subsection{Interpreting the Latent Groups}
\label{app:read-groups-lean}

The classification is fully data–driven, but it is not completely a  black box. 

\paragraph{What drives the classification.} As seen from Figure \ref{fig:Grouped_Frontiers_pihat}, the blue dots and red dots are well separated by a horizontal line in Panel (a)  and by a vertical line in Panel (e), but are evenly mixed otherwise. Consequently, the classification is mostly driven by coefficient $\hat{\pi}_{i1}$ before $B_{1}(\tau_{t})$ (time trend term in the intercept) and the coefficient $\hat{\pi}_{i8}$ before $y_{it2}$ (level term before non-consumer loans).

\paragraph{What differs technologically.}
Figures~\ref{fig:Grouped_Frontiers_Application_SE} display the group–specific, time–varying
coefficients $\{\widehat\beta_{k,\ell}(\tau)\}$ and intercepts $\widehat\alpha_k(\tau)$. Two facts are visible:
(i) the levels/curvature of input elasticities differ across groups, and (ii) their time profiles
(drifts) are not parallel. This is exactly the dimension our HAC step is designed to capture:
firms in the same group share the same coefficient paths; firms in different groups load inputs
and time differently (technology regimes).

We summarize certain statistics for each group below. These numbers indicate correlations only and do not suggest any causal relationship.  We find that banks in group 1 are, on average, larger, more efficient (closer to their frontier), and exhibit less variation in inefficiency than banks in group 2.

\paragraph{What differs in performance.}
To separate technology from performance,
we summarize the mean and standard deviation of $\widehat{\textrm{E}}\left(\alpha_i^{0}+u_{i}|\varepsilon_{i1},...,\varepsilon_{iT}\right)$ using (\ref{eq:inefficiency_post}),  within each group.
Group 1 has the mean (standard deviation) at $0.0923$ $(2.17\times10^{-4})$ with 113 observations, while Group 2 has its mean (standard deviation) at $0.1130$ $(0.0558)$ with 353 observations. The group means are statistically different at 5\% level of significance using 
\[\frac{\hat{\mu}_{1}-\hat{\mu}_{2}}{\sqrt{s_{1}^{2}/n_{1}+s_{2}^{2}/n_{2}}} = -6.6179.\]
From the mean values, banks in Group 2 appear less efficient than those in Group 1. The standard deviations also indicate that inefficiency varies more within Group 2 than within Group 1.

\paragraph{Who is more often in each group.}
We compare bank sizes across the two groups and find the following. For group 1, the mean (standard deviation) of $\log(\text{size})$ is $13.4866$ $(1.8689)$ based on 113 observations. For group 2, the corresponding mean (standard deviation) is $12.4813$ $(1.5012)$ based on 353 observations. The difference in group means is statistically significant at the 5\% level, with a t-statistic of $6.9721$. Thus, on average, banks in group 1 are larger.

\subsection{Additional Figure for the Empirical Application}\label{APP:application_fig}

\begin{figure}[H]
\caption{Inefficiency Estimates of 60 Banks in Descending Order}
\centering
\includegraphics[scale=0.65] {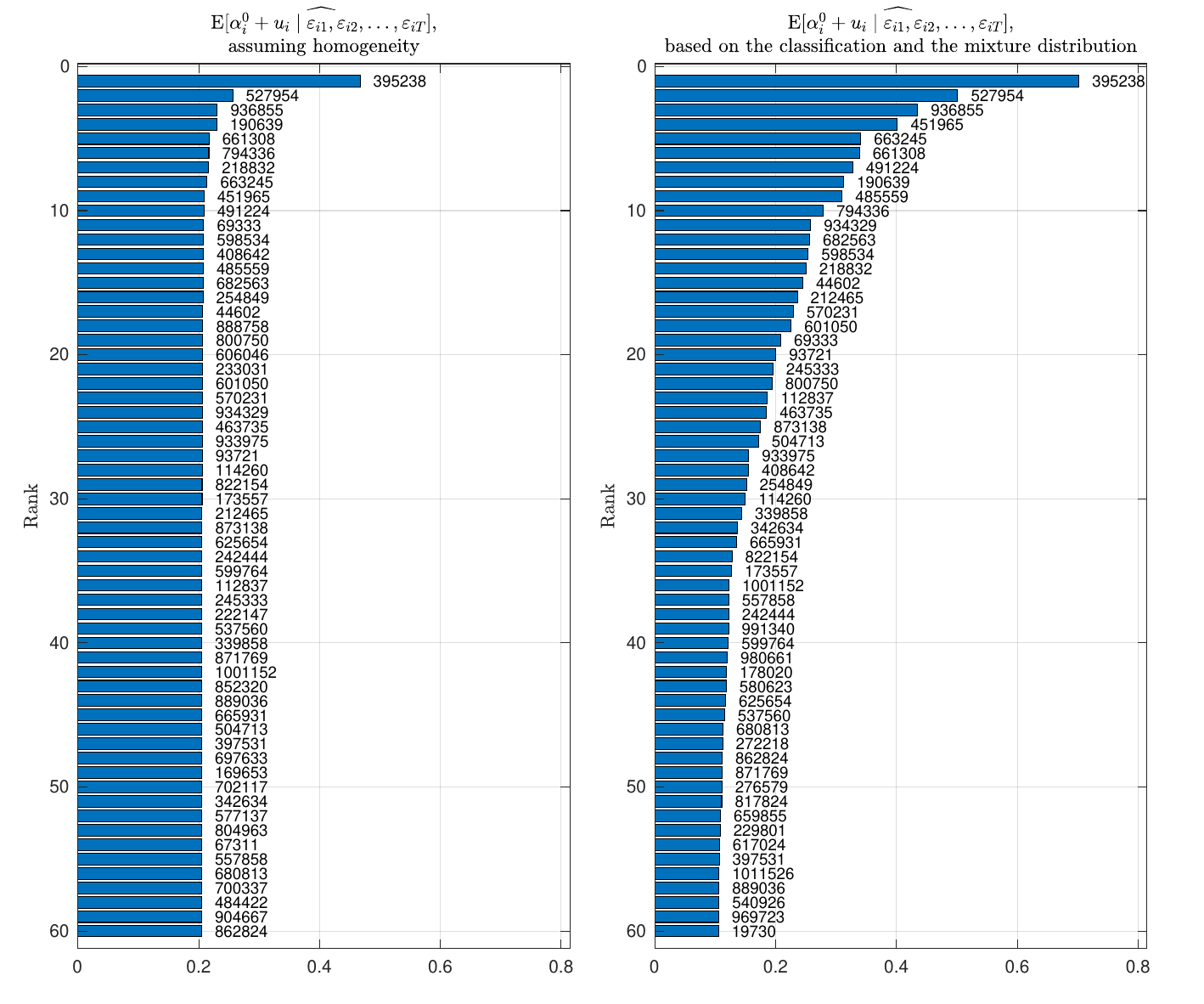}
\label{fig:Inefficiency_Ranking_Ranking}
\begin{tablenotes}
      \footnotesize
      \item 
      \emph{Note}: Left panel shows the inefficiency estimates in the homogeneous case where all banks are treated as one group with uniquely distributed error term, while the right panel shows the inefficiency estimates in the heterogeneous case where banks are classified into groups and inefficiency estimates to possess a mixture distribution structure. The numbers to the right of each bar corresponds to their respective bank ID, ``IDRSSD'' from \href{https://cdr.ffiec.gov/public/PWS/DownloadBulkData.aspx}{https://cdr.ffiec.gov/public/PWS/DownloadBulkData.aspx}.
\end{tablenotes}
\end{figure}
\newpage

\section{On \texorpdfstring{$\mathbb{I}$}{I} in Theorem \ref{TH:post-estimation} }

\label{App:derivatives}


Without loss of generality, we show the properties of $\mathbb{I}$
in the case when $\mathcal{K}^{*}=2$ to simplify notation. The cases
when $\mathcal{K}^{*}>2$ can be similarly shown.

\subsection{First and Second Order Derivatives}

\label{App:derivativesforms}

We use the notation in Appendix \ref{APP:likeli}. We assume $\alpha_{i}$,
$\beta_{i}$, and $\sigma_{vi}^{2}$ are known because they are not
the parameters of interest in this appendix.

We first discuss the case when the distribution does not have a latent
structure. The unknown parameters are $\left(\alpha^{0},\sigma_{u}^{2}\right)$
in this case. Note 
\[
\varepsilon_{it}=y_{it}-\alpha^{0}-\alpha\left(\tau_{t}\right)-x_{it}'\beta\left(\tau_{t}\right)=v_{it}-u_{i},
\]
\begin{align*}
f\left(\left.y_{i}\right|x_{i};\alpha^{0},\sigma_{u}^{2}\right) & =\frac{2}{\sigma_{vi}^{T-1}\sqrt{\sigma_{vi}^{2}+T\sigma_{u}^{2}}}\left[1-\Phi\left(\frac{\sum_{t=1}^{T}\varepsilon_{it}}{\sigma_{vi}\sqrt{\sigma_{vi}^{2}/\sigma_{u}^{2}+T}}\right)\right]\left[\frac{1}{\left(2\pi\right)^{T/2}}\exp\left(-\frac{\sum_{t=1}^{T}\varepsilon_{it}^{2}}{2\sigma_{vi}^{2}}\right)\right]\\
 & \times\exp\left(\frac{1}{2}\frac{\left(\sum_{t=1}^{T}\varepsilon_{it}\right)^{2}}{\sigma_{vi}^{2}\left(\sigma_{vi}^{2}/\sigma_{u}^{2}+T\right)}\right),
\end{align*}
and

\begin{align*}
\log f\left(\left.y_{i}\right|x_{i};\alpha^{0},\sigma_{u}^{2}\right)= & C-\frac{(T-1)}{2}\log\sigma_{vi}^{2}-\frac{1}{2}\text{\ensuremath{\log}}\left(\sigma_{vi}^{2}+T\sigma_{u}^{2}\right)\\
 & +\log\left[1-\Phi\left(-\frac{\mu_{i\ast}}{\sigma_{i\ast}}\right)\right]+\frac{1}{2}\left(\frac{\mu_{i\ast}}{\sigma_{i\ast}}\right)^{2}-\frac{\sum_{t=1}^{T}\varepsilon_{it}^{2}}{2\sigma_{vi}^{2}},
\end{align*}
with $\sigma_{i}^{2}=\sigma_{vi}^{2}+T\sigma_{u}^{2},$ $\rho_{i}=\sigma_{u}/\sigma_{vi},$
$\mu_{i\ast}=-\sigma_{u}^{2}\sum_{t=1}^{T}\varepsilon_{it}/\sigma_{i}^{2},$
$\sigma_{i\ast}^{2}=\sigma_{u}^{2}\sigma_{vi}^{2}/\sigma_{i}^{2}$,
and a $C$ that does not depend on parameters.

By some straightforward calculations,

\begin{align}
\frac{\partial}{\partial\sigma_{u}^{2}}\log f\left(\left.y_{i}\right|x_{i};\alpha^{0},\sigma_{u}^{2}\right) & =-\frac{1}{2}\frac{T}{\sigma_{vi}^{2}+T\sigma_{u}^{2}}+\frac{\phi\left(-\frac{\mu_{i\ast}}{\sigma_{i\ast}}\right)}{1-\Phi\left(-\frac{\mu_{i\ast}}{\sigma_{i\ast}}\right)}\frac{-\sigma_{vi}}{2\sigma_{u}^{4}\left(\sigma_{vi}^{2}/\sigma_{u}^{2}+T\right)^{3/2}}\sum_{t=1}^{T}\varepsilon_{it}\nonumber \\
 & +\frac{1}{2\sigma_{u}^{4}\left(\sigma_{vi}^{2}/\sigma_{u}^{2}+T\right)^{2}}\left(\sum_{t=1}^{T}\varepsilon_{it}\right)^{2},\label{eq:dlogdsigma}
\end{align}
\begin{align}
\frac{\partial}{\partial\alpha^{0}}\log f\left(\left.y_{i}\right|x_{i};\alpha^{0},\sigma_{u}^{2}\right) & =\frac{\phi\left(-\frac{\mu_{i\ast}}{\sigma_{i\ast}}\right)}{1-\Phi\left(-\frac{\mu_{i\ast}}{\sigma_{i\ast}}\right)}\frac{T}{\sigma_{vi}\sqrt{\sigma_{vi}^{2}/\sigma_{u}^{2}+T}}+\frac{1}{T^{-1}\sigma_{vi}^{2}+\sigma_{u}^{2}}\frac{1}{T}\left(\sum_{t=1}^{T}\varepsilon_{it}\right),\label{eq:dlogdc}
\end{align}
\begin{align*}
\frac{\partial^{2}}{\partial\sigma_{u}^{2}\partial\alpha^{0}}\log f\left(\left.y_{i}\right|x_{i};\alpha^{0},\sigma_{u}^{2}\right) & =g\left(\frac{\mu_{i\ast}}{\sigma_{i\ast}}\right)\frac{-T}{2\sigma_{u}^{4}\left(\sigma_{vi}^{2}/\sigma_{u}^{2}+T\right)^{2}}\sum_{t=1}^{T}\varepsilon_{it}+\frac{\phi\left(-\frac{\mu_{i\ast}}{\sigma_{i\ast}}\right)}{1-\Phi\left(-\frac{\mu_{i\ast}}{\sigma_{i\ast}}\right)}\frac{T\sigma_{vi}}{2\sigma_{u}^{4}\left(\sigma_{vi}^{2}/\sigma_{u}^{2}+T\right)^{3/2}}\\
 & -\frac{T}{\sigma_{u}^{4}\left(\sigma_{vi}^{2}/\sigma_{u}^{2}+T\right)^{2}}\left(\sum_{t=1}^{T}\varepsilon_{it}\right),
\end{align*}
\begin{align*}
\frac{\partial^{2}}{\partial(\alpha^{0})^{2}}\log f\left(\left.y_{i}\right|x_{i};\alpha^{0},\sigma_{u}^{2}\right) & =g\left(\frac{\mu_{i\ast}}{\sigma_{i\ast}}\right)\frac{T^{2}}{\sigma_{vi}^{2}\left(\sigma_{vi}^{2}/\sigma_{u}^{2}+T\right)}-\frac{1}{T^{-1}\sigma_{vi}^{2}+\sigma_{u}^{2}},
\end{align*}
and 
\begin{align*}
\frac{\partial^{2}}{\partial\left(\sigma_{u}^{2}\right)^{2}}\log f\left(\left.y_{i}\right|x_{i};\alpha^{0},\sigma_{u}^{2}\right) & =\frac{1}{2}\frac{T^{2}}{\left(\sigma_{vi}^{2}+T\sigma_{u}^{2}\right)^{2}}+g\left(\frac{\mu_{i\ast}}{\sigma_{i\ast}}\right)\frac{\sigma_{vi}^{2}}{4\sigma_{u}^{8}\left(\sigma_{vi}^{2}/\sigma_{u}^{2}+T\right)^{3}}\left(\sum_{t=1}^{T}\varepsilon_{it}\right)^{2}\\
 & +\frac{\phi\left(-\frac{\mu_{i\ast}}{\sigma_{i\ast}}\right)}{1-\Phi\left(-\frac{\mu_{i\ast}}{\sigma_{i\ast}}\right)}\frac{\left(\sigma_{vi}^{3}\sigma_{u}^{-4/3}+4T\sigma_{vi}\sigma_{u}^{2/3}\right)}{4\left(\sigma_{vi}^{2}\sigma_{u}^{2/3}+T\sigma_{u}^{8/3}\right)^{5/2}}\sum_{t=1}^{T}\varepsilon_{it}-\frac{T}{\left(\sigma_{vi}^{2}+T\sigma_{u}^{2}\right)^{3}}\left(\sum_{t=1}^{T}\varepsilon_{it}\right)^{2},
\end{align*}
where 
\[
g\left(\frac{\mu_{i\ast}}{\sigma_{i\ast}}\right)=\left.\frac{d}{ds}\frac{\phi\left(-s\right)}{1-\Phi\left(-s\right)}\right|_{s=\frac{\mu_{i\ast}}{\sigma_{i\ast}}}=\left.\frac{-\phi\left(-s\right)-\phi'\left(-s\right)\left[1-\phi\left(-s\right)\right]}{\left[1-\Phi\left(-s\right)\right]^{2}}\right|_{s=\frac{\mu_{i\ast}}{\sigma_{i\ast}}}.
\]

In the case when the distribution possesses a latent structure, the
log likelihood function becomes 
\[
\log\left[\tau f\left(y_{i}\left\vert x_{i};\alpha_{(1)}^{0},\sigma_{u(1)}^{2}\right.\right)+\left(1-\tau\right)f\left(\left.y_{i}\right|x_{i};\alpha_{(2)}^{0},\sigma_{u(2)}^{2}\right)\right]\equiv\log\left[\tau f_{1i}+\left(1-\tau\right)f_{2i}\right].
\]
Then, 
\[
\frac{\partial}{\partial\tau}\log\left[\tau f_{1i}+\left(1-\tau\right)f_{2i}\right]=\frac{f_{1i}-f_{2i}}{\tau f_{1i}+\left(1-\tau\right)f_{2i}},\text{ and}
\]
\[
\frac{\partial^{2}}{\partial\tau^{2}}\log\left[\tau f_{1i}+\left(1-\tau\right)f_{2i}\right]=-\frac{\left(f_{1i}-f_{2i}\right)^{2}}{\left[\tau f_{1i}+\left(1-\tau\right)f_{2i}\right]^{2}}.
\]
The derivatives with respect to other arguments $\alpha^{0}$ and
$\sigma_{u}^{2}$ can be derived using the chain rule. For example,
\[
\frac{\partial}{\partial\alpha_{(1)}^{0}}\log\left[\tau f_{1i}+\left(1-\tau\right)f_{2i}\right]=\frac{\tau\frac{\partial}{\partial\alpha_{(1)}^{0}}f_{1i}}{\tau f_{1i}+\left(1-\tau\right)f_{2i}}=\frac{\tau f_{1i}}{\tau f_{1i}+\left(1-\tau\right)f_{2i}}\cdot\frac{\partial}{\partial\alpha_{(1)}^{0}}\log f_{1i},\text{ and}
\]
\[
\frac{\partial}{\partial\sigma_{u(1)}^{2}}\log\left[\tau f_{1i}+\left(1-\tau\right)f_{2i}\right]=\frac{\tau\frac{\partial}{\partial\sigma_{u(1)}^{2}}f_{1i}}{\tau f_{1i}+\left(1-\tau\right)f_{2i}}=\frac{\tau f_{1i}}{\tau f_{1i}+\left(1-\tau\right)f_{2i}}\cdot\frac{\partial}{\partial\sigma_{u(1)}^{2}}\log f_{1i}.
\]

\subsection{Properties of \texorpdfstring{$\mathbb{I}$}{I}}

\label{app:IOmega}

We verify that $\mathbb{I}$ behaves like a regular positive definite
matrix in this section. The intuition of this result is as follows.
Note $\tilde{f}_{i}\left(\varrho_{0}\right)$ is the density function
for $\varepsilon_{it}=v_{it}-u_{i}$, and the parameters of interest
are on the distribution of $u_{i}$ only. In other words, $v_{it}$
are nuisance for the estimation. Mathematically, one can see that
the dominant components in $\frac{\partial}{\partial\varrho}\log\tilde{f}_{i}\left(\varrho\right)$
are functions of $u_{i}$ and do not grow as $T\rightarrow\infty$.
Thus, the convergence rate of $\hat{\varrho}$ is $\sqrt{N}$ based
on this intuition.

Without loss of generality, we assume that there is no group structure
for the frontiers and $\sigma_{vi}^{2}$, that is, $K^{*}=1$. We
show the result by assuming $\alpha^{0}=\alpha_{(1)}^{0}=\alpha_{(2)}^{0},$
and $\sigma_{u(1)}^{2}\neq\sigma_{u(2)}^{2}$. The case with $\alpha_{(1)}^{0}\neq\alpha_{(2)}^{0}$
can be similarly handled but with more tedious discussions. We show
the result by using the following identity 
\begin{align}
\mathbb{I} & =\left.-\text{E}\left[\frac{\partial^{2}}{\partial\varrho\partial\varrho'}\log\tilde{f}_{i}\left(\varrho\right)\right]\right|_{\varrho=\varrho^{0}}\nonumber \\
 & =\left.\text{E}\left[\frac{\partial}{\partial\varrho}\log\tilde{f}_{i}\left(\varrho\right)\cdot\frac{\partial}{\partial\varrho'}\log\tilde{f}_{i}\left(\varrho\right)\right]\right|_{\varrho=\varrho^{0}}.\label{eq:Iidentity}
\end{align}

From the forms of $\frac{\partial}{\partial\sigma_{uj}^{2}}\log\left[\tau f_{1}+\left(1-\tau\right)f_{2}\right],$
$\frac{\partial}{\partial\alpha_{j}^{0}}\log\left[\tau f_{1}+\left(1-\tau\right)f_{2}\right],$
$j=1,2$, and 
\[
\frac{\partial}{\partial\tau}\log\left[\tau f_{1}+\left(1-\tau\right)f_{2}\right]
\]
in Appendix \ref{App:derivativesforms}, they are clearly not linear
dependent. If $T$ is fixed, $\mathbb{I}$ is just the information
matrix for a regular likelihood function, and it is positive definite.
Our analysis is complicated by the fact that $T\rightarrow\infty$.
The diagonal of $\mathbb{I}$ might explode as $T\rightarrow\infty$.
Thus, it suffices to show that the diagonals of $\mathbb{I}$ behave
like regular positive constants, and we show that in the below.

The following is useful for the analysis. At the true values of parameters,
$\varepsilon_{it}=v_{it}-u_{i}.$ So 
\begin{align*}
T^{-1}\sum_{t=1}^{T}\varepsilon_{it} & =-u_{i}+T^{-1}\sum_{t=1}^{T}v_{it}\\
 & =-u_{i}+O_{P}\left(T^{-1/2}\right).
\end{align*}
Since $-u_{i}$ is negative, $T^{-1}\sum_{t=1}^{T}\varepsilon_{it}$
is negative with very high probability. Another implication is $\sum_{t=1}^{T}\varepsilon_{it}\propto_{P}-T+O_{P}\left(T^{1/2}\right).$
Using the definition of $\frac{\mu_{i\ast}}{\sigma_{i\ast}},$ 
\[
\frac{\mu_{i\ast}}{\sigma_{i\ast}}=\frac{-1}{\sigma_{vi}\sqrt{\sigma_{vi}^{2}/\sigma_{u}^{2}+T}}\sum_{t=1}^{T}\varepsilon_{it}\propto_{P}\sqrt{T}+O_{P}\left(1\right).
\]
The above implies $\Phi\left(-\frac{\mu_{i\ast}}{\sigma_{i\ast}}\right)=o_{P}\left(1\right),$
\[
\frac{\phi\left(-\frac{\mu_{i\ast}}{\sigma_{i\ast}}\right)}{1-\Phi\left(-\frac{\mu_{i\ast}}{\sigma_{i\ast}}\right)}\propto_{P}\exp\left(-T\right),\quad\frac{\phi\left(-\frac{\mu_{i\ast}}{\sigma_{i\ast}}\right)}{1-\Phi\left(-\frac{\mu_{i\ast}}{\sigma_{i\ast}}\right)}\sqrt{T}=o_{P}\left(1\right),
\]
and 
\[
g\left(\frac{\mu_{i\ast}}{\sigma_{i\ast}}\right)\propto_{P}\exp\left(-T\right).
\]

Applying the above for (\ref{eq:dlogdc}), 
\[
\frac{\partial}{\partial\alpha^{0}}\log f\left(\left.y_{i}\right|x_{i};\alpha^{0},\sigma_{u}^{2}\right)=-u_{i}+E\left(u_{i}\right)+o_{P}\left(1\right).
\]
Similarly (\ref{eq:dlogdsigma}) implies 
\[
\frac{\partial}{\partial\sigma_{u}^{2}}\log f\left(\left.y_{i}\right|x_{i};\alpha^{0},\sigma_{u}^{2}\right)=\frac{u_{i}^{2}}{2\sigma_{u}^{4}}-\frac{1}{2\sigma_{u}^{2}}+o_{P}\left(1\right).
\]
Some tedious yet straightforward calculations can yield 
\[
f\left(y_{i}\left\vert x_{i};\alpha^{0},\sigma_{u(1)}^{2}\right.\right)\propto_{P}f\left(\left.y_{i}\right|x_{i};\alpha^{0},\sigma_{u(2)}^{2}\right),
\]
and 
\begin{align*}
\frac{\partial}{\partial\tau}\log\left[\tau f_{1i}+\left(1-\tau\right)f_{2i}\right] & =\frac{f_{1i}-f_{2i}}{\tau f_{1i}+\left(1-\tau\right)f_{2i}}.\\
 & =\frac{1-\sqrt{\left.\sigma_{u(1)}^{2}\right/\sigma_{u(2)}^{2}}\exp\left(\frac{1}{2}\left(\frac{1}{\sigma_{u(1)}^{2}}-\frac{1}{\sigma_{u(2)}^{2}}\right)u_{i}^{2}\right)}{\tau+\left(1-\tau\right)\sqrt{\left.\sigma_{u(1)}^{2}\right/\sigma_{2u}^{2}}\exp\left(\frac{1}{2}\left(\frac{1}{\sigma_{u(1)}^{2}}-\frac{1}{\sigma_{u(2)}^{2}}\right)u_{i}^{2}\right)}\left(1+o_{P}\left(1\right)\right).
\end{align*}
The variance of the leading terms in $\frac{\partial}{\partial\alpha^{0}}\log f$,
$\frac{\partial}{\partial\sigma_{u}^{2}}\log f,$ and $\frac{\partial}{\partial\tau}\log\left[\tau f_{1i}+\left(1-\tau\right)f_{2i}\right]$
are functions of $u_{i}$, and they are bounded and bounded away from
zero.

Finally, since $f_{1i}\propto_{P}f_{2i},$ $\frac{\partial}{\partial\alpha_{j}^{0}}\log\left[\tau f_{1i}+\left(1-\tau\right)f_{2i}\right],$
$j=1,2,$ enjoys similar properties as $\frac{\partial}{\partial\alpha^{0}}\log f$.
This is also the case for $\frac{\partial}{\partial\sigma_{uj}^{2}}\log\left[\tau f_{1i}+\left(1-\tau\right)f_{2i}\right],$$j=1,2.$

From the form of the derivatives, clearly, $\frac{\partial}{\partial\varrho_{j}}\log\left[\tau f_{1i}+\left(1-\tau\right)f_{2i}\right]$,
$j=1,2,...,5,$ are not linearly dependent. Together with the fact
that the variance of them are bounded and bounded away from zero,
then for any $5\times1$ vector $a,$ 
\[
a'\left.\text{E}\left[\frac{\partial}{\partial\varrho}\log\tilde{f}_{i}\left(\varrho\right)\cdot\frac{\partial}{\partial\varrho'}\log\tilde{f}_{i}\left(\varrho\right)\right]\right|_{\varrho=\varrho^{0}}a\propto\left\Vert a\right\Vert ^{2}.
\]
This implies that $\mathbb{I}$ must be positive definite with bounded
eigenvalues, using the identity in (\ref{eq:Iidentity}).


\section{Technical Lemmas}

\label{APP:lemmaProofs}


We collect technical lemmas and their proofs in this section. For
an easier reference, we present the bodies of those lemmas first,
followed by proofs. We note the lemmas of probability bounds below
are related to those developed in \citet{Ataketal}.

\begin{remark}[Some inequalities]\label{RE:ineq}\emph{We may apply
the following inequalities in the proofs directly without referring
them back to here. First } 
\[
\Pr\left(X_{1}+X_{2}\geq C\right)\leq\Pr\left(X_{1}\geq\pi C\right)+\Pr\left(X_{2}\geq\left(1-\pi\right)C\right)
\]
\emph{for any constant }$\pi.$\emph{ That is because }$\left\{ X_{1}+X_{2}\geq C\right\} \subseteq\left\{ X_{1}\geq\pi C\right\} \cup\left\{ X_{1}\geq\left(1-\pi\right)C\right\} .$\emph{
Similarly, we have } 
\[
\Pr\left(\sum_{i=1}^{n}X_{i}\geq C\right)\leq\sum_{i=1}^{n}\Pr\left(X_{i}\geq Cn^{-1}\right)\ \text{\emph{and} }\Pr\left(\max_{1\leq i\leq n}X_{i}\geq C\right)\leq\sum_{i=1}^{n}\Pr\left(X_{i}\geq C\right).
\]
\emph{For any positive random variables }$X_{1}$\emph{\ and }$X_{2}$\emph{
and positive constants }$C_{1}$\emph{ and }$C_{2},$ 
\[
\Pr\left(X_{1}\cdot X_{2}\geq C_{1}\right)\leq\Pr\left(X_{1}\geq C_{1}/C_{2}\right)+\Pr\left(X_{2}\geq C_{2}\right),
\]
\emph{due to the fact that }$\left\{ X_{1}\cdot X_{2}\geq C_{1}\right\} \subseteq\left\{ X_{1}\geq C_{1}/C_{2}\right\} \cup\left\{ X_{2}\geq C_{2}\right\} .$\emph{
} \end{remark}

\begin{lemma}\label{LE:bound}Suppose $e_{it}$ satisfies the mixing
condition across $t$ in Assumption \ref{A:mixing} (ii) and is identically
distributed across $i$. In addition $\max_{t}\text{\emph{E}}\left(\left|e_{it}\right|^{C_{q}}\right)<\infty,C_{q}>2,$
$T=N^{C}$ for some $C>0$, $m\rightarrow\infty$ as $N\rightarrow\infty$,
then

(i) for any $\epsilon>0$ and any positive $1\leq\upsilon_{NT}\ll T^{1/2}\left(\log N\right)^{-2}$
\[
\Pr\left(\max_{i=1,...,N}\left|\frac{1}{T}\sum_{t=1}^{T}\left[e_{it}-\text{\emph{E}}\left(e_{it}\right)\right]\right|>\frac{\epsilon}{\upsilon_{NT}}\right)\lesssim\frac{NT\upsilon_{NT}^{C_{q}}\left(\log N\right)^{4C_{q}}}{T^{C_{q}}};
\]

(ii) there exists a positive $M,$ such that 
\[
\Pr\left(\max_{i=1,...,N}\max_{l=0,1,...,p}\left|\frac{1}{T}\sum_{t=1}^{T}e_{it}b_{il}\left(\tau_{t}\right)\right|>Mm^{-\kappa}\right)\lesssim\frac{NT\left(\log N\right)^{4C_{q}}}{T^{C_{q}}},
\]
where $b_{il}\left(\tau_{t}\right)$ is defined in (\ref{eq:bil}).

\end{lemma}

\begin{lemma}\label{LE:splines}Suppose that Assumption \ref{A:rank}
holds and $m/T\rightarrow0$. For any $\tilde{\pi}\in\mathbb{R}^{m\left(p+1\right)},$
\[
\frac{\underline{C}_{xx}}{2}\left\Vert \tilde{\pi}\right\Vert \leq\min_{1\leq i\leq N}\frac{1}{T}\sum_{t=1}^{T}\text{\emph{E}}\left(\tilde{\pi}'\tilde{z}_{it}\tilde{z}_{it}'\tilde{\pi}\right)\leq\max_{1\leq i\leq N}\frac{1}{T}\sum_{t=1}^{T}\text{\emph{E}}\left(\pi'\tilde{z}_{it}\tilde{z}_{it}'\pi\right)\leq2\bar{C}_{xx}\left\Vert \tilde{\pi}\right\Vert ,
\]
after some large $T$.

\end{lemma}

\begin{lemma}\label{LE:splines-sample-1}Suppose that Assumptions
\ref{A:mixing}, \ref{A:moment}, \ref{A:rank} and \ref{A:tuningPara}
(i) hold. Then

(i) for any $\epsilon>0,$ 
\[
\Pr\left(\max_{1\leq i\leq N}\sup_{\tilde{\pi}\in\mathbb{R}^{m\left(p+1\right)},\left\Vert \tilde{\pi}\right\Vert =1}\left|\frac{1}{T}\sum_{t=1}^{T}\tilde{\pi}'\tilde{z}_{it}\tilde{z}_{it}'\tilde{\pi}-\frac{1}{T}\sum_{t=1}^{T}\text{\emph{E}}\left(\tilde{\pi}'\tilde{z}_{it}\tilde{z}_{it}'\tilde{\pi}\right)\right|>\epsilon\right)\apprle\frac{Nm^{q/2+2}\left(\log N\right)^{2q}}{T^{q/2-1}};
\]

(ii) there exist some positive and finite $\underline{C}_{zz}$ and
$\bar{C}_{zz}$ such that 
\begin{align*}
 & \Pr\left(\underline{C}_{zz}\leq\min_{1\leq i\leq N}\mu_{\min}\left(\frac{1}{T}\sum_{t=1}^{T}\tilde{z}_{it}\tilde{z}_{it}'\right)\leq\max_{1\leq i\leq N}\mu_{\max}\left(\frac{1}{T}\sum_{t=1}^{T}\tilde{z}_{it}\tilde{z}_{it}'\right)\leq\bar{C}_{zz}\right)\\
= & 1-o\left(\frac{Nm^{q/2+2}\left(\log N\right)^{2q}}{T^{q/2-1}}\right)\text{.}
\end{align*}

\end{lemma}

\begin{lemma}\label{LE:bound-ours}Suppose Assumptions \ref{A:mixing},
\ref{A:moment}, \ref{A:coef} and \ref{A:tuningPara} (i) hold. Then

(i) for any $\epsilon>0,$ 
\[
\Pr\left(\max_{i=1,...,N}\left\Vert \frac{1}{T}\sum_{t=1}^{T}\tilde{z}_{it}v_{it}\right\Vert >\epsilon\right)\lesssim\frac{Nm^{q/4}\left(\log N\right)^{2q}}{T^{q/2-1}};
\]

(ii) for the $\xi_{it}$ defined in (\ref{eq:xi_define}), 
\[
\Pr\left(\max_{i=1,...,N}\left\Vert \frac{1}{T}\sum_{t=1}^{T}\tilde{z}_{it}\xi_{it}\right\Vert >\epsilon\right)\lesssim\frac{Nm\left(\log N\right)^{2q}}{T^{q/2-1}}.
\]
\end{lemma}

\begin{lemma}\label{LE:uni_converge}Suppose that Assumptions \ref{A:mixing},
\ref{A:moment}, \ref{A:rank}, \ref{A:coef}, and \ref{A:tuningPara}
hold. Then,

(i) for any small positive $\epsilon,$

\[
\Pr\left(\max_{1\leq i\leq N}\left\Vert \widehat{\tilde{\pi}}_{i}-\tilde{\pi}_{i}^{0}\right\Vert >\epsilon\right)=o\left(1\right);
\]

(ii) for any small positive $\epsilon,$ 
\[
\Pr\left(\max_{1\leq i\leq N}\left\Vert \hat{\sigma}_{vi}^{2}-\sigma_{vi}^{2}\right\Vert >\epsilon\right)=o\left(1\right).
\]

\end{lemma}

\begin{lemma}\label{LE:z_k'z_k}Suppose Assumption \ref{A:rank}
holds and $\underline{m}/T\rightarrow0$. Then after some large $T,$
\[
\frac{\underline{C}_{xx}}{2}\leq\min_{1\leq i\leq N}\mu_{\min}\left\{ \text{\emph{E}}\left(\frac{1}{T}\sum_{t=1}^{T}\underline{\ddot{z}}_{it}\underline{\ddot{z}}_{it}'\right)\right\} \leq\max_{1\leq i\leq N}\mu_{\max}\left\{ \text{\emph{E}}\left(\frac{1}{T}\sum_{t=1}^{T}\underline{\ddot{z}}_{it}\underline{\ddot{z}}_{it}'\right)\right\} \leq2\bar{C}_{xx},
\]
where $\underline{C}_{xx}$ and $\bar{C}_{xx}$ are defined in Lemma
\ref{LE:splines}.

\end{lemma}

\begin{lemma}\label{LE:splines-sample-2}Suppose that Assumptions
\ref{A:mixing}, \ref{A:moment}, \ref{A:rank}, \ref{A:additional-1}
and \ref{A:group} hold. Then for a fixed $k,$ $k=1,2,...,K^{*},$
and any $\epsilon>0,$ 
\[
\Pr\left(\sup_{\pi\in\mathbb{R}^{\underline{m}-1+\underline{m}p},\left\Vert \pi\right\Vert =1}\left|\frac{1}{N_{k}}\sum_{i\in G_{k|K^{*}}}\left[\frac{1}{T}\sum_{t=1}^{T}\pi'\underline{\ddot{z}}_{it}\underline{\ddot{z}}_{it}'\pi-\frac{1}{T}\sum_{t=1}^{T}\text{\emph{E}}\left(\pi'\underline{\ddot{z}}_{it}\underline{\ddot{z}}_{it}'\pi\right)\right]\right|>\epsilon\right)\apprle\frac{\underline{m}^{q/2+2}\left(\log N\right)^{2q}}{\left(NT\right)^{q/2-1}};
\]

\end{lemma}

\begin{lemma}\label{LE:Ak1}Suppose Assumptions \ref{A:mixing},
\ref{A:moment}, \ref{A:rank}, \ref{A:coef}, \ref{A:tuningPara},
\ref{A:additional-1}, \ref{A:group}, and \ref{A:tuning2} hold.
Then, the term 
\[
A_{k1}=\left(\sum_{i\in G_{k|K}}\sum_{t=1}^{T}\underline{\ddot{z}}_{it}\underline{\ddot{z}}_{it}'\right)^{-1}\left(\sum_{i\in G_{k|K}}\sum_{t=1}^{T}\underline{\ddot{z}}_{it}\ddot{\xi}_{it}\right)
\]
satisfies $\left\Vert A_{k1}\right\Vert =O_{P}\left(\underline{m}^{-\kappa}\right).$

\end{lemma}

\begin{lemma}\label{LE:Ak2}Suppose Assumptions \ref{A:mixing},
\ref{A:moment}, \ref{A:rank}, \ref{A:coef}, \ref{A:tuningPara},
\ref{A:additional-1}, \ref{A:group}, \ref{A:stochasticF} , and
\ref{A:tuning2} hold. For 
\[
A_{k2}=\left(\sum_{i\in G_{k|K}}\sum_{t=1}^{T}\underline{\ddot{z}}_{it}\underline{\ddot{z}}_{it}'\right)^{-1}\left(\sum_{i\in G_{k|K}}\sum_{t=1}^{T}\underline{\ddot{z}}_{it}\ddot{v}_{it}\right)
\]
and any $\left(p+1\right)\times1$ vector $a$, it satisfies 
\[
\left.\sqrt{\frac{N_{k}T}{\underline{m}}}a'\mathbb{M_{B}}\left(s\right)A_{k2}\right/\left\{ a'\left[\frac{\sigma_{v\left(k\right)}^{2}}{\underline{m}}\mathbb{M_{B}}\left(s\right)Q_{(k),zz}^{-1}\mathbb{M_{B}}\left(s\right)'\right]a\right\} ^{1/2}\stackrel{d}{\rightarrow}N\left(0,1\right).
\]

\end{lemma}

\begin{lemma} \label{LE:IC1}Suppose Assumptions \ref{A:mixing}
- \ref{A:group} hold. In addition, $\frac{\lambda_{NT}}{\sqrt{NT}}\rightarrow\infty$
and $\frac{\lambda_{NT}}{NT}\rightarrow0$. 
\[
\Pr\left(\text{IC}\left(K^{*},\lambda_{NT}\right)<\text{IC}\left(K,\lambda_{NT}\right),1\leq K\leq K^{*}-1\right)\rightarrow1.
\]

\end{lemma}

\begin{lemma} \label{LE:IC2}Suppose Assumptions \ref{A:mixing}
- \ref{A:group} hold. In addition, $\frac{\lambda_{NT}}{\sqrt{NT}}\rightarrow\infty$
and $\frac{\lambda_{NT}}{NT}\rightarrow0$. 
\[
\Pr\left(\text{IC}\left(K^{*},\lambda_{NT}\right)<\text{IC}\left(K,\lambda_{NT}\right),K^{*}+1\leq K\leq\bar{K}\right)\rightarrow1.
\]

\end{lemma}

\begin{lemma} \label{LE:IC3}Suppose Assumptions \ref{A:mixing}
- \ref{A:tuning2} hold. $\tilde{\lambda}_{NT}$ satisfies $\tilde{\lambda}_{NT}\rightarrow\infty$
and $\frac{\tilde{\lambda}_{NT}}{N}\rightarrow0$.

\noindent
\[
\Pr\left(\hat{\mathcal{K}}\left(\tilde{\lambda}_{NT}\right)=\mathcal{K}^{*}\right)\rightarrow1.
\]

\end{lemma} 
\begin{proof}[Proof of Lemma \ref{LE:bound}] \textbf{(i)} Set
\[
C_{NT}=\upsilon_{NT}^{-1}T\left(\log N\right)^{-4}.
\]
Define
\[
\mathbf{1}_{it}=\mathbf{1}\left(\left|e_{it}\right|<C_{NT}\right),
\qquad
\mathbf{\bar{1}}_{it}=1-\mathbf{1}_{it}.
\]
With this notation, further define
\[
e_{it}^{\left(1\right)}
=
e_{it}\mathbf{1}_{it}
-
\text{\emph{E}}\left(e_{it}\mathbf{1}_{it}\right),
\]
\[
e_{it}^{\left(2\right)}
=
e_{it}\mathbf{\bar{1}}_{it},
\qquad
e_{it}^{\left(3\right)}
=
-\text{\emph{E}}\left(e_{it}\mathbf{\bar{1}}_{it}\right).
\]
Then
\[
e_{it}-\text{\emph{E}}\left(e_{it}\right)
=
e_{it}^{\left(1\right)}
+
e_{it}^{\left(2\right)}
+
e_{it}^{\left(3\right)}.
\]

By the mixing condition in Assumption \ref{A:mixing}(ii), the moment
condition with $C_{q}>2$, and the covariance inequality for strong
mixing processes, there exists a positive finite constant $v_{0}$
such that, uniformly over $i,N$, and $T$,
\[
\max_{t}\text{\emph{Var}}\left(e_{it}^{\left(1\right)}\right)
+
2\max_{t}\sum_{s=t+1}^{T}
\left|
\text{\emph{Cov}}
\left(
e_{it}^{\left(1\right)},
e_{is}^{\left(1\right)}
\right)
\right|
\leq v_{0}.
\]

Using Bernstein inequality for strong mixing processes
(e.g., Theorem 2 in \citet{MerlevedeEtal}), there exist positive
constants $C_{1}$ and $C_{2}$ such that
\begin{align*}
&
\Pr\left(
\max_{i=1,\ldots,N}
\left|
\frac{1}{T}\sum_{t=1}^{T}e_{it}^{\left(1\right)}
\right|
>
\frac{\epsilon}{3\upsilon_{NT}}
\right)                                                    \\
\leq\;&
\sum_{i=1}^{N}
\Pr\left(
\left|
\sum_{t=1}^{T}e_{it}^{\left(1\right)}
\right|
>
\frac{T\epsilon}{3\upsilon_{NT}}
\right)                                                    \\
\leq\;&
N\exp\left\{
-
\frac{C_{1}T^{2}\epsilon^{2}}
{
9Tv_{0}\upsilon_{NT}^{2}
+
9C_{NT}^{2}\upsilon_{NT}^{2}
+
3T\epsilon\upsilon_{NT}C_{NT}(\log T)^{2}
}
\right\}.
\end{align*}
Since
\[
C_{NT}
=
\upsilon_{NT}^{-1}T(\log N)^{-4},
\]
we have
\[
C_{NT}^{2}\upsilon_{NT}^{2}
=
T^{2}(\log N)^{-8}
\]
and
\[
T\epsilon\upsilon_{NT}C_{NT}(\log T)^{2}
=
T^{2}\epsilon(\log N)^{-4}(\log T)^{2}.
\]
Moreover,
\[
T\upsilon_{NT}^{2}
=
o\left\{
T^{2}(\log N)^{-4}
\right\}
\]
because
\[
\upsilon_{NT}
\ll
T^{1/2}(\log N)^{-2}.
\]
Since $T=N^{C}$, we have $\log T=C\log N$, and hence
\[
(\log N)^{-4}(\log T)^{2}
\asymp
(\log N)^{-2}.
\]
It follows that, for some positive constant $C_{2}$,
\[
\Pr\left(
\max_{i=1,\ldots,N}
\left|
\frac{1}{T}\sum_{t=1}^{T}e_{it}^{\left(1\right)}
\right|
>
\frac{\epsilon}{3\upsilon_{NT}}
\right)
\leq
N\exp\left\{
-C_{2}(\log N)^{2}
\right\}
=
o\left(N^{-M}\right)
\]
for any sufficiently large $M>0$.

We next consider $e_{it}^{\left(2\right)}$. Notice that, if
\[
\max_{\substack{1\leq i\leq N\\1\leq t\leq T}}
|e_{it}|<C_{NT},
\]
then $e_{it}^{\left(2\right)}=0$ for every $i$ and $t$. Therefore,
\begin{align*}
&
\Pr\left(
\max_{i=1,\ldots,N}
\left|
\frac{1}{T}\sum_{t=1}^{T}e_{it}^{\left(2\right)}
\right|
>
\frac{\epsilon}{3\upsilon_{NT}}
\right)                                                    \\
\leq\;&
\Pr\left(
\max_{\substack{1\leq i\leq N\\1\leq t\leq T}}
|e_{it}|
\geq C_{NT}
\right)                                                    \\
\leq\;&
\sum_{i=1}^{N}\sum_{t=1}^{T}
\Pr\left(
|e_{it}|\geq C_{NT}
\right)                                                    \\
\leq\;&
\sum_{i=1}^{N}\sum_{t=1}^{T}
\frac{
\text{\emph{E}}\left(|e_{it}|^{C_{q}}\right)
}{
C_{NT}^{C_{q}}
}                                                          \\
\lesssim\;&
NTC_{NT}^{-C_{q}}                                          \\
=\;&
\frac{
NT\upsilon_{NT}^{C_{q}}
(\log N)^{4C_{q}}
}{
T^{C_{q}}
}.
\end{align*}

Finally, consider $e_{it}^{\left(3\right)}$. By the moment condition,
\begin{align*}
\left|
\frac{1}{T}\sum_{t=1}^{T}e_{it}^{\left(3\right)}
\right|
&\leq
\max_{t=1,\ldots,T}
\text{\emph{E}}
\left(
|e_{it}|\mathbf{\bar{1}}_{it}
\right)                                                     \\
&=
\max_{t=1,\ldots,T}
\text{\emph{E}}
\left[
|e_{it}|
\mathbf{1}\left(|e_{it}|\geq C_{NT}\right)
\right]                                                     \\
&\leq
\max_{t=1,\ldots,T}
\frac{
\text{\emph{E}}\left(|e_{it}|^{C_{q}}\right)
}{
C_{NT}^{C_{q}-1}
}                                                           \\
&\lesssim
C_{NT}^{-(C_{q}-1)}                                        \\
&=
\left(
\frac{\upsilon_{NT}}{T}
\right)^{C_{q}-1}
(\log N)^{4(C_{q}-1)}.
\end{align*}
Furthermore,
\begin{align*}
&
\upsilon_{NT}
\left(
\frac{\upsilon_{NT}}{T}
\right)^{C_{q}-1}
(\log N)^{4(C_{q}-1)}                                      \\
=\;&
\frac{
\upsilon_{NT}^{C_{q}}
(\log N)^{4(C_{q}-1)}
}{
T^{C_{q}-1}
}
=o(1).
\end{align*}
Indeed, writing
\[
\upsilon_{NT}
=
o\left\{
T^{1/2}(\log N)^{-2}
\right\},
\]
the expression on the preceding display is
\[
o\left\{
T^{1-C_{q}/2}
(\log N)^{2C_{q}-4}
\right\}
=
o(1),
\]
because $C_{q}>2$ and $T=N^{C}$. Thus, uniformly over
$i=1,\ldots,N$, for all sufficiently large $T$,
\[
\left|
\frac{1}{T}\sum_{t=1}^{T}e_{it}^{\left(3\right)}
\right|
<
\frac{\epsilon}{3\upsilon_{NT}}.
\]
Consequently,
\[
\Pr\left(
\max_{i=1,\ldots,N}
\left|
\frac{1}{T}\sum_{t=1}^{T}e_{it}^{\left(3\right)}
\right|
>
\frac{\epsilon}{3\upsilon_{NT}}
\right)
=
0
\]
for all sufficiently large $T$.

Combining the preceding three bounds yields
\begin{align*}
&
\Pr\left(
\max_{i=1,\ldots,N}
\left|
\frac{1}{T}\sum_{t=1}^{T}
\left[
e_{it}-\text{\emph{E}}\left(e_{it}\right)
\right]
\right|
>
\frac{\epsilon}{\upsilon_{NT}}
\right)                                                     \\
\leq\;&
\Pr\left(
\max_{i=1,\ldots,N}
\left|
\frac{1}{T}\sum_{t=1}^{T}e_{it}^{\left(1\right)}
\right|
>
\frac{\epsilon}{3\upsilon_{NT}}
\right)                                                     \\
&+
\Pr\left(
\max_{i=1,\ldots,N}
\left|
\frac{1}{T}\sum_{t=1}^{T}e_{it}^{\left(2\right)}
\right|
>
\frac{\epsilon}{3\upsilon_{NT}}
\right)                                                     \\
&+
\Pr\left(
\max_{i=1,\ldots,N}
\left|
\frac{1}{T}\sum_{t=1}^{T}e_{it}^{\left(3\right)}
\right|
>
\frac{\epsilon}{3\upsilon_{NT}}
\right)                                                     \\
\lesssim\;&
\frac{
NT\upsilon_{NT}^{C_{q}}
(\log N)^{4C_{q}}
}{
T^{C_{q}}
}.
\end{align*}

\textbf{(ii)} Recall (\ref{eq:bias_rate}), we can find a positive
$\bar{C_{b}}$ such that 
\[
\max_{l=0,...,p}\sup_{s\in\left[0,1\right]}\left|b_{il}\left(s\right)\right|\leq\bar{C_{b}}m^{-\kappa},
\]
because $p$ is finite. As a result, 
\[
\max_{i=1,...,N}\max_{l=0,1,...,p}\left|\frac{1}{T}\sum_{t=1}^{T}e_{it}b_{il}\left(\tau_{t}\right)\right|\leq\bar{C_{b}}m^{-\kappa}\max_{i=1,...,N}\frac{1}{T}\sum_{t=1}^{T}\left|e_{it}\right|.
\]
Applying (i) by setting $\upsilon_{NT}=1$ for series $\left|e_{it}\right|$
implies 
\[
\Pr\left(\max_{i=1,...,N}\left|\frac{1}{T}\sum_{t=1}^{T}\left|e_{it}\right|-\frac{1}{T}\sum_{t=1}^{T}\text{E}\left|e_{it}\right|\right|>\epsilon\right)\lesssim\frac{NT\left(\log N\right)^{4C_{q}}}{T^{C_{q}}}.
\]
Since $\max_{i=1,...,N}\frac{1}{T}\sum_{t=1}^{T}\text{E}\left|e_{it}\right|$
is bounded and $\epsilon$ can be any arbitrary small positive number,
we can set 
\[
M=2\bar{C_{b}}\cdot\max_{i=1,...,N}\frac{1}{T}\sum_{t=1}^{T}\text{E}\left|e_{it}\right|,
\]
so that the desired result holds. 
\end{proof}
\begin{proof}[Proof of Lemma \ref{LE:splines}]
Let $g_{l}\left(\tau_{t}\right)=\pi_{l}'\mathbb{B}^{m}\left(\tau_{t}\right)$,
where $\pi_{l}$ is $m\times1$ and denotes the corresponding elements
in $\tilde{\pi}$ for $x_{it,l}\otimes\mathbb{B}^{m}\left(\tau_{t}\right)$,
$l=0,1,...,p$.\footnote{We abuse the notation a bit by letting $x_{it,0}$ denote 1.}
Let $\mathbf{g}\left(\tau_{t}\right)\equiv\left(g_{0}\left(\tau_{t}\right),...,g_{p}\left(\tau_{t}\right)\right)'$.
In other words, $g_{l}\left(\tau_{t}\right)$ is the approximation
for $\beta_{l}\left(\tau_{t}\right)$ when $\pi_{l}$=$\pi_{l}^{0}$.
Then $\frac{1}{T}\sum_{t=1}^{T}\text{E}\left(\tilde{\pi}'\tilde{z}_{it}\tilde{z}_{it}'\tilde{\pi}\right)$
can be written as 
\begin{align*}
\frac{1}{T}\sum_{t=1}^{T}\text{E}\left(\tilde{\pi}'\tilde{z}_{it}\tilde{z}_{it}'\tilde{\pi}\right) & =\frac{1}{T}\sum_{t=1}^{T}\text{E}\left[\mathbf{g}\left(\tau_{t}\right)'\tilde{x}_{it}\tilde{x}_{it}'\mathbf{g}\left(\tau_{t}\right)\right]\\
 & =\frac{1}{T}\sum_{t=1}^{T}\mathbf{g}\left(\tau_{t}\right)'\text{E}\left(\tilde{x}_{it}\tilde{x}_{it}'\right)\mathbf{g}\left(\tau_{t}\right),
\end{align*}
By the rank condition in Assumption \ref{A:rank}, the above implies
\begin{align}
\underline{C}_{xx}\frac{1}{T}\sum_{t=1}^{T}\mathbf{g}\left(\tau_{t}\right)'\mathbf{g}\left(\tau_{t}\right) & \leq\min_{1\leq i\leq N}\frac{1}{T}\sum_{t=1}^{T}\text{E}\left(\tilde{\pi}'\tilde{z}_{it}\tilde{z}_{it}'\tilde{\pi}\right)\nonumber \\
 & \leq\max_{1\leq i\leq N}\frac{1}{T}\sum_{t=1}^{T}\text{E}\left(\tilde{\pi}'\tilde{z}_{it}\tilde{z}_{it}'\tilde{\pi}\right)\leq\bar{C}_{xx}\frac{1}{T}\sum_{t=1}^{T}\mathbf{g}\left(\tau_{t}\right)'\mathbf{g}\left(\tau_{t}\right).\label{eq:zpi}
\end{align}

Lemma A.4 in \citet{DongLinton2018} implies that, $T^{-1}\sum_{t=1}^{T}\mathbb{B}^{m}\left(\tau_{t}\right)\mathbb{B}^{m}\left(\tau_{t}\right)'=I_{m}+O\left(m/T\right)$.
Using it 
\begin{align*}
\frac{1}{T}\sum_{t=1}^{T}g_{l}\left(\tau_{t}\right)^{2} & =\frac{1}{T}\sum_{t=1}^{T}\tilde{\pi}_{l}'\mathbb{B}^{m}\left(\tau_{t}\right)\mathbb{B}^{m}\left(\tau_{t}\right)'\tilde{\pi}_{l}\\
 & =\tilde{\pi}_{l}'\left[I_{m}+O\left(m/T\right)\right]\tilde{\pi}_{l}=\left\Vert \tilde{\pi}_{l}\right\Vert \left[1+o\left(1\right)\right],
\end{align*}
due to $m/T\rightarrow0$. Substituting the above result into (\ref{eq:zpi})
yields the desired result by relaxing the lower and upper bounds to
$\underline{C}_{xx}/2\left\Vert \tilde{\pi}_{l}\right\Vert $ and
$2\bar{C}_{xx}\left\Vert \tilde{\pi}_{l}\right\Vert $, respectively. 
\end{proof}
\begin{proof}[Proof of Lemma \ref{LE:splines-sample-1}]
\textbf{(i)} is standard in the literature. Lemma 5 in \citet{FanEtal2011}
and Lemma S1.4 in \citet{SuEtal2024} show a similar result for variables
with sub-exponential tails. We prove this lemma for $\text{E}\left(\left|x_{it}\right|^{q}\right)<\infty$.
Note that $\tilde{z}_{it}=\left[\mathbb{B}^{m}\left(\tau_{t}\right)^{\prime},\left(x_{it}\otimes\mathbb{B}^{m}\left(\tau_{t}\right)\right)^{\prime}\right]$,
a $\left[\left(p+1\right)m\right]\times1$ vector, then 
\begin{align*}
\frac{1}{T}\sum_{t=1}^{T}\tilde{\pi}'\tilde{z}_{it}\tilde{z}_{it}'\tilde{\pi} & =\frac{1}{T}\sum_{t=1}^{T}\left[\sum_{j=0}^{p}\sum_{l=0}^{p}\sum_{s=0}^{m-1}\sum_{s'=0}^{m-1}\pi_{js}\pi_{ls'}B_{s}\left(\tau_{t}\right)B_{s'}\left(\tau_{t}\right)x_{itj}x_{itl}\right]\\
 & =\sum_{j=0}^{p}\sum_{l=0}^{p}\sum_{s=0}^{m-1}\sum_{s'=0}^{m-1}\left[\frac{1}{T}\sum_{t=1}^{T}\pi_{js}\pi_{ls'}B_{s}\left(\tau_{t}\right)B_{s'}\left(\tau_{t}\right)x_{itj}x_{itl}\right]
\end{align*}
where we let $x_{it0}$ denote the constant 1. Using the above, 
\begin{align*}
 & \frac{1}{T}\sum_{t=1}^{T}\tilde{\pi}'\tilde{z}_{it}\tilde{z}_{it}'\tilde{\pi}-\frac{1}{T}\sum_{t=1}^{T}\text{E}\left(\tilde{\pi}'\tilde{z}_{it}\tilde{z}_{it}'\tilde{\pi}\right)\\
= & \sum_{j=0}^{p}\sum_{l=0}^{p}\sum_{s=0}^{m-1}\sum_{s'=0}^{m-1}\pi_{js}\pi_{ls'}\left\{ \frac{1}{T}\sum_{t=1}^{T}B_{s}\left(\tau_{t}\right)B_{s'}\left(\tau_{t}\right)\left[x_{itj}x_{itl}-\text{E}\left(x_{itj}x_{itl}\right)\right]\right\} \\
\equiv & \sum_{j=0}^{p}\sum_{l=0}^{p}\sum_{s=0}^{m-1}\sum_{s'=0}^{m-1}\pi_{js}\pi_{ls'}I_{iss'},
\end{align*}
with 
\[
I_{iss'}\equiv\frac{1}{T}\sum_{t=1}^{T}B_{s}\left(\tau_{t}\right)B_{s'}\left(\tau_{t}\right)\left[x_{itj}x_{itl}-\text{E}\left(x_{itj}x_{itl}\right)\right].
\]
We first assume the event $\left\{ \max_{i=1,...,N}\max_{s,s'=1,...,m}\left|I_{iss'}\right|\leq\frac{\epsilon}{mp^{2}}\right\} $
holds; its probability bound will be shown later. Conditional on this
event, for all $\left\Vert \tilde{\pi}\right\Vert =1,$ 
\begin{align*}
\max_{i=1,...,N}\left|\sum_{j=0}^{p}\sum_{l=0}^{p}\sum_{s=0}^{m-1}\sum_{s'=0}^{m-1}\pi_{js}\pi_{ls'}I_{iss'}\right| & \leq\frac{\epsilon}{mp^{2}}\sum_{j=0}^{p}\sum_{l=0}^{p}\sum_{s=0}^{m-1}\sum_{s'=0}^{m-1}\left|\pi_{js}\pi_{ls'}\right|.\\
 & \leq\frac{\epsilon}{mp^{2}}\sum_{j=0}^{p}\sum_{l=0}^{p}\sum_{s=0}^{m-1}\left[\left|\pi_{js}\right|\sum_{s'=0}^{m-1}\left|\pi_{ls'}\right|\right]\\
 & \leq\frac{\epsilon}{mp^{2}}\sum_{j=0}^{p}\sum_{l=0}^{p}\sum_{s=0}^{m-1}\left|\pi_{js}\right|\left(\sum_{s'=0}^{m-1}\pi_{ls'}^{2}\right)^{1/2}\left(\sum_{s'=0}^{m-1}1^{2}\right)^{1/2}\\
 & =\frac{\epsilon}{\sqrt{m}p^{2}}\sum_{j=0}^{p}\sum_{l=0}^{p}\sum_{s=0}^{m-1}\left|\pi_{js}\right|\leq\epsilon,
\end{align*}
where the third and fourth lines hold by Cauchy-Schwarz inequality.
Thus 
\begin{align*}
 & \Pr\left(\max_{i=1,...,N}\sup_{\tilde{\pi}\in\mathbb{R}^{m\left(p+1\right)},\left\Vert \tilde{\pi}\right\Vert =1}\left|\sum_{j=0}^{p}\sum_{l=0}^{p}\sum_{s=0}^{m-1}\sum_{s'=0}^{m-1}\pi_{js}\pi_{ls'}I_{iss'}\right|\leq\epsilon\right)\\
\geq & \Pr\left(\max_{i=1,...,N}\max_{s,s'=1,...,m}\left|I_{iss'}\right|\leq\frac{\epsilon}{mp^{2}}\right)
\end{align*}
which is equivalent to 
\begin{align}
 & \Pr\left(\max_{i=1,...,N}\sup_{\tilde{\pi}\in\mathbb{R}^{m\left(p+1\right)},\left\Vert \tilde{\pi}\right\Vert =1}\left|\sum_{j=0}^{p}\sum_{l=0}^{p}\sum_{s=0}^{m-1}\sum_{s'=0}^{m-1}\pi_{js}\pi_{ls'}I_{iss'}\right|>\epsilon\right)\nonumber \\
< & \Pr\left(\max_{i=1,...,N}\max_{s,s'=1,...,m}\left|I_{iss'}\right|>\frac{\epsilon}{mp^{2}}\right).\label{eq:iiss}
\end{align}

Using the result in Lemma \ref{LE:bound} (i) (by setting $\upsilon_{NT}=m$),
the moment condition on $x$ ($C_{q}$ in Lemma \ref{LE:bound} is
$q/2$ in this case), and the fact that $p$ is fixed, not hard to
see that 
\begin{align*}
\Pr\left(\max_{i=1,...,N}\max_{s,s'=1,...,m}\left|I_{iss'}\right|>\frac{\epsilon}{mp^{2}}\right) & \leq\sum_{s,s'=1}^{m}\Pr\left(\max_{i=1,...,N}\left|I_{iss'}\right|>\frac{\epsilon}{p^{2}}\right)\\
 & \lesssim\sum_{s,s'=1}^{m}\frac{NTm^{q/2}\left(\log N\right)^{2q}}{T^{q/2}}\\
 & \lesssim\frac{NTm^{q/2+2}\left(\log N\right)^{2q}}{T^{q/2}}.
\end{align*}
Substitute it into (\ref{eq:iiss}), we obtain the desired result.

\textbf{(ii)} is an immediate result from Lemmas \ref{LE:splines}
and the result (i).

Suppose the complement of the event in (i) holds with a small $\epsilon>0$.
For any $\tilde{\pi}\in\mathbb{R}^{m\left(p+1\right)}$ and $\left\Vert \tilde{\pi}\right\Vert =1$,
\[
\max_{1\leq i\leq N}\tilde{\pi}'\left(\frac{1}{T}\sum_{t=1}^{T}\tilde{z}_{it}\tilde{z}_{it}'\right)\tilde{\pi}\leq\max_{1\leq i\leq N}\left(1+\epsilon\right)\frac{1}{T}\sum_{t=1}^{T}\text{E}\left(\tilde{\pi}'\tilde{z}_{it}\tilde{z}_{it}'\tilde{\pi}\right)\leq2\bar{C}_{xx}\left(1+\epsilon\right).
\]
Similarly, 
\[
\min_{1\leq i\leq N}\tilde{\pi}'\left(\frac{1}{T}\sum_{t=1}^{T}\tilde{z}_{it}\tilde{z}_{it}'\right)\tilde{\pi}\geq\min_{1\leq i\leq N}\left(1-\epsilon\right)\frac{1}{T}\sum_{t=1}^{T}\text{E}\left(\tilde{\pi}'\tilde{z}_{it}\tilde{z}_{it}'\tilde{\pi}\right)\geq\frac{1}{2}\bar{C}_{xx}\left(1-\epsilon\right).
\]
We show the result by setting $\underline{C}_{zz}=1/2\bar{C}_{xx}\left(1-\epsilon\right)$
and $2\bar{C}_{xx}\left(1+\epsilon\right)$. 
\end{proof}
\begin{proof}[Proof of Lemma \ref{LE:bound-ours}]
\textbf{(i)} $\tilde{z}_{it}v_{it}$ is $(p+1)m$$\times1$. Elements
in $\frac{1}{T}\sum_{t=1}^{T}\tilde{z}_{it}v_{it}$ are 
\[
\frac{1}{T}\sum_{t=1}^{T}x_{itj}B_{l}\left(\tau_{t}\right)v_{it},\text{ for }j=0,...,p\text{ and }l=0,...,m-1.
\]
The idea is to obtain the probability bound of 
\[
\max_{i=1,...,N}\max_{j=0,...,p}\max_{l=0,...,m-1}\left|\frac{1}{T}\sum_{t=1}^{T}x_{itj}B_{l}\left(\tau_{t}\right)v_{it}\right|>\epsilon/\sqrt{\left(p+1\right)m},
\]
with which we can derive the probability bound of the event $\max_{i=1,...,N}\left\Vert \frac{1}{T}\sum_{t=1}^{T}\tilde{z}_{it}v_{it}\right\Vert >\epsilon,$
because the latter can be implied by the former. By the moment condition
in Assumption \ref{A:moment}, Lemma \ref{LE:bound} (i) implies that
(by setting $\upsilon_{NT}=\sqrt{m}$ and $C_{q}=q/2$) 
\begin{align*}
 & \Pr\left(\max_{i=1,...,N}\max_{j=0,...,p}\max_{l=0,...,m-1}\left|\frac{1}{T}\sum_{t=1}^{T}x_{itj}B_{l}\left(\tau_{t}\right)v_{it}\right|>\frac{\epsilon}{\sqrt{\left(p+1\right)m}}\right)\\
\leq & \sum_{l=0}^{m-1}\sum_{j=0}^{p}\Pr\left(\max_{i=1,...,N}\left|\frac{1}{T}\sum_{t=1}^{T}x_{itj}B_{l}\left(\tau_{t}\right)v_{it}\right|>\frac{\epsilon}{\sqrt{\left(p+1\right)m}}\right)\lesssim\frac{NTm^{q/4+1}\left(\log N\right)^{2q}}{T^{q/2}},
\end{align*}
using the fact that $p$ is a fixed constant. As a result 
\begin{align*}
 & \Pr\left(\max_{i=1,...,N}\left\Vert \frac{1}{T}\sum_{t=1}^{T}\tilde{z}_{it}v_{it}\right\Vert >\epsilon\right)\\
\leq & \Pr\left(\max_{i=1,...,N}\max_{j=0,...,p}\max_{l=0,...,m-1}\left|\frac{1}{T}\sum_{t=1}^{T}x_{itj}B_{l}\left(\tau_{t}\right)v_{it}\right|>\frac{\epsilon}{\sqrt{\left(p+1\right)m}}\right)\\
\lesssim & \frac{NTm^{q/4+1}\left(\log N\right)^{2q}}{T^{q/2}}.
\end{align*}

\textbf{(ii)} Similarly, $\tilde{z}_{it}\xi_{it}$ is $(p+1)m$$\times1$.
Elements in $\frac{1}{T}\sum_{t=1}^{T}\tilde{z}_{it}\xi_{it}$ are
\begin{align*}
 & \frac{1}{T}\sum_{t=1}^{T}x_{itj}B_{l}\left(\tau_{t}\right)\xi_{it}=\frac{1}{T}\sum_{t=1}^{T}\sum_{j'=0}^{p}x_{itj}x_{itj'}B_{l}\left(\tau_{t}\right)b_{ij'}\left(\tau_{t}\right)
\end{align*}
$\text{ for }j=0,1,...,p\text{ and }l=0,1,...,m-1$. By the condition
$\kappa\geq1$ in Assumption \ref{A:coef}, so that $m^{-\kappa}\ll m^{-1/2}$.
Using Lemma \ref{LE:bound} (ii), 
\begin{align*}
 & \Pr\left(\max_{i=1,...,N}\max_{j=0,...,p}\max_{l=0,...,m-1}\left|\frac{1}{T}\sum_{t=1}^{T}\sum_{j'=0}^{p}x_{itj}x_{itj'}B_{l}\left(\tau_{t}\right)b_{ij'}\left(\tau_{t}\right)\right|>\frac{\epsilon}{\sqrt{\left(p+1\right)m}}\right)\\
\leq & \sum_{j=0}^{p}\sum_{l=0}^{m-1}\sum_{j'=0}^{p}\Pr\left(\max_{i=1,...,N}\left|\frac{1}{T}\sum_{t=1}^{T}x_{itj}x_{itj'}B_{l}\left(\tau_{t}\right)b_{ij'}\left(\tau_{t}\right)\right|>\frac{\epsilon}{\left(p+1\right)^{3/2}\sqrt{m}}\right)\\
\leq & \sum_{j=0}^{p}\sum_{l=0}^{m-1}\sum_{j'=0}^{p}\Pr\left(\max_{i=1,...,N}\left|\frac{1}{T}\sum_{t=1}^{T}x_{itj}x_{itj'}B_{l}\left(\tau_{t}\right)b_{ij'}\left(\tau_{t}\right)\right|>Mm^{-\kappa}\right)\\
\lesssim & \frac{NTm\left(\log N\right)^{2q}}{T^{q/2}}.
\end{align*}
for an arbitrary large $M$ after some large $T,$ where the third
line uses $m^{-\kappa}\ll m^{-1/2}$ and the last line holds by applying
Lemma \ref{LE:bound} (ii) with $C_{q}=q/2$ and the fact that $p$
is fixed.

We reach the desired result by 
\begin{align*}
 & \Pr\left(\max_{i=1,...,N}\left\Vert \frac{1}{T}\sum_{t=1}^{T}\tilde{z}_{it}\xi_{it}\right\Vert >\epsilon\right)\\
\leq & \Pr\left(\max_{i=1,...,N}\max_{l=0,...,m}\max_{j=0,...,p}\left|\frac{1}{T}\sum_{t=1}^{T}\sum_{j'=0}^{p}x_{itj}x_{itj'}B_{l}\left(\tau_{t}\right)b_{ij'}\left(\tau_{t}\right)\right|>\frac{\epsilon}{\sqrt{\left(p+1\right)m}}\right).
\end{align*}
\end{proof}
\begin{proof}[Proof of Lemma \ref{LE:uni_converge}]
\textbf{(i)} By (\ref{eq:pihat}) and (\ref{eq:yit_withbias}), we
substitute $y_{it}=\tilde{z}_{it}'\tilde{\pi}_{i}^{0}+\xi_{it}+v_{it}$
into $\widehat{\tilde{\pi}}_{i}$, and obtain 
\begin{align*}
\widehat{\tilde{\pi}}_{i} & -\tilde{\pi}_{i}^{0}=\left(\frac{1}{T}\sum_{t=1}^{T}\tilde{z}_{it}\tilde{z}_{it}'\right)^{-1}\left(\frac{1}{T}\sum_{t=1}^{T}\tilde{z}_{it}\xi_{it}\right)+\left(\frac{1}{T}\sum_{t=1}^{T}\tilde{z}_{it}\tilde{z}_{it}'\right)^{-1}\left(\frac{1}{T}\sum_{t=1}^{T}\tilde{z}_{it}v_{it}\right).
\end{align*}

Then 
\[
\left\Vert \widehat{\tilde{\pi}}_{i}-\tilde{\pi}_{i}^{0}\right\Vert \leq\mu_{\min}^{-1}\left(\frac{1}{T}\sum_{t=1}^{T}\tilde{z}_{it}\tilde{z}_{it}'\right)\left[\left\Vert \frac{1}{T}\sum_{t=1}^{T}\tilde{z}_{it}\xi_{it}\right\Vert +\left\Vert \frac{1}{T}\sum_{t=1}^{T}\tilde{z}_{it}v_{it}\right\Vert \right],
\]
which implies 
\begin{align*}
 & \Pr\left(\max_{1\leq i\leq N}\left\Vert \widehat{\tilde{\pi}}_{i}-\tilde{\pi}_{i}^{0}\right\Vert >\epsilon\right)\\
\leq & \Pr\left(\max_{1\leq i\leq N}\mu_{\min}^{-1}\left(\frac{1}{T}\sum_{t=1}^{T}\tilde{z}_{it}\tilde{z}_{it}'\right)>C_{zz}^{-1}\right)+\Pr\left(\max_{1\leq i\leq N}\left\Vert \frac{1}{T}\sum_{t=1}^{T}\tilde{z}_{it}\xi_{it}\right\Vert +\left\Vert \frac{1}{T}\sum_{t=1}^{T}\tilde{z}_{it}v_{it}\right\Vert >C_{zz}\epsilon\right)\\
\leq & \Pr\left(\max_{1\leq i\leq N}\mu_{\min}^{-1}\left(\frac{1}{T}\sum_{t=1}^{T}\tilde{z}_{it}\tilde{z}_{it}'\right)>C_{zz}^{-1}\right)+\Pr\left(\max_{1\leq i\leq N}\left\Vert \frac{1}{T}\sum_{t=1}^{T}\tilde{z}_{it}\xi_{it}\right\Vert >\frac{C_{zz}\epsilon}{2}\right)\\
 & +\Pr\left(\max_{1\leq i\leq N}\left\Vert \frac{1}{T}\sum_{t=1}^{T}\tilde{z}_{it}v_{it}\right\Vert >\frac{C_{zz}\epsilon}{2}\right)\\
= & o\left(1\right),
\end{align*}
by Lemmas \ref{LE:splines-sample-1} and \ref{LE:bound-ours}, and
the rates in Assumption \ref{A:tuningPara} (ii).

\textbf{(ii)} Note that 
\begin{align*}
\hat{\sigma}_{vi}^{2} & -\sigma_{vi}^{2}=\frac{1}{T-1}\sum_{t=1}^{T}\left(y_{it}-\tilde{z}_{it}'\widehat{\tilde{\pi}}_{i}\right)^{2}-\sigma_{vi}^{2}\\
 & =\frac{1}{T-1}\sum_{t=1}^{T}\left(\tilde{z}_{it}'\left(\tilde{\pi}_{i}^{0}-\widehat{\tilde{\pi}}_{i}\right)+\xi_{it}+v_{it}\right)^{2}-\sigma_{vi}^{2}\\
 & =\frac{1}{T-1}\sum_{t=1}^{T}v_{it}^{2}-\sigma_{vi}^{2}+\left(\tilde{\pi}_{i}^{0}-\widehat{\tilde{\pi}}_{i}\right)'\left(\frac{1}{T-1}\sum_{t=1}^{T}\tilde{z}_{it}\tilde{z}_{it}'\right)\left(\tilde{\pi}_{i}^{0}-\widehat{\tilde{\pi}}_{i}\right)+\frac{1}{T-1}\sum_{t=1}^{T}\xi_{it}^{2}\\
 & +2\left[\frac{1}{T-1}\sum_{t=1}^{T}\left(\xi_{it}+v_{it}\right)z_{it}'\right]\left(\tilde{\pi}_{i}^{0}-\widehat{\tilde{\pi}}_{i}\right)+2\frac{1}{T-1}\sum_{t=1}^{T}\xi_{it}v_{it}\\
 & \equiv A_{i1}+A_{i2}+A_{i3}+A_{i4}+A_{i5}.
\end{align*}
The uniform convergence of $A_{i1},A_{i2},A_{i3},A_{i4},\text{ and }A_{i5}$
can be similarly shown as in Lemmas \ref{LE:bound}, \ref{LE:splines-sample-1},
\ref{LE:bound-ours} and part (i) of this lemma. We omit the proof
for conciseness. Therefore 
\[
\Pr\left(\max_{i=1,...,N}\left|\hat{\sigma}_{vi}^{2}-\sigma_{vi}^{2}\right|>\epsilon\right)=o\left(1\right).
\]
\end{proof}
\begin{proof}[Proof of Lemma \ref{LE:z_k'z_k}]
Recall that Lemma \ref{LE:splines} shows that 
\[
\frac{\underline{C}_{xx}}{2}\leq\mu_{\min}\left\{ \text{E}\left(\frac{1}{T}\sum_{t=1}^{T}\tilde{\underline{z}}_{it}\tilde{\underline{z}}_{it}'\right)\right\} \leq\mu_{\max}\left\{ \text{E}\left(\frac{1}{T}\sum_{t=1}^{T}\tilde{\underline{z}}_{it}\tilde{\underline{z}}_{it}'\right)\right\} \leq2\bar{C}_{xx},
\]
after some large $T.$ Note $\tilde{z}_{it}=\left(1,z_{it}'\right)'.$
Thus, 
\[
\text{E}\left(\frac{1}{T}\sum_{t=1}^{T}\tilde{\underline{z}}_{it}\tilde{\underline{z}}_{it}'\right)=\left(\begin{array}{cc}
1 & \text{E}\left(\frac{1}{T}\sum_{t=1}^{T}\underline{z}_{it}'\right)\\
\text{E}\left(\frac{1}{T}\sum_{t=1}^{T}\underline{z}_{it}\right) & \text{E}\left(\frac{1}{T}\sum_{t=1}^{T}\underline{z}_{it}\underline{z}_{it}'\right)
\end{array}\right).
\]
With the block representation of $\text{E}\left(\frac{1}{T}\sum_{t=1}^{T}\tilde{\underline{z}}_{it}\tilde{\underline{z}}_{it}'\right)$
and its full rank condition, the inverse of it can be calculated as
(we only present its lower diagonal): 
\begin{align*}
\left[\text{E}\left(\frac{1}{T}\sum_{t=1}^{T}\tilde{\underline{z}}_{it}\tilde{\underline{z}}_{it}'\right)\right]^{-1} & =\left(\begin{array}{cc}
\cdot_{1\times1} & \cdot_{1\times\left(mp+m-1\right)}\\
\cdot_{\left(mp+m-1\right)\times1} & \left[\text{E}\left(\frac{1}{T}\sum_{t=1}^{T}\underline{z}_{it}\underline{z}_{it}'\right)-\text{E}\left(\frac{1}{T}\sum_{t=1}^{T}\underline{z}_{it}\right)\text{E}\left(\frac{1}{T}\sum_{t=1}^{T}\underline{z}_{it}'\right)\right]^{-1}
\end{array}\right)\\
 & =\left(\begin{array}{cc}
\cdot_{1\times1} & \cdot_{1\times\left(mp+m-1\right)}\\
\cdot_{\left(mp+m-1\right)\times1} & \left[\text{E}\left(\frac{1}{T}\sum_{t=1}^{T}\ddot{\underline{z}}_{it}\ddot{\underline{z}}_{it}'\right)\right]^{-1}
\end{array}\right).
\end{align*}

The above implies that $\text{E}\left(\frac{1}{T}\sum_{t=1}^{T}\ddot{z}_{it}\ddot{z}_{it}'\right)$
must be of full rank, and 
\[
\frac{\bar{C}_{xx}^{-1}}{2}\leq\min_{1\leq i\leq N}\mu_{\min}\left\{ \left[\text{E}\left(\frac{1}{T}\sum_{t=1}^{T}\underline{\ddot{z}}_{it}\underline{\ddot{z}}_{it}'\right)\right]^{-1}\right\} \leq\max_{1\leq i\leq N}\mu_{\max}\left\{ \left[\text{E}\left(\frac{1}{T}\sum_{t=1}^{T}\underline{\ddot{z}}_{it}\underline{\ddot{z}}_{it}'\right)\right]^{-1}\right\} \leq2\underline{C}_{xx}^{-1}.
\]
We obtain the desired result by $\min_{1\leq i\leq N}\mu_{\min}\left(A_{i}^{-1}\right)=1\left/\max_{1\leq i\leq N}\mu_{\max}\left(A_{i}\right)\right.$
for any full rank matrices $\left\{ A_{i}\right\} _{i=1}^{N}$. 
\end{proof}
\begin{proof}[Proof of Lemma \ref{LE:splines-sample-2}]
The proof here follows that of Lemma \ref{LE:splines-sample-1}.
We similarly define

\[
I_{ss'\left(k\right)}\equiv\frac{1}{N_{k}}\sum_{i\in G_{k|K^{*}}}\frac{1}{T}\sum_{t=1}^{T}\left\{ B_{s}\left(\tau_{t}\right)B_{s'}\left(\tau_{t}\right)\left[x_{itj}x_{itl}-\text{E}\left(x_{itj}x_{itl}\right)\right]\right\} .
\]
Its probability bound is 
\begin{align*}
\Pr\left(\max_{s,s'=1,...,\underline{m}}\left|I_{ss'\left(k\right)}\right|>\frac{\epsilon}{\underline{m}p^{2}}\right) & \leq\sum_{s,s'=1}^{\underline{m}}\Pr\left(\left|I_{ss'\left(k\right)}\right|>\frac{\epsilon}{\underline{m}p^{2}}\right)\\
 & \lesssim\frac{\underline{m}^{q/2+2}\left(\log N\right)^{2q}}{\left(NT\right)^{q/2-1}},
\end{align*}
where we use a similar argument as in the proof of (i) in Lemma \ref{LE:bound}
and we set $\upsilon_{NT}=\underline{m}p^{2}$, $C_{NT}=\upsilon_{NT}^{-1}NT\left(\log N\right)^{-4}$,
in addition, we use the assumptions that $T\propto N^{C}$ and $N_{k}\propto N$.

The desired result follows, using the same logic as in the proof of
(i) in Lemma \ref{LE:bound}. 
\end{proof}
\begin{proof}[Proof of Lemma \ref{LE:Ak1}]
To analyze $A_{k1}$, we introduce some new notation: 
\[
\ddot{\underline{Z}}_{\left(k\right)}\equiv\left(\underline{\ddot{z}}_{11},...,\underline{\ddot{z}}_{1T},....,\underline{\ddot{z}}_{N_{k}1},...,\underline{\ddot{z}}_{N_{k}T}\right)',
\]
a $N_{k}T\times\left(\underline{m}-1+\underline{m}p\right)$ matrix,
where we abuse the notation by letting $\underline{\ddot{z}}_{i1}$
in $\ddot{\underline{Z}}_{\left(k\right)}$ denote all observations
in $G_{k|K}$. Similarly 
\[
\ddot{\xi}_{\left(k\right)}\equiv\left(\ddot{\xi}_{11},...,\ddot{\xi}_{1T},...,\ddot{\xi}_{N_{k}1},...,\ddot{\xi}_{N_{k}T}\right)',
\]
a $N_{k}T\times1$ vector.

By Lemmas \ref{LE:z_k'z_k}\textbf{ }and \ref{LE:splines-sample-2}
on\textbf{ $\frac{1}{N_{k}T}\ddot{Z}_{\left(k\right)}'\ddot{Z}_{\left(k\right)},$}
with\textbf{ }very probability, 
\begin{equation}
\mu_{\min}\left(\frac{1}{N_{k}T}\ddot{Z}_{\left(k\right)}'\ddot{Z}_{\left(k\right)}\right)=\mu_{\min}\left(\frac{1}{N_{k}T}\sum_{i\in G_{k|K^{*}}}\sum_{t=1}^{T}\underline{\ddot{z}}_{it}\underline{\ddot{z}}_{it}'\right)\geq\frac{\underline{C}_{xx}}{2}-\epsilon\geq\frac{\underline{C}_{xx}}{3},\label{eq:zk'zk}
\end{equation}
by setting a small enough $\epsilon$, and similarly 
\[
\mu_{\max}\left(\frac{1}{N_{k}T}\ddot{Z}_{\left(k\right)}'\ddot{Z}_{\left(k\right)}\right)\leq3\bar{C}_{xx}.
\]
Thus, with very high probability, 
\begin{align}
\left\Vert A_{k1}\right\Vert  & =\left[\frac{1}{N_{k}T}\ddot{\xi}_{\left(k\right)}'\ddot{Z}_{\left(k\right)}\left(\frac{1}{N_{k}T}\ddot{Z}_{\left(k\right)}'\ddot{Z}_{\left(k\right)}\right)^{-1}\left(\frac{1}{N_{k}T}\ddot{Z}_{\left(k\right)}'\ddot{Z}_{\left(k\right)}\right)^{-1}\frac{1}{N_{k}T}\ddot{Z}_{\left(k\right)}'\ddot{\xi}_{\left(k\right)}\right]^{1/2}\nonumber \\
 & \leq3\underline{C}_{xx}^{-1}\left[\frac{1}{N_{k}T}\ddot{\xi}_{\left(k\right)}'\left(\frac{1}{N_{k}T}\ddot{Z}_{\left(k\right)}\ddot{Z}_{\left(k\right)}'\right)\ddot{\xi}_{\left(k\right)}\right]^{1/2}\nonumber \\
 & \leq3^{3/2}\underline{C}_{xx}^{-1}\bar{C}_{xx}^{1/2}\left[\frac{1}{N_{k}T}\ddot{\xi}_{\left(k\right)}'\ddot{\xi}_{\left(k\right)}\right]^{1/2}\nonumber \\
 & =3^{3/2}\underline{C}_{xx}^{-1}\bar{C}_{xx}^{1/2}\left[\frac{1}{N_{k}T}\sum_{i\in G_{k|K}}\sum_{t=1}^{T}\ddot{\xi}_{it}^{2}\right]^{1/2}.\label{eq:Ak1bound}
\end{align}

By the rate in (\ref{eq:bias_rate_group}), 
\begin{align}
\frac{1}{N_{k}T}\sum_{i\in G_{k|K}}\sum_{t=1}^{T}\ddot{\xi}_{it}^{2} & \leq\left[\frac{1}{N_{k}T}\sum_{i\in G_{k|K}}\sum_{t=1}^{T}\left(1+\sum_{l=1}^{p}\left|x_{itl}\right|\right)^{2}\right]\cdot O\left(m^{-2\kappa}\right)\nonumber \\
 & =O_{P}\left(m^{-2\kappa}\right),\label{eq:xiit2}
\end{align}
where the second line holds by the moment condition in Assumption
\ref{A:moment}. Substitute (\ref{eq:xiit2}) back to (\ref{eq:Ak1bound}),
and we obtain the desired result. 
\end{proof}
\begin{proof}[Proof of Lemma \ref{LE:Ak2}]
Recall that 
\[
\mathbb{M_{B}}\left(s\right)=\left(\begin{array}{cccc}
\mathbb{B}_{-0}^{\underline{m}}\left(s\right)' & 0 & \cdots & 0\\
0 & \mathbb{B}^{\underline{m}}\left(s\right)' & \cdots & 0\\
\vdots & \vdots & \ddots & \vdots\\
0 & 0 & \cdots & \mathbb{B}^{\underline{m}}\left(s\right)'
\end{array}\right)_{\left(p+1\right)\times\left(\underline{m}-1+\underline{m}p\right)},
\]
and 
\[
Q_{(k),zz}=\frac{1}{N_{k}T}\sum_{i\in G_{k|K}}\sum_{t=1}^{T}\underline{\ddot{z}}_{it}\underline{\ddot{z}}_{it}'.
\]
For any $\left(p+1\right)\times1$ vector $a$, 
\begin{align}
\sqrt{\frac{N_{k}T}{\underline{m}}}a'\mathbb{M_{B}}\left(s\right)A_{k2} & =a'\mathbb{M_{B}}\left(s\right)\left(\frac{1}{N_{k}T}\sum_{i\in G_{k|K}}\sum_{t=1}^{T}\underline{\ddot{z}}_{it}\underline{\ddot{z}}_{it}'\right)^{-1}\left(\frac{1}{\sqrt{N_{k}T\underline{m}}}\sum_{i\in G_{k|K}}\sum_{t=1}^{T}\underline{\ddot{z}}_{it}\ddot{v}_{it}\right)\nonumber \\
 & =a'\mathbb{M_{B}}\left(s\right)Q_{(k),zz}^{-1}\left(\frac{1}{\sqrt{N_{k}T\underline{m}}}\sum_{i\in G_{k|K}}\sum_{t=1}^{T}\underline{\ddot{z}}_{it}v_{it}\right)\nonumber \\
 & =\frac{1}{\sqrt{N_{k}}}\sum_{i\in G_{k|K}}\left\{ \frac{1}{\sqrt{T\underline{m}}}a'\mathbb{M_{B}}\left(s\right)Q_{(k),zz}^{-1}\sum_{t=1}^{T}\underline{\ddot{z}}_{it}v_{it}\right\} ,\label{eq:Ak2-1}
\end{align}
where the second line uses $\sum_{t=1}^{T}\underline{\ddot{z}}_{it}\ddot{v}_{it}=\sum_{t=1}^{T}\underline{\ddot{z}}_{it}v_{it}$.
By the i.i.d. assumption across $i$ and $t$ on $v_{it}$ and its
independence with $x,$ the conditional variance of this term can
be calculated as 
\begin{equation}
\text{Var}\left(\left.\sqrt{\frac{N_{k}T}{\underline{m}}}a'\mathbb{M_{B}}\left(s\right)A_{k2}\right|x_{1},...,x_{N}\right)=a'\left[\frac{\sigma_{v\left(k\right)}^{2}}{\underline{m}}\mathbb{M_{B}}\left(s\right)Q_{(k),zz}^{-1}\mathbb{M_{B}}\left(s\right)'\right]a.\label{eq:var(Ak2)}
\end{equation}
The above is finite and proportional to $\left\Vert a\right\Vert ^{2}.$
To see it, (\ref{eq:zk'zk}) implies that with very high probability
\begin{align*}
\text{Var}\left(\left.a'\mathbb{M_{B}}\left(s\right)A_{k2}\right|x_{1},...,x_{N}\right) & \geq3\underline{C}_{xx}^{-1}\frac{\sigma_{v\left(k\right)}^{2}}{\underline{m}}a'\mathbb{M_{B}}\left(s\right)\mathbb{M_{B}}\left(s\right)'a\\
 & =3\underline{C}_{xx}^{-1}\sigma_{v\left(k\right)}^{2}\left[\frac{1}{\underline{m}}\mathbb{B}_{-0}^{\underline{m}}\left(s\right)'\mathbb{B}_{-0}^{\underline{m}}\left(s\right)a_{0}^{2}+\frac{1}{\underline{m}}\mathbb{B}^{\underline{m}}\left(s\right)'\mathbb{B}^{\underline{m}}\left(s\right)\sum_{l=1}^{p}a_{l}^{2}\right]\\
 & \propto\sum_{l=0}^{p}a_{l}^{2}=\left\Vert a\right\Vert ^{2},
\end{align*}
where we abuse the notation a bit by letting $a=\left(a_{0},a_{1},...,a_{p}\right)'$.
We can similarly verify the Lindeberg condition for the last term
in (\ref{eq:Ak2-1}) as in Lemma A.8 in \citet{HuangEtal2004} or
Lemma A.8 in \citet{SuEtal2019}. We omit this verification process
due to the similarity. We then apply the Lindeberg Central Limit Theorem,
and obtain 
\[
\left.\sqrt{\frac{N_{k}T}{\underline{m}}}a'\mathbb{M_{B}}\left(s\right)A_{k2}\right/\left\{ a'\left[\frac{\sigma_{v\left(k\right)}^{2}}{\underline{m}}\mathbb{M_{B}}\left(s\right)Q_{(k),zz}^{-1}\mathbb{M_{B}}\left(s\right)'\right]a\right\} ^{1/2}\stackrel{d}{\rightarrow}N\left(0,1\right).
\]
\end{proof}
\begin{proof}[Proof of Lemma \ref{LE:IC1}]
We only show the result that 
\[
\Pr\left(\text{IC}\left(K^{*},\lambda_{NT}\right)<\text{IC}\left(K^{*}-1,\lambda_{NT}\right)\right)\rightarrow1.
\]
Other cases are similar. Define an event 
\[
\mathcal{M}_{1}=\left\{ \text{Two groups merged and other groups corrected classified}\right\} .
\]
By the nature of the HAC algorithm, and the proof in Theorem \ref{TH:classify},
we can similarly show that 
\[
\Pr\left(\mathcal{M}_{1}\right)\rightarrow1.
\]
Recall that 
\[
\mathcal{M}\equiv\left\{ \left(\hat{G}_{1|K^{*}},\hat{G}_{2|K^{*}},\ldots,\hat{G}_{K^{*}|K^{*}}\right)=\left(G_{1|K^{*}},G_{2|K^{*}},\ldots,G_{K^{*}|K^{*}}\right)\right\} .
\]
We first show the result conditional on $\mathcal{M}$ and $\mathcal{M}_{1},$
then we show the result unconditionally. Without loss of generality,
assume that observations in group $K^{*}-1$ and $K^{*}$ are merged.
Thus, 
\begin{align}
 & \frac{\text{IC}\left(K^{*},\lambda_{NT}\right)-\text{IC}\left(K^{*}-1,\lambda_{NT}\right)}{\left(N_{K^{*}-1}+N_{K^{*}}\right)T}\nonumber \\
 & =\frac{1}{\left(N_{K^{*}-1}+N_{K^{*}}\right)T}\left\{ \sum_{k=K^{*}-1}^{K^{*}}N_{k}T\log\left(\hat{\sigma}_{v\left(k|K^{*}\right)}\right)+\lambda_{NT}-\left(N_{K^{*}-1}+N_{K^{*}}\right)T\log\left(\hat{\sigma}_{v\left(K^{*}-1|K^{*}-1\right)}\right)\right\} \nonumber \\
 & =\frac{N_{K^{*}-1}}{N_{K^{*}-1}+N_{K^{*}}}\log\left(\hat{\sigma}_{v\left(K^{*}-1|K^{*}\right)}\right)+\frac{N_{K^{*}}}{N_{K^{*}-1}+N_{K^{*}}}\log\left(\hat{\sigma}_{v\left(K^{*}|K^{*}\right)}\right)-\log\left(\hat{\sigma}_{v\left(K^{*}-1|K^{*}-1\right)}\right)+o\left(1\right),\label{eq:IC_diff1}
\end{align}
where the last line uses $\lambda_{NT}=o\left(NT\right)$ and $N_{k}\propto N.$

We claim the result below and we defer its proof to the end: 
\begin{equation}
\hat{\sigma}_{v\left(K^{*}-1|K^{*}-1\right)}^{2}=\frac{N_{K^{*}-1}}{N_{K^{*}-1}+N_{K^{*}}}\left(\hat{\sigma}_{v\left(K^{*}-1|K^{*}\right)}^{2}+\Delta_{1}^{2}\right)+\frac{N_{K^{*}}}{N_{K^{*}-1}+N_{K^{*}}}\left(\hat{\sigma}_{v\left(K^{*}|K^{*}\right)}^{2}+\Delta_{2}^{2}\right)+o_{P}\left(1\right),\label{eq:sigma_vk-1}
\end{equation}
where 
\begin{align*}
\Delta_{1}^{2} & \equiv\frac{1}{N_{K^{*}-1}T}\sum_{i\in G_{K^{*}-1|K^{*}}}\sum_{t=1}^{T}\left[\underline{\ddot{z}}_{it}'\left(\pi_{\left(K^{*}-1|K^{*}-1\right)}^{0*}-\pi_{\left(K^{*}-1|K^{*}\right)}^{0*}\right)\right]^{2},\text{ and}\\
\Delta_{2}^{2} & \equiv\frac{1}{N_{K^{*}}T}\sum_{i\in G_{K^{*}|K^{*}}}\sum_{t=1}^{T}\left[\underline{\ddot{z}}_{it}'\left(\pi_{\left(K^{*}-1|K^{*}-1\right)}^{0*}-\pi_{\left(K^{*}|K^{*}\right)}^{0*}\right)\right]^{2},
\end{align*}
and $\pi$$^{0*}$$_{\left(K^{*}-1|K^{*}-1\right)}$ denotes the estimand
of $\pi$ for $G_{K^{*}-1|K^{*}}\cup G_{K^{*}|K^{*}}$.

Using (\ref{eq:sigma_vk-1}) and the Jensen's inequality, 
\begin{align}
\log\hat{\sigma}_{v\left(K^{*}-1|K^{*}-1\right)}^{2} & \geq\frac{N_{K^{*}-1}}{N_{K^{*}-1}+N_{K^{*}}}\log\left(\hat{\sigma}_{v\left(K^{*}-1|K^{*}\right)}^{2}+\Delta_{1}^{2}\right)\nonumber \\
 & +\frac{N_{K^{*}}}{N_{K^{*}-1}+N_{K^{*}}}\log\left(\hat{\sigma}_{v\left(K^{*}|K^{*}\right)}^{2}+\Delta_{2}^{2}\right)+o_{P}\left(1\right).\label{eq:log_sigma}
\end{align}
Since 
\[
\left\{ \alpha_{\left(K^{*}-1\right)}^{*},\beta_{\left(K^{*}-1\right)}^{\ast},\sigma_{v\left(K^{*}-1\right)}^{\ast2}\right\} \neq\left\{ \alpha_{\left(K^{*}\right)}^{*},\beta_{\left(K^{*}\right)}^{\ast},\sigma_{v\left(K^{*}\right)}^{\ast2}\right\} ,
\]
we either have $\sigma_{v\left(K^{*}-1\right)}^{\ast2}\neq\sigma_{v\left(K^{*}\right)}^{\ast2}$
or $\left\{ \alpha_{\left(K^{*}-1\right)}^{*},\beta_{\left(K^{*}-1\right)}^{\ast}\right\} \neq\left\{ \alpha_{\left(K^{*}\right)}^{*},\beta_{\left(K^{*}\right)}^{\ast}\right\} $,
so that either $\hat{\sigma}_{v\left(K^{*}-1|K^{*}\right)}^{2}\neq\hat{\sigma}_{v\left(K^{*}|K^{*}\right)}^{2}$
or $\max\left\{ \Delta_{1}^{2},\Delta_{2}^{2}\right\} >0$ holds with
very high probability. Together with this, we can change ``$\geq$''
in (\ref{eq:log_sigma}) to ``$>$'', so that 
\begin{equation}
\log\hat{\sigma}_{v\left(K^{*}-1|K^{*}-1\right)}^{2}>\frac{N_{K^{*}-1}}{N_{K^{*}-1}+N_{K^{*}}}\log\left(\hat{\sigma}_{v\left(K^{*}-1|K^{*}\right)}^{2}\right)+\frac{N_{K^{*}}}{N_{K^{*}-1}+N_{K^{*}}}\log\left(\hat{\sigma}_{v\left(K^{*}|K^{*}\right)}^{2}\right)+o_{P}\left(1\right)\label{eq:sigma_log}
\end{equation}
Using (\ref{eq:sigma_log}), (\ref{eq:IC_diff1}) implies that 
\[
\text{IC}\left(K^{*},\lambda_{NT}\right)-\text{IC}\left(K^{*}-1,\lambda_{NT}\right)<0
\]
holds with very high probability after some large $N$ and $T$.

We have shown the desired result conditional on $\mathcal{M}$ and
$\mathcal{M}_{1}$. Since 
\[
\Pr\left(\mathcal{M}\cap\mathcal{M}_{1}\right)\geq1-\Pr\left(\mathcal{M}^{c}\right)-\Pr\left(\mathcal{M}_{1}^{c}\right)\rightarrow1,
\]
the desired result then holds unconditionally.

We finish the proof by showing the claim in (\ref{eq:sigma_vk-1}).
Then as in the proof of (ii) in Theorem \ref{TH:post-estimation}
(the decomposition of $\hat{\sigma}_{v\left(k|K^{*}\right)}^{2}$),
we have 
\begin{align*}
 & \sum_{i\in G_{K^{*}-1|K^{*}}\cup G_{K^{*}|K^{*}}}\sum_{t=1}^{T}\left(\ddot{y}_{it}-\underline{\ddot{z}}_{it}'\hat{\pi}_{\left(K^{*}-1|K^{*}-1\right)}\right)^{2}\\
= & \sum_{i\in G_{K^{*}-1|K^{*}}}\sum_{t=1}^{T}\left[\underline{\ddot{z}}_{it}'\left(\hat{\pi}_{\left(K^{*}-1|K^{*}-1\right)}-\pi_{\left(K^{*}-1|K^{*}-1\right)}^{0*}\right)+\underline{\ddot{z}}_{it}'\left(\pi_{\left(K^{*}-1|K^{*}-1\right)}^{0*}-\pi_{\left(K^{*}-1|K^{*}\right)}^{0*}\right)+\ddot{\xi}_{it}+\ddot{v}_{it}\right]^{2}\\
 & +\sum_{i\in\cup G_{K^{*}|K^{*}}}\sum_{t=1}^{T}\left[\underline{\ddot{z}}_{it}'\left(\hat{\pi}_{\left(K^{*}-1|K^{*}-1\right)}-\pi_{\left(K^{*}-1|K^{*}-1\right)}^{0*}\right)+\underline{\ddot{z}}_{it}'\left(\pi_{\left(K^{*}-1|K^{*}-1\right)}^{0*}-\pi_{\left(K^{*}|K^{*}\right)}^{0*}\right)+\ddot{\xi}_{it}+\ddot{v}_{it}\right]^{2}.
\end{align*}
Using the decomposition above and by the same analysis for $\hat{\sigma}_{v\left(k|K^{*}\right)}^{2}$
in the proof of (ii) in Theorem \ref{TH:post-estimation}, we can
extract the leading term of the following: 
\begin{align*}
\hat{\sigma}_{v\left(K^{*}-1|K^{*}-1\right)}^{2}= & \frac{1}{\left(N_{K^{*}-1}+N_{K^{*}}\right)\left(T-1\right)}\sum_{i\in G_{K^{*}-1|K^{*}}\cup G_{K^{*}|K^{*}}}\sum_{t=1}^{T}\left(\ddot{y}_{it}-\underline{\ddot{z}}_{it}'\hat{\pi}_{\left(K^{*}-1|K^{*}-1\right)}\right)^{2}\\
= & \frac{1}{\left(N_{K^{*}-1}+N_{K^{*}}\right)\left(T-1\right)}\sum_{i\in G_{K^{*}-1|K^{*}}}\sum_{i\in G_{K^{*}|K^{*}}}\sum_{t=1}^{T}\ddot{v}_{it}^{2}+\Delta_{1}^{2}+\Delta_{2}^{2}+o_{P}\left(1\right)\\
= & \frac{N_{K^{*}-1}}{N_{K^{*}-1}+N_{K^{*}}}\hat{\sigma}_{v\left(K^{*}-1|K^{*}\right)}^{2}+\Delta_{1}^{2}+\frac{N_{K^{*}}}{N_{K^{*}-1}+N_{K^{*}}}\hat{\sigma}_{v\left(K^{*}|K^{*}\right)}^{2}+\Delta_{2}^{2}+o_{P}\left(1\right).
\end{align*}
where the second and third lines repeatedly use the arguments for
the analysis of $\hat{\sigma}_{v\left(k|K^{*}\right)}^{2}$. The above
is the desired result. 
\end{proof}
\begin{proof}[Proof of Lemma \ref{LE:IC2}]
We only show the result for $K=K^{*}+1.$ Other cases are similar.
Define an event 
\[
\mathcal{M}_{2}=\left\{ \text{One group separated into two groups and other groups corrected classified}\right\} .
\]
By the nature of the HAC algorithm, and the proof in Theorem \ref{TH:classify},
we can similarly show that 
\[
\Pr\left(\mathcal{M}_{2}\right)\rightarrow1.
\]

We show the result conditional on $\mathcal{M}$ and $\mathcal{M}_{2}.$
Without loss of generality, assume that observations in group $K^{*}$
are separated into two groups, and the numbers of observations in
groups $K^{*}$ and $K^{*}+1$ are denoted as $N_{K^{*}1}$ and $N_{K^{*}2}$,
respectively. Clearly $N_{K^{*}}=N_{K^{*}1}+N_{K^{*}2}$. Since these
two groups come from the same underlying group, the parameters of
interest are the same. Using the rates in Theorem \ref{TH:post-estimation},
\[
\hat{\sigma}_{v\left(K^{*}|K^{*}+1\right)}^{2}-\sigma_{v\left(K^{*}\right)}^{*2}=O_{P}\left(\frac{1}{\sqrt{N_{K^{*}1}T}}\right)\text{ and }\hat{\sigma}_{v\left(K^{*}+1|K^{*}+1\right)}^{2}-\sigma_{v\left(K^{*}\right)}^{*2}=O_{P}\left(\frac{1}{\sqrt{N_{K^{*}2}T}}\right).
\]
Since $\hat{\sigma}_{v\left(K^{*}|K^{*}\right)}^{2}-\sigma_{v\left(K^{*}\right)}^{*2}=O_{P}\left(\frac{1}{\sqrt{N_{K^{*}}T}}\right),$
with the above, we have 
\[
\hat{\sigma}_{v\left(K^{*}|K^{*}+1\right)}^{2}-\hat{\sigma}_{v\left(K^{*}|K^{*}\right)}^{2}=O_{P}\left(\frac{1}{\sqrt{N_{K^{*}1}T}}\right)\text{ and }\hat{\sigma}_{v\left(K^{*}+1|K^{*}+1\right)}^{2}-\hat{\sigma}_{v\left(K^{*}|K^{*}\right)}^{2}=O_{P}\left(\frac{1}{\sqrt{N_{K^{*}2}T}}\right)
\]

Then 
\begin{align*}
 & \frac{\text{IC}\left(K^{*}+1,\lambda_{NT}\right)-\text{IC}\left(K^{*},\lambda_{NT}\right)}{N_{K^{*}}T}\\
= & \frac{1}{N_{K^{*}}T}\left\{ N_{K^{*}1}T\log\left(\hat{\sigma}_{v\left(K^{*}|K^{*}+1\right)}\right)+N_{K^{*}2}T\log\left(\hat{\sigma}_{v\left(K^{*}+1|K^{*}+1\right)}\right)+\lambda_{NT}\right.\left.-N_{k}T\log\left(\hat{\sigma}_{v\left(K^{*}|K^{*}\right)}\right)\right\} \\
= & \frac{N_{K^{*}1}}{N_{K^{*}}}\left[\log\left(\hat{\sigma}_{v\left(K^{*}|K^{*}+1\right)}^{2}\right)-\log\left(\hat{\sigma}_{v\left(K^{*}|K^{*}\right)}^{2}\right)\right]+\frac{N_{K^{*}2}}{N_{K^{*}}}\left[\log\left(\hat{\sigma}_{v\left(K^{*}+1|K^{*}+1\right)}^{2}\right)-\log\left(\hat{\sigma}_{v\left(K^{*}|K^{*}\right)}^{2}\right)\right]+\frac{\lambda_{NT}}{N_{K^{*}}T}\\
= & \frac{N_{K^{*}1}}{N_{K^{*}}}O_{P}\left(\frac{1}{\sqrt{N_{K^{*}1}T}}\right)+\frac{N_{K^{*}2}}{N_{K^{*}}}O_{P}\left(\frac{1}{\sqrt{N_{K^{*}2}T}}\right)+\frac{\lambda_{NT}}{N_{K^{*}}T}\\
= & \left(\sqrt{\frac{N_{K^{*}1}}{N_{K^{*}}}}+\sqrt{\frac{N_{K^{*}2}}{N_{K^{*}}}}\right)\cdot O_{P}\left(\frac{1}{\sqrt{N_{K^{*}}T}}\right)+\frac{\lambda_{NT}}{N_{K^{*}}T}\\
= & O_{P}\left(\frac{1}{\sqrt{N_{K^{*}}T}}\right)+\frac{\lambda_{NT}}{N_{K^{*}}T}>0,
\end{align*}
with very high probability, where the last line holds due to $\sqrt{\frac{N_{K^{*}1}}{N_{K^{*}}}}+\sqrt{\frac{N_{K^{*}2}}{N_{K^{*}}}}\leq2$,
$\lambda_{NT}\gg$$\sqrt{NT},$ and $N_{K^{*}}\propto N.$

We have shown the desired result conditional on $\mathcal{M}$ and
$\mathcal{M}_{2}$. Since 
\[
\Pr\left(\mathcal{M}\cap\mathcal{M}_{2}\right)\geq1-\Pr\left(\mathcal{M}^{c}\right)-\Pr\left(\mathcal{M}_{2}^{c}\right)\rightarrow1,
\]
the desired result then holds unconditionally. 
\end{proof}
\begin{proof}[Proof of Lemma \ref{LE:IC3}]
Without loss of generality, we ignore the group structure of other
parameters for the following discussion. Moreover, we assume that
we know $\left(\alpha(\cdot),\beta(\cdot),\sigma_{v}^{2}\right)$.
It is innocuous because $\hat{\vartheta}$ converges to the true value
faster than $\sqrt{N}$ which is the convergence rate of $\left(\hat{\alpha}^{0},\hat{\sigma}_{u}^{2}\right)$.

As in Lemmas \ref{LE:IC1} and \ref{LE:IC2}, we only show that 
\[
\Pr\left(\widetilde{\mathrm{IC}}\left(\mathcal{K}^{*},\tilde{\lambda}_{NT}\right)<\widetilde{\mathrm{IC}}\left(\mathcal{K}^{*}-1,\tilde{\lambda}_{NT}\right)\right)\rightarrow1,
\]
and 
\[
\Pr\left(\widetilde{\mathrm{IC}}\left(\mathcal{K}^{*},\tilde{\lambda}_{NT}\right)<\widetilde{\mathrm{IC}}\left(\mathcal{K}^{*}+1,\tilde{\lambda}_{NT}\right)\right)\rightarrow1.
\]
As such, we divide the proof into two parts with Part 1 showing the
first result and Part 2 for the other.

\noindent\textbf{Part 1.} We show in Appendix \ref{app:IOmega} that
the information matrix is a well-behaved positive definite matrix
with diagonals being finite constants, even though $T\rightarrow\infty.$
Recall that 
\[
\varrho^{0}=\left(\alpha_{(1)}^{0},\sigma_{u(1)}^{2},...,\alpha_{(\mathcal{K^{*}})}^{0},\sigma_{u(\mathcal{K}^{*})}^{2},\tau_{1}^{0},...,\tau_{\mathcal{\mathcal{K}^{*}}-1}^{0}\right).
\]
The true density function for $y$ is 
\[
\tilde{f}\left(y\left|x;\varrho^{0},\vartheta\right.\right),
\]
where we abuse the notation by letting $\vartheta\equiv\left(\sigma_{v}^{2},\alpha(\cdot),\beta(\cdot)\right)$.
By definition, $\varrho^{0}$ uniquely maximizes 
\[
\log\left[\tilde{f}\left(y\left|x;\varrho,\vartheta\right.\right)\right].
\]
For a $\varrho$ that is in a small neighborhood of $\varrho^{0},$
\[
\log\left[\tilde{f}\left(y\left|x;\varrho,\vartheta\right.\right)\right]-\log\left[\tilde{f}\left(y\left|x;\varrho^{0},\vartheta\right.\right)\right]=-\left(\varrho-\varrho^{0}\right)'\mathbb{I}\left(\varrho-\varrho^{0}\right)+o\left(\left\Vert \varrho-\varrho^{0}\right\Vert ^{2}\right).
\]
As mentioned at the beginning, $\mathbb{I}$ is positive definite.
Further, we restrict our attention to a compact support of $\varrho,$
therefore, there exists a $C$ such that 
\begin{equation}
\log\left[\tilde{f}\left(y\left|x;\varrho,\vartheta\right.\right)\right]-\log\left[\tilde{f}\left(y\left|x;\varrho^{0},\vartheta\right.\right)\right]\leq-C<0,\label{eq:log_diff_1-1}
\end{equation}
for all $\varrho$ outside of a small neighborhood of $\varrho^{0}.$

When we restrict the mixture distribution to come from $\mathcal{K}^{*}-1$
distributions, clearly, the estimate converges to a parameter that
is outside the small neighborhood of $\varrho^{0}.$ By the consistency
and inequality in (\ref{eq:log_diff_1-1}), 
\[
\tilde{\text{IC}}\left(\mathcal{K}^{*},\tilde{\lambda}_{NT}\right)-\tilde{\text{IC}}\left(\mathcal{K}^{*}-1,\tilde{\lambda}_{NT}\right)\leq-C\cdot N+o_{P}\left(N\right)+\tilde{\lambda}_{NT}.
\]
Further, $\tilde{\lambda}_{NT}\ll N,$ the above implies the first
desired result.

\noindent\textbf{Part 2.} When we force $\mathcal{K}=\mathcal{K}^{*}+1,$
denote 
\[
\varsigma^{0}=\left(\alpha_{(1)}^{0},\sigma_{u(1)}^{2},...,\alpha_{(\mathcal{K^{*}}+1)}^{0},\sigma_{u(\mathcal{K}^{*}+1)}^{2},\tau_{1}^{0},...,\tau_{\mathcal{\mathcal{K}^{*}}}^{0}\right),
\]
as the population parameter. It is not unique but we use the one that
the estimated parameters converge to. $\varsigma^{0}$ is essentially
the same as $\varrho^{0}$ in terms of mixture distribution; mathematically,
\begin{equation}
\tilde{f}\left(y_{i}\left\vert x_{i};\varsigma^{0},\vartheta\right.\right)=\tilde{f}\left(y_{i}\left\vert x_{i};\varrho^{0},\vartheta\right.\right),\label{eq:same_density}
\end{equation}
where as before $\vartheta=\left(\sigma_{v}^{2},\alpha(\cdot),\beta(\cdot)\right)$.
To see that, on the one hand 
\[
\textrm{E}\left[\log\tilde{f}\left(y_{i}\left\vert x_{i};\varsigma^{0},\vartheta\right.\right)\right]\geq\textrm{E}\left[\log\tilde{f}\left(y_{i}\left\vert x_{i};\varrho^{0},\vartheta\right.\right)\right],
\]
because $\varsigma^{0}$ has more degrees of freedom than that of
$\varrho^{0}$ . On the other hand, 
\[
\textrm{E}\left[\log\tilde{f}\left(y_{i}\left\vert x_{i};\varsigma^{0},\vartheta\right.\right)\right]\leq\textrm{E}\left[\log\tilde{f}\left(y_{i}\left\vert x_{i};\varrho^{0},\vartheta\right.\right)\right],
\]
because $\tilde{f}\left(y_{i}\left\vert x_{i};\varrho^{0},\vartheta\right.\right)$
is the true underlying distribution. Therefore, (\ref{eq:same_density})
must hold.

Let $\hat{\varsigma}$ be the estimates when $\mathcal{K}=\mathcal{K}^{*}+1$.
Then, using (\ref{eq:same_density}), 
\begin{align*}
 & \frac{1}{N}\sum_{i=1}^{N}\log\tilde{f}\left(y_{i}\left\vert x_{i};\hat{\varsigma},\vartheta\right.\right)-\frac{1}{N}\sum_{i=1}^{N}\log\tilde{f}\left(y_{i}\left\vert x_{i};\hat{\varrho},\vartheta\right.\right)\\
= & \frac{1}{N}\sum_{i=1}^{N}\left[\log\tilde{f}\left(y_{i}\left\vert x_{i};\hat{\varsigma},\vartheta\right.\right)-\log\tilde{f}\left(y_{i}\left\vert x_{i};\varsigma^{0},\vartheta\right.\right)\right]\\
 & -\frac{1}{N}\sum_{i=1}^{N}\left[\log\tilde{f}\left(y_{i}\left\vert x_{i};\hat{\varrho},\vartheta\right.\right)-\log\tilde{f}\left(y_{i}\left\vert x_{i};\varrho^{0},\vartheta\right.\right)\right]\\
= & \left(\frac{1}{N}\sum_{i=1}^{N}\left.\frac{\partial}{\partial\varsigma}\log\tilde{f}_{i}\right|_{\varsigma=\varsigma^{0}}\right)'\left(\hat{\varsigma}-\varsigma^{0}\right)+O_{P}\left(\left\Vert \hat{\varsigma}-\varsigma^{0}\right\Vert ^{2}\right)\\
 & +\left(\frac{1}{N}\sum_{i=1}^{N}\left.\frac{\partial}{\partial\varrho}\log\tilde{f_{i}}\right|_{\varrho=\varrho^{0}}\right)'\left(\hat{\varrho}-\varrho^{0}\right)+O_{P}\left(\left\Vert \hat{\varrho}-\varrho^{0}\right\Vert ^{2}\right)\\
= & O_{P}\left(N^{-1}\right),
\end{align*}
where the last line holds by the first order condition and the independence
across $N$ so that 
\[
\frac{1}{N}\sum_{i=1}^{N}\left.\frac{\partial}{\partial\varsigma}\log f_{i}\right|_{\varsigma=\varsigma^{0}}=O_{P}\left(N^{-1/2}\right)\text{ and }\frac{1}{N}\sum_{i=1}^{N}\left.\frac{\partial}{\partial\varrho}\log\tilde{f_{i}}\right|_{\varrho=\varrho^{0}}=O_{P}\left(N^{-1/2}\right),
\]
$\left\Vert \hat{\varsigma}-\varsigma^{0}\right\Vert =O_{P}\left(N^{-1/2}\right)$,
and $\left\Vert \hat{\varrho}-\varrho^{0}\right\Vert =O_{P}\left(N^{-1/2}\right)$.
Using the above, 
\[
\widetilde{\mathrm{IC}}\left(\mathcal{K}^{*},\tilde{\lambda}_{NT}\right)-\widetilde{\mathrm{IC}}\left(\mathcal{K}^{*}+1,\tilde{\lambda}_{NT}\right)=-\tilde{\lambda}_{NT}+O_{P}\left(1\right).
\]
Then 
\[
\Pr\left(\widetilde{\mathrm{IC}}\left(\mathcal{K}^{*},\tilde{\lambda}_{NT}\right)<\widetilde{\mathrm{IC}}\left(\mathcal{K}^{*}+1,\tilde{\lambda}_{NT}\right)\right)\rightarrow1,
\]
due to $\tilde{\lambda}_{NT}\rightarrow\infty.$ 
\end{proof}

\end{document}